\documentclass[graybox,envcountchap,sectrefs]{svmono}

\usepackage{makeidx}         
\usepackage[utf8]{inputenc}
\usepackage{bm}
\usepackage{bbm}
\usepackage{mwe}
\usepackage{amsmath}
\usepackage{amssymb}
\usepackage{booktabs}
\usepackage{amsfonts}
\usepackage{multirow}
\usepackage{multicol}        
\usepackage[ruled]{algorithm2e}
\usepackage{algorithmic}
\usepackage[subrefformat=parens,labelformat=parens]{subfig}
\usepackage{amsthm}
\usepackage{url}
\usepackage{hyperref}
\usepackage{wrapfig}
\usepackage{longtable}
\usepackage[sectionbib]{natbib}
\usepackage{chapterbib}

\makeindex

\hypersetup{
    colorlinks,
    citecolor=black,
    filecolor=black,
    linkcolor=black,
    urlcolor=black
}

\newcommand\numberthis{\addtocounter{equation}{1}\tag{\theequation}}

\SetKwComment{Comment}{$\triangleright$\ }{}

\title{Foundations of Vector Retrieval}

\DeclareMathOperator*{\argmax}{arg\,max}
\DeclareMathOperator*{\argmin}{arg\,min}
\DeclareMathOperator*{\var}{Var}
\DeclareMathOperator*{\ev}{\mathbb{E}}
\DeclareMathOperator*{\probability}{\mathbb{P}}
\newcommand{\splade}{\textsc{Splade}}
\newcommand{\esplade}{\textsc{Efficient Splade}}

\begin{document}

\author{Sebastian Bruch}
\maketitle

\frontmatter
\preface
\addcontentsline{toc}{chapter}{Preface}

We are witness to a few years of remarkable developments
in Artificial Intelligence with the use of advanced machine learning algorithms,
and in particular, \emph{deep learning}. Gargantuan, complex neural networks
that can learn through self-supervision---and quickly so with the aid of specialized
hardware---transformed the research landscape so dramatically that, overnight it seems,
many fields experienced not the usual, incremental progress, but rather a leap forward.
Machine translation, natural language understanding, information retrieval,
recommender systems, and computer vision are but a
few examples of research areas that have had to grapple with the shock.
Countless other disciplines beyond computer science such as robotics, biology, and chemistry
too have benefited from deep learning.

These neural networks and their training algorithms may be complex, and
the scope of their impact broad and wide,
but nonetheless they are simply functions in a high-dimensional space.
A trained neural network takes a \emph{vector} as input, crunches and transforms
it in various ways, and produces another vector, often in some other space.
An image may thereby be turned into a vector, a song into a sequence of vectors,
and a social network as a structured collection of vectors.
It seems as though much of human knowledge,
or at least what is expressed as text, audio, image, and video,
has a vector representation in one form or another.

It should be noted that representing data as vectors is not unique to neural networks
and deep learning. In fact, long before learnt vector representations of pieces of
data---what is commonly known as ``embeddings''---came along, data
was often encoded as hand-crafted \emph{feature} vectors. Each feature quantified
into continuous or discrete values some facet of the
data that was deemed relevant to a particular task (such as classification or regression).
Vectors of that form, too, reflect
our understanding of a real-world object or concept.

If new and old knowledge can be squeezed into a collection of
learnt or hand-crafted vectors, what useful things does that enable us to do?
A metaphor that might help us think about that question is this:
An ever-evolving database full of such vectors that capture
various pieces of data can be understood as a \emph{memory} of sorts.
We can then recall information from this memory to answer questions,
learn about past and present events, reason about new problems,
generate new content, and more.

\section*{Vector Retrieval}

Mathematically, ``recalling information'' translates to finding vectors that are
most \emph{similar} to a \emph{query} vector. The query vector represents what
we wish to know more about, or recall information for. So, if we have a
particular question in mind, the query is the vector representation of that question.
If we wish to know more about an event, our query is that event expressed as a vector.
If we wish to predict the function of a protein, perhaps we may learn a thing or two
from known proteins that have a similar structure to the one in question,
making a vector representation of the structure of our new protein a query.

Similarity is then a function of two vectors, quantifying how similar two vectors are.
It may, for example, be based on the Euclidean distance
between the query vector and a database vector, where similar vectors have a smaller
distance. Or it may instead be based on the inner product between two vectors.
Or their angle. Whatever function we use to measure similarity between pieces
of data defines the structure of a database.

Finding $k$ vectors from a database that have the highest similarity to a query vector
is known as the top-$k$ retrieval problem.
When similarity is based on the Euclidean distance, 
the resulting problem is known as \emph{nearest neighbor} search.
Inner product for similarity leads to a problem known as \emph{maximum inner product search}.
Angular distance gives \emph{maximum cosine similarity} search.
These are mathematical formulations of the mechanism we called ``recalling information.''

\bigskip

The need to search for similar vectors from a large database
arises in virtually every single one of our online transactions.
Indeed, when we search the web for information about a topic, the search engine itself
performs this similarity search over millions of web documents to find what may
lexically or semantically match our query. Recommender systems find the most similar
items to your browsing history by encoding items as vectors and, effectively, searching
through a database of such items. Finding an old photo in a photo library,
as another routine example, boils down to performing a similarity search over vector
representations of images.

A neural network that is trained to perform a general task such as question-answering,
could conceivably augment its view of the world by ``recalling'' information from such a database
and finding answers to new questions.
This is particularly useful for \emph{generative} agents such as chatbots
who would otherwise be frozen in time, and whose knowledge limited to
what they were exposed to during their training. With a vector database on the side,
however, they would have access to real-time information and can deduce new observations
about content that is new to them.
This is, in fact, the cornerstone of what is known as retrieval-augmented generation,
an emerging learning paradigm.

\bigskip

Finding the most similar vectors
to a query vector is easy when the database is small or when time is not of the essence:
We can simply compare every vector
in the database with the query and sort them by similarity. When the database grows large
and the time budget is limited, as is
often the case in practice, a na\"ive, exhaustive comparison of a query with database vectors
is no longer realistic.
That is where \textbf{vector retrieval} algorithms become relevant.

For decades now, research on vector retrieval has sought to improve the efficiency
of search over large vector databases. The resulting literature is rich with
solutions ranging from heavily theoretical results to performant empirical heuristics.
Many of the proposed algorithms have undergone rigorous benchmarking
and have been challenged in competitions at major conferences.
Technology giants and startups alike have invested heavily in developing
open-source libraries and managed infrastructure that offer fast and scalable
vector retrieval.

That is not the end of that story, however. Research continues to date.
In fact, how we do vector retrieval today faces a stress-test as
databases grow orders of magnitude larger than ever before.
None of the existing methods, for example, proves easy to scale to a
database of billions of high-dimensional vectors, or a database
whose records change frequently. 

\section*{About This Monograph}
The need to conduct more research underlines the importance of making the existing
literature more readily available and the research area more inviting.
That is partially fulfilled with existing surveys that report
the state of the art at various points in time. However, these publications
are typically focused on a single class of vector retrieval algorithms,
and compare and contrast published methods by their empirical performance alone.
Importantly, no manuscript has yet summarized major algorithmic milestones
in the vast vector retrieval literature, or has been prepared to serve as a reference
for new and established researchers.

That gap is what this monograph intends to close.
With the goal of presenting the fundamentals of vector retrieval as a sub-discipline,
this manuscript delves into important data structures and algorithms that have
emerged in the literature to solve the vector retrieval problem efficiently and effectively.

\subsection*{Structure}
This monograph is divided into four parts. The first part introduces the problem
of vector retrieval and formalizes the concepts involved. The second part delves into
retrieval algorithms that help solve the vector retrieval problem efficiently and effectively.
Part three is devoted to vector compression. Finally, the fourth part presents a review
of background material in a series of appendices.

\subsubsection*{Introduction}
We start with a thorough \textbf{introduction} to the problem itself in Chapter~\ref{chapter:flavors}
where we define the various flavors of vector retrieval.
We then elaborate what is so difficult about the problem in high-dimensional spaces
in Chapter~\ref{chapter:instability}.

In fact, sometimes high-dimensional spaces are hopeless. However, in reality data often
lie on some low-dimensional space, even though their na\"ive vector representations
are in high dimensions. In those cases, it turns out, we can do much better. Exactly how
we characterize this low \textbf{``intrinsic'' dimensionality} is the topic of
Chapter~\ref{chapter:intrinsic-dimensionality}.

\subsubsection*{Retrieval Algorithms}
With that foundation in place and the question clearly formulated,
the second part of the monograph explores the different classes of existing solutions
in great depth.
We close each chapter with a summary of algorithmic insights.
There, we will also discuss what remains challenging and explore future research directions.

We start with \textbf{branch-and-bound} algorithms in Chapter~\ref{chapter:branch-and-bound}.
The high-level idea is to lay a hierarchical mesh over the space, then given a query point
navigate the hierarchy to the cell that likely contains the solution.
We will see, however, that in high dimensions, the basic forms of these methods become highly inefficient
to the point where an exhaustive search likely performs much better.

Alternatively, instead of laying a mesh over the space, we may define a fixed
number of buckets and map data points to these buckets with the property that,
if two data points are close to each other according to the distance function,
they are more likely to be mapped to the same bucket. When processing a query,
we find which bucket it maps to and search the data points in that bucket.
This is the intuition that led to the family of \textbf{Locality Sensitive Hashing} (LSH)
algorithms---a topic we will discuss in depth in Chapter~\ref{chapter:lsh}.

Yet another class of ideas adopts the view that data points are nodes in a graph.
We place an edge between two nodes if they are among each others' nearest neighbors.
When presented with a query point, we enter the graph through one of the nodes
and greedily traverse the edges by taking the edge that leads to the minimum distance
with the query. This process is repeated until we are stuck in some (local) optima.
This is the core idea in \textbf{graph} algorithms, as we will learn in Chapter~\ref{chapter:graph}.

The final major approach is the simplest of all: Organize the data points into small clusters during
pre-processing. When a query point arrives, solve the``cluster retrieval'' problem first, then
solve retrieval on the chosen clusters.
We will study this \textbf{clustering} method in detail in
Chapter~\ref{chapter:ivf}.

As we examine vector retrieval algorithms, it is inevitable that we must
ink in extra pages to discuss why similarity based on inner product is special
and why it poses extra challenges for the algorithms in each
category---many of these difficulties
will become clear in the introductory chapters.

There is, however, a special class of algorithms specifically for inner product.
\textbf{Sampling} algorithms take advantage of the linearity of inner product to
reduce the dependence of the time complexity on the number of dimensions.
We will review example algorithms in Chapter~\ref{chapter:sampling}.

\subsubsection*{Compression}

The third part of this monograph concerns the storage of vectors and their distance computation.
After all, the vector retrieval problem is not just concerned with
the time complexity of the retrieval process itself, but also aims to reduce the
size of the data structure that helps answer queries---known as the index.
Compression helps that cause.

In Chapter~\ref{chapter:quantization} we will review how vectors
can be \textbf{quantized} to reduce the size of the index
while simultaneously facilitating fast computation of the distance function
in the compressed domain! That is what makes quantization effective but challenging.

Related to the topic of compression is the concept of \textbf{sketching}.
Sketching is a technique to project a high-dimensional vector into
a low-dimensional vector, called a \emph{sketch}, such that certain properties
(e.g., the $L_2$ norm, or inner products between any two vectors) are \emph{approximately}
preserved. This probabilistic method of reducing dimensionality naturally
connects to vector retrieval.
We offer a peek into the vast sketching literature in Chapter~\ref{chapter:sketching}
and discuss its place in the vector retrieval research.
We do so with a particular focus on \emph{sparse} vectors in an inner product
space---contrasting sketching with quantization methods that are more appropriate
for \emph{dense} vectors.

\subsection*{Objective}

It is important to stress, however, that the purpose of this monograph is \emph{not}
to provide a comprehensive survey or comparative analysis of every published work
that has appeared in the vector retrieval literature.
There is simply too many empirical works with volumes of heuristics and engineering
solutions to cover. Instead, we will give an in-depth, didactic treatment of foundational
ideas that have caused a seismic shift in how we approach the problem,
and the theory that underpins them.

By consolidating these ideas, this monograph hopes to make this fascinating field
more inviting---especially to the uninitiated---and enticing as a research topic
to new and established researchers. We hope the reader will find that this
monograph delivers on these objectives.

\subsection*{Intended Audience}

This monograph is intended as an introductory text
for graduate students who wish to embark on research on vector retrieval.
It is also meant to serve as a self-contained reference that captures
important developments in the field, and as such, may be useful to established
researchers as well.

As the work is geared towards researchers, however, it naturally emphasizes the theoretical
aspects of algorithms as opposed to their empirical behavior or experimental
performance. We present theorems and their proofs, for example.
We do not, on the other hand, present experimental results or compare
algorithms on datasets systematically. There is also no discussion around the
use of the presented algorithms in practice, notes on implementation and libraries, or
practical insights and heuristics that are often critical to making these algorithms
work on real data.
As a result, practitioners or applied researchers may not find
the material immediately relevant.

Finally, while we make every attempt to articulate the theoretical results
and explain the proofs thoroughly, having some familiarity with linear algebra
and probability theory helps digest the results more easily.
We have included a review of the relevant concepts and results from
these subjects in Appendices~\ref{appendix:probability}
(probability),~\ref{appendix:measure} (concentration inequalities),
and~\ref{appendix:linear-algebra} (linear algebra) for convenience.
Should the reader wish to skip the proofs, however, the narrative
should still paint a complete picture of how each algorithm works.

\extrachap{Acknowledgements}
\addcontentsline{toc}{chapter}{Acknowledgements}

I am forever indebted to my dearest colleagues Edo Liberty, Amir Ingber,
Brian Hentschel, and Aditya Krishnan. This incredible but humble group
of scholars at Pinecone are generous with their time and knowledge,
patiently teaching me what I do not know,
and letting me use them as a sounding board without fail.
Their encouragement throughout the process of writing this manuscript,
too, was the force that drove this work to completion.

I am also grateful to Claudio Lucchese, a dear friend, a co-author,
and a professor of computer science at the Ca' Foscari University of Venice, Italy.
I conceived of the idea for this monograph as I lectured at Ca' Foscari
on the topic of retrieval and ranking, upon Claudio's kind invitation.

I would not be writing these words were it not for
the love, encouragement, and wisdom of Franco Maria Nardini,
of the ISTI CNR in Pisa, Italy. In the mad and often maddening world of research,
Franco is the one knowledgeable and kind soul who
restores my faith in research and guides me as I navigate the landscape.

Finally, there are no words that could possibly convey my deepest gratitude
to my partner, Katherine, for always supporting me and my ambitions;
for showing by example what dedication, tenacity, and grit ought to mean;
and for finding me when I am lost.
\chapter*{Notation}
\label{appendix:notation}
\addcontentsline{toc}{chapter}{Notation}

This section summarizes the special symbols and notation used throughout this work.
We often repeat these definitions in context as a reminder, especially if we choose to
abuse notation for brevity or other reasons.

\begin{svgraybox}
    Paragraphs that are highlighted in a gray box such as this contain important
    statements, often conveying key findings or observations, or a detail that will
    be important to recall in later chapters.
\end{svgraybox}

\subsection*{Terminology}
We use the terms ``vector'' and ``point'' interchangeably.
In other words, we refer to an ordered list of $d$ real values
as a $d$-dimensional vector or a point in $\mathbb{R}^d$.

We say that a point is a \emph{data} point if it is part of the collection
of points we wish to sift through. It is a \emph{query} point if it is the input
to the search procedure, and for which we are expected to return the top-$k$
similar data points from the collection.

\subsection*{Symbols}

\subsubsection*{Reserved Symbols}
\begin{longtable*}{p{0.3\linewidth}p{0.7\linewidth}}
$\mathcal{X}$ & Used exclusively to denote a collection of vectors. \\
$m$ & We use this symbol exclusively to denote the cardinality of a collection of data points, $\mathcal{X}$. \\
$q$ & Used singularly to denote a query point. \\
$d$ & We use this symbol exclusively to refer to the number of dimensions. \\
$e_1, e_2, \ldots, e_d$ & Standard basis vectors in $\mathbb{R}^d$ \\
\end{longtable*}

\subsubsection*{Sets}
\begin{longtable*}{p{0.3\linewidth}p{0.7\linewidth}}
$\mathcal{J}$ & Calligraphic font typically denotes sets. \\
$\lvert \cdot \rvert$ & The cardinality (number of items) of a finite set. \\
$[n]$ & The set of integers from $1$ to $n$ (inclusive): $\{ 1, 2, 3, \ldots, n\}$. \\
$B(u, r)$ & The closed ball of radius $r$ centered at point $u$: $\{ v \;|\; \delta(u, v) \leq r \}$ where
$\delta(\cdot, \cdot)$ is the distance function.\\
$\setminus$ & The set difference operator: $\mathcal{A} \setminus \mathcal{B} = \{ x \in \mathcal{A} \;|\; x \notin \mathcal{B} \}$. \\
$\triangle$ & The symmetric difference of two sets. \\
$\mathbbm{1}_p$ & The indicator function. It is $1$ if the predicate $p$ is true, and $0$ otherwise. \\
\end{longtable*}

\subsubsection*{Vectors and Vector Space}
\begin{longtable*}{p{0.3\linewidth}p{0.7\linewidth}}
$[a, b]$ & The closed interval from $a$ to $b$. \\
$\mathbb{Z}$ & The set of integers. \\
$\mathbb{R}^d$ & $d$-dimensional Euclidean space. \\
$\mathbb{S}^{d-1}$ & The hypersphere in $\mathbb{R}^d$. \\
$u, v, w$ & Lowercase letters denote vectors. \\
$u_i, v_i, w_i$ & Subscripts identify a specific coordinate of a vector, so that $u_i$ is the $i$-th coordinate of vector $u$. \\
\end{longtable*}

\subsubsection*{Functions and Operators}
\begin{longtable*}{p{0.3\linewidth}p{0.7\linewidth}}
$\mathit{nz}(\cdot)$ & The set of non-zero coordinates of a vector: $\mathit{nz}(u) = \{ i \;|\; u_i \neq 0 \}$. \\
$\delta(\cdot, \cdot)$ & We use the symbol $\delta$ exclusively to denote the distance function,
taking two vectors and producing a real value. \\
$J(\cdot, \cdot)$ & The Jaccard similarity index of two vectors: $J(u, v) = \lvert \mathit{nz}(u) \cap \mathit{nz}(v) \rvert / \lvert \mathit{nz}(u) \cup \mathit{nz}(v) \rvert$. \\
$\langle \cdot, \cdot \rangle$ & Inner product of two vectors: $\langle u, v \rangle = \sum_i u_i v_i$. \\
$\lVert \cdot \rVert_p$ & The $L_p$ norm of a vector: $\lVert u \rVert_p = (\sum_i \lvert u_i \rvert^p)^{1/p}$. \\
$\oplus$ & The concatenation of two vectors. If $u, v \in \mathbb{R}^d$, then $u \oplus v \in \mathbb{R}^{2d}$. \\
\end{longtable*}

\subsubsection*{Probabilities and Distributions}
\begin{longtable*}{p{0.3\linewidth}p{0.7\linewidth}}
$\ev[\cdot]$ & The expected value of a random variable. \\
$\var[\cdot]$ & The variance of a random variable. \\
$\probability[\cdot]$ & The probability of an event. \\
$\land$, $\lor$ & Logical AND and OR operators. \\
$Z$ & We generally use uppercase letters to denote random variables. \\
\end{longtable*}

\tableofcontents

\mainmatter
\begin{partbacktext}
\part{Introduction}
\end{partbacktext}

\chapter{Vector Retrieval}
\label{chapter:flavors}

\abstract{
This chapter sets the stage for the remainder of this monograph.
It explains where vectors come from, how they have come to represent
data of any modality, and why they are a useful mathematical tool in machine learning.
It then describes the structure we typically expect from a collection
of vectors: that similar objects get vector representations that are close to 
each other in an inner product or metric space.
We then define the problem of top-$k$ retrieval over a well-structured
collection of vectors, and explore its different flavors, including
approximate retrieval.
}

\section{Vector Representations}
We routinely use ordered lists of numbers, or \emph{vectors}, to describe objects of any
shape or form. Examples abound. Any geographic location on earth can be recognized as a vector
consisting of its latitude and longitude. A desk can be described as 
a vector that represents its dimensions, area, color, and other quantifiable properties.
A photograph as a list of pixel values that together paint a picture.
A sound wave as a sequence of frequencies.

Vector representations of objects have long been an integral part of the machine learning
literature. Indeed, a classifier, a regression model, or a ranking function learns patterns
from, and acts on, vector representations of data.
In the past, this vector representation of an object was nothing more than
a collection of its \emph{features}.
Every feature described some facet of the object (for example, the color intensity of a pixel
in a photograph) as a continuous or discrete value.
The idea was that, while individual features describe only a small part of the object,
together they provide sufficiently powerful statistics about the object and its properties
for the machine learnt model to act on.

The features that led to the vector representation of an object were generally hand-crafted functions.
To make sense of that, let us consider a text document in English.
Strip the document of grammar and word order, and we end up with a \emph{set} of words,
more commonly known as a ``bag of words.'' This set can be summarized as a histogram.

If we designated every term in the English vocabulary to be
a dimension in a (naturally) high-dimensional space,
then the histogram representation of the document can be encoded as a vector.
The resulting vector has relatively few non-zero coordinates,
and each non-zero coordinate records the frequency of a term present in the document.
This is illustrated in Figure~\ref{figure:flavors:text-sparse-vector} for a toy example.
More generally, non-zero values may be a function of a term's frequency in 
the document and its propensity in a collection---that is, the likelihood of encountering
the term~\citep{salton1988term}.

\begin{figure}[t]
    \centering
    \includegraphics[width=0.8\linewidth]{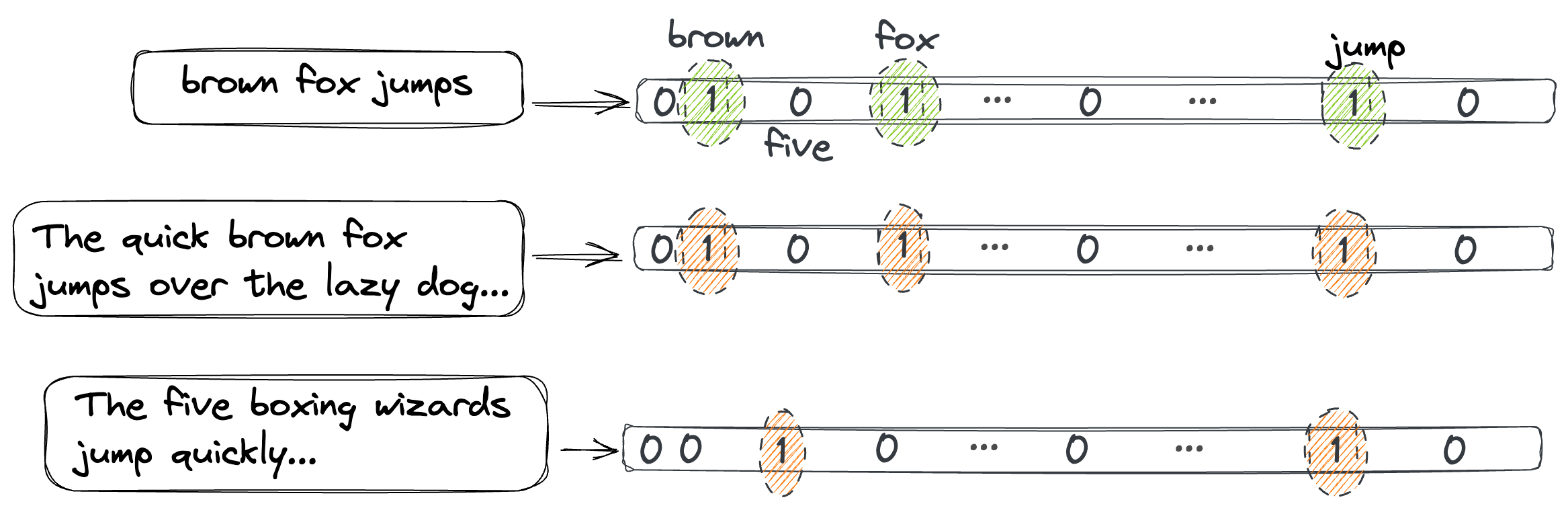}
    \caption{Vector representation of a piece of text by adopting
    a ``bag of words'' view: A text document, when stripped of grammar and word order,
    can be thought of as a vector, where each coordinate represents a term in our vocabulary
    and its value records the frequency of that term in the document or some function of it.
    The resulting vectors are typically \emph{sparse}; that is, they have very few non-zero coordinates.}
    \label{figure:flavors:text-sparse-vector}
\end{figure}

\bigskip

The advent of \emph{deep learning} and, in particular, Transformer-based models~\citep{vaswani2017attention}
brought about vector representations that are beyond the elementary formation above.
The resulting representation is often, as a single entity, referred to as an \emph{embedding},
instead of a ``feature vector,''
though the underlying concept remains unchanged: an object is encoded as a real $d$-dimensional vector,
a point in $\mathbb{R}^d$.

Let us go back to the example from earlier to see how the embedding of a text document
could be different from its representation as a frequency-based feature vector.
Let us maintain the one-to-one mapping between coordinates
and terms in the English vocabulary. Remember that in the ``lexical'' representation from earlier,
if a coordinate was non-zero, that implied that the corresponding term was present in the document
and its value indicated its frequency-based feature. Here we instead \emph{learn} to turn
coordinates on or off and, when we turn a coordinate on, we want its value to predict the
significance of the corresponding term based on semantics and contextual information.
For example, the (absent) synonyms of a (present) term may get a non-zero value, and terms that offer
little discriminative power in the given context become $0$ or close to it.
This basic idea has been explored extensively by many recent models of text~\citep{sparterm,formal2021splade,formal2022splade,zhuang2022reneuir,dai2020sigir,coil,mallia2021learning,zamani2018cikm,unicoil}
and has been shown to produce effective representations.

Vector representations of text need not be sparse.
While sparse vectors with dimensions that are grounded in the vocabulary
are inherently \emph{interpretable}, text documents can also be represented with
lower-dimensional \emph{dense} vectors (where every coordinate is \emph{almost surely} non-zero).
This is, in fact, the most dominant form of vector representation of text
documents in the literature~\citep{lin2021pretrained,karpukhin-etal-2020-dense,xiong2021approximate,reimers-2019-sentence-bert,santhanam-etal-2022-colbertv2,colbert2020khattab}. Researchers have also explored \emph{hybrid}
representations of text where vectors have a dense subspace and an orthogonal sparse
subspace~\citep{chen2022ecir,bruch2023fusion,wang2021bert,Kuzi2020LeveragingSA,karpukhin-etal-2020-dense,Ma2021ARS,Ma2020HybridFR,Wu2019EfficientIP}.

Unsurprisingly, the same embedding paradigm can be extended to other data modalities beyond text:
Using deep learning models, one may embed images, videos, and audio recordings
into vectors. In fact, it is even possible to project different
data modalities (e.g., images and text) together into the same vector space
and preserve some property of interest~\citep{multimodal2020Zhang,guo2019multimodal}.

\begin{svgraybox}
It appears, then, that vectors are everywhere.
Whether they are the result of hand-crafted functions that
capture features of the data or are the output of learnt models;
whether they are dense, sparse, or both,
they are effective representations of data of any modality.
\end{svgraybox}

But what precisely is the point of turning every piece of data into a vector?
One answer to that question takes us to the fascinating world of \emph{retrieval}.

\section{Vectors as Units of Retrieval}

It would make for a vapid exercise if all we had were vector representations of data
without any structure governing a collection of them.
To give a collection of points some structure, we must first ask ourselves what goal
we are trying to achieve by turning objects into vectors.
It turns out, we often intend for the vector representation of two \emph{similar} objects to
be ``close'' to each other according to some well-defined distance function.

That is the structure we desire: Similarity in the vector space must imply similarity
between objects. So, as we engineer features to be extracted from an object, or
design a protocol to learn a model to produce embeddings of data,
we must choose the dimensionality $d$ of the target space
(a subset of $\mathbb{R}^d$) along with a distance function $\delta(\cdot, \cdot)$.
Together, these define an inner product or metric space.

\bigskip

Consider again the lexical representation of a text document
where $d$ is the size of the English vocabulary. Let $\delta$ be the
distance variant of the Jaccard index,
$\delta(u, v) = - J(u, v) \triangleq - \lvert \mathit{nz}(u) \cap \mathit{nz}(v) \rvert / \lvert \mathit{nz}(u) \cup \mathit{nz}(v) \rvert$,
where $\mathit{nz}(u) = \{ i \;|\; u_i \neq 0 \}$ with $u_i$ denoting the $i$-th
coordinate of vector $u$.

In the resulting space, if vectors $u$ and $v$ have a smaller distance
than vectors $u$ and $w$,
then we can clearly conclude that the document represented by $u$ is lexically
more similar to the one represented by $v$ than it is to the document $w$ represents.
That is because the distance (or, in this case, similarity) function reflects
the amount of overlap between the terms present in one document with another.

\bigskip

We should be able to make similar arguments given a semantic embedding of text documents.
Again consider the sparse embeddings with $d$ being the size of the vocabulary,
and more concretely, take \splade{}~\citep{formal2021splade} as a concrete example.
This model produces real-valued sparse vectors in an inner product space.
In other words, the objective of its learning procedure is to
maximize the inner product between similar vectors,
where the inner product between two vectors $u$ and $v$ is
denoted by $\langle u, v \rangle$ and is computed using $\sum_i u_i v_i$.

In the resulting space, if $u$, $v$, and $w$ are generated by \splade{}
with the property that $\langle u, v \rangle > \langle u, w \rangle$, then we can conclude
that, according to \splade{}, documents represented by $u$ and $v$ are semantically more
similar to each other than $u$ is to $w$.
There are numerous other examples of models that optimize for the angular distance
or Euclidean distance ($L_2$) between vectors to preserve (semantic) similarity.

\bigskip

What can we do with a well-characterized collection of vectors that represent real-world
objects? Quite a lot, it turns out. One use case is the topic of this monograph:
the fundamental problem of retrieval.

\begin{svgraybox}
We are often interested in finding $k$ objects that have the highest
degree of similarity to a query object.
When those objects are represented by vectors in a collection $\mathcal{X}$,
where the distance function $\delta(\cdot, \cdot)$ is reflective of similarity,
we may formalize this top-$k$ question mathematically as finding the $k$
minimizers of distance with the query point!
\end{svgraybox}

We state that formally in the following definition:

\begin{definition}[Top-$k$ Retrieval]
\label{definition:flavors:top-k-retrieval}
Given a distance function $\delta(\cdot, \cdot)$, we wish to pre-process
a collection of data points $\mathcal{X} \subset \mathbb{R}^d$
in time that is polynomial in $\lvert \mathcal{X} \rvert$ and $d$,
to form a data structure (the ``index'') whose size is polynomial in
$\lvert \mathcal{X} \rvert$ and $d$, so as to efficiently solve
the following in time $o(\lvert \mathcal{X} \rvert d)$
for an arbitrary query $q \in \mathbb{R}^d$:
\begin{equation}
    \label{equation:flavors:top-k-retrieval}
    \argmin^{(k)}_{u \in \mathcal{X}} \delta(q, u).
\end{equation}
\end{definition}

A web search engine, for example, finds the most relevant documents to your
query by first formulating it as a top-$k$ retrieval problem
over a collection of (not necessarily text-based) vectors.
In this way, it quickly finds the subset of documents from the entire web that
may satisfy the information need captured in your query.
Question answering systems, conversational agents (such as Siri, Alexa, and ChatGPT),
recommendation engines, image search, outlier detectors,
and myriad other applications that are at the forefront
of many online services and in many consumer gadgets
depend on data structures and algorithms that can
answer the top-$k$ retrieval question as efficiently and as effectively as possible.

\section{Flavors of Vector Retrieval}

We create an instance of the deceptively simple
problem formalized in Definition~\ref{definition:flavors:top-k-retrieval}
the moment we acquire a collection of vectors $\mathcal{X}$ together with
a distance function $\delta$.
In the remainder of this monograph, we assume that there is some function,
either manually engineered or learnt, that transforms objects into vectors.
So, from now on, $\mathcal{X}$ is a given.

The distance function then, specifies the flavor of the top-$k$ retrieval problem we need to solve.
We will review these variations and explore what each entails.

\subsection{Nearest Neighbor Search}
In many cases, the distance function is derived from a proper metric
where non-negativity, coincidence, symmetry, and triangle inequality hold for $\delta$.
A clear example of this is the $L_2$ distance:
$\delta(u, v) = \lVert u - v \rVert_2$. The resulting problem,
illustrated for a toy example in Figure~\subref*{figure:flavors:flavors:knn},
is known as $k$-Nearest Neighbors ($k$-NN) search:
\begin{equation}
    \label{equation:flavors:knn}
    \argmin^{(k)}_{u \in \mathcal{X}} \lVert q - u \rVert_2
    = \argmin^{(k)}_{u \in \mathcal{X}} \lVert q - u \rVert_2^2.
\end{equation}

\subsection{Maximum Cosine Similarity Search}

The distance function may also be the angular distance between vectors, which is again a proper metric.
The resulting minimization problem can be stated as follows,
though its equivalent maximization problem (involving the cosine of the angle between
vectors) is perhaps more recognizable:
\begin{equation}
    \label{equation:flavors:kmcs}
    \argmin^{(k)}_{u \in \mathcal{X}} 1 - \frac{\langle q, u \rangle}{\lVert q \rVert_2 \lVert u \rVert_2}
    = \argmax^{(k)}_{u \in \mathcal{X}} \frac{\langle q, u \rangle}{\lVert u \rVert_2}.
\end{equation}
The latter is referred to as the $k$-Maximum Cosine Similarity ($k$-MCS) problem.
Note that, because the norm of the query point, $\lVert q \rVert_2$, is a constant in the
optimization problem, it can simply be discarded; the resulting distance function is
rank-equivalent to the angular distance. Figure~\subref*{figure:flavors:flavors:kmcs}
visualizes this problem on a toy collection of vectors.

\begin{figure}[t]
    \centering
    \subfloat[\textsc{$k$-NN}]{
        \label{figure:flavors:flavors:knn}
        \includegraphics[width=0.32\linewidth]{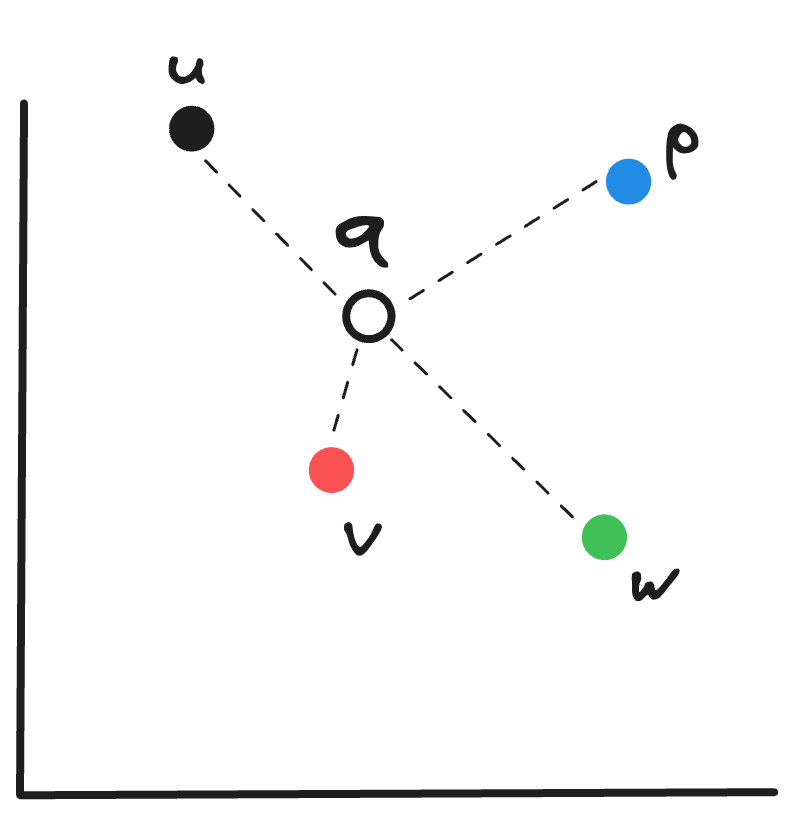}
    }
    \subfloat[\textsc{$k$-MCS}]{
        \label{figure:flavors:flavors:kmcs}
        \includegraphics[width=0.32\linewidth]{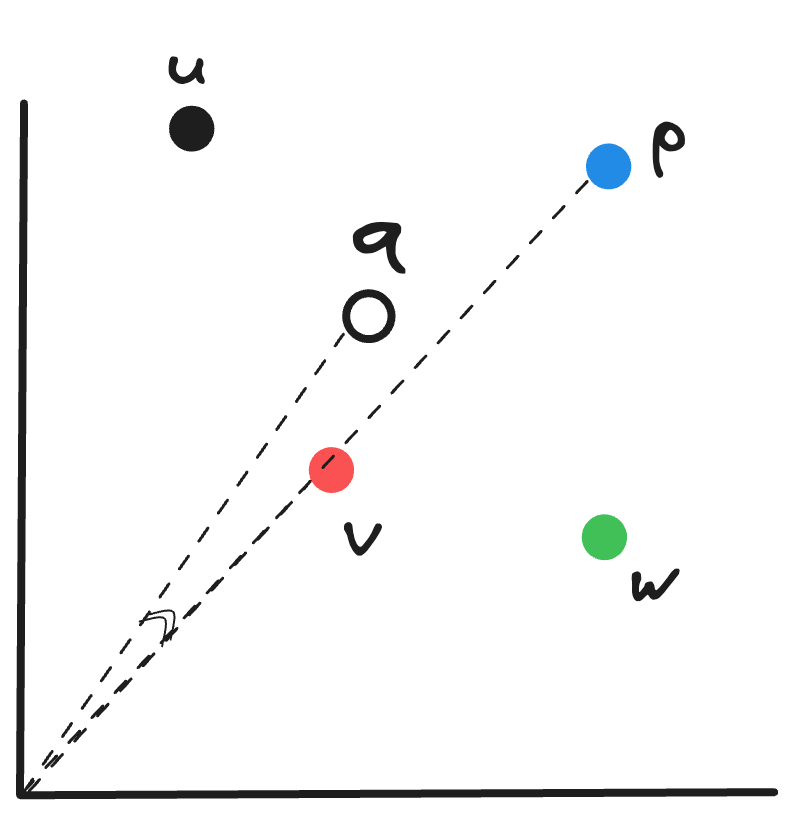}
    }
    \subfloat[\textsc{$k$-MIPS}]{
        \label{figure:flavors:flavors:kmips}
        \includegraphics[width=0.32\linewidth]{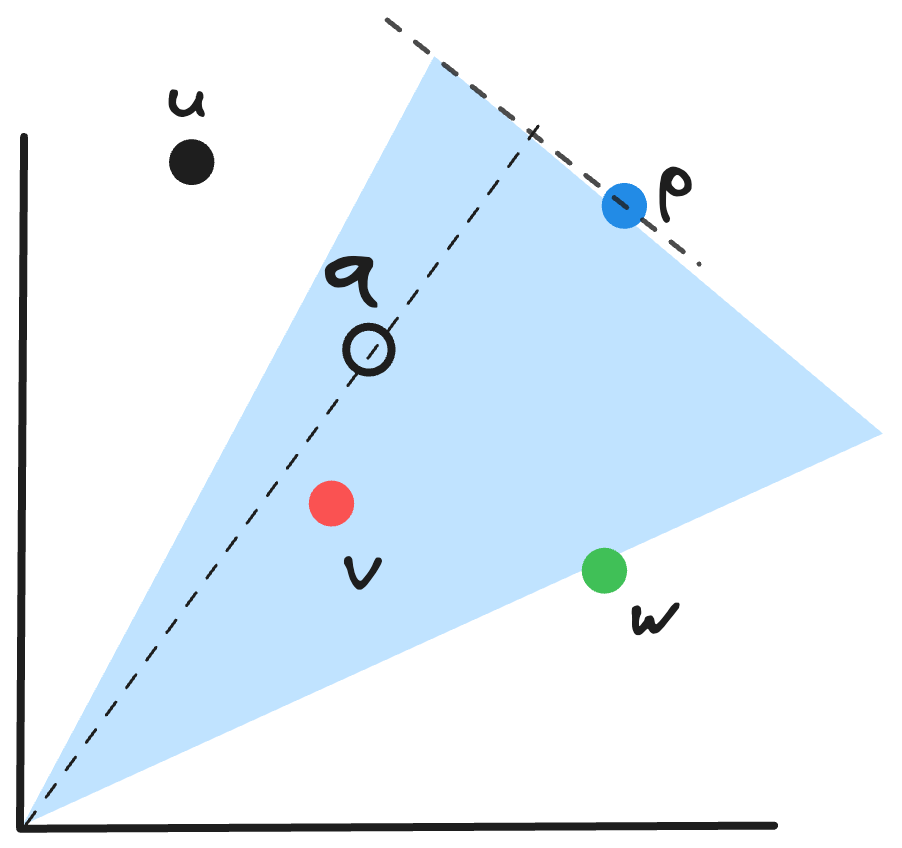}
    }
    \caption{Variants of vector retrieval for a toy vector collection in $\mathbb{R}^2$.
    In Nearest Neighbor search, we find the data point whose $L_2$ distance to the
    query point is minimal ($v$ for top-$1$ search). In Maximum Cosine Similarity search, we instead find the point whose
    angular distance to the query point is minimal ($v$ and $p$ are equidistant from the query).
    In Maximum Inner Product Search, we find a vector that maximizes the inner product with the query
    vector. This can be understood as letting the hyperplane orthogonal to the query point sweep the
    space towards the origin; the first vector to touch the sweeping plane is the maximizer of inner product.
    Another interpretation is this: the shaded region in the figure contains all the points $y$
    for which $p$ is the answer to $\argmax_{x \in \{ u, v, w, p\}} \langle x, y \rangle$.}
    \label{figure:flavors:flavors}
\end{figure}

\subsection{Maximum Inner Product Search}
\label{section:flavors:flavors:mips}
Both of the problems in Equations~(\ref{equation:flavors:knn}) and~(\ref{equation:flavors:kmcs})
are special instances of a more general problem known as $k$-Maximum Inner Product Search ($k$-MIPS):
\begin{equation}
    \label{equation:flavors:kmips}
    \argmax^{(k)}_{u \in \mathcal{X}} \langle q, u \rangle.
\end{equation}
This is easy to see for $k$-MCS: If, in a pre-processing step,
we $L_2$-normalized all vectors in $\mathcal{X}$ so that $u$ is transformed to $u^\prime = u / \lVert u \rVert_2$,
then $\lVert u^\prime \rVert_2 = 1$ and therefore Equation~(\ref{equation:flavors:kmcs}) reduces to Equation~(\ref{equation:flavors:kmips}).

As for a reduction of $k$-NN to $k$-MIPS, we can expand Equation~(\ref{equation:flavors:knn})
as follows:
\begin{align*}
    \argmin^{(k)}_{u \in \mathcal{X}} \lVert q - u \rVert_2^2 &=
        \argmin^{(k)}_{u \in \mathcal{X}} \lVert q \rVert_2^2 - 2\langle q, u \rangle + \lVert u \rVert_2^2\\
    &= \argmax^{(k)}_{u^\prime \in \mathcal{X}^\prime} \langle q^\prime, u^\prime \rangle,
\end{align*}
where we have discarded the constant term, $\lVert q \rVert_2^2$,
and defined $q^\prime \in \mathbb{R}^{d+1}$ as the concatenation of $q \in \mathbb{R}^d$ and
a $1$-dimensional vector with value $-1/2$ (i.e., $q^\prime = [q, -1/2]$),
and $u^\prime \in \mathbb{R}^{d+1}$ as $[u, \lVert u \rVert_2^2]$.

The $k$-MIPS problem, illustrated on a toy collection in
Figure~\subref*{figure:flavors:flavors:kmips}, does not come about just as the result of the reductions
shown above. In fact, there exist embedding models (such as \splade{}, as discussed earlier)
that learn vector representations with respect to inner product as the distance function.
In other words, $k$-MIPS is an important problem in its own right.

\subsubsection{Properties of MIPS}

In a sense, then, it is sufficient to solve the $k$-MIPS problem
as it is the umbrella problem for much of vector retrieval.
Unfortunately, $k$-MIPS is a much harder problem than the other variants.
That is because inner product is not a proper metric.
In particular, it is not non-negative and does not satisfy the triangle inequality, so that
$\langle u, v \rangle \nless \langle u, w \rangle + \langle w, v \rangle$ in general.

\begin{svgraybox}
Perhaps more troubling is the fact that even ``coincidence'' is not guaranteed.
In other words, it is not true in general that a vector $u$ maximizes
inner product with itself: $u \neq \argmax_{v \in \mathcal{X}} \langle v, u \rangle$!
\end{svgraybox}

As an example, suppose $v$ and $p = \alpha v$ for some $\alpha > 1$ are vectors
in the collection $\mathcal{X}$---a case demonstrated in Figure~\subref*{figure:flavors:flavors:kmips}.
Clearly, we have that $\langle v, p \rangle = \alpha \langle v, v \rangle > \langle v, v \rangle$,
so that $p$ (and not $v$) is the solution to MIPS\footnote{When $k=1$, we drop the symbol $k$ from the name of the retrieval problem. So we write MIPS instead of $1$-MIPS.}
for the query point $v$.

In high-enough dimensions and under certain statistical conditions, however,
coincidence is reinstated for MIPS with high probability. One such case is stated in the 
following theorem.

\begin{theorem}
    \label{theorem:flavors:coincidence}
    Suppose data points $\mathcal{X}$ are independent and identically distributed (\emph{iid}) in each dimension
    and drawn from a zero-mean distribution. Then, for any $u \in \mathcal{X}$:
    \begin{equation*}
        \lim_{d \rightarrow \infty} \probability\big[ u = \argmax_{v \in \mathcal{X}}  \langle u, v \rangle ] = 1.
    \end{equation*}
\end{theorem}
\begin{proof}
    Denote by $\var[\cdot]$ and $\ev[\cdot]$ the variance and expected value operators.
    By the conditions of the theorem, it is clear that $\ev[\langle u, u \rangle] = d \ev[Z^2]$ where
    $Z$ is the random variable that generates each coordinate of the vector. We can also see that
    $\ev[\langle u, X \rangle] = 0$ for a random data point $X$, and that
    $\var[\langle u, X \rangle] = \lVert u \rVert_2^2 \ev[Z^2]$.

    We wish to claim that $u \in \mathcal{X}$ is the solution to a MIPS problem where $u$ is also the query point.
    That happens if and only if every other vector in $\mathcal{X}$ has an inner product with $u$ that is smaller than
    $\langle u, u\rangle$. So that:
    \begin{align*}
        \probability\big[ u &= \argmax_{v \in \mathcal{X}}  \langle u, v \rangle ] =
            \probability\big[ \langle u, v \rangle \leq \langle u, u \rangle \quad \forall \; v \in \mathcal{X} ] = \\
            &1 - \probability\big[ \exists \; v \in \mathcal{X} \; \mathit{s.t.} \quad \langle u, v \rangle > \langle u, u \rangle] \geq && \text{(by Union Bound)} \\
            &1 - \sum_{v \in \mathcal{X}} \probability\big[ \langle u, v \rangle > \langle u, u \rangle] = && \text{(by \emph{iid})} \\
            &1 - \lvert \mathcal{X} \rvert \probability\big[ \langle u, X \rangle > \langle u, u \rangle].
    \end{align*}
    Let us turn to the last term and bound the probability for a random data point:
    \begin{equation*}
        \probability\big[\langle u, X \rangle > \langle u, u \rangle] =
        \probability\big[ \underbrace{\langle u, X \rangle - \langle u, u \rangle + d \ev[Z^2]}_{Y} > d \ev[Z^2] \big].
    \end{equation*}
    The expected value of $Y$ is $0$. Denote by $\sigma^2$ its variance. By the application of the one-sided
    Chebyshev's inequality,\footnote{The one-sided Chebyshev's inequality for a random variable $X$ with
    mean $\mu$ and variance $\sigma^2$ states that $\probability\big[ X - \mu > t\big] \leq \sigma^2 / \big( \sigma^2 + t^2 \big)$.}
    we arrive at the following bound:
    \begin{equation*}
        \probability\big[\langle u, X \rangle > \langle u, u \rangle] \leq \frac{\sigma^2}{\sigma^2 + d^2 \ev[Z^2]^2}.
    \end{equation*}
    Note that, $\sigma^2$ is a function of the sum of \emph{iid} random variables, and, as such, grows linearly with $d$.
    In the limit. this probability tends to $0$. We have thus shown that
    $\lim_{d \rightarrow \infty} \probability\big[ u = \argmax_{v \in \mathcal{X}}  \langle u, v \rangle ] \geq 1$ which concludes
    the proof.
\end{proof}

\begin{figure}[t]
    \centering
    \subfloat[\textsc{Synthetic}]{
        \label{figure:flavors:coincidence:synthetic}
        \includegraphics[width=0.52\linewidth]{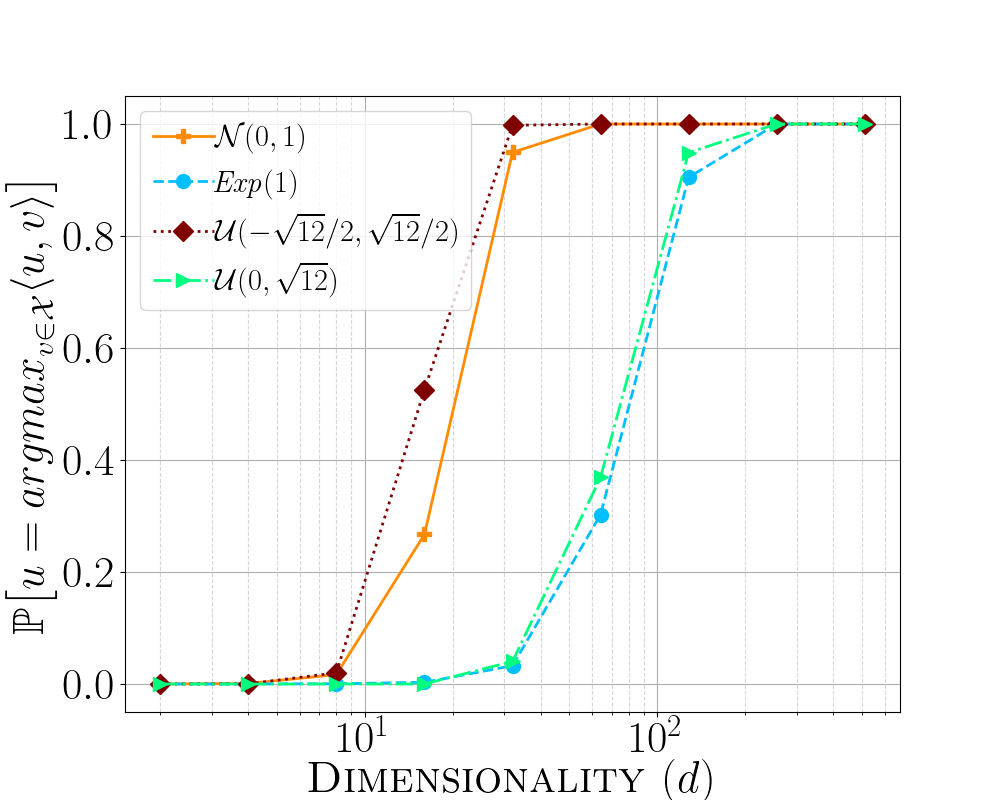}
    }\hspace*{-1.5em}
    \subfloat[\textsc{Real}]{
        \label{figure:flavors:coincidence:real}
        \includegraphics[width=0.52\linewidth]{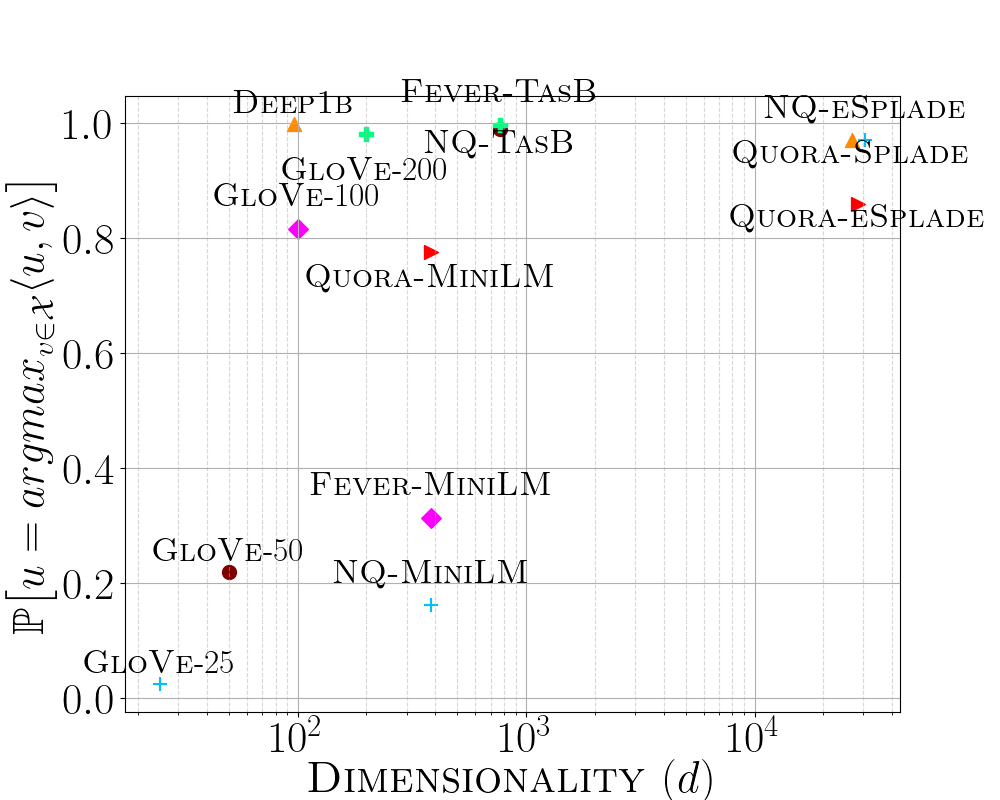}
    }
    \caption{Probability that $u \in \mathcal{X}$ is the solution to MIPS over $\mathcal{X}$ with query $u$
    versus the dimensionality $d$, for various synthetic and real collections $\mathcal{X}$.
    For synthetic collections, $\lvert \mathcal{X} \rvert = 100{,}000$. Appendix~\ref{appendix:collections}
    gives a description of the real collections. Note that, for real collections, we estimate the reported
    probability by sampling $10{,}000$ data points and using them as queries. Furthermore,
    we do not pre-process the vectors---importantly, we do not $L_2$-normalize the collections.}
    \label{figure:flavors:coincidence}
\end{figure}

\subsubsection{Empirical Demonstration of the Lack of Coincidence}
Let us demonstrate the effect of Theorem~\ref{theorem:flavors:coincidence}
empirically. First, let us choose distributions that meet
the requirements of the theorem: a Gaussian distribution
with mean $0$ and variance $1$, and a uniform distribution over $[-\sqrt{12}/2, \sqrt{12}/2]$
(with variance $1$) will do.
For comparison, choose another set of distributions that do not have the requisite properties:
Exponential with rate $1$ and uniform over $[0, \sqrt{12}]$.
Having fixed the distributions, we next sample $100{,}000$ random vectors
from them to form a collection $\mathcal{X}$.
We then take each data point, use it as a query in MIPS over $\mathcal{X}$, and
report the proportion of data points that are solutions to their own search.

Figure~\subref*{figure:flavors:coincidence:synthetic} illustrates the results of this experiment.
As expected, for the Gaussian and centered uniform distributions, the ratio of interest approaches $1$
when $d$ is sufficiently large. Surprisingly, even when the distributions do not strictly satisfy
the conditions of the theorem, we still observe the convergence of that ratio to $1$. So it appears that
the requirements of Theorem~\ref{theorem:flavors:coincidence} are more forgiving than one may imagine.

We also repeat the exercise above on several real-world collections, a description of which
can be found in Appendix~\ref{appendix:collections} along with salient statistics.
The results of these experiments are visualized in Figure~\subref*{figure:flavors:coincidence:real}.
As expected, whether a data point maximizes inner product with itself entirely depends
on the underlying data distribution. We can observe that, for some collections
in high dimensions, we are likely to encounter coincidence in the sense we defined
earlier, but for others that is clearly not the case. It is important to keep this
difference between synthetic and real collections in mind when designing experiments
that evaluate the performance of MIPS systems.

\section{Approximate Vector Retrieval}
\label{chapter:flavors:approximate}

Saying one problem is harder than another neither implies that we cannot approach
the harder problem, nor does it mean that the ``easier'' problem is easy to solve.
In fact, none of these variants of vector retrieval ($k$-NN, $k$-MCS, and $k$-MIPS)
can be solved exactly \emph{and} efficiently in high dimensions.
Instead, we must either accept that the solution would be inefficient
(in terms of space- or time-complexity), or allow some degree of error.

\begin{figure}[t]
    \centering
    \subfloat[\textsc{$k$-NN}]{
        \label{figure:flavors:flavors:approximate-knn}
        \includegraphics[width=0.32\linewidth]{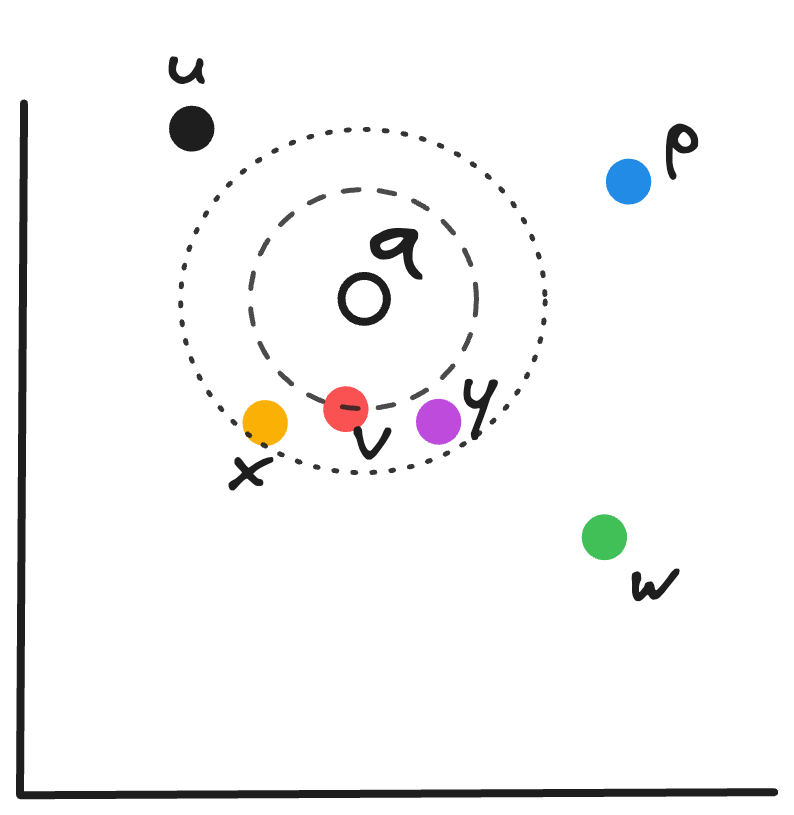}
    }
    \subfloat[\textsc{$k$-MCS}]{
        \label{figure:flavors:flavors:approximate-kmcs}
        \includegraphics[width=0.32\linewidth]{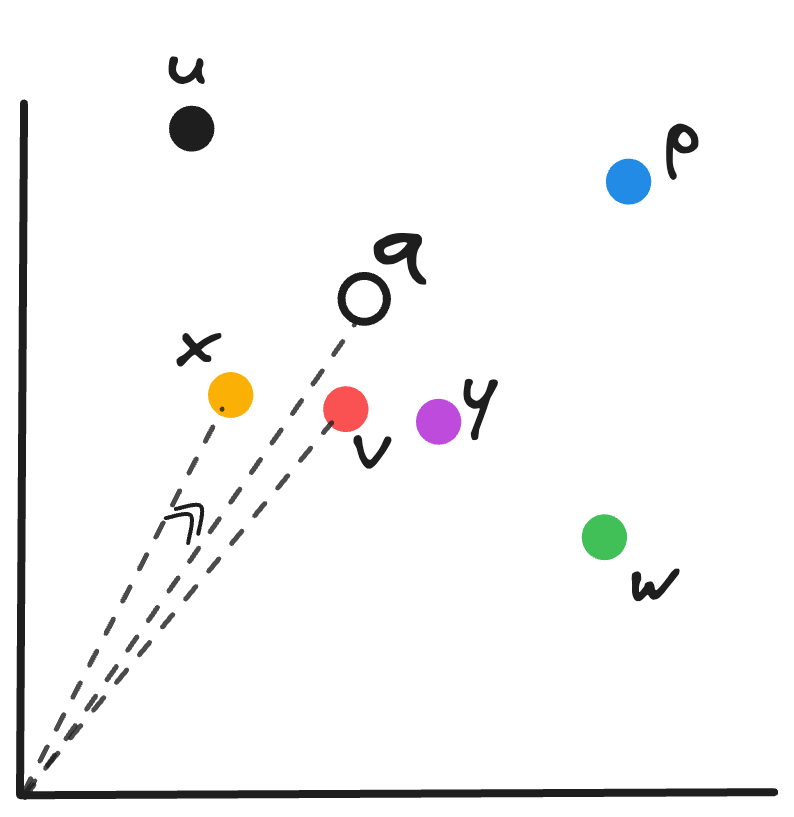}
    }
    \subfloat[\textsc{$k$-MIPS}]{
        \label{figure:flavors:flavors:approximate-kmips}
        \includegraphics[width=0.32\linewidth]{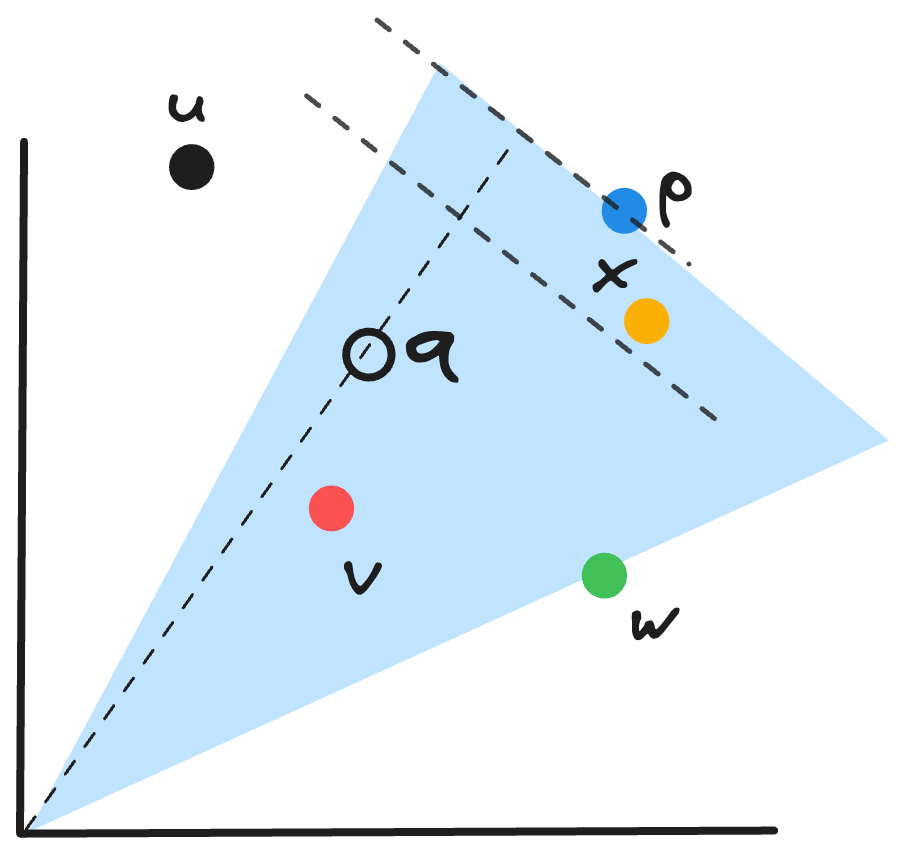}
    }
    \caption{Approximate variants of top-$1$ retrieval for a toy collection in $\mathbb{R}^2$.
    In NN, we admit vectors that are at most $\epsilon$ away from the optimal solution. As such,
    $x$ and $y$ are both valid solutions as they are in a ball with radius $(1+\epsilon) \delta(q, x)$
    centered at the query.
    Similarly, in MCS, we accept a vector (e.g., $x$) if its angle with the query point is at most $1 + \epsilon$
    greater than the angle between the query and the optimal vector (i.e., $v$).
    For the MIPS example, assuming that the inner product of query and $x$ is at most
    $(1 - \epsilon)$-times the inner product of query and $p$, then $x$ is an acceptable solution.}
    \label{figure:flavors:flavors-approximate}
\end{figure}

The first case of solving the problem exactly but inefficiently is uninteresting: If we are looking
to find the solution for $k=1$, for example, it is enough to compute the distance function
for every vector in the collection and the query, resulting in linear complexity.
When $k > 1$, the total time complexity is $\mathcal{O}(\lvert \mathcal{X} \rvert d \log k)$,
where $\lvert \mathcal{X} \rvert$ is the size of the collection.
So it typically makes more sense to investigate the second strategy of admitting error.

That argument leads naturally to the class of $\epsilon$-\emph{approximate} vector retrieval problems.
This idea can be formalized rather easily for the special case where $k=1$:
The approximate solution for the top-$1$ retrieval is satisfactory so long as the vector $u$
returned by the algorithm is at most $(1 + \epsilon)$ factor farther than the
optimal vector $u^\ast$, according to $\delta(\cdot, \cdot)$ and for some arbitrary
$\epsilon > 0$:
\begin{equation}
    \label{equation:flavors:approximate-top-k-retrieval}
    \delta(q, u) \leq (1 + \epsilon) \delta(q, u^\ast).
\end{equation}
Figure~\ref{figure:flavors:flavors-approximate} renders the solution space for
an example collection in $\mathbb{R}^2$.

The formalism above extends to the more general case where $k > 1$ in an obvious way:
a vector $u$ is a valid solution to the $\epsilon$-approximate top-$k$ problem if its distance
to the query point is at most $(1 + \epsilon)$ times the distance to the $k$-th
optimal vector. This is summarized in the following definition:

\begin{definition}[$\epsilon$-Approximate Top-$k$ Retrieval]
\label{definition:flavors:approximate-top-k-retrieval}
Given a distance function $\delta(\cdot, \cdot)$, we wish to pre-process
a collection of data points $\mathcal{X} \subset \mathbb{R}^d$
in time that is polynomial in $\lvert \mathcal{X} \rvert$ and $d$,
to form a data structure (the ``index'') whose size is polynomial in
$\lvert \mathcal{X} \rvert$ and $d$, so as to efficiently solve
the following in time $o(\lvert \mathcal{X} \rvert d)$
for an arbitrary query $q \in \mathbb{R}^d$ and $\epsilon > 0$:
\begin{equation*}
    \mathcal{S} =\argmin^{(k)}_{u \in \mathcal{X}}  \delta(q, u),
\end{equation*}
such that for all $u \in \mathcal{S}$,
Equation~(\ref{equation:flavors:approximate-top-k-retrieval}) is satisfied
where $u^\ast$ is the $k$-th optimal vector obtained by solving the problem
in Definition~\ref{definition:flavors:top-k-retrieval}.
\end{definition}

Despite the extension to top-$k$ above, it is more common
to characterize the effectiveness of an approximate top-$k$ solution
as the percentage of correct vectors that are present in the solution. Concretely,
if $\mathcal{S} = \argmax^{(k)}_{u \in \mathcal{X}} \delta(q, u)$ is the exact set of top-$k$ vectors,
and $\tilde{\mathcal{S}}$ is the approximate set, then the \emph{accuracy} of the approximate algorithm
can be reported as $\lvert \mathcal{S} \cap \tilde{\mathcal{S}} \rvert / k$.\footnote{
This quantity is also known as \emph{recall} in the literature, because we
are counting the number of vectors our algorithm recalls from the exact solution set.}

This monograph primarily studies the approximate\footnote{
We drop $\epsilon$ from the name when it is clear from context.} retrieval problem.
As such, while we state a retrieval problem using the $\argmax$ or $\argmin$
notation, we are generally only interested in approximate solutions to it.

\bibliographystyle{abbrvnat}
\bibliography{biblio}

\chapter{Retrieval Stability in High Dimensions}
\label{chapter:instability}

\abstract{
We are about to embark on a comprehensive survey and analysis of vector retrieval methods
in the remainder of this monograph. It may thus sound odd to suggest that you may not need
any of these clever ideas in order to perform vector retrieval.
Sometimes, under bizarrely general conditions that we will explore formally in this chapter,
an exhaustive search (where we compute the distance between query and every data point,
sort, and return the top $k$) is likely to perform much better in both accuracy and search latency!
The reason why that may be the case has to do with the approximate nature of algorithms
and the oddities of high dimensions. We elaborate this point by focusing on the top-$1$
case.}

\section{Intuition}
\label{chapter:instability:stability:intuition}
Consider the case of proper distance functions where $\delta(\cdot, \cdot)$ is a metric.
Recall from Equation~(\ref{equation:flavors:approximate-top-k-retrieval}) that a vector
$u$ is an acceptable $\epsilon$-approximate solution if its distance to the query $q$ according to $\delta(\cdot, \cdot)$
is at most $(1 + \epsilon) \delta(q, u^\ast)$, where $u^\ast$ is the optimal vector
and $\epsilon$ is an arbitrary parameter.
As shown in Figure~\subref*{figure:flavors:flavors:approximate-knn} for NN,
this means that, if you centered an $L_p$ ball around $q$ with
radius $\delta(q, (1 + \epsilon) u^\ast)$, then $u$ is in that ball.

So, what if we find ourselves in a situation where no matter how small $\epsilon$ is,
too many vectors, or indeed \emph{all} vectors, from our collection $\mathcal{X}$ end up
in the $(1+\epsilon)$-enlarged ball? Then, by definition, every vector is an $\epsilon$-approximate nearest neighbor of $q$!

In such a configuration of points,
it is questionable whether the notion of ``nearest neighbor'' has any meaning at all:
If the query point were perturbed by some noise as small as $\epsilon$,
then its true nearest neighbor would suddenly change, making NN \emph{unstable}.
Because of that instability, any approximate algorithm will need to examine a large portion or nearly all
of the data points anyway, reducing thereby to a procedure that performs
more poorly than exhaustive search.

That sounds troubling. But when might we experience that phenomenon? That is the
question~\cite{beyer1999nnMeaningful} investigate in their seminal paper.

\begin{svgraybox}
It turns out, one scenario where vector retrieval becomes unstable as dimensionality $d$ increases is
if a) data points are \emph{iid} in each dimension,
b) query points are similarly drawn \emph{iid} in each dimension, and
c) query points are independent of data points. This includes many
synthetic collections that are, even today, routinely but inappropriately used
for evaluation purposes.
\end{svgraybox}

On the other hand, when data points form clusters and query points
fall into these same clusters, then the (approximate) ``nearest cluster''
problem is meaningful---but not necessarily the approximate NN problem.
So while it makes sense to use approximate algorithms to obtain the nearest
cluster, search within clusters may as well be exhaustive.
This, as we will learn in Chapter~\ref{chapter:ivf},
is the basis for a popular and effective class of vector retrieval
algorithms on real collections.

\section{Formal Results}
More generally, vector retrieval becomes unstable in high dimensions
when the variance of the distance between query and data
points grows substantially more slowly than its expected value.
That makes sense. Intuitively, that means that more and more data points
fall into the $(1 + \epsilon)$-enlarged ball centered at the query.
This can be stated formally as the following theorem due to~\cite{beyer1999nnMeaningful},
extended to any general distance function $\delta(\cdot, \cdot)$.

\begin{theorem}
\label{theorem:instability:beyer}
    Suppose $m$ data points $\mathcal{X} \subset \mathbb{R}^d$ are drawn iid from a data distribution
    and a query point $q$ is drawn independent of data points from any distribution.
    Denote by $X$ a random data point. If
    \begin{equation*}
        \lim_{d \rightarrow \infty} \var\big[ \delta(q, X) \big] / \ev\big[ \delta(q, X) \big]^2 = 0,
    \end{equation*}
    then for any $\epsilon > 0$,
    $\lim_{d \rightarrow \infty} \probability\big[ \delta(q, X) \leq (1 + \epsilon) \delta(q, u^\ast) \big] = 1,$
    where $u^\ast$ is the vector closest to $q$.
\end{theorem}
\begin{proof}
    Let $\delta_\ast = \max_{u \in \mathcal{X}} \delta(q, u)$ and $\delta^\ast = \min_{u \in \mathcal{X}} \delta(q, u)$.
    If we could show that, for some $d$-dependent \emph{positive} $\alpha$ and $\beta$ such that $\beta/\alpha = 1 + \epsilon$,
    $\lim_{d \rightarrow \infty} \probability\big[ \alpha \leq \delta^\ast \leq \delta_\ast \leq \beta \big] = 1$,
    then we are done. That is because, in that case $\delta_\ast/\delta^\ast \leq \beta / \alpha = 1 + \epsilon$ almost surely
    and the claim follows.
    
    From the above, all that we need to do is to find $\alpha$ and $\beta$ for a given $d$.
    Intuitively, we want the interval $[\alpha, \beta]$ to contain $\ev[\delta(q, X)]$, because
    we know from the condition of the theorem that the distances should concentrate around their mean.
    So $\alpha = (1 - \eta) \ev[\delta(q, X)]$ and $\beta = (1 + \eta) \ev[\delta(q, X)]$ for some $\eta$
    seems like a reasonable choice. Letting $\eta = \epsilon / (\epsilon + 2)$ gives us the desired ratio:
    $\beta/\alpha = 1 + \epsilon$.
    
    Now we must show that $\delta_\ast$ and $\delta^\ast$ belong to our
    chosen $[\alpha, \beta]$ interval almost surely in the limit. That happens if \emph{all} distances belong to 
    that interval. So:
    \begin{align*}
        \lim_{d \rightarrow \infty} &\probability\big[ \alpha \leq \delta^\ast \leq \delta_\ast \leq \beta \big] = \\
        &\lim_{d \rightarrow \infty} \probability\Big[ \delta(q, u) \in [\alpha, \beta] \quad \forall \; u \in \mathcal{X} \Big] = \\
        &\lim_{d \rightarrow \infty} \probability\Big[ (1 - \eta) \ev[\delta(q, X)] \leq \delta(q, u) \leq (1 + \eta) \ev[\delta(q, X)] \; \forall \; u \in \mathcal{X} \Big] = \\
        &\lim_{d \rightarrow \infty} \probability\Big[ \big\lvert \delta(q, u) - \ev[\delta(q, X)] \big\rvert \leq \eta \ev[\delta(q, X)] \quad \forall \; u \in \mathcal{X} \Big].
    \end{align*}
    It is now easier to work with the complementary event:
    \begin{equation*}
        1 - \lim_{d \rightarrow \infty} \probability\Big[ \exists \; u \in \mathcal{X} \; \mathit{s.t.} \quad \big\lvert \delta(q, u) - \ev[\delta(q, X)] \big\rvert > \eta \ev[\delta(q, X)] \Big].
    \end{equation*}
    Using the Union Bound, the probability above is greater than or equal to the following:
    \begin{align*}
    \lim_{d \rightarrow \infty} &\probability\big[ \alpha \leq \delta^\ast \leq \delta_\ast \leq \beta \big] \geq \\
        &1 - \lim_{d \rightarrow \infty} \sum_{u \in \mathcal{X}} \probability\Big[\big\lvert \delta(q, u) - \ev[\delta(q, X)] \big\lvert > \eta \ev[\delta(q, X)] \Big] = \\
        &1 - \lim_{d \rightarrow \infty} \sum_{u \in \mathcal{X}} \probability\Big[\big( \delta(q, u) - \ev[\delta(q, X)] \big)^2 > \eta^2 \ev[\delta(q, X)]^2 \Big].
    \end{align*}
    Note that, $q$ is independent of data points and that data points are \emph{iid} random variables.
    Therefore, $\delta(q, u)$'s are random variables drawn \emph{iid} as well.
    Furthermore, by assumption $\ev[\delta(q, X)]$ exists, making it possible to apply Markov's inequality to obtain the following bound:
    \begin{align*}
    \lim_{d \rightarrow \infty} &\probability\big[ \alpha \leq \delta^\ast \leq \delta_\ast \leq \beta \big] \geq \\
        &1 - \lim_{d \rightarrow \infty} \lvert \mathcal{X} \rvert \probability\Big[\big( \delta(q, X) - \ev[\delta(q, X)] \big)^2 > \eta^2 \ev[\delta(q, X)]^2 \Big] \geq\\
        &1 - \lim_{d \rightarrow \infty} m \frac{1}{\eta^2 \ev \big[\delta(q, X)\big]^2} \ev \Big[ \big( \delta(q, u) - \ev[\delta(q, X)] \big)^2 \Big] =\\
        &1 - \lim_{d \rightarrow \infty} m \frac{\var[\delta(q, X)]}{\eta^2 \ev[\delta(q, X)]^2}.
    \end{align*}
    By the conditions of the theorem, $\var[\delta(q, X)]/\ev[\delta(q, X)]^2 \rightarrow 0$ as 
    $d \rightarrow \infty$, so that the last expression tends to $1$ in the limit. That concludes the proof.
\end{proof}

We mentioned earlier that if data and query points are independent of each other
and that vectors are drawn \emph{iid} in each dimension, then vector retrieval becomes
unstable. For NN with the $L_p$ norm, it is easy to show that such a configuration
satisfies the conditions of Theorem~\ref{theorem:instability:beyer}, hence the instability.
Consider the following for $\delta(q, u) = \lVert q - u \rVert_p^p$:
\begin{align*}
    \lim_{d \rightarrow \infty} &\frac{\var \big[ \lVert q - u \rVert_p^p \big]}{ \ev \big[ \lVert q - u \rVert_p^p \big]^2} = 
    \lim_{d \rightarrow \infty} \frac{\var \big[ \sum_i ( q_i - u_i )^p \big]}{ \ev \big[ \sum_i ( q_i - u_i )^p \big]^2} = \\
    &\lim_{d \rightarrow \infty} \frac{\sum_i \var \big[ ( q_i - u_i )^p \big]}{ \big(\sum_i \ev \big[ ( q_i - u_i )^p \big] \big)^2} && \text{(by independence)} \\
    &\lim_{d \rightarrow \infty} \frac{d \sigma^2}{d^2 \mu^2} = 0,
\end{align*}
where we write $\sigma^2 = \var[(q_i - u_i)^p]$ and $\mu = \ev[(q_i - u_i)^p]$.

When $\delta(q, u) = - \langle q, u \rangle$, the same conditions result in retrieval instability:
\begin{align*}
    \lim_{d \rightarrow \infty} &\frac{\var \big[ \langle q, u \rangle \big]}{ \ev \big[ \langle q, u \rangle \big]^2} = 
    \lim_{d \rightarrow \infty} \frac{\var \big[ \sum_i q_i u_i \big]}{ \ev \big[ \sum_i q_i u_i \big]^2} = \\
    &\lim_{d \rightarrow \infty} \frac{\sum_i \var \big[ q_i u_i \big]}{ \big(\sum_i \ev \big[ q_i u_i \big] \big)^2} && \text{(by independence)} \\
    &\lim_{d \rightarrow \infty} \frac{d \sigma^2}{d^2 \mu^2} = 0,
\end{align*}
where we write $\sigma^2 = \var[q_i u_i]$ and $\mu = \ev[q_i u_i]$.

\section{Empirical Demonstration of Instability}

Let us examine the theorem empirically.
We simulate the NN setting with $L_2$ distance and report the results in Figure~\ref{figure:instability:instability}.
In these experiments, we sample $1{,}000{,}000$ data points with each coordinate drawing its value independently
from the same distribution, and $1{,}000$ query points sampled similarly. We then compute the minimum and maximum distance
between each query point and the data collection, measure the ratio between them, and report the mean
and standard deviation of the ratio across queries. We repeat this exercise for various values of dimensionality $d$
and render the results in Figure~\subref*{figure:instability:instability:ratio}. Unsurprisingly, this ratio tends
to $1$ as $d \rightarrow \infty$, as predicted by the theorem.

\begin{figure}[t]
    \centering
    \subfloat[$\delta_\ast / \delta^\ast$]{
        \label{figure:instability:instability:ratio}
        \includegraphics[width=0.52\linewidth]{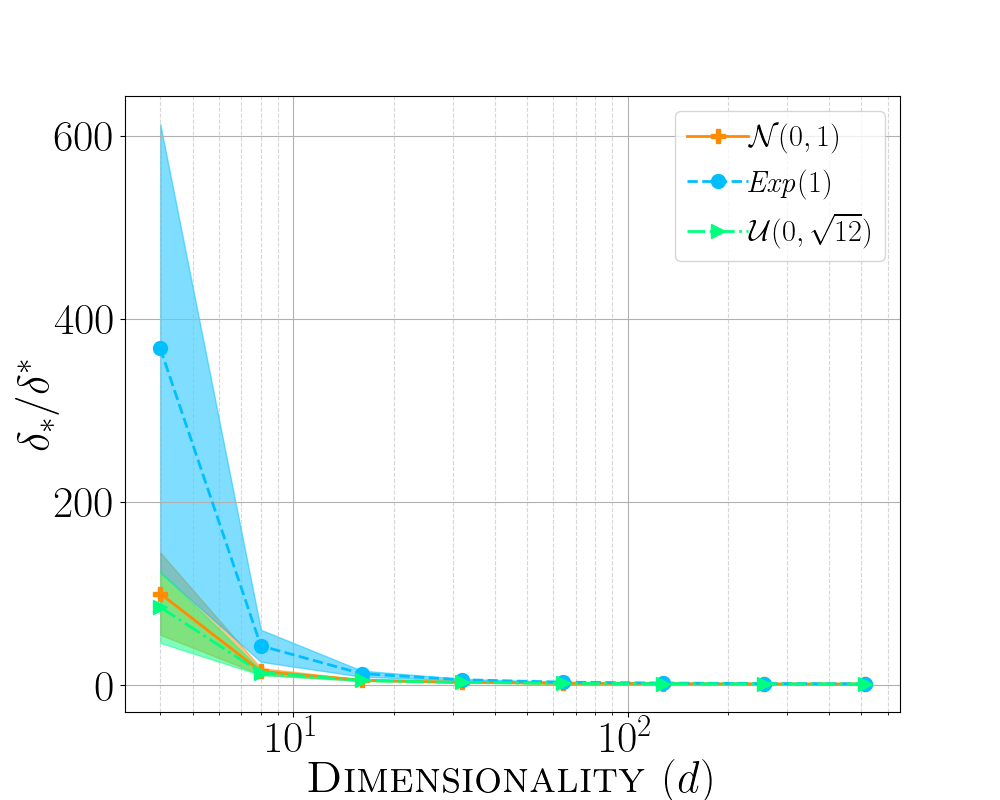}
    }\hspace*{-1.5em}
    \subfloat[\textsc{Percent Approximate Solutions}]{
        \label{figure:instability:instability:count}
        \includegraphics[width=0.52\linewidth]{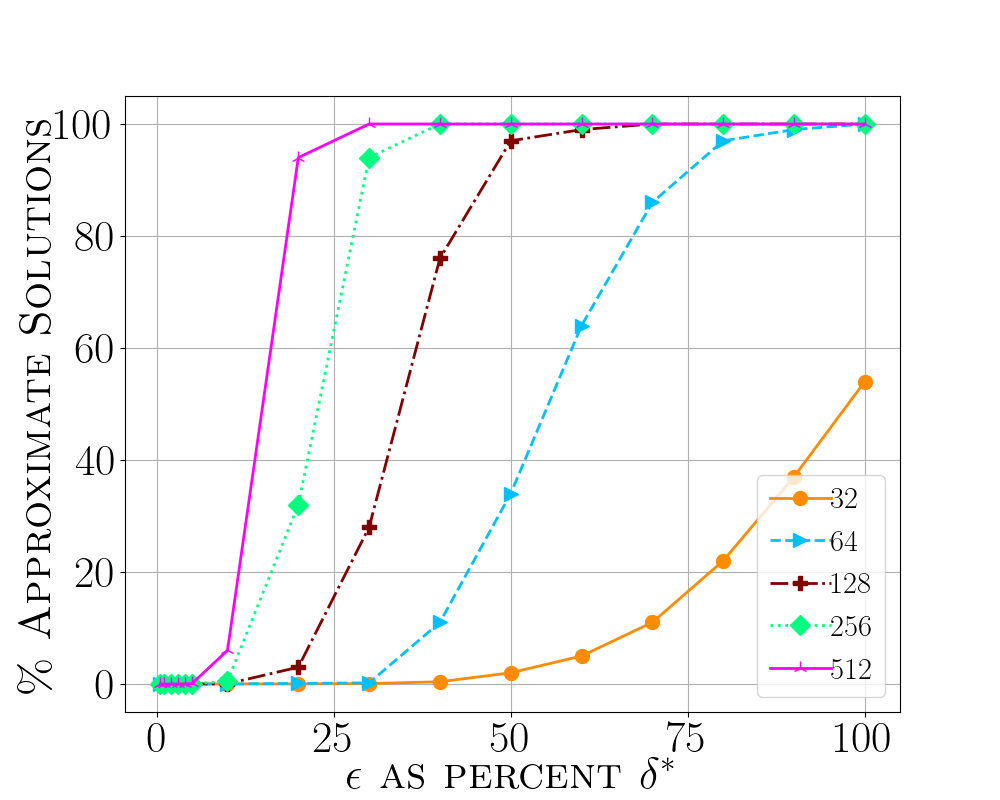}
    }
    \caption{Simulation results for Theorem~\ref{theorem:instability:beyer} applied to NN with $L_2$ distance.
    \emph{Left}: The ratio between the maximum
    distance between a query and data points $\delta_\ast$, to the minimum distance $\delta^\ast$.
    The shaded region shows one standard deviation.
    As dimensionality increases, this ratio tends to $1$. \emph{Right}: The percentage of data points whose
    distance to a query is at most $(1 + \epsilon/100) \delta^\ast$, visualized for the Gaussian
    distribution---the trend is similar for other distributions. As $d$ increases, more vectors
    fall into the enlarged ball, making them valid solutions to the approximate NN problem.}
    \label{figure:instability:instability}
\end{figure}

Another way to understand this result is to count the number of data points that qualify
as approximate nearest neighbors. The theory predicts that, as $d$ increases, we can find a smaller
$\epsilon$ such that nearly all data points fall within $(1 + \epsilon) \delta^\ast$ distance from the query.
The results of our experiments confirm this phenomenon; we have plotted the results for the Gaussian distribution
in Figure~\subref*{figure:instability:instability:count}.

\subsection{Maximum Inner Product Search}

In the discussion above, we established that retrieval becomes unstable in
high dimensions if the data satisfies certain statistical conditions.
That meant that the difference between the maximum and the minimum distance
grows just as fast as the magnitude of the minimum distance, so that any approximate
solution becomes meaningless.

\begin{svgraybox}
The instability statement does not necessarily imply, however, that the distances
become small or converge to a certain value. But as we see in this section,
inner product in high dimensions does become smaller and smaller as a function of $d$.
\end{svgraybox}

The following theorem summarizes this phenomenon for a unit query point and bounded
data points.
Note that, the condition that $q$ is a unit vector is not restrictive in any way,
as the norm of the query point does not affect the retrieval outcome.

\begin{theorem}
    \label{theorem:instability:orthogonality-random-vectors}
    If $m$ data points with bounded norms,
    and a unit query vector $q$ are drawn \emph{iid}
    from a spherically symmetric\footnote{A distribution is spherically symmetric
    if it remains invariant under an orthogonal transformation.}
    distribution in $\mathbb{R}^d$, then:
    \begin{equation*}
        \lim_{d \rightarrow \infty} \probability\big[ \langle q, X \rangle > \epsilon \big] = 0.
    \end{equation*}
\end{theorem}
\begin{proof}
    By spherical symmetry, it is easy to see that $\ev[\langle q, X \rangle] = 0$.
    The variance of the inner product is then equal to $\ev[\langle q, X \rangle^2]$,
    which can be expanded as follows.
    
    First, find an orthogonal transformation $\Gamma: \mathbb{R}^d \rightarrow \mathbb{R}^d$ that maps
    the query point $q$ to the first standard basis (i.e., $e_1 = [1, 0, 0, \ldots, 0] \in \mathbb{R}^d$).
    Due to spherical symmetry, this transformation does not change the data distribution.
    Now, we can write:
    \begin{align*}
        \ev[\langle q, X \rangle^2] &= \ev[ \langle \Gamma q, \Gamma X \rangle^2] =
        \ev[(\Gamma X)_1^2] = \\
        &\ev[\frac{1}{d} \sum_{i=1}^d (\Gamma X)_i^2] = \frac{1}{d} \lVert X \rVert_2^2.
    \end{align*}
    In the above, the third equality is due to the fact that the distribution of the (transformed)
    vectors is the same in every direction. Because $\lVert X \rVert$ is bounded by assumption, the
    variance of the inner product between $q$ and a random data point tends to $0$ as $d \rightarrow \infty$.
    The claim follows.
\end{proof}

\begin{svgraybox}
The proof of Theorem~\ref{theorem:instability:orthogonality-random-vectors} tells us that
the variance of inner product grows as a function of $1/d$ and $\lVert X \rVert_2^2$.
So if our vectors have bounded norms, then we can find a $d$ such that inner products are
arbitrarily close to $0$. This is yet another reason that approximate MIPS could become meaningless.
But if our data points are clustered in (near) orthogonal subspaces, then approximate
MIPS over clusters makes sense---though, again, MIPS within clusters would be unstable.
\end{svgraybox}

\bibliographystyle{abbrvnat}
\bibliography{biblio}

\chapter{Intrinsic Dimensionality}
\label{chapter:intrinsic-dimensionality}

\abstract{
We have seen that high dimensionality poses difficulties for vector retrieval.
Yet, judging by the progression from hand-crafted feature vectors to sophisticated
embeddings of data, we detect a clear trend towards higher dimensional representations of data.
How worried should we be about this ever increasing dimensionality? This chapter
explores that question. Its key message is that, even though data points may appear
to belong to a high-dimensional space, they actually lie on or near a low-dimensional
manifold and, as such, have a low \emph{intrinsic} dimensionality. This chapter
then formalizes the notion of intrinsic dimensionality and presents a mathematical
framework that will be useful in analyses in future chapters.
}

\section{High-Dimensional Data and Low-Dimensional Manifolds}

We talked a lot about the difficulties of answering $\epsilon$-approximate top-$k$
questions in high dimensions. We said, in certain situations,
the question itself becomes meaningless and
retrieval falls apart. For MIPS, in particular,
we argued in Theorem~\ref{theorem:instability:orthogonality-random-vectors}
that points become nearly orthogonal almost surely
as the number of dimensions increases. But how concerned should we
be, especially given the ever-increasing dimensionality of
vector representations of data? Do our data points really
live in such extremely high-dimensional spaces? Are all the dimensions
necessary to preserving the structure of our data or do our data points
have an intrinsically smaller dimensionality?

The answer to these questions is sometimes obvious.
If a set of points in $\mathbb{R}^d$ lie strictly in a flat subspace 
$\mathbb{R}^{d_\circ}$ with $d_\circ < d$, then one can simply drop
the ``unused'' dimensions---perhaps after a rotation.
This could happen if a pair of coordinates are correlated, for instance.
No matter what query vector we are performing retrieval
for or what distance function we use, the top-$k$ set does not change
whether the unused dimensions are taken into account or the vectors
corrected to lie in $\mathbb{R}^{d_\circ}$.

Other times the answer is intuitive but not so obvious.
When a text document is represented as a sparse vector,
all the document's information is contained entirely in the vector's
non-zero coordinates. The coordinates that are $0$ do not contribute
to the representation of the document in any way.
In a sense then, the intrinsic dimensionality of a collection
of such sparse vectors is in the order of the number of non-zero
coordinates, rather than the nominal dimensionality of the space
the points lie in.

\begin{svgraybox}
It appears then that there are instances where a collection of points
have a superficially large number of dimensions, $d$, but that, in fact,
the points lie in a lower-dimensional space with dimensionality $d_\circ$.
We call $d_\circ$ the \emph{intrinsic dimensionality} of the point set.
\end{svgraybox}

This situation, where the intrinsic dimensionality of data is lower than
that of the space, arises more commonly than one imagines.
In fact, so common is this phenomenon that in statistical learning theory,
there are special classes of algorithms~\citep{ma2012manifold}
designed for data collections that lie on
or near a low-dimensional submanifold of $\mathbb{R}^d$
despite their apparent arbitrarily high-dimensional representations.

In the context of vector retrieval, too, the concept of
intrinsic dimensionality often plays an important role.
Knowing that data points have a low intrinsic dimensionality means
we may be able to reduce dimensionality without (substantially) losing
the geometric structure of the data, including interpoint distances.
But more importantly, we can design algorithms specifically for
data with low intrinsic dimensionality, as we will see in later chapters.
In our analysis of many of these algorithms, too, we often resort
to this property to derive meaningful bounds and make assertions about
their performance.

Doing so, however, requires that we formalize the notion of intrinsic dimensionality.
We often do not have a characterization of the submanifold itself,
so we need an alternate way of characterizing the low-dimensional
structure of our data points. In the remainder of this chapter,
we present two common (and related) definitions of intrinsic dimensionality
that will be useful in subsequent chapters.

\section{Doubling Measure and Expansion Rate}
\label{section:intrinsic-dimensionality:doubling-measure}

\cite{karger2002growth-restricted-metrics} characterize intrinsic dimensionality
as the growth or \emph{expansion} rate of a point set. To understand what that means
intuitively, place yourself somewhere in the data collection, draw a ball around
yourself, and count how many data points are in that ball. Now expand the radius
of this ball by a factor $2$, and count again. The count of data points in a
``growth-restricted'' point set should increase \emph{smoothly}, rather than suddenly,
as we make this ball larger.

\begin{svgraybox}
In other words, data points ``come into view,'' as~\cite{karger2002growth-restricted-metrics}
put it, at a constant rate as we expand our view, regardless of where we are located.
We will not encounter massive holes in the space where there are no data points,
followed abruptly by a region where a large number of vectors are concentrated.
\end{svgraybox}

The formal definition is not far from the intuitive description above.
In fact, expansion rate as defined by~\cite{karger2002growth-restricted-metrics}
is an instance of the following more general definition of a \emph{doubling measure},
where the measure $\mu$ is the counting measure over a collection of points $\mathcal{X}$.

\begin{definition}
    \label{definition:doubling-measure}
    A distribution $\mu$ on $\mathbb{R}^d$ is a \emph{doubling measure} if there is a constant $d_\circ$
    such that, for any $r > 0$ and $x \in \mathbb{R}^d$, $\mu(B(x, 2r)) \leq 2^{d_\circ} \mu(B(x, r))$.
    The constant $d_\circ$ is said to be the \emph{expansion rate} of the distribution.
\end{definition}

One can think of the expansion rate $d_\circ$ as a dimension of sorts.
In fact, as we will see later, several
works~\citep{dasgupta2015rptrees,karger2002growth-restricted-metrics,covertrees}
use this notion of intrinsic dimensionality to design algorithms for
top-$k$ retrieval or utilize it to derive performance guarantees for
vector collections that are drawn from a doubling measure.
That is the main reason we review this definition of intrinsic dimensionality
in this chapter.

While the expansion rate is a reasonable way of describing the structure of a set of
points, it is unfortunately not a stable indicator. It can suddenly blow up,
for example, by the addition of a single point to the set. As a concrete example,
consider the set of integers between $\lvert r \rvert$ and $\lvert 2r \rvert$
for any arbitrary value of $r$: $\mathcal{X} = \{ u \in \mathbb{Z} \;|\; r < \lvert u \rvert < 2r \}$.
The expansion rate of the resulting set is constant because no matter
which point we choose as the center of our ball, and regardless of our choice
of radius, doubling the radius brings points into view at a constant rate.

What happens if we added the origin to the set, so that our set becomes $\{ 0 \} \cup \mathcal{X}$?
If we chose $0$ as the center of the ball, and set its radius to $r$,
we have a single point in the resulting ball. The moment we double $r$,
the resulting ball will contain the entire set! In other words, the expansion rate
of the updated set is $\log m$ (where $m = \lvert \mathcal{X} \rvert$).

It is easy to argue that a subset of a set with bounded expansion rate
does not necessarily have a bounded expansion rate itself.
This unstable behavior is less than ideal, which is why a more robust notion
of intrinsic dimensionality has been developed.
We will introduce that next.

\section{Doubling Dimension}
\label{section:intrinsic-dimensionality:doubling-dimension}

Another idea to formalize intrinsic dimensionality that has worked well
in algorithmic design and anlysis is the \emph{doubling dimension}. It was introduced
by~\cite{gupta2003doublingDimension} but is closely related to the
Assouad dimension~\citep{Assouad1983}. It is defined as follows.

\begin{definition}
    \label{definition:doubling-dimension}
    A set $\mathcal{X} \subset \mathbb{R}^d$ is said to have doubling dimension $d_\circ$ if
    $B(\cdot, 2r) \cap \mathcal{X}$, the intersection of any ball of radius $2r$ with the set,
    can be covered by at most $2^{d_\circ}$ balls of radius $r$.
\end{definition}

The base $2$ in the definition above can be replaced with any other
constant $k$: The doubling dimension of $\mathcal{X}$ is $d_\circ$ if
the intersection of any ball of radius $r$ with the set
can be covered by $\mathcal{O}(k^{d_\circ})$ balls of radius $r/k$.
Furthermore, the definition can be easily extended to any metric space,
not just $\mathbb{R}^d$ with the Euclidean norm.

The doubling dimension is a different notion from the expansion rate as defined in
Definition~\ref{definition:doubling-measure}. The two, however, are in some sense related,
as the following lemma shows.

\begin{lemma}
    \label{lemma:intrinsic-dimensionality:doubling-relationship}
    The doubling dimension, $d_\circ$ of any finite metric $(X, \delta)$
    is bounded above by its expansion rate, $d_\circ^{\textsc{kr}}$ times $4$:
    $d_\circ \leq 4 d_\circ^{\textsc{kr}}$.
\end{lemma}
\begin{proof}
    Fix a ball $B(u, 2r)$ and let $S$ be its $r$-net. That is,
    $S \subset X$, the distance
    between any two points in $S$ is at least $r$, and
    $\mathcal{X} \subseteq \bigcup_{u \in S} B(u, r)$. We have that:
    \begin{equation*}
        B(u, 2r) \subset \bigcup_{v \in S} B(v, r) \subset B(u, 4r).
    \end{equation*}
    By definition of the expansion rate, for every $v \in S$:
    \begin{equation*}
        \big\lvert B(u, 4r) \big\rvert \leq \big\lvert B(v, 8r) \big\rvert
        \leq 2^{4 d_\circ^{\textsc{kr}}} \big\lvert B(v, \frac{r}{2}) \big\rvert.
    \end{equation*}
    Because the balls $B(v, r/2)$ for all $v \in S$ are disjoint, it follows that
    $\lvert S \rvert \leq 2^{4 d_\circ^{\textsc{kr}}}$, so that
    $2^{4 d_\circ^{\textsc{kr}}}$ many balls of radius $r$ cover $B(u, 2r)$.
    That concludes the proof.
\end{proof}

\begin{svgraybox}
The doubling dimension and expansion rate both quantify the intrinsic dimensionality of a point set.
But Lemma~\ref{lemma:intrinsic-dimensionality:doubling-relationship}
shows that, the class of doubling metrics (i.e., metric spaces with a constant doubling dimension)
contains the class of metrics with a bounded expansion rate.
\end{svgraybox}

The converse of the above lemma is not true. In other words,
there are sets that have a bounded doubling dimension, but whose
expansion rate is unbounded. The set,
$\mathcal{X} = \{ 0 \} \cup \{ u \in \mathbb{Z} \;|\; r < \lvert u \rvert < 2r \}$,
from the previous section is one example where this happens.
From our discussion above, its expansion rate is $\log \lvert \mathcal{X} \rvert$.
It is easy to see that the doubling dimension of this set, however, is constant.

\subsection{Properties of the Doubling Dimension}

It is helpful to go over a few concrete examples of point sets with bounded
doubling dimension in order to understand a few properties of this definition of
intrinsic dimensionality.
We will start with a simple example: a line segment in $\mathbb{R}^d$ with the Euclidean
norm.

If the set
$\mathcal{X}$ is a line segment, then its intersection with a ball of radius $r$
is itself a line segment. Clearly, the intersection set can be covered with two
balls of radius $r/2$. Therefore, the doubling dimension $d_\circ$ of $\mathcal{X}$
is $1$.

We can extend that result to any affine set in $\mathbb{R}^d$ to obtain the following
property:

\begin{lemma}
    \label{lemma:intrinsic-dimensionality:k-dimensional-flat}
    A $k$-dimensional flat in $\mathbb{R}^d$ has doubling dimension $\mathcal{O}(k)$.
\end{lemma}
\begin{proof}
    The intersection of a ball in $\mathbb{R}^d$ and a $k$-dimensional flat is a ball in
    $\mathbb{R}^k$.
    It is a well-known result that the size of an $\epsilon$-net of a 
    unit ball in $\mathbb{R}^k$ is at most $(C/\epsilon)^k$ for some small
    constant $C$. As such, a ball of radius $r$ can be covered with $2^{\mathcal{O}(k)}$
    balls of radius $r/2$, implying the claim.
\end{proof}

The lemma above tells us that the doubling dimension of a set in the Euclidean space
is at most some constant factor larger than the natural dimension of the space;
note that this was not the case for the expansion rate.
Another important property that speaks to the stability of the doubling dimension
is the following, which is trivially true:

\begin{lemma}
    Any subset of a set with doubling dimension $d_\circ$ itself has doubling dimension $d_\circ$.
\end{lemma}

The doubling dimension is also robust under the addition of points to the set,
as the following result shows.

\begin{lemma}
    \label{lemma:intrinsic-dimensionality:union}
    Suppose sets $\mathcal{X}_i$ for $i \in [n]$ each have doubling dimension $d_\circ$.
    Then their union has doubling dimension at most $d_\circ + \log n$.
\end{lemma}
\begin{proof}
    For any ball $B$ of radius $r$, $B \cap \mathcal{X}_i$ can be covered with
    $2^{d_\circ}$ balls of half the radius. As such, at most $n2^{d_\circ}$ balls
    of radius $r/2$ are needed to cover the union. The doubling dimension of the union
    is therefore $d_\circ + \log n$.
\end{proof}

One consequence of the previous two lemmas is the following statement concerning sparse
vectors:

\begin{lemma}
    \label{lemma:intrinsic-dimensionality:sparse}
    Suppose that $\mathcal{X} \subset \mathbb{R}^d$ is a collection of sparse vectors,
    each having at most $n$ non-zero coordinates. Then the doubling dimension of $\mathcal{X}$
    is at most $Ck + k \log d$ for some constant $C$.
\end{lemma}
\begin{proof}
    $\mathcal{X}$ is the union of $\binom{d}{n} \leq d^n$ $n$-dimensional flats.
    Each of these flats has doubling dimension $Ck$ for some universal constant $C$,
    by Lemma~\ref{lemma:intrinsic-dimensionality:k-dimensional-flat}.
    By the application of Lemma~\ref{lemma:intrinsic-dimensionality:union},
    we get that the doubling dimension of $\mathcal{X}$ is at most $Cn + n \log d$.
\end{proof}

\begin{svgraybox}
    Lemma~\ref{lemma:intrinsic-dimensionality:sparse} states that collections of
    sparse vectors in the Euclidean space are naturally described by the doubling dimension. 
\end{svgraybox}

\bibliographystyle{abbrvnat}
\bibliography{biblio}

\begin{partbacktext}
\part{Retrieval Algorithms}
\end{partbacktext}

\chapter{Branch-and-Bound Algorithms}
\label{chapter:branch-and-bound}

\abstract{
One of the earliest approaches to the top-$k$ retrieval problem
is to partition the vector space recursively into smaller regions
and, each time we do so, make note of their geometry.
During search, we eliminate the regions
whose shape indicates they cannot contain or overlap with the solution set.
This chapter covers algorithms that embody this approach
and discusses their exact and approximate variants.
}

\section{Intuition}
\label{section:branch-and-bound:intuition}

Suppose there was some way to split a collection $\mathcal{X}$ into
two sub-collections, $\mathcal{X}_l$ and $\mathcal{X}_r$, such that
$\mathcal{X} = \mathcal{X}_l \cup \mathcal{X}_r$ and that the two
sub-collections have roughly the same size. In general, we can relax the splitting
criterion so the two sub-collections are not necessarily partitions; that is,
we may have $\mathcal{X}_l \cap \mathcal{X}_r \neq \emptyset$. We may also
split the collection into more than two sub-collections. For the moment,
though, assume we have two sub-collections that do not overlap.

Suppose further that, we could geometrically characterize \emph{exactly} the regions that contain
$\mathcal{X}_l$ and $\mathcal{X}_r$. For example, when $\mathcal{X}_l \cap \mathcal{X}_r = \emptyset$,
these regions partition the space and may therefore be characterized by a separating hyperplane.
Call these regions $\mathcal{R}_l$ and $\mathcal{R}_r$, respectively.
The separating hyperplane forms a \emph{decision boundary} that helps us determine
if a vector falls into $\mathcal{R}_l$ or $\mathcal{R}_r$.

In effect, we have created a binary tree of depth $1$ where the root node has a decision boundary
and each of the two leaves contains data points that fall into its region.
This is illustrated in Figure~\subref*{figure:branch-and-bound:motivation:single-split}.

Now suppose we have a query point $q$ somewhere in the space and that we are interested
in finding the top-$1$ data point with respect to a proper distance function $\delta(\cdot, \cdot)$.
$q$ falls either in $\mathcal{R}_l$ or $\mathcal{R}_r$;
suppose it is in $\mathcal{R}_l$. We determine that by evaluating the decision boundary in the 
root of the tree and navigating to the appropriate leaf.
Now, we solve the exact top-$1$ retrieval problem over $\mathcal{X}_l$
to obtain the optimal point in that region $u_l^\ast$, then make a note of this minimum distance obtained,
$\delta(q, u_l^\ast)$.

\begin{figure}[t]
    \centering
    \subfloat[]{
        \label{figure:branch-and-bound:motivation:single-split}
        \includegraphics[width=0.7\linewidth]{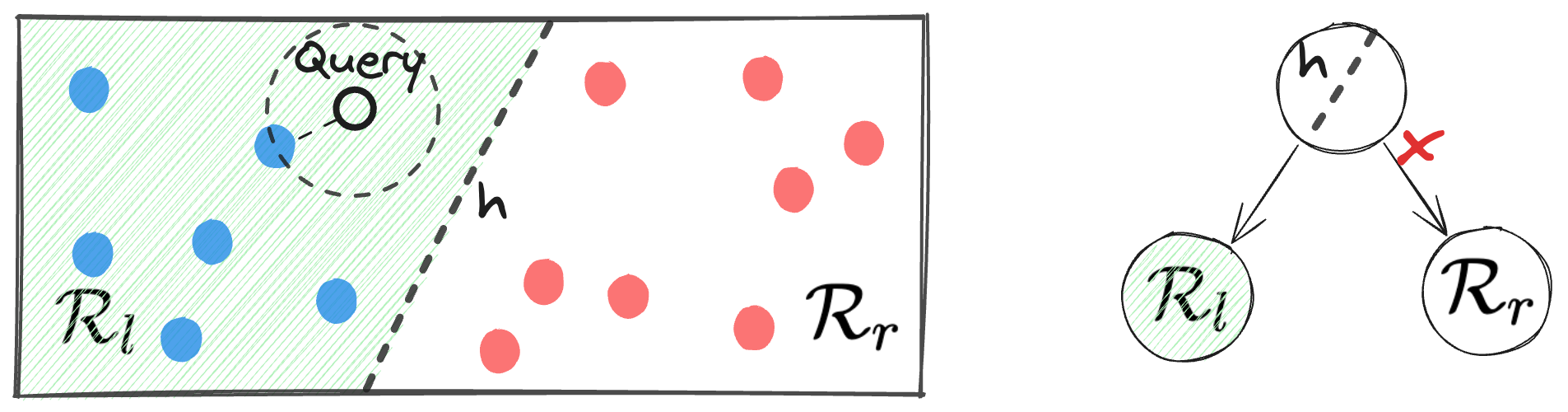}
    }

    \subfloat[]{
        \label{figure:branch-and-bound:motivation:multiple-split}
        \includegraphics[width=0.7\linewidth]{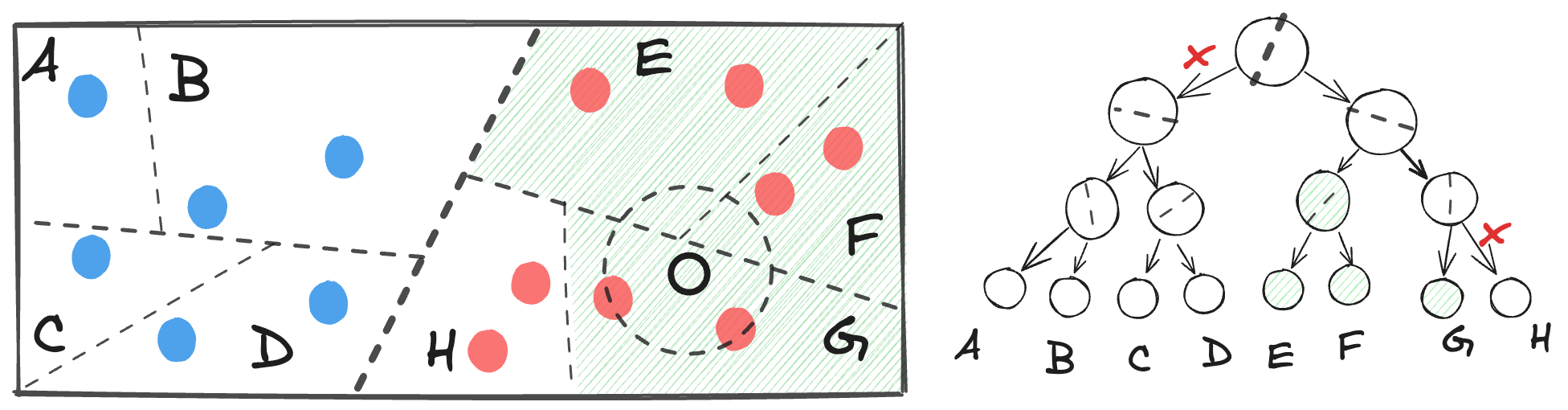}
    }
    \caption{Illustration of a general branch-and-bound method on a toy collection in $\mathbb{R}^2$.
    In (a), $\mathcal{R}_l$ and $\mathcal{R}_r$ are separated by the dashed line $h$.
    The distance between query $q$ and the closest vector in $\mathcal{R}_l$ is less than the distance between $q$ and
    $h$. As such, we do not need to search for the top-$1$ vector
    over the points in $\mathcal{R}_r$, so that the right branch of the tree is pruned.
    In (b), the regions are recursively split until each
    terminal region contains at most two data points.
    We then find the distance between $q$ and the data points in the region that contains $q$, $G$.
    If the ball around $q$ with this distance as its radius does not intersect a region, we can safely prune
    that region---regions that are not shaded in the figure.
    Otherwise, we may have to search it during the certification process.}
    \label{figure:branch-and-bound:motivation}
\end{figure}

At this point, if it turns out that $\delta(q, u_l^\ast) < \delta(q, \mathcal{R}_r)$\footnote{
The distance between a point $u$ and a set $\mathcal{S}$ is defined as
$\delta(u, \mathcal{S}) = \inf_{v \in \mathcal{S}} \delta(u, v)$.}
then we have found the optimal point and do not need to search the data points in $\mathcal{X}_r$ at all!
That is because, the $\delta$-ball\footnote{The ball centered at $u$ with radius $r$ with respect to metric $\delta$
is $\{ x \;|\; \delta(u, x) \leq r \}$.}
centered at $q$ with radius $\delta(q, u_l^\ast)$ is contained entirely in $\mathcal{R}_l$,
so that no point from $\mathcal{R}_r$ can have a shorter distance to $q$ than $u_l^\ast$.
Refer again to Figure~\subref*{figure:branch-and-bound:motivation:single-split} for an
illustration of this scenario.

If, on the other hand, $\delta(q, u_l^\ast) \geq \delta(q, \mathcal{R}_r)$, then
we proceed to solve the top-$1$ problem over $\mathcal{X}_r$ as well and
compare the solution with $u_l^\ast$ to find the optimal vector.
We can think of the comparison between $\delta(q, u_l^\ast)$ with $\delta(q, \mathcal{R}_r)$
as backtracking to the parent node of $\mathcal{R}_l$ in the equivalent tree---which is the root---and comparing
$\delta(q, u_l^\ast)$ with the distance of $q$ with the decision boundary.
This process of backtracking and deciding to prune a branch or search it
\emph{certifies} that $u_l^\ast$ is indeed optimal, thereby solving the top-$1$ problem exactly.

We can extend the framework above easily by recursively splitting the two sub-collections
and characterizing the regions containing the resulting partitions. This leads to a (balanced)
binary tree where each internal node has a decision boundary---the separating hyperplane of its child regions.
We may stop splitting a node if it has fewer than $m_\circ$ points.
This extension is rendered in Figure~\subref*{figure:branch-and-bound:motivation:multiple-split}.

The retrieval process is the same but needs a little more care:
Let the query $q$ traverse the tree from root to leaf,
where each internal node determines if $q$ belongs to
the ``left'' or ``right'' sub-regions and \emph{routes} $q$ accordingly.
Once we have found the leaf (terminal) region that contains $q$, we find the
candidate vector $u^\ast$, then backtrack and certify that $u^\ast$ is indeed optimal.

During the backtracking, at each internal node, we compare the distance between $q$ and
the current candidate with the distance between $q$ and the region on the other side of
the decision boundary. As before, that comparison results in either pruning a branch or searching it
to find a possibly better candidate. The certification process stops when we find ourselves
back in the root node with no more branches to verify, at which point we have found the optimal solution.

The above is the logic that is at the core of branch-and-bound algorithms for top-$k$
retrieval~\citep{dasgupta2015rptrees,kdtree,ram2019revisiting_kdtree,mtrees,vptrees,liu2004SpillTree,panigrahy2008improved-kdtree,conetrees,xbox-tree}.
The specific instances of this framework differ in terms of how they
\emph{split} a collection and the details of the \emph{certification} process.
We will review key algorithms that belong to this family in the remainder of this chapter.
We emphasize that, most branch-and-bound algorithms only address the NN
problem in the Euclidean space (so that $\delta(u, v) = \lVert u - v \rVert_2$)
or in growth-restricted measures~\citep{karger2002growth-restricted-metrics,clarkson1997,krauthgamer2004navigatingnets}
but where the metric is nonetheless proper.

\section{\texorpdfstring{$k$}{k}-dimensional Trees}
\label{section:branch-and-bound:kd-tree}

The $k$-dimensional Tree or $k$-d Tree~\citep{kdtree} is a special instance of the framework
described above wherein the distance function is Euclidean
and the space is recursively partitioned into hyper-rectangles.
In other words, the decision boundaries in a $k$-d Tree are axis-aligned hyperplanes.

Let us consider its simplest construction for $\mathcal{X} \subset \mathbb{R}^d$.
The root of the tree is a node that represents the entire space, which naturally
contains the entire data collection. Assuming that the size of the collection
is greater than $1$, we follow a simple procedure to split the node:
We select one coordinate axis and partition the collection at the
median of data points along the chosen direction.
The process recurses on each newly-minted node, with nodes at the same depth in the tree using the same
coordinate axis for splitting, and where we go through the coordinates in a round-robin
manner as the tree grows. We stop splitting a node further if it contains a single
data point ($m_\circ = 1$), then mark it as a leaf node.

A few observations that are worth noting.
By choosing the median point to split on, we guarantee that the tree is balanced.
That together with the fact that $m_\circ = 1$ implies that
the depth of the tree is $\log m$ where $m = \lvert \mathcal{X} \rvert$.
Finally, because nodes in each level of the tree split on the same coordinate,
every coordinate is split in $(\log m) /d$ levels. These will become important
in our analysis of the algorithm.

\subsection{Complexity Analysis}
The $k$-d Tree data structure is fairly simple to construct. It is also
efficient: Its space complexity given a set of $m$ vectors is $\Theta(m)$
and its construction time has complexity $\Theta(m \log m)$.\footnote{
The time to construct the tree depends on the complexity of the subroutine
that finds the median of a set of values.}

The search algorithm, however, is not so easy to analyze in general.
\cite{freidman1977kdtree_proof} claimed that the expected search complexity
is $\mathcal{O}(\log m)$ for $m$ data points that are sampled uniformly
from the unit hypercube. While uniformity is an unrealistic assumption,
it is necessary for the analysis of the average case.
On the other hand, no generality is lost by the assumption that vectors
are contained in the hypercube. That is because, we can always scale every data
point by a constant factor into the unit hypercube---a transformation
that does not affect the pairwise distances between vectors.
Let us now discuss the sketch of the proof of their claim.

Let $\delta^\ast = \min_{u \in \mathcal{X}} \lVert q - u \rVert_2$ be the optimal distance to a query $q$.
Consider the ball of radius $\delta^\ast$ centered at $q$ and denote it by $B(q, \delta^\ast)$.
It is easy to see that the number of leaves we may need to visit in order to certify
an initial candidate is upper-bounded by the number of leaf regions (i.e., $d$-dimensional boxes)
that touch $B(q, \delta^\ast)$. That quantity itself is upper-bounded by the number of boxes
that touch the smallest hypercube that contains $B(q, \delta^\ast)$. If we calculated this
number, then we have found an upper-bound on the search complexity.

Following the argument above,~\cite{freidman1977kdtree_proof} show that---with
very specific assumptions on the density of vectors in the space, which do not
necessarily hold in high dimensions---the quantity of interest is upper-bounded
by the following expression:
\begin{equation}
    \label{equation:branch-and-bound:kdtree:original_upperbound}
    \big( 1 + G(d)^{1/d} \big)^d,
\end{equation}
where $G(d)$ is the ratio between the volume of the hypercube that contains $B(q, \delta^\ast)$
and the volume of $B(q, \delta^\ast)$ itself. Because $G(d)$ is independent of $m$,
and because visiting each leaf takes $\mathcal{O}(\log m)$ (i.e., the depth of the tree)
operations, they conclude that the complexity of the algorithm is $\mathcal{O}(\log m)$.

\subsection{Failure in High Dimensions}
The argument above regarding the search time complexity of the algorithm
fails in high dimensions. Let us elaborate why in this section.

Let us accept that the assumptions that enabled the proof above hold
and focus on $G(d)$. The volume of a hypercube in $d$ dimensions
with sides that have length $2\delta^\ast$ is $(2\delta^\ast)^d$. The volume of $B(q, \delta^\ast)$
is $\pi^{d/2} {\delta^\ast}^d / \Gamma(d/2 + 1)$, where $\Gamma$ denotes the Gamma function.
For convenience, suppose that $d$ is even,
so that $\Gamma(d/2 + 1) = (d/2)!$. As such $G(d)$, the ratio between the two volumes,
is:
\begin{equation}
    G(d) = \frac{2^d (d/2)!}{\pi^{d/2}}.
\end{equation}
Plugging this back into Equation~(\ref{equation:branch-and-bound:kdtree:original_upperbound}),
we arrive at:
\begin{align*}
    \big( 1 + G(d)^{1/d} \big)^d &= \big( 1 + \frac{2}{\sqrt{\pi}} (d/2)!^{\frac{1}{d}} \big)^d \\
    &= \mathcal{O}((\frac{2}{\sqrt{\pi}})^d (\frac{d}{2})!) \\
    &= \mathcal{O}((\frac{2}{\sqrt{\pi}})^d d^{\frac{d + 1}{2}}),
\end{align*}
where in the third equality we used Stirling's formula,
which approximates $n!$ as $\sqrt{2 \pi n} (\frac{n}{e})^n$, to expand $(d/2)!$
as follows:
\begin{align*}
    (\frac{d}{2})! &\approx \sqrt{2 \pi d/2} (\frac{d}{2e})^{d/2} \\
    &= \sqrt{\pi} \frac{1}{(2e)^{d/2}} d^{\frac{d + 1}{2}} \\
    &= \mathcal{O}(d^\frac{d + 1}{2}).
\end{align*}

\begin{svgraybox}
The above shows that, the number of leaves that may be visited during the certification
process has, asymptotically, an exponential dependency on $d$. That does not bode well
for high dimensions.
\end{svgraybox}

There is an even simpler argument to make to show that in high dimensions the search algorithm
must visit at least $2^d$ data points during the certification process.
Our argument is as follows. We will show in Lemma~\ref{lemma:branch-and-bound:expected-distance}
that, with high probability, the distance between the query point $q$ and a randomly drawn
data point concentrates sharply on $\sqrt{d}$. This implies that $B(q, \delta^\ast)$ has a radius
that is larger than $1$ with high probability. Noting that the side of the unit hypercube is $2$, it follows that
$B(q, \delta^\ast)$ crosses decision boundaries across every dimension, making it necessary to visit
the corresponding partitions for certification.

Finally, because each level of the tree splits
on a single dimension, the reasoning above means that the certification process must visit
$\Omega(d)$ levels of the tree. As a result, we visit at least $2^{\Omega(d)}$ data points.
Of course, in high dimensions, we often have far fewer than $2^d$ data points,
so that we end up visiting every vector during certification.

\begin{lemma}
    \label{lemma:branch-and-bound:expected-distance}
    The distance $r$ between a randomly chosen point and its nearest neighbor among $m$ points drawn
    uniformly at random from the unit hypercube is $\Theta\big( \sqrt{d} / m^{1/d}\big)$
    with probability at least $1 - \mathcal{O}(1/2^d)$.
\end{lemma}
\begin{proof}
    Consider the ball of radius $r$ in $d$-dimensional unit hypercube with volume $1$. Suppose, for notational
    convenience, that $d$ is even---whether it is odd or even does not change our asymptotic
    conclusions. The volume of this ball is:
    \begin{equation*}
        \frac{\pi^{d/2} r^d}{(d/2)!}.
    \end{equation*}
    Since we have $m$ points in the hypercube, the expected number of points that are contained in
    the ball of radius $r$ is therefore:
    \begin{equation*}
        \frac{\pi^{d/2} r^d}{(d/2)!} m.
    \end{equation*}
    As a result, the radius $r$ for which the ball contains one point in expectation is:
    \begin{align*}
        \frac{\pi^{d/2} r^d}{(d/2)!} m = 1 &\implies r^d = \Theta\big(\frac{1}{m} (\frac{d}{2})!\big)\\
        &\implies r = \Theta\big( \frac{1}{m^{1/d}} (\frac{d}{2})!^{1/d} \big).
    \end{align*}
    Using Stirling's formula and letting $\Theta$ consume the constants and small
    factors completes the claim that $r = \sqrt{d} / m^{1/d}$.

    All that is left is bounding the probability of the event that $r$ takes on
    the above value. For that, consider first the ball of radius $r/2$.
    The probability that this ball contains at least one
    point is at most $1/2^d$. To see this, note that the probability that a single point
    falls into this ball is:
    \begin{equation*}
        \frac{\pi^{d/2} (r/2)^d}{(d/2)!} = \frac{1}{2^d} \underbrace{\frac{\pi^{d/2} r^d}{(d/2)!}}_{1/m}.
    \end{equation*}
    By the Union Bound, the probability that at least one point out of $m$ points
    falls into this ball is at most $m \times 1/(m 2^d) = 1/2^d$.

    Next, consider the ball of radius $2r$. The probability that it contains no points
    at all is at most $(1 - 2^d/m)^m \approx \exp(-2^d) \leq 1/2^d$, where we used the approximation that
    $\big(1 - 1/x\big)^x \approx \exp(-1)$ and the fact that $\exp(-x) \leq 1/x$.
    To see why, it is enough to compute the probability that a single point does not fall into a
    ball of radius $2r$, then by independence we arrive at the joint probability above.
    That probability is $1$ minus the probability that the point falls into the ball,
    which is itself:
    \begin{equation*}
        \frac{\pi^{d/2} (2r)^d}{(d/2)!} = 2^d \underbrace{\frac{\pi^{d/2} r^d}{(d/2)!}}_{1/m},
    \end{equation*}
    hence the total probability $(1 - 2^d/m)^m$.

    We have therefore shown that the probability that the distance of interest is $r$
    is extremely high and in the order of $1 - 1/2^d$, completing the proof.
\end{proof}

\section{Randomized Trees}

As we explained in Section~\ref{section:branch-and-bound:kd-tree},
the search algorithm over a $k$-d Tree ``index'' consists of two operations:
A single root-to-leaf traversal of the tree followed by backtracking to certify
the candidate solution. As the analysis presented in the same section shows, it
is the certification procedure that may need to visit virtually all data points.
It is therefore not surprising that~\cite{liu2004SpillTree} report that,
in their experiments with low-dimensional vector collections (up to $30$ dimensions),
nearly $95\%$ of the search time is spent in the latter phase.

That observation naturally leads to the following question: What if we eliminated
the certification step altogether? In other words, when given a query $q$, the search
algorithm simply finds the cell that contains $q$ in $\mathcal{O}(\log m/m_\circ)$
time (where $m=\lvert \mathcal{X} \rvert$), then returns the solution
from among the $m_\circ$ vectors in that cell---a
strategy~\cite{liu2004SpillTree} call \emph{defeatist} search.

As~\cite{panigrahy2008improved-kdtree} shows for uniformly distributed vectors,
however, the failure probability of the defeatist
method is unacceptably high. That is primarily because, when a query is close to a
decision boundary, the optimal solution may very well be on the other side.
Figure~\subref*{figure:branch-and-bound:randomized:failure} illustrates this phenomenon.
As both the construction and search algorithms are \emph{deterministic},
such a failure scenario is intrinsic to the algorithm and
cannot be corrected once the tree has been constructed. Decision boundaries
are hard and fast rules.

\begin{figure}[t]
    \centering
    \subfloat[]{
        \label{figure:branch-and-bound:randomized:failure}
        \includegraphics[width=0.47\linewidth]{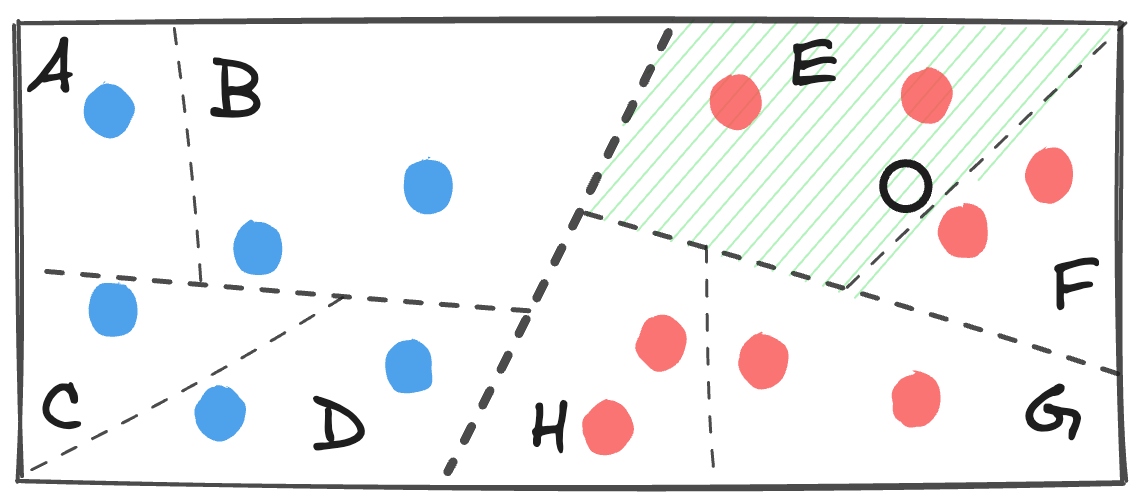}
    }
    \subfloat[]{
        \includegraphics[width=0.47\linewidth]{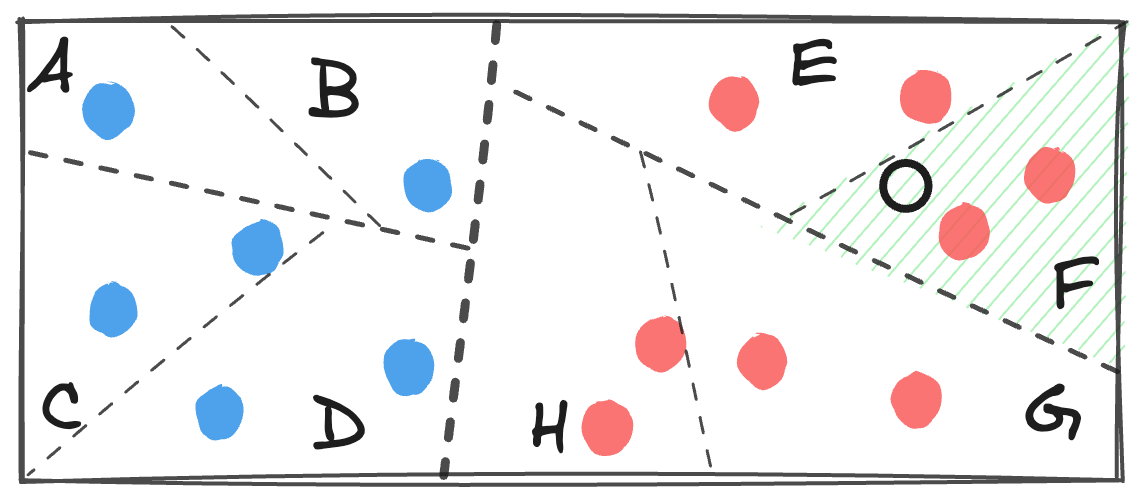}
    }
    \caption{Randomized construction of $k$-d Trees for a fixed collection of vectors (filled circles).
    Decision boundaries take random directions
    and are planted at a randomly-chosen point near the median. Repeating this procedure
    results in multiple ``index'' structures of the vector collection. Performing a ``defeatist'' search
    repeatedly for a given query (the empty circles) then leads to a higher probability of success.}
    \label{figure:branch-and-bound:randomized}
\end{figure}

Would the situation be different if tree construction was a randomized algorithm?
We could, for instance, let a random subset of the data points that are close to
each decision boundary fall into both ``left'' and ``right'' sub-regions as we split an internal node.
As another example, we could place the decision boundaries at randomly chosen points
close to the median, and have them take a randomly chosen direction.
We illustrate the latter in Figure~\ref{figure:branch-and-bound:randomized}.

\begin{svgraybox}
Such randomized decisions mean that, every time we construct a $k$-d Tree, we would obtain a different
index of the data. Furthermore, by building a forest of randomized $k$-d Trees and repeating
the defeatist search algorithm, we may be able to lower the failure probability!
\end{svgraybox}

These, as we will learn in this section, are indeed successful ideas that have
been extensively explored in the literature~\citep{liu2004SpillTree,ram2019revisiting_kdtree,dasgupta2015rptrees}.

\subsection{Randomized Partition Trees}

Recall that a decision boundary in a $k$-d Tree is an axis-aligned hyperplane that is placed at the median point
of the projection of data points onto a coordinate axis.
Consider the following adjustment to that procedure, due originally to~\cite{liu2004SpillTree}
and further refined by~\cite{dasgupta2015rptrees}.
Every time a node whose region is $\mathcal{R}$ is to be split, we first draw a \emph{random direction}
$u$ by sampling a vector from the $d$-dimensional unit sphere, and a scalar $\beta \in [1/4, 3/4]$ uniformly at random.
We then project all data points that are in $\mathcal{R}$ onto $u$,
and obtain the $\beta$-fractile of the projections, $\theta$. $u$ together with $\theta$ form the decision
boundary. We then proceed as before to partition the data points in $\mathcal{R}$, by following the rule
$\langle u, v \rangle \leq \theta$ for every data point $v \in \mathcal{R}$.
A node turns into a leaf if it contains a maximum of $m_\circ$ vectors.

The procedure above gives us what~\cite{dasgupta2015rptrees} call a \emph{Randomized Partition} (RP) Tree.
You have already seen a visual demonstration of two RP Trees in Figure~\ref{figure:branch-and-bound:randomized}.
Notice that, by requiring $u$ to be a standard basis vector, fixing one $u$ per each level of the tree,
and letting $\beta=0.5$, we reduce an RP Tree to the original $k$-d Tree.

What is the probability that a defeatist search over a single RP Tree fails to return the correct nearest neighbor?
\cite{dasgupta2015rptrees} proved that, that probability is related to the following potential function:
\begin{equation}
    \label{equation:branch-and-bound:randomized:potential}
    \Phi(q, \mathcal{X}) = \frac{1}{m} \sum_{i=1}^{m}
    \frac{\lVert q - x^{(\pi_1)} \rVert_2}{\lVert q - x^{(\pi_i)} \rVert_2},
\end{equation}
where $m = \lvert \mathcal{X} \rvert$, $\pi_1$ through $\pi_m$ are indices that
sort data points by increasing distance to the query point $q$, so that $x^{(\pi_1)}$
is the closest data point to $q$.

\begin{svgraybox}
Notice that, a value of $\Phi$ that is close to $1$ implies that nearly
all data points are at the same distance from $q$. In that case, as we saw in Chapter~\ref{chapter:instability},
NN becomes unstable and approximate top-$k$ retrieval becomes meaningless.
When $\Phi$ is closer to $0$, on the other hand, the optimal vector is well-separated
from the rest of the collection.

Intuitively then, $\Phi$ is reflective of the difficulty or stability of the NN problem for a given query point.
It makes sense then, that the probability of failure for $q$ is related to this notion of difficulty
of NN search: intuitively, when the nearest neighbor is far from other vectors, a defeatist search is more likely
to yield the correct solution.
\end{svgraybox}

\subsubsection{A Potential Function to Quantify the Difficulty of NN Search}

Before we state the relationship between the failure probability and the potential function above more concretely,
let us take a detour and understand where the expression for $\Phi$ comes from. All the arguments that we are about
to make, including the lemmas and theorems, come directly from~\cite{dasgupta2015rptrees}, 
though we repeat them here using our adopted notation for completeness. We also present an
expanded proof of the formal results which follows the original proofs but elaborates on some of the steps.

\medskip

Let us start with a simplified setup where $\mathcal{X}$ consists of just two vectors $x$ and $y$.
Suppose that for a query point $q$, $\lVert q - x \rVert_2 \leq \lVert q - y \rVert_2$.
It turns out that, if we chose a random direction $u$ and projected $x$ and $y$ onto it,
then the probability that the projection of $y$ onto $u$ lands somewhere in between the projections of
$q$ and $x$ onto $u$ is a function of the potential function of Equation~(\ref{equation:branch-and-bound:randomized:potential}).
The following lemma formalizes this relationship.

\begin{lemma}
    \label{lemma:branch-and-bound:randomized:two-points}
    Suppose $q, x, y \in \mathbb{R}^d$ and $\lVert q - x \rVert_2 \leq \lVert q - y \rVert_2$.
    Let $U \in \mathbb{S}^{d-1}$ be a random direction and define
    $\overset{\angle}{v} = \langle U, v \rangle$.
    The probability that $\overset{\angle}{y}$ is between $\overset{\angle}{q}$ and $\overset{\angle}{x}$
    is:
    \begin{align*}
        \probability\big[
            \big( \overset{\angle}{q} \leq \overset{\angle}{y} \leq \overset{\angle}{x} \big) \lor
            \big( \overset{\angle}{x} \leq \overset{\angle}{y} \leq \overset{\angle}{q} \big) \big] &= \\
        \frac{1}{\pi} \arcsin \Biggl( 2 \Phi(q, \{ x, y\})
            & \sqrt{1 - \Big(
                \frac{\langle q - x, y - x \rangle}{\lVert q - x \rVert_2 \lVert y - x \rVert_2}
            \Big)^2}
        \Biggr).
    \end{align*}
\end{lemma}
\begin{proof}
Assume, without loss of generality, that $U$ is sampled from a $d$-dimensional standard Normal distribution: $U \sim \mathcal{N}(\mathbf{0}, I_d)$.
That assumption is inconsequential because normalizing $U$ by its $L_2$ norm gives a vector that lies on $\mathbb{S}^{d-1}$ as required.
But because the norm of $U$ does not affect the argument we need not explicitly perform the normalization.

Suppose further that we translate all vectors by $q$,
and redefine $q \triangleq \mathbf{0}$, $x \triangleq x - q$, and $y \triangleq y - q$.
We then rotate the vectors so that $x = \lVert x \rVert_2 e_1$, where $e_1$ is the first standard basis vector.
Neither the translation nor the rotation affects pairwise distances, and as such, no generality is lost due to these transformations.

Given this arrangement of vectors, it will be convenient to write $U = (U_1, U_{\setminus 1})$ and $y = (y_1, y_{\setminus 1})$
so as to make explicit the first coordinate of each vector (denoted by subscript $1$)
and the remaining coordinates (denoted by subscript $\setminus 1$).

It is safe to assume that $y_{\setminus 1} \neq \mathbf{0}$. The reason is that, if that were not the case,
the two vectors $x$ and $y$ have an intrinsic dimensionality of $1$ and are thus on a line.
In that case, no matter which direction $U$ we choose, $\overset{\angle}{y}$ will not fall between
$\overset{\angle}{x}$ and $\overset{\angle}{q} = \mathbf{0}$.

We can now write the probability of the event of interest as follows:
\begin{align*}
    \probability&\Big[
        \underbrace{
            \big( \overset{\angle}{q} \leq \overset{\angle}{y} \leq \overset{\angle}{x} \big) \lor
            \big( \overset{\angle}{x} \leq \overset{\angle}{y} \leq \overset{\angle}{q} \big) }_{E}\Big] = \\
    &\probability\Big[
            \big( 0 \leq \langle U, y \rangle \leq \lVert x \rVert_2 U_1 \big) \lor
            \big( \lVert x \rVert_2 U_1 \leq \langle U, y \rangle \leq 0 \big) \Big].
\end{align*}
By expanding $\langle U, y \rangle = y_1 U_1 + \langle U_{\setminus 1}, y_{\setminus 1} \rangle$, it becomes clear that the expression
above measures the probability that $\langle U_{\setminus 1}, y_{\setminus 1} \rangle$ falls in the interval
$\big(-y_1 \lvert U_1 \rvert, (\lVert x \rVert_2 - y_1)\lvert U_1 \rvert \big)$ when $U_1 \geq 0$ or
$\big(-(\lVert x \rVert_2 - y_1) \lvert U_1 \rvert, y_1 \lvert U_1 \rvert\big)$ otherwise.
As such $\probability[E]$ can be rewritten as follows:
\begin{align*}
    \probability[E] &=
    \probability\Big[
            -y_1 \lvert U_1 \rvert \leq \langle U_{\setminus 1}, y_{\setminus 1} \rangle \leq (\lVert x \rVert_2 - y_1)\lvert U_1 \rvert \;\big|\; U_1 \geq 0 \Big] \probability\big[ U_1 \geq 0 \big] \\
    &+ \probability\Big[-(\lVert x \rVert_2 - y_1)\lvert U_1 \rvert \leq \langle U_{\setminus 1}, y_{\setminus 1} \rangle \leq y_1 \lvert U_1 \rvert \;\big|\; U_1 < 0 \Big] \probability\big[ U_1 < 0 \big].
\end{align*}
First, note that $U_1$ is independent of $U_{\setminus 1}$ given that they are sampled from $\mathcal{N}(\mathbf{0}, I_d)$.
Second, observe that $\langle U_{\setminus 1}, y_{\setminus 1} \rangle$ is distributed as $\mathcal{N}(\mathbf{0}, \lVert y_{\setminus 1} \rVert_2^2)$,
which is symmetric, so that the two intervals have the same probability mass.
These two observations simplify the expression above, so that $\probability[E]$ becomes:
\begin{align*}
    \probability[ E ] &= \probability\Big[ -y_1 \lvert U_1 \rvert \leq \langle U_{\setminus 1}, y_{\setminus 1} \rangle \leq (\lVert x \rVert_2 - y_1)\lvert U_1 \rvert \Big] \\
    &= \probability\Big[ -y_1 \lvert Z \rvert \leq \lVert y_{\setminus 1} \rVert_2 Z^\prime \leq (\lVert x \rVert_2 - y_1)\lvert Z \rvert \Big] \\
    &= \probability\Bigg[ \frac{\lvert Z' \rvert}{\lvert Z \rvert} \in \Big( -\frac{y_1}{\lVert y_{\setminus 1} \rVert_2}, \frac{\lVert x \rVert_2 - y_1}{\lVert y_{\setminus 1} \rVert_2} \Big) \Bigg],
\end{align*}
where $Z$ and $Z^\prime$ are independent random variables drawn from $\mathcal{N}(0, 1)$.

Using the fact that the ratio of two independent Gaussian random variables follows a standard Cauchy distribution,
we can calculate $\probability[ E ]$ as follows:
\begin{align*}
    \probability[ E ] &= \int_{-y_1/\lVert y_{\setminus 1} \rVert_2}^{\big(\lVert x \rVert_2 - y_1\big)/\lVert y_{\setminus 1} \rVert_2} \frac{d \omega}{\pi (1 + \omega^2)} \\
    &= \frac{1}{\pi} \Bigg[ \arctan\Big( \frac{\lVert x \rVert_2 - y_1}{\lVert y_{\setminus 1} \rVert_2} \Big) - \arctan \Big( -\frac{y_1}{\lVert y_{\setminus 1} \rVert_2} \Big) \Bigg] \\
    &= \frac{1}{\pi} \arctan\Bigg( \frac{\lVert x \rVert_2 \lVert y_{\setminus 1} \rVert_2}{\lVert y \rVert_2^2 - y_1 \lVert x \rVert_2} \Bigg) \\
    &= \frac{1}{\pi} \arcsin\Bigg( \frac{\lVert x \rVert_2}{\lVert y \rVert_2} \sqrt{\frac{\lVert y \rVert_2^2 - y_1^2}{\lVert y \rVert_2^2 - 2y_1 \lVert x \rVert_2 + \lVert x \rVert_2^2}} \Bigg).
\end{align*}
In the third equality, we used the fact that $\arctan a + \arctan b = \arctan (a + b)/(1 - ab)$, and in the fourth equality
we used the identity $\arctan a = \arcsin a / \sqrt{1 + a^2}$. Substituting $y_1 = \langle y, x \rangle / \lVert x \rVert_2$
and noting that $x$ and $y$ have been shifted by $q$ completes the proof.
\end{proof}

\begin{corollary}
\label{lemma:branch-and-bound:randomized:two-points-corollary}
    In the same configuration as in Lemma~\ref{lemma:branch-and-bound:randomized:two-points}:
    \begin{align*}
        \frac{2}{\pi} \Phi(q, \{ x, y\})
            & \sqrt{1 - \Big(
                \frac{\langle q - x, y - x \rangle}{\lVert q - x \rVert_2 \lVert y - x \rVert_2}
            \Big)^2} \leq& \\
        &\probability\Big[
            \big( \overset{\angle}{q} \leq \overset{\angle}{y} \leq \overset{\angle}{x} \big) \lor
            \big( \overset{\angle}{x} \leq \overset{\angle}{y} \leq \overset{\angle}{q} \big) \Big] \leq 
            \Phi(q, \{ x, y\})
    \end{align*}
\end{corollary}
\begin{proof}
Applying the inequality $\theta \geq \sin \theta \geq 2 \theta/\pi$ for $0 \leq \theta \leq \pi/2$ to Lemma~\ref{lemma:branch-and-bound:randomized:two-points}
implies the claim.
\end{proof}

Now that we have examined the case of $\mathcal{X} = \{ x, y\}$, it is easy to extend
the result to a configuration of $m$ vectors.

\begin{theorem}
    \label{theorem:branch-and-bound:randomized:m-points}
    Suppose $q \in \mathbb{R}^d$, $\mathcal{X} \subset \mathbb{R}^d$ is a set of $m$ vectors,
    and $x^\ast \in \mathcal{X}$ the nearest neighbor of $q$.
    Let $U \in \mathbb{S}^{d-1}$ be a random direction, define $\overset{\angle}{v} = \langle U, v \rangle$,
    and let $\overset{\angle}{\mathcal{X}} = \{ \overset{\angle}{x} \;|\; x \in \mathcal{X} \}$.
    Then:
    \begin{equation*}
        \mathbb{E}_{U} \big[ \textit{fraction of $\overset{\angle}{\mathcal{X}}$ that is between
        $\overset{\angle}{q}$ and $\overset{\angle}{x^\ast}$} \big] \leq \frac{1}{2} \Phi(q, \mathcal{X}).
    \end{equation*}
\end{theorem}
\begin{proof}
    Let $\pi_1$ through $\pi_m$ be indices that order the elements of $\mathcal{X}$ by increasing distance to $q$,
    so that $x^\ast = x^{(\pi_1)}$. Denote by $Z_i$ the event that $\langle U, x^{(\pi_i)} \rangle$ falls between $\overset{\angle}{x^\ast}$
    and $\overset{\angle}{q}$. By Corollary~\ref{lemma:branch-and-bound:randomized:two-points-corollary}:
    \begin{equation*}
        \probability[Z_i] \leq \frac{1}{2} \frac{\lVert q - x^{(\pi_1)} \rVert_2}{\lVert q - x^{(\pi_i)} \rVert_2}.
    \end{equation*}
    We can now write the expectation of interest as follows:
    \begin{equation*}
        \sum_{i=2}^{m} \frac{1}{m} \probability[ Z_i ] \leq \frac{1}{2} \Phi(q, \mathcal{X}).
    \end{equation*}
\end{proof}.

\begin{corollary}
    \label{corollary:branch-and-bound:randomized:m-points}
    Under the assumptions of Theorem~\ref{theorem:branch-and-bound:randomized:m-points}, for any $\alpha \in (0, 1)$
    and any $s$-subset $S$ of $\mathcal{X}$ that contains $x^\ast$:
    \begin{equation*}
        \probability\big[ \textit{at least $\alpha$ fraction of $\overset{\angle}{S}$ is between
        $\overset{\angle}{q}$ and $\overset{\angle}{x^\ast}$} \big] \leq \frac{1}{2\alpha} \Phi_s(q, \mathcal{X}),
    \end{equation*}
    where:
    \begin{equation*}
        \Phi_s(q, \mathcal{X}) = \frac{1}{s} \sum_{i=1}^{s}
    \frac{\lVert q - x^{(\pi_1)} \rVert_2}{\lVert q - x^{(\pi_i)} \rVert_2},
    \end{equation*}
    and $\pi_1$ through $\pi_s$ are indices of the $s$ vectors in $\mathcal{X}$ that are closest to $q$, ordered
    by increasing distance.
\end{corollary}
\begin{proof}
    Apply Theorem~\ref{theorem:branch-and-bound:randomized:m-points} to the set $S$ to obtain:
    \begin{equation*}
        \mathbb{E}\big[ \textit{fraction of $\overset{\angle}{S}$ that is between
        $\overset{\angle}{q}$ and $\overset{\angle}{x^\ast}$} \big] \leq \frac{1}{2} \Phi(q, S) \leq \frac{1}{2} \Phi_s(q, \mathcal{X}).
    \end{equation*}
    Using Markov's inequality (i.e., $\probability[Z > \alpha] \leq \mathbb{E}[Z]/\alpha$) completes the proof.
\end{proof}

The above is where the potential function of Equation~(\ref{equation:branch-and-bound:randomized:potential}) first emerges
in its complete form for an arbitrary collection of vectors and its subsets.
As we see, $\Phi$ bounds the expected number of vectors whose projection onto
a random direction $U$ falls between a query point and its nearest neighbor.

\begin{svgraybox}
The reason this expected value is important (which subsequently justifies the importance of
$\Phi$) has to do with the fact that decision boundaries are planted at some $\beta$-fractile point of
the projections. As such, a bound on the number of points that fall between
$q$ and its nearest neighbor serves as a tool to bound the odds that the decision boundary may separate $q$ from
its nearest neighbor, which is the failure mode we wish to quantify.
\end{svgraybox}

\subsubsection{Probability of Failure}
We are now ready to use Theorem~\ref{theorem:branch-and-bound:randomized:m-points} and
Corollary~\ref{corollary:branch-and-bound:randomized:m-points} to derive the failure probability
of the defeatist search over an RP Tree.
To that end, notice that, the path from the root to a leaf is a sequence of $\log_{1/\beta} (m/m_\circ) $
independent decisions that involve randomly projected data points. So if we were able to bound the failure probability
of a single node, we can apply the union bound and obtain a bound on the failure probability of the tree.
That is the intuition that leads to the following result.

\begin{theorem}
    The probability that an RP Tree built for collection $\mathcal{X}$ of $m$ vectors fails to find the nearest neighbor
    of a query $q$ is at most:
    \begin{equation*}
        \sum_{l = 0}^{\ell} \Phi_{\beta^l m} \ln \frac{2e}{\Phi_{\beta^l m}},
    \end{equation*}
    with $\beta = 3/4$ and $\ell = \log_{1/\beta} \big( m / m_\circ \big)$, and where we use the shorthand
    $\Phi_s$ for $\Phi_s(q, \mathcal{X})$.
\end{theorem}
\begin{proof}
    Consider an internal node of the RP Tree that contains $q$ and $s$ data points including $x^\ast$, the nearest neighbor of $q$.
    If the decision boundary at this node separates $q$ from $x^\ast$, then the defeatist search will fail.
    We therefore seek to quantify the probability of that event.

    Denote by $F$ the fraction of the $s$ vectors that, once projected onto the random direction $U$ associated with the node,
    fall between $q$ and $x^\ast$. Recall that, the split threshold associated with the node is drawn uniformly from an interval
    of mass $1/2$. As such, the probability that $q$ is separated from $x^\ast$ is at most
    $F/(1/2)$. By integrating over $F$, we obtain:
    \begin{align*}
        \probability\big[ \textit{$q$ is separated from $x^\ast$} \big] &\leq \int_0^1 \probability\big[ F = f\big] \frac{f}{1/2} df \\
        &= 2 \int_0^1 \probability\big[ F > f \big] df \\
        &\leq 2 \int_0^1 \min \Big( 1, \frac{\Phi_s}{2f} \Big) df \\
        &= 2 \int_0^{\Phi_s / 2} df + 2 \int_{\Phi_s / 2}^ 1 \frac{\Phi_s}{2f} df \\
        &= \Phi_s \ln \frac{2e}{\Phi_s}.
    \end{align*}
    The first equality uses the definition of expectation for a positive random variable,
    while the second inequality uses Corollary~\ref{corollary:branch-and-bound:randomized:m-points}.
    Applying the union bound to a path from root to leaf,
    and noting that the size of the collection that falls into each node
    drops geometrically per level by a factor of at least $3/4$ completes the proof.
\end{proof}

We are thus able to express the failure probability as a function of $\Phi$,
a quantity that is defined for a particular $q$ and a concrete collection of vectors.
If we have a model of the data distribution, it may be possible to
state more general bounds by bounding $\Phi$ itself.~\cite{dasgupta2015rptrees}
demonstrate examples of this for two practical data distributions.
Let us review one such example here.

\subsubsection{Data Drawn from a Doubling Measure}
\label{section:branch-and-bound:doubling-measure}
Throughout our analysis of $k$-d Trees in Section~\ref{section:branch-and-bound:kd-tree},
we considered the case where data points are uniformly distributed in $\mathbb{R}^d$.
As we argued in Chapter~\ref{chapter:intrinsic-dimensionality}, in many practical situations,
however, even though vectors are represented in $\mathbb{R}^d$,
they actually lie in some low-dimensional manifold with \emph{intrinsic dimension}
$d_\circ$ where $d_\circ \ll d$.
This happens, for example, when data points are drawn from a \emph{doubling measure}
with low dimension as defined in Definition~\ref{definition:doubling-measure}.

\begin{svgraybox}
\cite{dasgupta2015rptrees} prove that, if a collection of $m$ vectors is sampled
from a doubling measure with dimension $d_\circ$,
then $\Phi$ can be bounded from above roughly by $(1/m)^{1/d_\circ}$. The following
theorem presents their claim.
\end{svgraybox}

\begin{theorem}
    Suppose a collection $\mathcal{X}$ of $m$ vectors is drawn from $\mu$, a continuous, doubling measure
    on $\mathbb{R}^d$ with dimension $d_\circ \geq 2$. For an arbitrary $\delta \in (0, 1/2)$,
    with probability at least $1 - 3\delta$, for all $2 \leq s \leq m$:
    \begin{equation*}
        \Phi_s(q, \mathcal{X}) \leq 6 \Bigg( \frac{2}{s} \ln \frac{1}{\delta} \Bigg)^{1/d_\circ}.
    \end{equation*}
\end{theorem}

Using the result above,~\cite{dasgupta2015rptrees} go on to prove that,
under the same conditions, with probability at least $1 - 3\delta$,
the failure probability of an RP Tree is bounded above by:
\begin{equation*}
    c_\circ (d_\circ + \ln m_\circ) \Bigg( \frac{8 \max (1, \ln 1/\delta)}{m_\circ} \Bigg)^{1/d_\circ},
\end{equation*}
where $c_\circ$ is an absolute constant, and $m_\circ \geq c_\circ 3^{d_\circ} \max (1, \ln 1/\delta)$.

\begin{svgraybox}
The results above tell us that, so long as the space has a small intrinsic
dimension, we can make the probability of failing to find the optimal solution
arbitrarily small.
\end{svgraybox}

\subsection{Spill Trees}
The Spill Tree~\citep{liu2004SpillTree} is another randomized variant of the $k$-d Tree.
The algorithm to construct a Spill Tree comes with a hyperparameter $\alpha \in [0, 1/2]$
that is typically a small constant closer to $0$.
Given an $\alpha$, the Spill Tree modifies the tree construction algorithm of the $k$-d Tree as follows.
When splitting a node whose region is $\mathcal{R}$, we first project all vectors contained
in $\mathcal{R}$ onto a random direction $U$, then find the median of the resulting distribution.
However, instead of \emph{partitioning} the vectors based on which side of the median they are on,
the algorithm forms two overlapping sets. The ``left'' set contains all vectors in $\mathcal{R}$
whose projection onto $U$ is smaller than the $(1/2 + \alpha)$-fractile point of the distribution,
while the ``right'' set consists of those that fall to the right of the $(1/2 - \alpha)$-fractile point.
As before, a node becomes a leaf when it has a maximum of $m_\circ$ vectors.

\begin{figure}[t]
    \centering
    \includegraphics[width=0.6\linewidth]{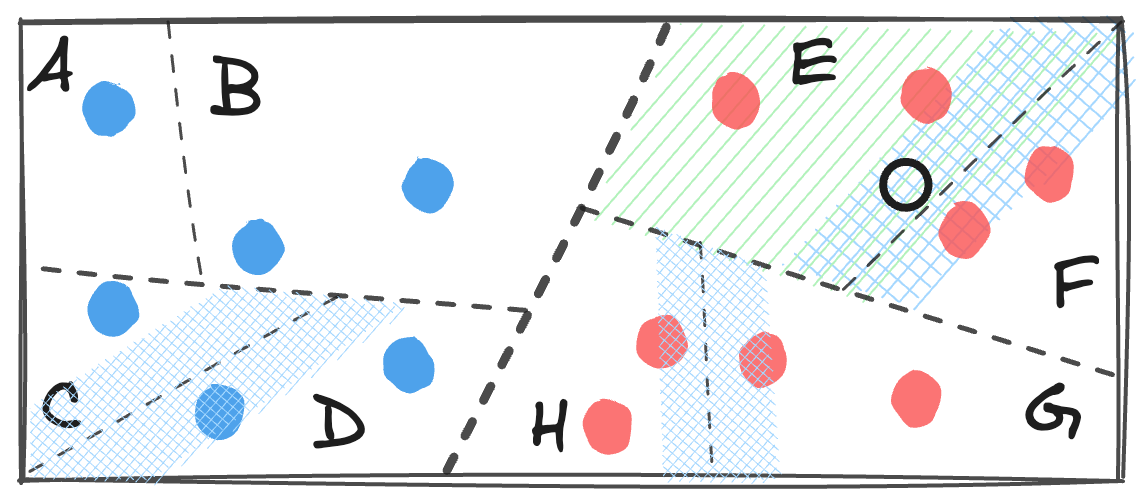}
    \caption{Defeatist search over a Spill Tree. In a Spill Tree, vectors that are
    close to the decision boundary are, in effect, duplicated, with their copy ``spilling'' over
    to the other side of the boundary. This is depicted for a few example regions
    as the blue shaded area that straddles the decision boundary: vectors that fall into the shaded
    area belong to neighboring regions. For example, regions $G$ and $H$ share two vectors.
    As such, a defeatist search for the example query (the empty circle) looks through not just the region
    $E$ but its extended region that overlaps with $F$.}
    \label{figure:branch-and-bound:spill-tree}
\end{figure}

During search, the algorithm performs a defeatist search by routing the query point $q$
based on a comparison of its projection onto the random direction associated with each node,
and the \emph{median} point. It is clear that with this strategy, if the nearest neighbor
of $q$ is close to the decision boundary of a node, we do not increase the likelihood of
failure whether we route $q$ to the left child or to the right one. Figure~\ref{figure:branch-and-bound:spill-tree}
shows an example of the defeatist search over a Spill Tree.

\subsubsection{Space Overhead}
One obvious downside of the Spill Tree is that a single data point
may end up in multiple leaf nodes, which increases the space complexity.
We can quantify that by noting that the depth of the tree on a collection of $m$ vectors
is at most $\log_{1/(1/2 + \alpha)} (m/m_\circ)$, so that the total number of vectors in all leaves
is:
\begin{equation*}
    m_\circ 2^{\log_{1/(1/2 + \alpha)} (m/m_\circ)} = m_\circ \big( \frac{m}{m_\circ} \big)^{\log_{1/(1/2 + \alpha)} 2} =
    m_\circ \big( \frac{m}{m_\circ} \big)^{1/(1 - \log(1 + 2\alpha))}.
\end{equation*}
As such, the space complexity of a Spill Tree is $\mathcal{O}(m^{1/(1 - \log(1 + 2 \alpha))})$.

\subsubsection{Probability of Failure}
The defeatist search over a Spill Tree fails to return the nearest neighbor $x^\ast$
if the following event takes place at any of the nodes that contains $q$ and $x^\ast$.
That is the event where the projections of $q$ and $x^\ast$ are separated by the median \emph{and}
where the projection of $x^\ast$ is separated from the median by at least $\alpha$-fraction of the vectors.
That event happens when the projections of $q$ and $x^\ast$ are separated by
at least $\alpha$-fraction of the vectors in some node along the path.

The probability of the event above can be bounded by Corollary~\ref{corollary:branch-and-bound:randomized:m-points}.
By applying the union bound to a root-to-leaf path, and noting that the size of the collection
reduces at each level by a factor of at least $1/2 + \alpha$, we obtain the following result:

\begin{theorem}
    The probability that a Spill Tree built for collection $\mathcal{X}$ of $m$ vectors fails to find the nearest neighbor
    of a query $q$ is at most:
    \begin{equation*}
        \sum_{l = 0}^{\ell} \frac{1}{2\alpha} \Phi_{\beta^l m}(q, \mathcal{X}),
    \end{equation*}
    with $\beta = 1/2 + \alpha$ and $\ell = \log_{1/\beta} \big( m / m_\circ \big)$.
\end{theorem}

\section{Cover Trees}
The branch-and-bound algorithms we have reviewed thus far divide a
collection recursively into exactly two sub-collections, using a 
hyperplane as a decision boundary. Some also have a certification process
that involves backtracking from a leaf node whose region contains a query
to the root node. As we noted in Section~\ref{section:branch-and-bound:intuition},
however, none of these choices is absolutely necessary. In fact,
branching and bounding can be done entirely differently. We review
in this section a popular example that deviates from that pattern,
a data structure known as the Cover Tree~\citep{covertrees}.

It is more intuitive to describe the Cover Tree, as well as the construction
and search algorithms over it, in the abstract first. This is
what~\cite{covertrees} call the \emph{implicit} representation.
Let us first describe its structure, then review its properties
and explain the relevant algorithms,
and only then discuss how the abstract tree can be implemented concretely.

\subsection{The Abstract Cover Tree and its Properties}

The abstract Cover Tree is a tree structure with infinite depth
that is defined for a proper metric $\delta(\cdot, \cdot)$.
Each level of the tree is numbered by an integer that starts
from $\infty$ at the level of the root node and decrements, to $-\infty$,
at each subsequent level.
Each node represents a single data point.
If we denote the collection of nodes on level $\ell$ by $C_\ell$,
then $C_\ell$ is a \emph{set}, in the sense that the
data points represented by those nodes are distinct.
But $C_\ell \subset C_{\ell - 1}$, so that once a node
appears in level $\ell$, it necessarily appears in levels $(\ell - 1)$
onward. That implies that, in the abstract Cover Tree, $C_\infty$ contains
a single data point, and $C_{-\infty} = \mathcal{X}$ is the entire collection.

\begin{figure}[t]
    \centering
    \includegraphics[width=0.7\linewidth]{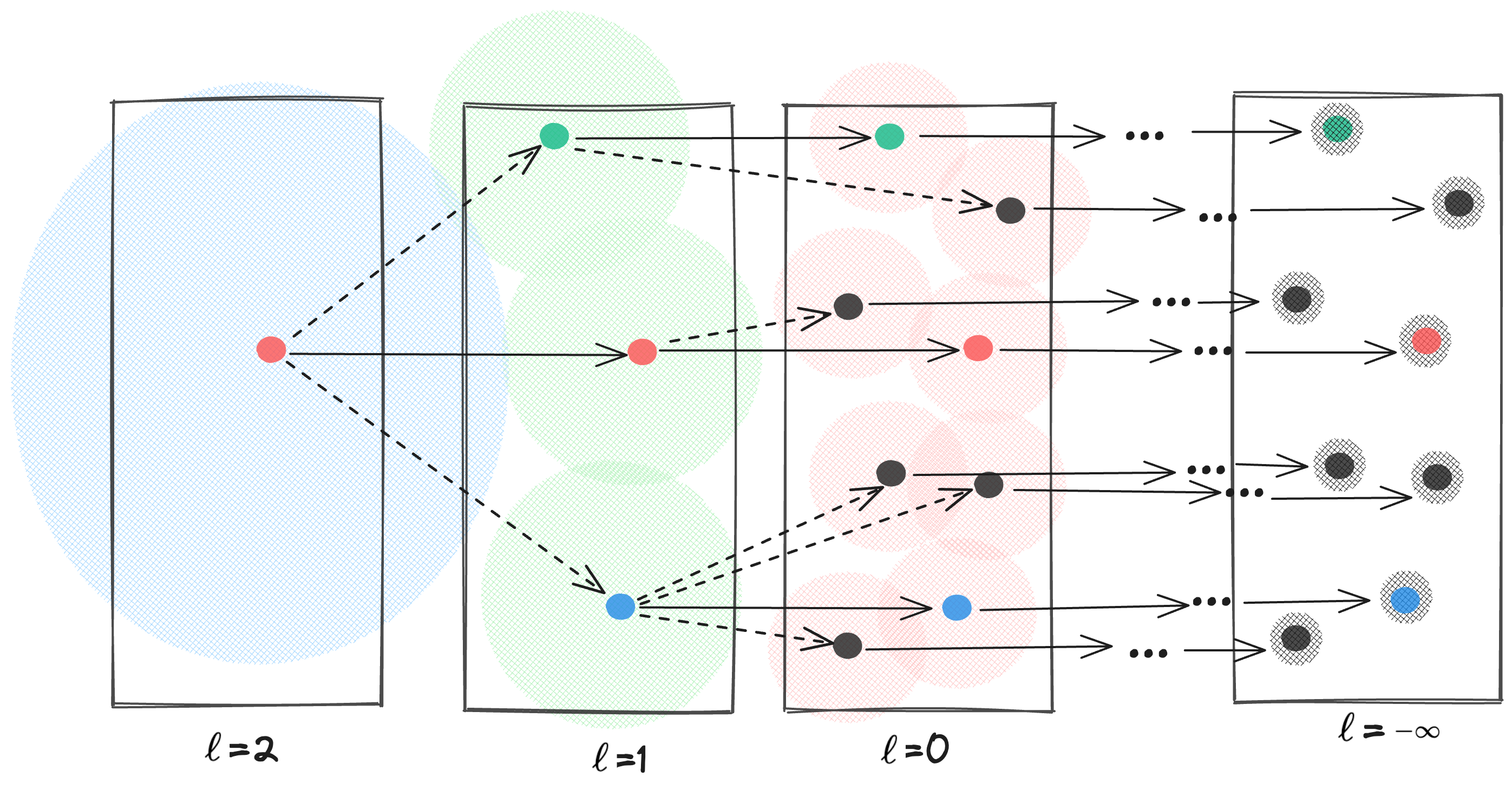}
    \caption{Illustration of the abstract Cover Tree for a collection of $8$ vectors.
    Nodes on level $\ell$ of the tree are separated by at least $2^\ell$ by the separation invariant.
    Nodes on level $\ell$ cover nodes on level $(\ell - 1)$ with a ball of radius
    at most $2^\ell$ by the covering invariant.
    Once a node appears in the tree, it will appear on all subsequent
    levels as its own child (solid arrows), by the nesting invariant.}
    \label{figure:branch-and-bound:abstract-cover-tree}
\end{figure}

This structure, which is illustrated in Figure~\ref{figure:branch-and-bound:abstract-cover-tree}
for an example collection of vectors, obeys three invariants. That is, all algorithms
that construct the tree or manipulate it in any way must guarantee
that the three properties are not violated. These invariants are:
\begin{itemize}
    \item \textbf{Nesting}: As we noted, $C_\ell \subset C_{\ell - 1}$.
    \item \textbf{Covering}: For every node $u \in C_{\ell - 1}$
    there is a node $v \in C_\ell$ such that $\delta(u, v) < 2^\ell$.
    In other words, every node in the next level $(\ell - 1)$ of the tree is ``covered''
    by an open ball of radius $2^\ell$ around a node in the current level, $\ell$.
    \item \textbf{Separation}: All nodes on the same level $\ell$ are separated by
    a distance of at least $2^\ell$. Formally, if $u, v \in C_\ell$, then $\delta(u, v) > 2^\ell$.
\end{itemize}

\subsection{The Search Algorithm}
We have seen what a Cover Tree looks like and what properties it is guaranteed to maintain.
Given this structure, how do we find the nearest neighbor of a query point?
That turns out to be a fairly simple algorithm as shown in Algorithm~\ref{algorithm:branch-and-bound:cover-tree:nn-search}.

\begin{algorithm}[!t]
\SetAlgoLined
{\bf Input: }{Cover Tree with metric $\delta(\cdot, \cdot)$; query point $q$.}\\
\KwResult{Exact NN of $q$.}

\begin{algorithmic}[1]

\STATE $Q_\infty \leftarrow C_\infty$ \Comment*[l]{\footnotesize $C_\ell$ is the set of nodes on level $\ell$}

\FOR{$\ell$ from $\infty$ to $-\infty$} \label{algorithm:branch-and-bound:cover-tree:nn-search:loop}
    \STATE $Q \leftarrow \{ \textsc{Children}(v) \;|\; v \in Q_\ell \}$ \Comment*[l]{\footnotesize \textsc{Children}$(\cdot)$ returns the children of its argument.}
    \STATE $Q_{\ell - 1} \leftarrow \{ u \;|\; \delta(q, u) \leq \delta(q, Q) + 2^\ell \}$ \Comment*[l]{\footnotesize $\delta(u, S) \triangleq \min_{v \in S} \delta(u, v)$.}\label{algorithm:branch-and-bound:cover-tree:nn-search:pruning}
\ENDFOR

\RETURN $\argmin_{u \in Q_{-\infty}} \delta(q, u)$ \label{algorithm:branch-and-bound:cover-tree:nn-search:return}
\end{algorithmic}
\caption{Nearest Neighbor search over a Cover Tree.}
\label{algorithm:branch-and-bound:cover-tree:nn-search}
\end{algorithm}

Algorithm~\ref{algorithm:branch-and-bound:cover-tree:nn-search} always maintains a current
set of candidates in $Q_\ell$ as it visits level $\ell$ of the tree. In each iteration of the
loop on Line~\ref{algorithm:branch-and-bound:cover-tree:nn-search:loop}, it creates a temporary
set---denoted by $Q$---by collecting the children of all nodes in $Q_\ell$. It then prunes the nodes
in $Q$ based on the condition on Line~\ref{algorithm:branch-and-bound:cover-tree:nn-search:pruning}.
Eventually, the algorithm returns the exact nearest neighbor of query $q$ by
performing exhaustive search over the nodes in $Q_{-\infty}$.

Let us understand why the algorithm is correct. In a way, it is enough to argue that
the pruning condition on Line~\ref{algorithm:branch-and-bound:cover-tree:nn-search:pruning}
never discards an ancestor of the nearest neighbor. If that were the case, we are done
proving the correctness of the algorithm: $Q_{-\infty}$ is guaranteed to have the nearest neighbor,
at which point we will find it on Line~\ref{algorithm:branch-and-bound:cover-tree:nn-search:return}.

The fact that Algorithm~\ref{algorithm:branch-and-bound:cover-tree:nn-search} never prunes
the ancestor of the solution is easy to establish.
To see how, consider the distance between $u \in C_{\ell-1}$ and any of its descendants, $v$.
The distance between the two vectors is bounded as follows:
$\delta(u, v) \leq \sum_{l = \ell - 1}^{-\infty} 2^l = 2^\ell$.
Furthermore, because $\delta$ is proper, by triangle inequality, we know
that: $\delta(q, u^\ast) \leq \delta(q, Q) + \delta(Q, u^\ast)$, where
$u^\ast$ is the solution and a descendant of $u \in C_{\ell - 1}$.
As such, any candidate whose distance is greater than
$\delta(q, Q) + \delta(Q, u^\ast) \leq \delta(q, Q) + 2^\ell$ can be safely
pruned.

\bigskip

The search algorithm has an $\epsilon$-approximate variant too.
To obtain a solution that is at most $(1 + \epsilon) \delta(q, u^\ast)$
away from $q$, assuming $u^\ast$ is the optimal solution, we need 
only to change the termination condition on Line~\ref{algorithm:branch-and-bound:cover-tree:nn-search:loop},
by exiting the loop as soon as $\delta(q, Q_{\ell}) \geq 2^{\ell + 1} (1 + 1/\epsilon)$.
Let us explain why the resulting algorithm is correct.

Suppose that the algorithm terminates early when it reaches level $\ell$.
That means that $2^{\ell + 1} (1 + 1/\epsilon) \leq \delta(q, Q_\ell)$.
We have already seen that $\delta(q, Q_\ell) \leq 2^{\ell + 1}$, and by triangle
inequality, that $\delta(q, Q_\ell) \leq \delta(q, u^\ast) + 2^{\ell + 1}$.
So we have bounded $\delta(q, Q_\ell)$ from below and above, resulting
in the following inequality:
\begin{equation*}
    2^{\ell + 1} \Big( 1 + \frac{1}{\epsilon} \Big) \leq \delta(q, u^\ast) + 
    2^{\ell + 1} \implies 2^{\ell + 1} \leq \epsilon \delta(q, u^\ast).
\end{equation*}
Putting all that together, we have shown that $\delta(q, Q_\ell) \leq (1 + \epsilon) \delta(q, u^\ast)$,
so that Line~\ref{algorithm:branch-and-bound:cover-tree:nn-search:return} returns an $\epsilon$-approximate
solution.

\subsection{The Construction Algorithm}
Inserting a single vector into the Cover Tree ``index'' is a procedure
that is similar to the search algorithm but is better conceptualized recursively,
as shown in Algorithm~\ref{algorithm:branch-and-bound:cover-tree:insertion}.

\begin{algorithm}[!t]
\SetAlgoLined
{\bf Input:}{ Cover Tree $\mathcal{T}$ with metric $\delta(\cdot, \cdot)$; New Vector $p$;
Level $\ell$; Candidate set $Q_\ell$.}\\
\KwResult{Cover Tree containing $p$.}

\begin{algorithmic}[1]

\STATE $Q \leftarrow \{ \textsc{Children}(u) \;|\; u \in Q_\ell \}$ \label{algorithm:branch-and-bound:cover-tree:insertion:q-formation}

\IF{$\delta(p, Q) > 2^\ell$} \label{algorithm:branch-and-bound:cover-tree:insertion:termination}
    \RETURN $\bowtie$
\ELSE
    \STATE $Q_{\ell - 1} \leftarrow \{ u \in Q \;|\; \delta(p, u) \leq 2^\ell \}$ \label{algorithm:branch-and-bound:cover-tree:insertion:pruning}
    \IF{\textbf{Insert}$(\mathcal{T}, p, Q_{\ell - 1}, \ell - 1) =\; \bowtie \;\land\; \delta(p, Q_\ell) \leq 2^\ell$} \label{algorithm:branch-and-bound:cover-tree:insertion:success}
        \STATE Choose $u \in Q_\ell$ such that $\delta(p, u) \leq 2 ^\ell$
        \STATE Add $p$ to \textsc{Children}$(u)$
        \RETURN $\blacklozenge$
    \ELSE
        \RETURN $\bowtie$
    \ENDIF
\ENDIF

\end{algorithmic}
\caption{Insertion of a vector into a Cover Tree.}
\label{algorithm:branch-and-bound:cover-tree:insertion}
\end{algorithm}

It is important to note that the procedure in
Algorithm~\ref{algorithm:branch-and-bound:cover-tree:insertion} assumes
that the point $p$ is not present in the tree. That is a harmless assumption
as the existence of $p$ can be checked by a simple invocation of
Algorithm~\ref{algorithm:branch-and-bound:cover-tree:nn-search}. We can therefore
safely assume that $\delta(p, Q)$ for any $Q$ formed on
Line~\ref{algorithm:branch-and-bound:cover-tree:insertion:q-formation} is strictly positive.
That assumption guarantees that the algorithm eventually terminates.
That is because $\delta(p, Q) > 0$ so that ultimately we will invoke the algorithm
with a value $\ell$ such that $\delta(p, Q) > 2^\ell$, at which point
Line~\ref{algorithm:branch-and-bound:cover-tree:insertion:termination} terminates
the recursion.

We can also see why Line~\ref{algorithm:branch-and-bound:cover-tree:insertion:success} is bound to
evaluate to \textsc{True} at some point during the execution of the algorithm.
That is because there must exist a level $\ell$ such that $2^{\ell - 1} < \delta(p, Q) \leq 2^\ell$.
That implies that the point $p$ will ultimately be inserted into the tree.

What about the three invariants of the Cover Tree? We must now show that the resulting
tree maintains those properties: nesting, covering, and separation.
The covering invariant is immediately guaranteed as a result of Line~\ref{algorithm:branch-and-bound:cover-tree:insertion:success}. The nesting invariant too is trivially maintained because
we can insert $p$ as its own child for all subsequent levels.

What remains is to show that the insertion algorithm maintains the separation property. 
To that end, suppose $p$ has been inserted into $C_{\ell-1}$ and consider 
its sibling $u \in C_{\ell-1}$. If $u \in Q$, then it is clear that
$\delta(p, u) > 2^{\ell - 1}$ because Line~\ref{algorithm:branch-and-bound:cover-tree:insertion:success}
must have evaluated to \textsc{True}. On the other hand, if $u \notin Q$,
that means that there was some $\ell^\prime > \ell$ where some ancestor of
$u$, $u^\prime \in C_{\ell^\prime - 1}$, was pruned on
Line~\ref{algorithm:branch-and-bound:cover-tree:insertion:pruning},
so that $\delta(p, u^\prime) > 2^{\ell^\prime}$.
Using the covering invariant, we can deduce that:
\begin{align*}
    \delta(p, u) &\geq \delta(p, u^\prime) - \sum_{l = \ell^\prime - 1}^{\ell} 2^l \\
    &= \delta(p, u^\prime) - (2^{\ell^\prime} - 2^\ell) \\
    &> 2^{\ell^\prime} - (2^{\ell^\prime} - 2^\ell) = 2^\ell.
\end{align*}
That concludes the proof that $\delta(p, C_{\ell - 1}) > 2^{\ell - 1}$,
showing that Algorithm~\ref{algorithm:branch-and-bound:cover-tree:insertion} maintains
the separation invariant.

\subsection{The Concrete Cover Tree}
The abstract tree we described earlier has infinite depth. While that representation
is convenient for explaining the data structure and algorithms that operate on it,
it is not practical. But it is easy to derive a concrete instance of the data structure,
without changing the algorithmic details, to obtain what~\cite{covertrees}
call the \emph{explicit} representation.

One straightforward way of turning the abstract Cover Tree into a concrete one
is by turning a node into a (terminal) leaf if it is its only child---recall that,
a node in the abstract Cover Tree is its own child, indefinitely. For example,
in Figure~\ref{figure:branch-and-bound:abstract-cover-tree}, all nodes on level
$0$ would become leaves and the Cover Tree would end at that depth.
We leave it as an exercise to show that the concrete representation of the tree
does not affect the correctness of Algorithms~\ref{algorithm:branch-and-bound:cover-tree:nn-search}
and~\ref{algorithm:branch-and-bound:cover-tree:insertion}.

The concrete form is not only important for making the data structure practical,
it is also necessary for analysis. For example, as~\cite{covertrees} prove that
the space complexity of the concrete Cover Tree is $\mathcal{O}(m)$ with $m=\lvert \mathcal{X} \rvert$,
whereas the abstract form is infinitely large.
The time complexity of the insertion and search algorithms also use the
concrete form, but they further require assumptions on the data distribution.
\cite{covertrees} present their analysis for vectors that are drawn from
a doubling measure, as we have defined in Definition~\ref{definition:doubling-measure}.
However, their claims have been disputed~\citep{curtin2016phd}
by counter-examples~\citep{elkin2022counterexamples},
and corrected in a recent work~\citep{elkin2023compressed-cover-trees}.

\section{Closing Remarks}

This chapter has only covered algorithms that convey the foundations of
a branch-and-bound approach to NN search. Indeed, we left out a number of
alternative constructions that are worth mentioning as we close this chapter.

\subsection{Alternative Constructions and Extensions}

The standard $k$-d Tree itself, as an example, can be instantiated
by using a different splitting procedure, such as splitting on the axis
along which the data exhibits the greatest spread.
PCA Trees~\citep{Sproull1991pcatrees}, PAC Trees~\citep{pactrees}, and
Max-Margin Trees~\citep{maxmargintrees} offer other ways of choosing the axis
or direction along which the algorithm partitions the data.
Vantage-point Trees~\citep{vptrees}, as another example, follow the same
iterative procedure as $k$-d Trees, but partition the space using hyperspheres
rather than hyperplanes.

\medskip

There are also various other randomized constructions of tree index structures
for NN search. \cite{panigrahy2008improved-kdtree}, for instance,
construct a standard $k$-d Tree over the original data points
but, during search, perturb the query point. Repeating the perturb-then-search
scheme reduces the failure probability of a defeatist search over the $k$-d Tree.

\cite{lshvsrptrees} proposes a different variant of the RP Tree where,
instead of a random projection, they choose the principal direction corresponding
to the largest eigenvalue of the covariance of the vectors that fall into a
node. This is equivalent to the PAC Tree~\citep{pactrees} with the exception
that the splitting threshold (i.e., the $\beta$-fractile point) is chosen randomly,
rather than setting it to the median point. \cite{lshvsrptrees} shows that,
with the modified algorithm, a smaller ensemble of trees is necessary to
reach high retrieval accuracy, as compared with the original RP Tree construction.

\cite{sparse-rp-trees} improve the space complexity
of RP Trees by replacing the $d$-dimensional \emph{dense} random direction
with a \emph{sparse} random projection using Fast Johnson-Lindenstrauss
Transform~\citep{fjlt}. The result is that, every internal node of the tree
has to store a sparse vector whose number of non-zero coordinates is far less
than $d$. This space-efficient variant of the RP Tree offers virtually the
same theoretical guarantees as the original RP Tree structure.

\cite{ram2019revisiting_kdtree} improve the running time of the
NN search over an RP Tree (which is $\mathcal{O}(d \log m)$ for $m = \lvert \mathcal{X} \rvert$)
by first randomly rotating the
vectors in a pre-processing step, then applying the standard $k$-d Tree to the
rotated vectors.
They show that, such a construction leads to a search time complexity of $\mathcal{O}(d\log d + \log m)$
and offers the same guarantees on the failure probability as the RP Tree.

\medskip

Cover Trees too have been the center of much research.
As we have already mentioned, many subsequent works~\citep{elkin2022counterexamples,elkin2023compressed-cover-trees,curtin2016phd} investigated
the theoretical results presented in the original paper~\citep{covertrees}
and corrected or improved the time complexity bounds on the insertion
and search algorithms.
\cite{faster-cover-trees} simplified the structure of the concrete Cover Tree
to make its implementation more efficient and cache-aware.
\cite{parallel-cover-trees} proposed parallel insertion and deletion
algorithms for the Cover Tree to scale the algorithm to real-world vector collections.
We should also note that the Cover Tree itself is an extension (or, rather,
a simplification) of Navigating Nets~\citep{krauthgamer2004navigatingnets},
which itself has garnered much research.

\medskip

It is also possible to extend the framework to MIPS.
That may be surprising. After all, the machinery of the branch-and-bound
framework rests on the assumption that the
distance function has all the nice properties we expect from a metric space.
In particular, we take for granted that the distance is non-negative and that
distances obey the triangle inequality. As we know, however, none of these
properties holds when the distance function is inner product.

As~\cite{xbox-tree} show, however, it is possible to apply a rank-preserving
transformation to vectors such that solving MIPS over the original space
is equivalent to solving NN over the transformed space. \cite{conetrees}
take a different approach and derive bounds on the inner product between
an arbitrary query point and vectors that are contained in a ball associated
with an internal node of the tree index. This bound allows the certification
process to proceed as usual. Nonetheless, these methods face the same challenges
as $k$-d Trees and their variants.

\subsection{Future Directions}

The literature on branch-and-bound algorithms for top-$k$ retrieval is rather mature
and stable at the time of this writing.
While publications on this fascinating class of algorithms
continue to date, most recent works either improve the theoretical analysis of existing
algorithms (e.g.,~\citep{elkin2023compressed-cover-trees}),
improve their implementation (e.g.,~\citep{ram2019revisiting_kdtree}),
or adapt their implementation to other computing paradigms such as distributed systems
(e.g.,~\citep{parallel-cover-trees}).

Indeed, such research is essential. Tree indices are---as the reader will undoubtedly
learn after reading this monograph---among the few retrieval algorithms that rest on
a sound theoretical foundation. Crucially, their implementations too reflect those
theoretical principles: There is little to no gap between theoretical tree indices and
their concrete forms. Improving their theoretical guarantees and modernizing their
implementation, therefore, makes a great deal of sense,
especially so because works like~\citep{ram2019revisiting_kdtree} show
how competitive tree indices can be in practice.

An example area that has received little attention concerns the
data structure that materializes a tree index.
In most works, trees appear in their na\"ive form and are processed trivially.
That is, a tree is simply a collection of if-else blocks,
and is evaluated from root to leaf, one node at a time. The vectors in the leaf
of a tree, too, are simply searched exhaustively.
Importantly, the knowledge that one tree is often insufficient and that a forest
of trees is often necessary to reach an acceptable retrieval accuracy, is not
taken advantage of.
This insight was key in improving forest traversal in the learning-to-rank
literature~\citep{quickscorer,ye2018rapidscorer}, in particular when
a batch of queries is to be processed simultaneously.
It remains to be seen if a more efficient tree traversal algorithm
can unlock the power of tree indices.

Perhaps more importantly, the algorithms we studied in this chapter give
us an arsenal of theoretical tools that may be of independent interest.
The concepts such as partitioning, spillage, and $\epsilon$-nets that are so
critical in the development of many of the algorithms we saw earlier,
are useful not only in the context of trees, but also in other classes
of retrieval algorithms. We will say more on that in future chapters.

\bibliographystyle{abbrvnat}
\bibliography{biblio}

\chapter{Locality Sensitive Hashing}
\label{chapter:lsh}

\abstract{
In the preceding chapter, we delved into algorithms that inferred the
geometrical shape of a collection of vectors and condensed it into
a navigable structure. In many cases, the algorithms were designed for exact
top-$k$ retrieval, but could be modified to provide guarantees on approximate
search. This section, instead, explores an entirely different idea that is
probabilistic in nature and, as such, is designed specifically for approximate
top-$k$ retrieval from the ground up.
}

\section{Intuition}
\label{section:lsh:intuition}
Let us consider the intuition behind what is known as \emph{Locality Sensitive Hashing}
(LSH)~\citep{lsh} first.
Define $b$ separate ``buckets.'' Now, suppose there exists a mapping $h(\cdot)$
from vectors in $\mathbb{R}^d$ to these buckets, such that every vector is placed into
a single bucket: $h: \mathbb{R}^d \rightarrow [b]$. Crucially, assume that vectors that are closer to each other
according to the distance function $\delta(\cdot, \cdot)$, are more likely to be placed
into the same bucket. In other words, the probability that two vectors collide
increases as $\delta$ decreases.

Considering the setup above, indexing is simply a matter of applying $h$ to
all vectors in the collection $\mathcal{X}$ and making note of the resulting placements.
Retrieval for a query $q$ is also straightforward: Perform exact
search over the data points that are in the bucket $h(q)$.
The reason this procedure works with high probability is because it
is more likely for the mapping $h$ to place $q$ in a bucket that contains
its nearest neighbors, so that an exact search over the $h(q)$ bucket yields the correct
top-$k$ vectors with high likelihood. This is visualized in Figure~\subref*{figure:lsh:intuition:single-dimensional}.

\begin{figure}[t]
    \centering
    \subfloat[]{
        \label{figure:lsh:intuition:single-dimensional}
        \includegraphics[width=0.49\linewidth]{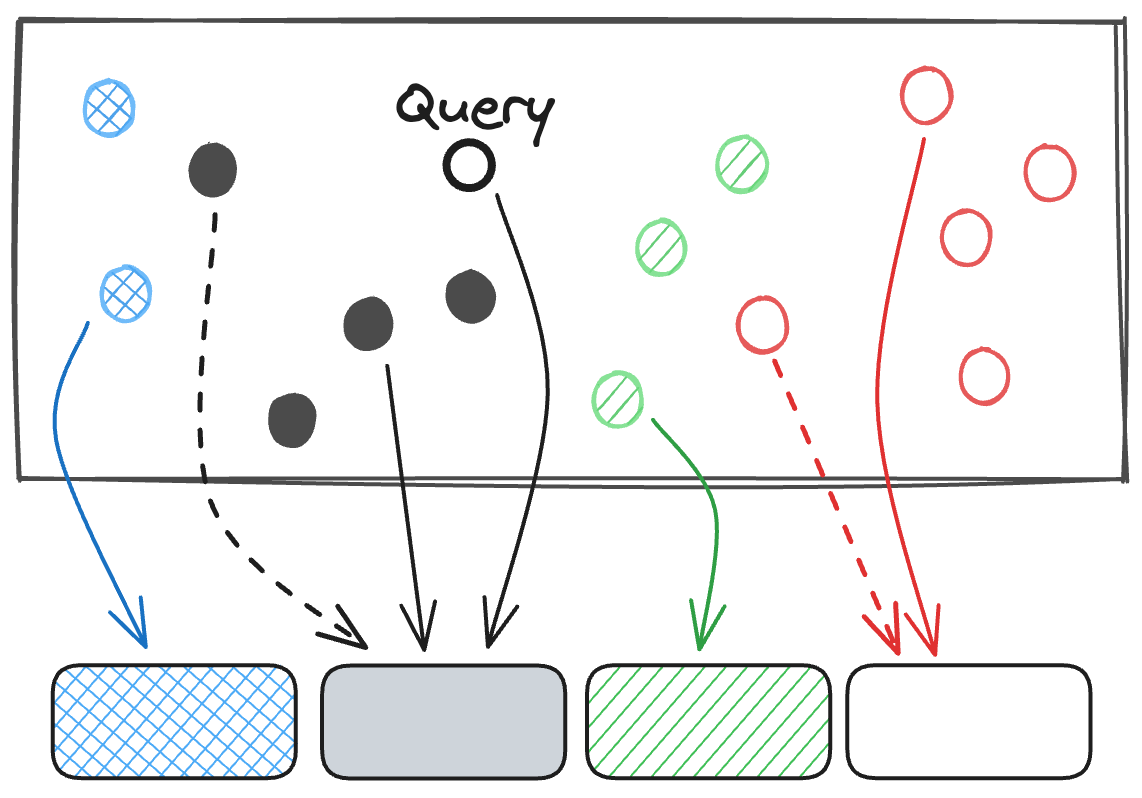}
    }
    \subfloat[]{
        \label{figure:lsh:intuition:multi-dimensional}
        \includegraphics[width=0.49\linewidth]{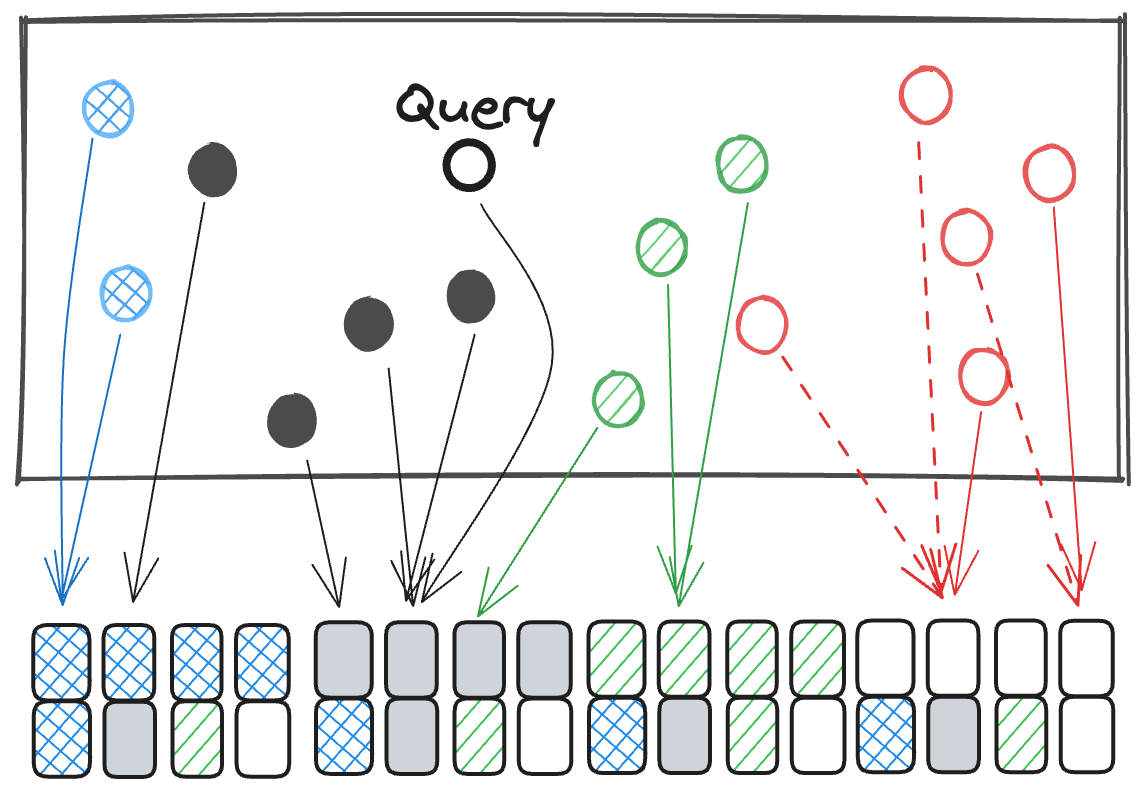}
    }
    \caption{Illustration of Locality Sensitive Hashing. In (a), a function
    $h: \mathbb{R}^2 \rightarrow \{ 1, 2, 3, 4 \}$ maps vectors to four buckets.
    Ideally, when two points are closer to each other, they are more likely to be placed
    in the same bucket. But, as the dashed arrows show, some vectors end up in less-than-ideal
    buckets. When retrieving the top-$k$ vectors for a query $q$, we search through the data
    vectors that are in the bucket $h(q)$. Figure (b) depicts an extension of the framework
    where each bucket is the vector $[h_1(\cdot), h_2(\cdot)]$ obtained from two independent mappings
    $h_1$ and $h_2$.
    }
    \label{figure:lsh:intuition}
\end{figure}

It is easy to extend this setup to ``multi-dimensional'' buckets in the following sense.
If $h_i$'s are independent functions that have the desired property above (i.e., increased
chance of collision with smaller $\delta$), we may define
a bucket in $[b]^\ell$ as the vector mapping $g(\cdot) = [h_1(\cdot), h_2(\cdot), \ldots, h_\ell(\cdot)]$.
Figure~\subref*{figure:lsh:intuition:multi-dimensional} illustrates this extension for $\ell=2$.
The indexing and search procedures work in much the same way. But now, there are
presumably fewer data points in each bucket, and
spurious collisions (i.e., vectors that were mapped to the same bucket 
but that are far from each other according to $\delta$) are less likely to occur.
In this way, we are likely to reduce the overall search time and increase
the accuracy of the algorithm.

Extending the framework even further, we can repeat the process above $L$ times
by constructing independent mappings $g_1(\cdot)$ through $g_L(\cdot)$ from
individual mappings $h_{ij}(\cdot)$ ($1 \leq i \leq L$ and $1 \leq j \leq \ell$),
all of which possessing the property of interest. Because the mappings are independent,
repeating the procedure many times increases the probability of obtaining a high
retrieval accuracy.

That is the essence of the LSH approach to top-$k$ retrieval.
Its key ingredient is the family $\mathcal{H}$ of functions $h_{ij}$'s that have the stated
property \emph{for a given distance function}, $\delta$.
This is the detail that is studied in the remainder of this section.
But before we proceed to define $\mathcal{H}$ for different distance functions,
we will first give a more rigorous description of the algorithm.

\section{Top-\texorpdfstring{$k$}{k} Retrieval with LSH}

Earlier, we described informally the class of mappings that are at the core
of LSH, as hash functions that preserve the distance between points. That is,
the likelihood that such a hash function places two points in the same bucket
is a function of their distance. Let us formalize that notion first in the
following definition, due to~\cite{lsh}.

\begin{definition}[$(r, (1+\epsilon)r, p_1, p_2)$-Sensitive Family]
    \label{definition:lsh:hash-family}
    A family of hash functions $\mathcal{H} = \{ h:\; \mathbb{R}^d \rightarrow [b] \}$
    is called $(r, (1 + \epsilon)r, p_1, p_2)$-sensitive for a distance function $\delta(\cdot, \cdot)$,
    where $\epsilon > 0$ and $0 < p_1, p_2 < 1$, if for any two points $u, v \in \mathbb{R}^d$:
    \begin{itemize}
        \item $\delta(u, v) \leq r \implies \probability_\mathcal{H}\big[ h(u) = h(v) \big] \geq p_1$; and,
        \item $\delta(u, v) > (1 + \epsilon)r \implies \probability_\mathcal{H}\big[ h(u) = h(v) \big] \leq p_2$.
    \end{itemize}
\end{definition}

It is clear that such a family is useful only when $p_1 > p_2$.
We will see examples of $\mathcal{H}$ for different distance functions later in this section.
For the time being, however, suppose such a family of functions exists for any $\delta$ of interest.

The indexing algorithm remains as described before.
Fix parameters $\ell$ and $L$ to be determined later in
this section. Then define the vector function
$g(\cdot) = \big[ h_1(\cdot), h_2(\cdot), \ldots, h_\ell(\cdot) \big]$
where $h_i \in \mathcal{H}$. Now, construct $L$ such functions $g_1$ through $g_L$, and process
the data points in collection $\mathcal{X}$ by evaluating $g_i$'s and placing
them in the corresponding multi-dimensional bucket.

\begin{svgraybox}
In the end, we have effectively
built $L$ tables, each mapping buckets to a list of data points that
fall into them. Note that, each of the $L$ tables holds a copy of the collection,
but where each table organizes the data points differently.
\end{svgraybox}

\subsection{The Point Location in Equal Balls Problem}
Our intuitive description of retrieval using LSH ignored a minor technicality
that we must elaborate in this section. In particular, as is clear from 
Definition~\ref{definition:lsh:hash-family}, a family $\mathcal{H}$ has
a dependency on the distance $r$. That means any instance of the family provides guarantees
only with respect to a specific $r$. Consequently, any index obtained
from a family $\mathcal{H}$, too, is only useful in the context of a fixed $r$.

It appears, then, that the LSH index is not in and of itself sufficient
for solving the $\epsilon$-approximate retrieval problem of
Definition~\ref{definition:flavors:approximate-top-k-retrieval} directly.
But, it is enough for solving an easier \emph{decision problem} that is
known as Point Location in Equal Balls (PLEB), defined as follows:

\begin{definition}[$(r, (1+\epsilon)r)$-Point Location in Equal Balls]
    \label{definition:lsh:pleb}
    For a query point $q$ and a collection $\mathcal{X}$,
    if there is a point $u \in \mathcal{X}$ such that $\delta(q, u) \leq r$,
    return \textsc{Yes} and any point $v$ such that $\delta(q, v) < (1+\epsilon)r$.
    Return \textsc{No} if there are no such points.
\end{definition}

The algorithm to solve the $(r, (1+\epsilon)r)$-PLEB problem for a query point $q$
is fairly straightforward. It involves evaluating $g_i$'s on $q$ and exhaustively searching the
corresponding buckets in order. We may terminate early after visiting at
most $4L$ data points. For every examined data point $u$, the algorithm
returns \textsc{Yes} if $\delta(q, u) \leq (1+\epsilon)r$, and \textsc{No} otherwise.

\subsubsection{Proof of Correctness}

Suppose there exits a point $u^\ast \in \mathcal{X}$ such that
$\delta(q, u^\ast) \leq r$. The algorithm above is correct, in the sense
that it returns a point $u$ with $\delta(q, u) \leq (1+\epsilon)r$,
if we choose $\ell$ and $L$ such that the following two properties hold
with constant probability:
\begin{itemize}
    \item $\exists \;i \in [L] \textit{ s.t. } g_i(u^\ast) = g_i(q)$; and,
    \item $\sum_{j=1}^{L} \Big\lvert \Big( \mathcal{X} \setminus B(q, (1+\epsilon)r) \Big) \cap g_j^{-1}(g_j(q)) \Big\rvert \leq 4L$, where $g_j^{-1}(g_j(q))$ is the set of vectors in bucket $g_j(q)$.
\end{itemize}

The first property ensures that, as we traverse the $L$ buckets associated
with the query point, we are likely to visit either the optimal point $u^\ast$,
or some other point whose distance to $q$ is at most $(1 + \epsilon)r$.
The second property guarantees that with constant probability, there are no more
than $4L$ points in the candidate buckets that are $(1+\epsilon)r$ away from $q$.
As such, we are likely to find a solution before visiting $4L$ points.

We must therefore prove that for some $\ell$ and $L$ the above properties hold.
The following claim shows one such configuration.

\begin{theorem}
    \label{theorem:lsh:pleb-configuration}
    Let $\rho=\ln p_1 / \ln p_2$ and $m=\lvert \mathcal{X} \rvert$.
    Set $L=m^\rho$ and $\ell=\log_{1/p_2} m$. The properties above hold
    with constant probability for a $(r, (1+\epsilon)r, p_1, p_2)$-sensitive LSH family.
\end{theorem}
\begin{proof}
    Consider the first property. We have, from Definition~\ref{definition:lsh:hash-family},
    that, for any $h_i \in \mathcal{H}$:
    \begin{equation*}
        \probability\Big[ h_i(u^\ast) = h_i(q) \Big] \geq p_1.
    \end{equation*}
    That holds simply because $u^\ast \in B(q, r)$. That implies:
    \begin{equation*}
        \probability\Big[ g_i(u^\ast) = g_i(q) \Big] \geq p_1^\ell.
    \end{equation*}
    As such:
    \begin{equation*}
        \probability\Big[ \exists \; i \in [L] \textit{ s.t. } g_i(u^\ast) = g_i(q) \Big] \geq 1 - (1 - p_1^\ell)^L.
    \end{equation*}
    Substituting $\ell$ and $L$ with the expressions given in the theorem gives:
    \begin{equation*}
        \probability\Big[ \exists \; i \in [L] \textit{ s.t. } g_i(u^\ast) = g_i(q) \Big] \geq 1 - (1 - \frac{1}{m^\rho})^{m^\rho} \approx 1 - \frac{1}{e},
    \end{equation*}
    proving that the property of interest holds with constant probability.

    Next, consider the second property. For any point $v$ such that $\delta(q, v) > (1 + \epsilon)r$,
    Definition~\ref{definition:lsh:hash-family} tells us that:
    \begin{align*}
        \probability\Big[ h_i(v) &= h_i(q) \Big] \leq p_2 \implies \probability\Big[ g_i(v) = g_i(q) \Big] \leq p_2^\ell \\
        &\implies \probability\Big[ g_i(v) = g_i(q) \Big] \leq \frac{1}{m} \\
        &\implies \ev\Big[ \Big\lvert v \textit{ s.t. } g_i(v) = g_i(q) \land \delta(q, v) > (1 + \epsilon)r \Big\rvert \;|\; g_i \Big] \leq 1 \\
        &\implies \ev\Big[ \Big\lvert v \textit{ s.t. } g_i(v) = g_i(q) \land \delta(q, v) > (1 + \epsilon)r \Big\rvert \Big] \leq L,
    \end{align*}
    where the last expression follows by the linearity of expectation when applied to all $L$ buckets.
    By Markov's inequality, the probability that there are more than $4L$ points for which $\delta(q, v) > (1+\epsilon)r$
    but that map to the same bucket as $q$ is at most $1/4$. That completes the proof.
\end{proof}

\subsubsection{Space and Time Complexity}

The algorithm terminates after visiting at most $4L$ vectors in the candidate buckets.
Given the configuration of Theorem~\ref{theorem:lsh:pleb-configuration}, this means that
the time complexity of the algorithm for query processing is $\mathcal{O}(d m^\rho)$,
which is sub-linear in $m$.

As for space complexity of the algorithm, note that the index stores each
data point $L$ times. That implies the space required to build an LSH index
has complexity $\mathcal{O}(mL) = \mathcal{O}(m^{1 + \rho})$,
which grows super-linearly with $m$. This growth rate can easily become
prohibitive~\citep{gionis1999hashing,buhler2001lsh-comparison},
particularly because it is often necessary to increase $L$ to reach a higher
accuracy, as the proof of Theorem~\ref{theorem:lsh:pleb-configuration} shows.
How do we reduce this overhead and still obtain sub-linear query time?
That is a question that has led to a flurry of research in the past.

One direction to address that question is to modify the search algorithm so that it visits multiple
buckets from each of the $L$ tables, instead of examining just a single bucket per table.
That is the idea first explored by~\cite{entropyLSH}. In that work, the search
algorithm is the same as in the standard version presented above, but in addition to searching the
buckets for query $q$, it also performs many search operations for perturbed copies of $q$.
While theoretically interesting, their method proves difficult to use in practice.
That is because, the amount of noise needed to perturb a query depends on the distance
of the nearest neighbor to $q$---a quantity that is unknown \emph{a priori}.
Additionally, it is likely that a single bucket may be visited many times over
as we invoke the search procedure on the copies of $q$.

Later,~\cite{multiprobeLSH} refined that theoretical result and presented a method
that, instead of perturbing queries \emph{randomly} and performing multiple hash computations
and search invocations, utilizes a more efficient approach
in deciding which buckets to probe within each table. In particular,
their ``multi-probe LSH'' first finds the bucket associated with $q$, say $g_i(q)$.
It then additionally visits other ``adjacent'' buckets where a bucket is adjacent
if it is more likely to hold data points that are close to the vectors in $g_i(q)$.

The precise way their algorithm arrives at a set of adjacent buckets
depends on the hash family itself. In their work,~\cite{multiprobeLSH}
consider only a hash family for the Euclidean distance, and take advantage
of the fact that adjacent buckets (which are in $[b]^\ell$) differ in each
coordinate by at most $1$---this becomes clearer when we review the LSH family for
Euclidean distance in Section~\ref{section:lsh:euclidean}.
This scheme was shown empirically to reduce by \emph{an order of magnitude}
the total number of hash tables that is required to achieve an accuracy greater
than $0.9$ on high-dimensional datasets.

\medskip

Another direction is to improve the guarantees of the LSH family itself.
As Theorem~\ref{theorem:lsh:pleb-configuration} indicates, $\rho = \log p_1 / \log p_2$
plays a critical role in the efficiency and effectiveness of the search algorithm,
as well as the space complexity of the data structure.
It makes sense, then, that improving $\rho$ leads to smaller space overhead.
Many works have explored advanced LSH families to do just
that~\citep{andoni2008near-optimal,andoni2014beyond,andoni2015cross-polytope-lsh}.
We review some of these methods in more detail later in this chapter.

\subsection{Back to the Approximate Retrieval Problem}
A solution to PLEB of Definition~\ref{definition:lsh:pleb}
is a solution to $\epsilon$-approximate top-$k$ retrieval
only if $r = \delta(q, u^\ast)$, where $u^\ast$ is the $k$-th minimizer
of $\delta(q, \cdot)$. But we do not know the minimal distance in advance!
That begs the question: How does solving the PLEB problem
help us solve the $\epsilon$-approximate retrieval problem?

\cite{lsh} argue that an efficient solution to this decision version of the 
problem leads directly to an efficient solution to the original problem.
In effect, they show that $\epsilon$-approximate retrieval can be reduced
to PLEB. Let us review one simple, albeit inefficient reduction.

Let $\delta_\ast = \max_{u, v \in \mathcal{X}} \delta(u, v)$ and
$\delta^\ast = \min_{u, v \in \mathcal{X}} \delta(u, v)$. Denote by
$\Delta$ the aspect ratio: $\Delta = \delta_\ast / \delta^\ast$.
Now, define a set of distances $\mathcal{R} = \{ (1+\epsilon)^0, (1 + \epsilon)^1, \ldots, \Delta \}$,
and construct $\lvert \mathcal{R} \rvert$ LSH indices for each $r \in \mathcal{R}$.

Retrieving vectors for query $q$ is a matter of performing
binary search over $\mathcal{R}$ to find the minimal
distance such that PLEB succeeds and returns a point $u \in \mathcal{X}$.
That point $u$ is the solution to the $\epsilon$-approximate retrieval
problem! It is easy to see that such a reduction adds to the time complexity
by a factor of $\mathcal{O}(\log \log_{1+\epsilon} \Delta)$, and to the space
complexity by a factor of $\mathcal{O}(\log_{1+\epsilon} \Delta)$.

\section{LSH Families}
We have studied how LSH solves the PLEB problem of Definition~\ref{definition:lsh:pleb},
analyzed its time and space complexity, and reviewed how a solution to PLEB leads
to a solution to the $\epsilon$-approximate top-$k$ retrieval problem of Definition~\ref{definition:flavors:approximate-top-k-retrieval}.
Throughout that discussion, we took for granted the existence of an LSH family
that satisfies Definition~\ref{definition:lsh:hash-family} for a distance function
of interest. In this section, we review example families and unpack their construction
to complete the picture.

\subsection{Hamming Distance}
We start with the simpler case of Hamming distance over the space of binary vectors.
That is, we assume that $\mathcal{X} \subset \{0, 1\}^d$ and $\delta(u, v) = \lVert u - v \rVert_1$,
measuring the number of coordinates in which the two vectors $u$ and $v$ differ.
For this setup, a hash family that maps a vector to one of its coordinates at
random---a technique that is also known as \emph{bit sampling}---is an LSH family~\citep{lsh},
as the claim below shows.

\begin{theorem}
For $\mathcal{X} \subset \{ 0, 1 \}^d$ equipped with the Hamming distance,
the family $\mathcal{H} = \{ h_i \;|\; h_i(u) = u_i, \; 1 \leq i \leq d \}$
is $(r, (1 + \epsilon)r, 1 - r/d, 1 - (1+\epsilon)r/d)$-sensitive.
\end{theorem}
\begin{proof}
    The proof is trivial. For a given $r$ and two vectors $u, v \in \{ 0, 1\}^d$,
    if $\lVert u - v \rVert_1 \leq r$, then $\probability\Big[ h_i(u) \neq h_i(v) \Big] \leq r/d$,
    so that $\probability\Big[ h_i(u) = h_i(v) \Big] \geq 1 - r/d$, and therefore $p_1 = 1 - r/d$.
    $p_2$ is derived similarly.
\end{proof}

\subsection{Angular Distance}
\label{section:lsh:angular}
Consider next the angular distance between two real vectors $u, v \in \mathbb{R}^d$, defined as:
\begin{equation}
    \label{equation:lsh:angular-distance}
    \delta(u, v) = \arccos \Big( \frac{\langle u, v \rangle}{\lVert u \rVert_2 \lVert v \rVert_2} \Big).
\end{equation}

\begin{figure}[t]
    \centering
    \subfloat[Hyperplane LSH]{
        \label{figure:lsh:angular:hyplerplane}
        \includegraphics[width=0.4\linewidth]{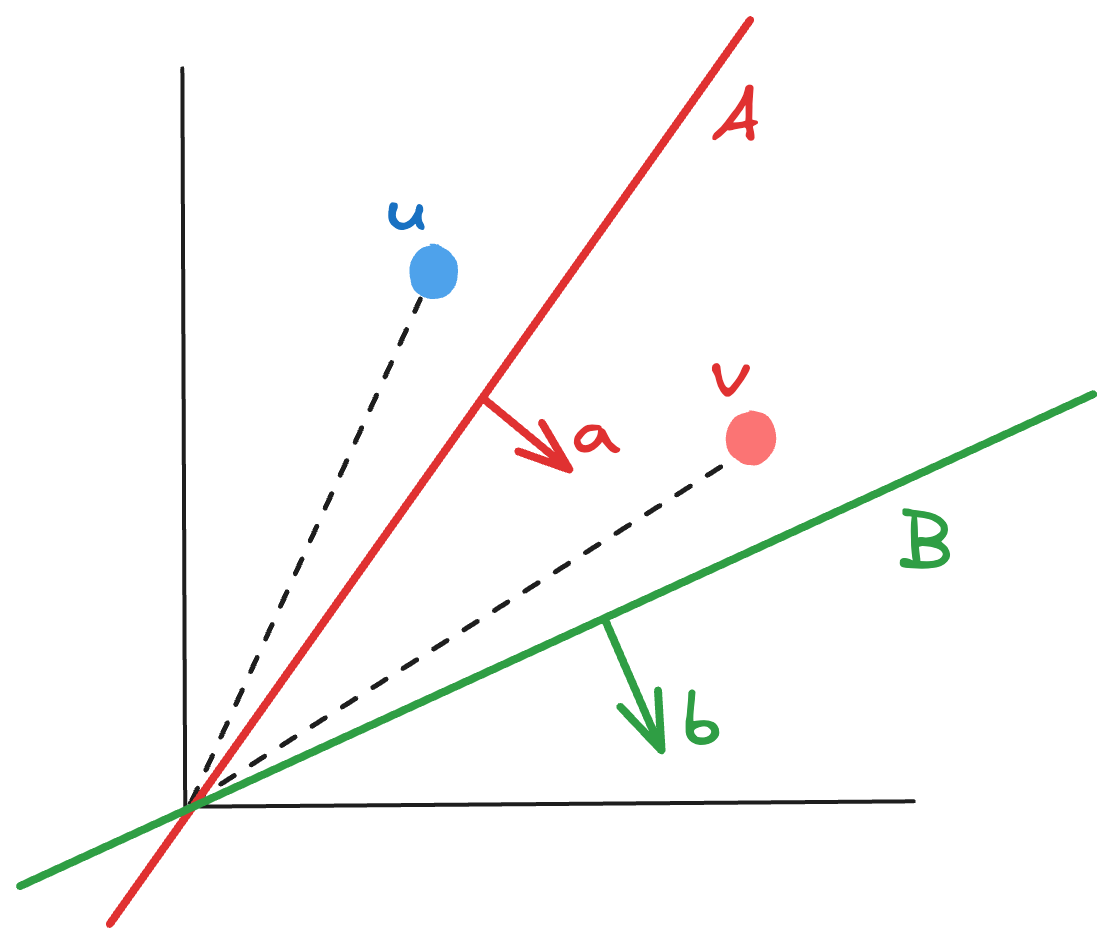}
    }
    \subfloat[Cross-polytope LSH]{
        \label{figure:lsh:angular:cross-polytope}
        \includegraphics[width=0.4\linewidth]{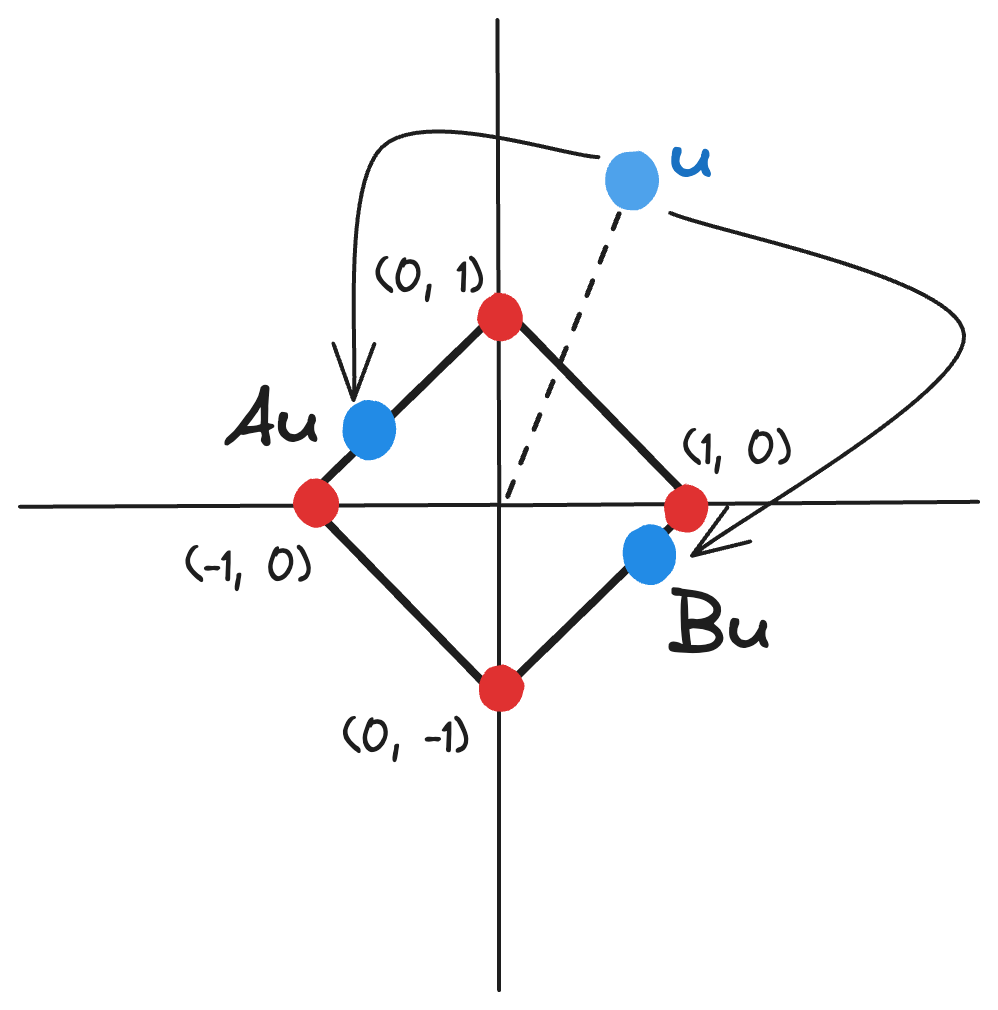}
    }
    \caption{Illustration of hyperplane and cross-polytope LSH functions
    for angular distance in $\mathbb{R}^2$. In hyperplane LSH, we draw random
    directions ($a$ and $b$) to define hyperplanes ($A$ and $B$), and record
    $+1$ or $-1$ depending on which side of the hyperplane a vector ($u$ and $v$)
    lies. For example, $h_a(u)=-1$, $h_a(v)=+1$, and $h_b(u)=h_b(v)=-1$. It is easy
    to see that the probability of a hash collision for two vectors $u$ and $v$ correlates
    with the angle between them. A cross-polytope LSH function, on the other hand,
    randomly rotates and normalizes (using matrix $A$ or $B$) the vector ($u$),
    and records the closest standard basis vector as its hash. Note that, the cross-polytope
    is the $L_1$ ball, which in $\mathbb{R}^2$ is a rotated square. As an example,
    $h_A(u) = -e_1$ and $h_B(u) = +e_1$.}
    \label{figure:lsh:angular}
\end{figure}

\subsubsection{Hyperplane LSH}
For this distance function, one simple LSH family is the set
of hash functions that project a vector
onto a randomly chosen direction and record the sign of the projection.
Put differently, a hash function in this family
is characterized by a random hyperplane, which is in turn defined by a unit vector
sampled uniformly at random. When applied to an input vector $u$,
the function returns a binary value (from $\{ -1, 1\}$) indicating on which
side of the hyperplane $u$ is located. This procedure, which is known as \emph{sign random projections}
or \emph{hyperplane LSH}~\citep{charikar2002rounding-algorithms},
is illustrated in Figure~\subref*{figure:lsh:angular:hyplerplane} and
formalized in the following claim.

\begin{theorem}
    For $\mathcal{X} \subset \mathbb{R}^d$ equipped with the angular distance of
    Equation~(\ref{equation:lsh:angular-distance}),
    the family $\mathcal{H} = \{ h_r \;|\; h_r(u) = \textsc{Sign}(\langle r, u \rangle), \; r \sim \mathbb{S}^{d - 1} \}$
    is $(\theta, (1 + \epsilon)\theta, 1 - \theta/\pi, 1 - (1+\epsilon)\theta/\pi)$-sensitive
    for $\theta \in [0, \pi]$, and $\mathbb{S}^{d-1}$ denoting the $d$-dimensional hypersphere.
\end{theorem}
\begin{proof}
    If the angle between two vectors is $\theta$, then the probability that a randomly
    chosen hyperplane lies between them is $\theta / \pi$. As such, the probability that
    they lie on the same side of the hyperplane is $1 - \theta / \pi$. The claim follows.
\end{proof}

\subsubsection{Cross-polytope LSH}
There are a number of other hash families for the angular distance in addition to the basic
construction above. \emph{Spherical LSH}~\citep{andoni2014beyond} is one example, albeit a purely
theoretical one---a single hash computation from that family alone is considerably more expensive
than an exhaustive search over a million data points~\citep{andoni2015cross-polytope-lsh}!

What is known as \emph{Cross-polytope LSH}~\citep{andoni2015cross-polytope-lsh,terasawa2007spherical-lsh}
offers similar guarantees as the Spherical LSH but is a more practical construction.
A function from this family randomly rotates an input vector first, then outputs the closest
signed standard basis vector ($e_i$'s for $1 \leq i \leq d$) as the hash value.
This is illustrated for $\mathbb{R}^2$ in Figure~\subref*{figure:lsh:angular:cross-polytope},
and stated formally in the following result.

\begin{theorem}
    For $\mathcal{X} \subset \mathbb{S}^{d-1}$ equipped with the angular distance of
    Equation~(\ref{equation:lsh:angular-distance}) or equivalently the Euclidean distance,
    the following family constitutes an LSH:
    \begin{equation*}
        \mathcal{H} = \{ h_R \;|\; h_R(u) = \argmin_{e \in \{ \pm e_i \}_{i=1}^d} \lVert e - \frac{Ru}{\lVert R u \rVert_2} \rVert, R \in \mathbb{R}^{d \times d},\; R_{ij} \sim \mathcal{N}(0, 1) \},
    \end{equation*}
    where $\mathcal{N}(0, 1)$ is the standard Gaussian distribution.
    The probability of collision for unit vectors $u, v \in \mathbb{S}^{d-1}$ with $\lVert u - v \rVert < \tau$ is:
    \begin{equation*}
        \ln \frac{1}{\probability\Big[ h_R(u) = h_R(v) \Big]} = \frac{\tau^2}{4 - \tau^2} \ln d + \mathcal{O}_\tau\Big( \ln \ln d \Big).
    \end{equation*}
    Importantly:
    \begin{equation*}
        \rho = \frac{\log p_1}{\log p_2} = \frac{1}{(1 + \epsilon)^2} \frac{4 - (1 + \epsilon)^2 r^2}{4 - r^2} + o(1).
    \end{equation*}
\end{theorem}
\begin{proof}
    We wish to show that, for two unit vectors $u, v \in \mathbb{S}^{d-1}$ with $\lVert u - v \rVert < \tau$,
    the expression above for the probability of a hash collision is correct.
    That, indeed, completes the proof of the theorem itself. To show that, we will
    take advantage of the spherical symmetry of Gaussian random variables---we used this property
    in the proof of Theorem~\ref{theorem:instability:orthogonality-random-vectors}.

    By the spherical symmetry of Gaussians, without loss of generality, we can assume that
    $u = e_1$, the first standard basis, and $v = \alpha e_1 + \beta e_2$,
    where $\alpha^2 + \beta^2 = 1$ (so that $v$ has unit norm) and $(\alpha - 1)^2 + \beta^2 = \tau^2$
    (because the distance between $u$ and $v$ is $\tau$).

    Let us now model the collision probability as follows:
    \begin{align*}
        \probability\Big[ &h(u) = h(v) \Big] = 2d \probability\Big[ h(u) = h(v) = e_1 \Big] \\
        &= 2d \probability_{X, Y \sim \mathcal{N}(0, I)}\Big[ \forall \; i,\; \lvert X_i \rvert \leq X_1 \land
        \lvert \alpha X_i + \beta Y_i \rvert \leq \alpha X_1 + \beta Y_1 \Big] \\
        &= 2d \ev_{X_1, Y_1 \sim \mathcal{N}(0, 1)} \Bigg[ 
            \probability_{X_2, Y_2}\Big[ \lvert X_2 \rvert \leq X_1 \land \lvert \alpha X_2 + \beta Y_2 \rvert \leq \alpha X_1 + \beta Y_1 \Big]^{d-1}
        \Bigg]. \numberthis \label{equation:lsh:cross-polytope:collision-prob}
    \end{align*}
    The first equality is due again to the spherical symmetry of the hash functions
    and the fact that there are $2d$ signed standard basis vectors.
    The second equality simply uses the expressions for $u=e_1$ and $v=\alpha e_1 + \beta e_2$.
    The final equality follows because of the independence of the coordinates of $X$ and $Y$,
    which are sampled from a $d$-dimensional isotropic Gaussian distribution.

    \begin{figure}[t]
        \centering
        \subfloat[]{
            \label{figure:lsh:angular:cross-polytope:planar-set}
            \includegraphics[width=0.4\linewidth]{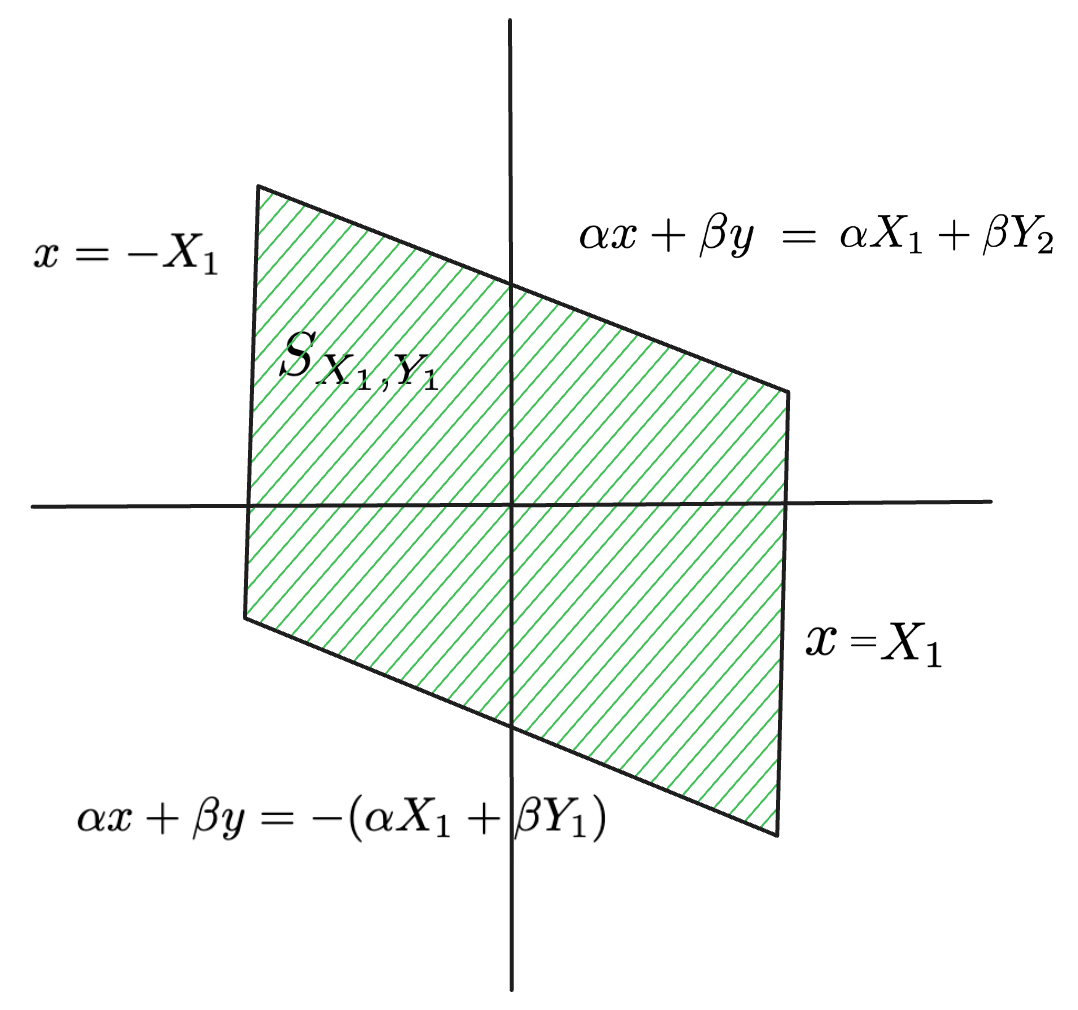}
        }
        \subfloat[]{
            \label{figure:lsh:angular:cross-polytope:lower-bound}
            \includegraphics[width=0.4\linewidth]{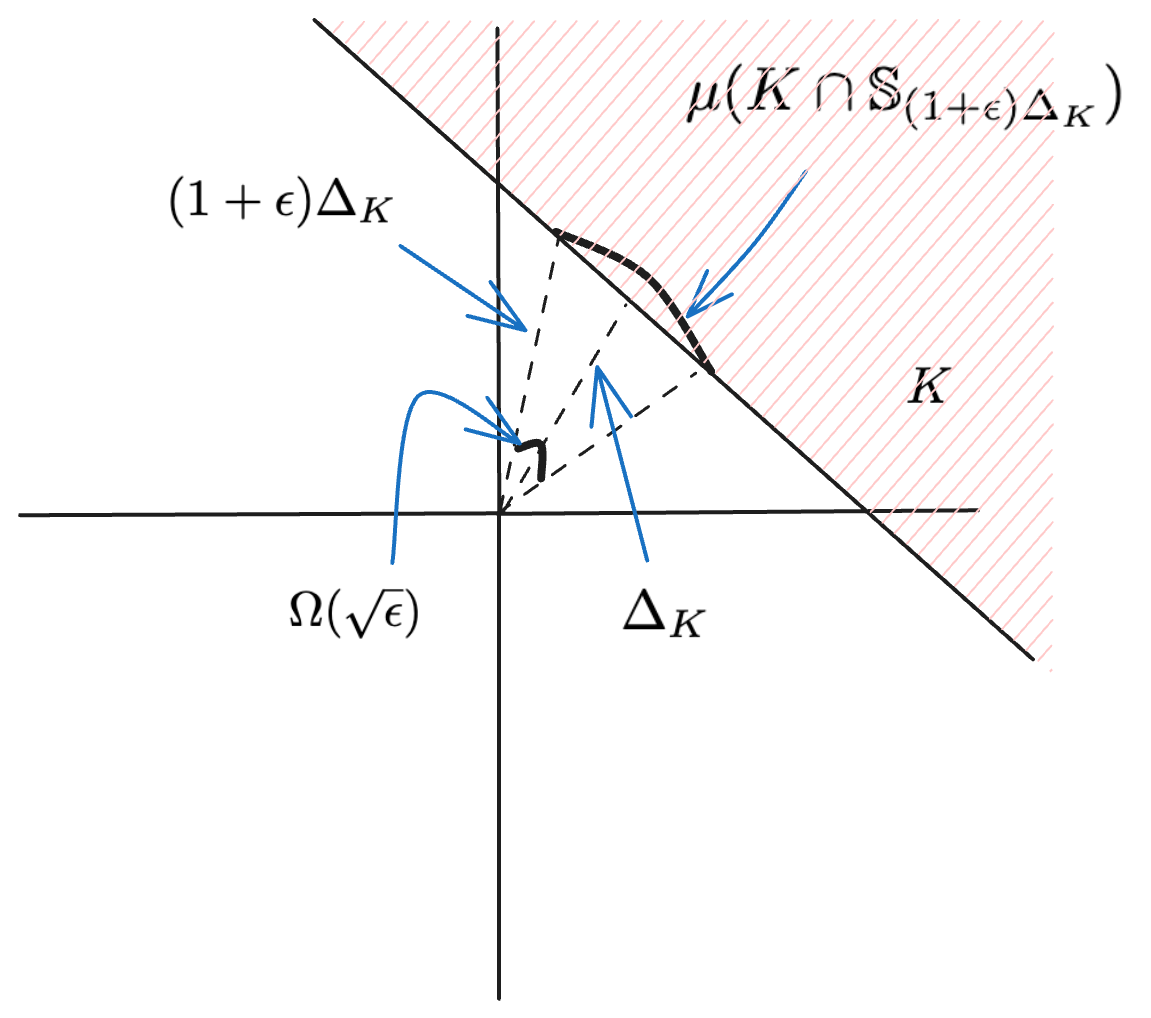}
        }
        \caption{Illustration of the set $S_{X_1, Y_1} = \{ \lvert x \rvert \leq X_1 \land \lvert \alpha x + \beta y \rvert \leq \alpha X_1 + \beta Y_1 \}$ in (a). Figure (b) visualizes the derivation of Equation~(\ref{equation:lsh:angular:cross-polytope:proof-lower-bound}).}
        \label{figure:lsh:angular:cross-polytope:proof}
    \end{figure}

    The innermost term in Equation~(\ref{equation:lsh:cross-polytope:collision-prob}) is
    the Gaussian measure of the closed, convex set
    $\{ \lvert x \rvert \leq X_1 \land \lvert \alpha x + \beta y \rvert \leq \alpha X_1 + \beta Y_1 \}$,
    which is a bounded plane in $\mathbb{R}^2$. This set, which we denote by $S_{X_1, Y_1}$,
    is illustrated in Figure~\subref*{figure:lsh:angular:cross-polytope:planar-set}.
    Then we can expand Equation~(\ref{equation:lsh:cross-polytope:collision-prob}) as follows:
    \begin{align}
        \label{equation:lsh:angular:cross-polytope:prob-collision-expanded}
        2d \ev_{X_1, Y_1 \sim \mathcal{N}(0, 1)} &\Bigg[ 
            \probability_{X_2, Y_2}\Big[ S_{X_1, Y_1} \Big]^{d-1}
        \Bigg] \\
        &= 2d \int_0^1 \probability_{X_1, Y_1 \sim \mathcal{N}(0, 1)} \Big[ \probability [ S_{X_1, Y_1} ] \geq t^{\frac{1}{d - 1}} \Big] dt.
    \end{align}
    We therefore need to expand $\probability[S_{X_1, Y_1}]$ in order to complete the expression above.
    The rest of the proof derives that quantity.
    
    \textbf{Step 1}.
    Consider $\probability[S_{X_1, Y_1}] = \mathcal{G}(S_{X_1, Y_1})$, which is the standard Gaussian
    measure of the set $S_{X_1, Y_1}$. In effect, we are interested in $\mathcal{G}(S)$ for some
    bounded convex subset $S \subset \mathbb{R}^2$. We need the following lemma to derive an expression
    for $\mathcal{G}(S)$. But first define $\mu_A(r)$ as the Lebesgue measure of the intersection of a circle
    of radius $r$ ($\mathbb{S}_r$) with the set $A$, normalized by the circumference of $\mathbb{S}_r$,
    so that $0 \leq \mu_A(r) \leq 1$ is a probability measure:
    \begin{equation*}
        \mu_A(r) \triangleq \frac{\mu(A \cap \mathbb{S}_r)}{2 \pi r},
    \end{equation*}
    and denote by $\Delta_A$ the distance from the origin to A (i.e., $\Delta_A \triangleq \inf \{ r > 0 \;|\; \mu_A(r) > 0 \}$).

    \begin{lemma}
        For the closed set $A \subset \mathbb{R}^2$ with $\mu_A(r)$ non-decreasing:
        \begin{equation*}
            \sup_{r > 0} \Big( \mu_A(r) \cdot e^{-r^2/2} \Big) \leq \mathcal{G}(A) \leq e^{-\Delta^2_A/2}.
        \end{equation*}
    \end{lemma}
    \begin{proof}
        The upper-bound can be derived as follows:
        \begin{equation*}
            \mathcal{G}(A) = \int_0^\infty r \mu_A(r) \cdot e^{-r^2/2} dr \leq \int_{\Delta_A}^\infty r e^{-r^2/2} dr = e^{-\Delta_A^2/2}.
        \end{equation*}
        For the lower-bound:
        \begin{equation*}
            \mathcal{G}(A) = \int_0^\infty r \mu_A(r) \cdot e^{-r^2/2} dr \geq \mu_A(r^\prime) \int_{r^\prime}^\infty r e^{-r^2/2} dr =
            \mu_A(r^\prime) e^{-(r^\prime)^2/2},
        \end{equation*}
        for all $r^\prime > 0$. The inequality holds because $\mu_A(\cdot)$ is non-decreasing.
    \end{proof}

    Now, $K^\complement \triangleq S_{X_1, Y_1}$ is a convex set, so for its complement,
    $K \subset \mathbb{R}^2$, $\mu_K(\cdot)$ is non-decreasing.
    Using the above lemma, that fact implies the following for small $\epsilon$:
    \begin{equation*}
        \Omega(\sqrt{\epsilon} \cdot e^{-(1+\epsilon)^2\Delta^2_K/2}) \leq \mathcal{G}(K) \leq e^{-\Delta_K^2/2}.
    \end{equation*}
    The lower-bound uses the fact that $\mu_K \Big( (1 + \epsilon) \Delta_K \Big) = \Omega (\sqrt{\epsilon})$,
    because:
    \begin{equation}
        \label{equation:lsh:angular:cross-polytope:proof-lower-bound}
        \mu(K \cap \mathbb{S}_{(1+\epsilon)\Delta_K}) = (1 + \epsilon) \Delta_K \arccos \Big(
            \frac{\Delta_K}{(1 + \epsilon) \Delta_K} \Big) \approx (1 + \epsilon) \Delta_K \sqrt{\epsilon}.
    \end{equation}
    See Figure~\subref*{figure:lsh:angular:cross-polytope:lower-bound} for a helpful illustration.

    Since we are interested in the measure of $K^\complement = S_{X_1, Y_1}$, we can apply
    the result above directly to obtain:
    \begin{equation}
        \label{equation:lsh:angular:cross-polytope:proof-measure-planar-set}
        1 - e^{-\Delta(u, v)^2/2} \leq \probability [S_{X_1, Y_1}] \leq
        1 - \Omega\Big( \sqrt{\epsilon} \cdot e^{-(1+\epsilon)^2\Delta(u, v)^2/2} \Big),
    \end{equation}
    where we use the notation $\Delta_{K} = \Delta(u, v) = \min \{ u, \alpha u + \beta v \}$.

    \textbf{Step 2}. For simplicity, first consider the side of Equation~(\ref{equation:lsh:angular:cross-polytope:proof-measure-planar-set})
    that does not depend on $\epsilon$,
    and substitute that into Equation~(\ref{equation:lsh:angular:cross-polytope:prob-collision-expanded}).
    We obtain:
    \begin{align*}
        2d \int_0^1 \probability_{X_1, Y_1 \sim \mathcal{N}(0, 1)} &\Big[ \probability [ S_{X_1, Y_1} ] \geq t^{\frac{1}{d - 1}} \Big] dt \\
        &=
        2d \int_0^1 \probability_{X_1, Y_1 \sim \mathcal{N}(0, 1)} \Big[ e^{-\Delta(X_1, Y_1)^2/2} \leq 1 - t^{\frac{1}{d - 1}}  \Big] dt \\
        &= 2d \int_0^1 \probability_{X_1, Y_1 \sim \mathcal{N}(0, 1)} \Big[ \Delta(X_1, Y_1) \geq \sqrt{-2 \log \Big( 1 - t^{\frac{1}{d - 1}} \Big)}  \Big] dt. \numberthis \label{equation:lsh:angular:cross-polytope:prob-collision-noepsilon-substitute}
    \end{align*}

    \textbf{Step 3}. We are left with bounding $\probability[ \Delta(X_1, Y_1) \geq \theta ]$.
    $\Delta(X_1, Y_1) \geq \theta$ is, by definition, the set that is the intersection of two
    half-planes: $X_1 \geq \theta$ and $\alpha X_1 + \beta Y_1 \geq \theta$. If we denote this set
    by $K$, then we are again interested in the Gaussian measure of $K$. For small $\epsilon$,
    we can apply the lemma above to show that:
    \begin{equation}
        \Omega \Big( \epsilon e^{-(1 + \epsilon)^2 \Delta_K^2} \Big) \leq \mathcal{G}(K) \leq e^{-\Delta^2_K /2},
    \end{equation}
    where the constant factor in $\Omega$ depends on the angle between the two half-planes. That is
    because $\mu(K \cap \mathbb{S}_{(1+ \epsilon) \Delta_K})$ is $\epsilon$ times that angle.

    It is easy to see that $\Delta_K^2 = \frac{4}{4 - \tau^2} \cdot \theta^2$, so that we arrive at the following
    for small $\epsilon$ and every $\theta \geq 0$:
    \begin{equation}
        \label{equation:lsh:angular:cross-polytope:proof-bound-delta}
        \Omega_\tau \Big( \epsilon \cdot e^{-(1 + \epsilon)^2 \cdot \frac{4}{4 - \tau^2}\cdot \frac{\theta^2}{2}} \Big) \leq
        \probability_{X_1, Y_1 \sim \mathcal{N}(0, 1)} \Big[ \Delta(X_1, Y_1) \geq \theta \Big] \leq
        e^{- \frac{4}{4 - \tau^2}\cdot \frac{\theta^2}{2}}.
    \end{equation}

    \textbf{Step 4}. Substituting Equation~(\ref{equation:lsh:angular:cross-polytope:proof-bound-delta}) into
    Equation~(\ref{equation:lsh:angular:cross-polytope:prob-collision-noepsilon-substitute}) yields:
    \begin{align*}
        2d \int_0^1 \probability_{X_1, Y_1 \sim \mathcal{N}(0, 1)} &\Big[ \Delta(X_1, Y_1) \geq \sqrt{-2 \log \Big( 1 - t^{\frac{1}{d - 1}} \Big)}  \Big] dt \\
        &= 2d \int_0^1 \Big( 1 - t^{\frac{1}{d - 1}} \Big)^{\frac{4}{4 - \tau^2}}dt \\
        &= 2d(d-1) \int_0^1 (1 - x)^{\frac{4}{4 - \tau^2}} x^{d-2} dt \\
        &= 2d (d - 1) B\Big( \frac{8 - \tau^2}{4 - \tau^2}; d - 1 \Big) \\
        &= 2d \Theta_\tau(1) d^{-\frac{4}{4 - \tau^2}},
    \end{align*}
    where $B$ denotes the Beta function and the last step uses the Stirling approximation.

    The result above can be expressed as follows:
    \begin{equation*}
        \ln \frac{1}{\probability [ h(u) = h(v)]} = \frac{\tau^2}{4 - \tau^2} \ln d \pm \mathcal{O}_\tau(1).
    \end{equation*}

    \textbf{Step 5}. Repeating Steps 2 through 4 with the expressions that involve $\epsilon$ in
    Equations~(\ref{equation:lsh:angular:cross-polytope:proof-measure-planar-set})
    and~(\ref{equation:lsh:angular:cross-polytope:proof-bound-delta}) gives the desired result.
\end{proof}

Finally, \cite{andoni2015cross-polytope-lsh} show that,
instead of applying a random rotation using Gaussian random variables, it is sufficient to
use a pseudo-random rotation based on Fast Hadamard Transform. In effect,
they replace the random Gaussian matrix $R$ in the construction above with
three consecutive applications of $HD$, where $H$ is the Hadamard matrix
and $D$ is a random diagonal sign matrix (where the entries on the diagonal take
values from $\{\pm 1\}$).

\subsection{Euclidean Distance}
\label{section:lsh:euclidean}
\cite{datar2004pstable-lsh} proposed the first LSH family for the Euclidean distance,
$\delta(u, v) = \lVert u - v \rVert_2$. Their construction relies on the notion of
\emph{$p$-stable distributions} which we define first.

\begin{definition}[$p$-stable Distribution]
    A distribution $\mathcal{D}_p$ is said to be $p$-stable if $\sum_{i=1}^n \alpha_i Z_i$,
    where $\alpha_i \in \mathbb{R}$ and $Z_i \sim \mathcal{D}_p$, has the same distribution
    as $\lVert \alpha \rVert_p Z$, where $\alpha = [\alpha_1, \alpha_2, \ldots, \alpha_n ]$
    and $Z \sim \mathcal{D}_p$. As an example, the Gaussian distribution is $2$-stable.
\end{definition}

Let us state this property slightly differently so it is easier to understand its
connection to LSH. Suppose we have an arbitrary vector $u \in \mathbb{R}^d$.
If we construct a $d$-dimensional random vector $\alpha$ whose coordinates are independently sampled from
a $p$-stable distribution $\mathcal{D}_p$, then the inner product $\langle \alpha, u \rangle$ is
distributed according to $\lVert u \rVert_p Z$ where $Z \sim \mathcal{D}_p$.
By linearity of inner product, we can also see that $\langle \alpha, u \rangle - \langle \alpha, v \rangle$,
for two vectors $u, v \in \mathbb{R}^d$, is distributed as $\lVert u - v \rVert_p Z$.
This particular fact plays an important role in the proof of the following result.

\begin{theorem}
    For $\mathcal{X} \subset \mathbb{R}^d$ equipped with the Euclidean distance,
    a $2$-stable distribution $\mathcal{D}_2$, and the uniform distribution $U$ over
    the interval $[0, r]$,
    the following family is $(r, (1 + \epsilon)r, p(r), p((1 + \epsilon) r))$-sensitive:
    \begin{equation*}
     \mathcal{H} = \{ h_{\alpha, \beta} \;|\; h_{\alpha, \beta}(u) = 
    \lfloor \frac{\langle \alpha, u \rangle + \beta }{r} \rfloor, \; \alpha \in \mathbb{R}^d,\; \alpha_i \sim \mathcal{D}_2,\;
    \beta \sim U[0, r] \},
    \end{equation*}
    where:
    \begin{equation*}
        p(x) = \int_{t=0}^{r} \frac{1}{x} f\Big( \frac{t}{x} \Big) \Big( 1 - \frac{t}{r} \Big) dt,
    \end{equation*}
    and $f$ is the probability density function of the \emph{absolute value} of $\mathcal{D}_2$.
\end{theorem}
\begin{proof}
    The key to proving the claim is modeling the probability of a hash collision for two
    arbitrary vectors $u$ and $v$: $\probability\Big[ h_{\alpha, \beta}(u) = h_{\alpha, \beta}(v) \Big]$.
    That event can be expressed as follows:
    \begin{align*}
        \probability\Big[ h_{\alpha, \beta}(u) &= h_{\alpha, \beta}(v) \Big] = 
        \probability\Big[ \Big\lfloor \frac{\langle \alpha, u \rangle + \beta}{r} \Big\rfloor =  \Big\lfloor \frac{\langle \alpha, v \rangle + \beta}{r} \Big\rfloor \Big] \\
        &= \probability\Big[ \underbrace{\lvert \langle \alpha, u - v \rangle \rvert < r}_{\textit{Event A}} \;\land \\
        &\underbrace{\langle \alpha, u \rangle + \beta \textit{ and } \langle \alpha, v \rangle + \beta \textit{ do not straddle an integer }}_{\textit{Event B}} \Big].
    \end{align*}
    Using the $2$-stability of $\alpha$, Event A is equivalent to $\lVert u - v \rVert_2 \lvert Z \rvert < r$,
    where $Z$ is drawn from $\mathcal{D}_2$. The probability of the complement of Event B is simply
    the ratio between $\langle \alpha, u - v \rangle$ and $r$. Putting all that together, we obtain that:
    \begin{align*}
        \probability\Big[ h_{\alpha, \beta}(u) = h_{\alpha, \beta}(v) \Big] &= 
        \int_{z = 0}^{\frac{r}{\lVert u - v \rVert_2}} f(z) \Big( 1 - \frac{z \lVert u - v \rVert_2}{r} \Big) dz \\
        &= \int_{t = 0}^{r} \frac{1}{\lVert u - v \rVert_2} f(\frac{t}{\lVert u - v \rVert_2}) \Big( 1 - \frac{t}{r} \Big) dt,
    \end{align*}
    where we derived the last equality by the variable change $t = z \lVert u - v \rVert_2$.
    Therefore, if $\lVert u - v \rVert \leq x$:
    \begin{equation*}
        \probability\Big[ h_{\alpha, \beta}(u) = h_{\alpha, \beta}(v) \Big] \geq
        \int_{t = 0}^{r} \frac{1}{x} f(\frac{t}{x}) \Big( 1 - \frac{t}{r} \Big) dt = p(x).
    \end{equation*}
    It is easy to complete the proof from here.
\end{proof}

\subsection{Inner Product}
\label{section:lsh:ip}

Many of the arguments that establish the existence of an LHS family for a distance function
of interest rely on triangle inequality. Inner product as a measure of similarity, however,
does not enjoy that property. As such, developing an LSH family for inner product requires
that we somehow transform the problem from MIPS to NN search or MCS search,
as was the case in Chapter~\ref{chapter:branch-and-bound}.

Finding the right transformation that results in improved hash quality---as determined by
$\rho$---is the question that has been explored by several works in the
past~\citep{Neyshabur2015lsh-mips,shrivastava2015alsh,shrivastava2014alsh,yan2018norm-ranging-lsh}.

Let us present a simple example. Note that, we may safely assume that queries are unit
vectors (i.e., $q \in \mathbb{S}^{d-1}$), because the norm of the query does not
change the outcome of MIPS. 

Now, define the transformation $\phi_d: \mathbb{R}^d \rightarrow \mathbb{R}^{d + 1}$,
first considered by~\cite{xbox-tree}, as follows: $\phi_d(u) = [u, \; \sqrt{1 - \lVert u \rVert_2^2}]$.
Apply this transformation to data points in $\mathcal{X}$.
Clearly, $\lVert \phi_d(u) \rVert_2 = 1$ for all $u \in \mathcal{X}$.
Separately, pad the query points with a single $0$: $\phi_q(v) = [v; 0] \in \mathbb{R}^{d+1}$.

We can immediately verify that $\langle q, u \rangle = \langle \phi_q(q), \phi_d(u) \rangle$
for a query $q$ and data point $u$.
But by applying the transformations $\phi_d(\cdot)$ and $\phi_q(\cdot)$,
we have reduced the problem to MCS! As such, we may use any of existing LSH families
that we have seen for angular distance in Section~\ref{section:lsh:angular} for MIPS.

\medskip

There has been much debate over the suitability of the standard LSH framework
for inner product, with some works extending the framework to what is known as \emph{asymmetric}
LSH~\citep{shrivastava2014alsh,shrivastava2015alsh}. It turns out, however, that none of that
is necessary. In fact, as~\cite{Neyshabur2015lsh-mips} argued formally and demonstrated empirically,
the simple scheme we described above sufficiently addresses MIPS.

\section{Closing Remarks}
Much like branch-and-bound algorithms, an LSH approach to top-$k$ retrieval
rests on a solid theoretical foundation. There is a direct link between all that is
developed theoretically and the accuracy of an LSH-based top-$k$ retrieval system.

Like tree indices, too, the LSH literature is arguably mature.
There is therefore not a great deal of open questions left to investigate in its foundation,
with many recent works instead exploring learnt hash functions or
its applications in other domains.

What remains open and exciting in the context of top-$k$ retrieval, however,
is the possibility of extending the theory of LSH to explain the success of
other retrieval algorithms. We will return to this discussion in
Chapter~\ref{chapter:ivf}.

\bibliographystyle{abbrvnat}
\bibliography{biblio}

\chapter{Graph Algorithms}
\label{chapter:graph}

\abstract{
We have seen two major classes of algorithms that approach the top-$k$ retrieval
problem in their own unique ways. One recursively partitions a vector collection to model its
geometry, and the other hashes the vectors into predefined buckets to reduce the search space.
Our next class of algorithms takes yet a different view of the question.
At a high level, our third approach is to ``walk'' through a collection,
hopping from one vector to another, where every hop
gets us \emph{spatially} closer to the optimal solution.
This chapter reviews algorithms that use a graph data structure
to implement that idea.
}

\section{Intuition}
\label{section:graph:intuition}

The most natural way to understand a spatial walk through a collection
of vectors is by casting it as traversing a (directed) connected graph.
As we will see, whether the graph is directed or not depends on the specific algorithm itself.
But the graph must regardless be \emph{connected}, so that there always exists at least
one path between every pair of nodes. This ensures that we can walk through the graph
no matter where we begin our traversal.

Let us write $G(\mathcal{V}, \mathcal{E})$ to refer to such a graph,
whose set of \emph{vertices} or \emph{nodes} are denoted by $\mathcal{V}$,
and its set of edges by $\mathcal{E}$.
So, for $u, v \in \mathcal{V}$ in a directed graph,
if $(u, v) \in \mathcal{E}$, we may freely move from node $u$ to node $v$.
Hopping from $v$ to $u$ is not possible if $(v, u) \notin \mathcal{E}$.
Because we often need to talk about the set of nodes that can be reached by a single hop
from a node $u$---known as the neighbors of $u$---we give it a special symbol and
define that set as follows: $N(u) = \{ v \;|\; (u, v) \in \mathcal{E} \}$.

The idea behind the algorithms in this chapter is to construct a graph in the pre-processing phase
and use that as an index of a vector collection for top-$k$ retrieval.
To do that, we must decide what is a node in the graph (i.e., define the set $\mathcal{V}$),
how nodes are linked to each other ($\mathcal{E}$), and, importantly, what the search algorithm looks like.

\smallskip

The set of nodes $\mathcal{V}$ is easy to construct:
Simply designate every vector in the collection $\mathcal{X}$
as a unique node in $G$, so that $\lvert \mathcal{X} \rvert = \lvert \mathcal{V} \rvert$.
There should, therefore, be no ambiguity if we referred to a node
as a vector. We use both terms interchangeably.

What properties should the edge set $\mathcal{E}$ have?
To get a sense of what is required of the edge set,
it would help to consider the search algorithm first.
Suppose we are searching for the top-$1$ vector closest to query $q$,
and assume that we are, at the moment, at an arbitrary node $u$ in $G$.

From node $u$, we can have a look around and assess if any of our neighbors
in $N(u)$ is closer to $q$. By doing so, we find ourselves in one of two situations.
Either we encounter no such neighbor, so that $u$ has the smallest
distance to $q$ among its neighbors. If that happens, ideally, we want $u$ to also have
the smallest distance to $q$ among \emph{all} vectors. In other words, in an ideal graph,
a local optimum coincides with the global optimum.

\begin{figure}[t]
    \centering
    \includegraphics[width=0.7\linewidth]{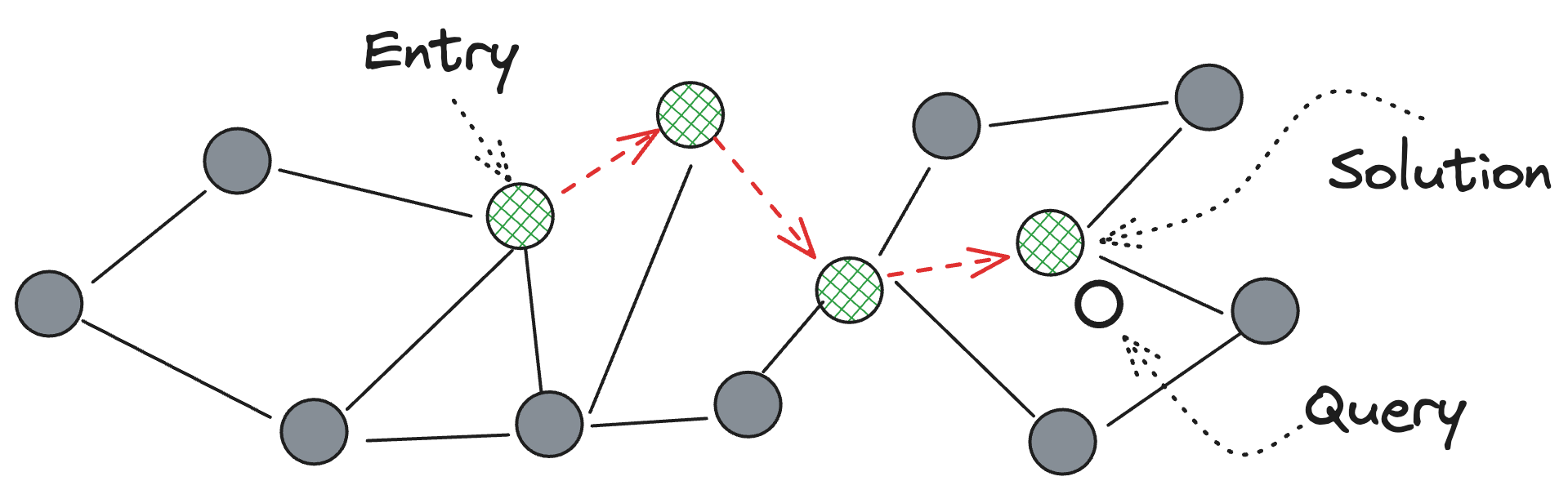}
    \caption{Illustration of the greedy traversal algorithm for finding
    the top-$1$ solution on an example (undirected) graph. The procedure enters
    the graph from an arbitrary ``entry'' node. It then compares the distance of the node to query $q$
    with the distance of its neighbors to $q$, and either terminates if no neighbor is closer
    to $q$ than the node itself, or advances to the closest neighbor. It repeats this procedure
    until the terminal condition is met. The research question in
    this chapter concerns the construction of the edge set: How do we construct a sparse graph
    which can be traversed greedily while providing guarantees on the (near-)optimality
    of the solution}
    \label{figure:graphs:greedy-traversal}
\end{figure}

Alternatively, we may find one such neighbor $v \in N(u)$ for which $\delta(q, v) < \delta(q, u)$
and $v = \argmin_{w \in N(u)} \delta(q, w)$.
In that case, the ideal graph is one where the following event takes place: If we moved from
$u$ to $v$, and repeated the process above in the context of $N(v)$ and so on, we will ultimately
arrive at a local optimum (which, by the previous condition, is the global optimum).
Terminating the algorithm then would therefore give us the optimal solution to the
top-$1$ retrieval problem.

\begin{svgraybox}
Put differently, in an ideal graph, if moving from a node to any of its neighbors
does not get us spatially closer to $q$, it is because the current
node is the optimal solution to the top-$1$ retrieval problem for $q$.
\end{svgraybox}

On a graph with that property,
the procedure of starting from any node in the graph,
hopping to a neighbor that is closer to $q$, and repeating this procedure until no such neighbor
exists, gives the optimal solution.
That procedure is the familiar \emph{best-first-search} algorithm,
which we illustrate on a toy graph in Figure~\ref{figure:graphs:greedy-traversal}.
That will be our base search algorithm for top-$1$ retrieval.

\begin{algorithm}[!t]
\SetAlgoLined
{\bf Input: }{Graph $G=(\mathcal{V}, \mathcal{E})$ over collection $\mathcal{X}$ 
with distance $\delta(\cdot, \cdot)$; query point $q$; entry node $s \in \mathcal{V}$;
retrieval depth $k$.}\\
\KwResult{Exact top-$k$ solution for $q$.}

\begin{algorithmic}[1]

\STATE $Q \leftarrow \{ s \}$ \Comment*[l]{\footnotesize $Q$ is a priority queue}

\WHILE{$Q$ changed in the previous iteration}
    \STATE $\mathcal{S} \leftarrow \bigcup_{u \in Q} N(u)$
    \STATE $v \leftarrow \argmin_{u \in \mathcal{S}} \delta(q, u)$

    \STATE $Q.\textsc{Insert}(v)$
    \IF{$\lvert Q \rvert \geq k$}
        \STATE $Q.\textsc{Pop}()$ \Comment*[l]{\footnotesize Removes the node with the largest distance to $q$}
    \ENDIF
\ENDWHILE

\RETURN $Q$
\end{algorithmic}
\caption{Greedy search algorithm for top-$k$ retrieval over a graph index.}
\label{algorithm:graphs:greedy-search}
\end{algorithm}

Extending the search algorithm to top-$k$ requires a minor modification to the procedure above.
It begins by initializing a priority queue of size $k$. When we visit a new node, we add it to
the queue if its distance with $q$ is smaller than the minimum distance among the nodes already in the queue.
We keep moving from a node in the queue to its neighbors until the queue stabilizes (i.e., no
unseen neighbor of any of the nodes in the queue has a smaller distance to $q$).
This is described in Algorithm~\ref{algorithm:graphs:greedy-search}.

Note that, assuming $\delta(\cdot, \cdot)$ is proper, it is easy to see that
the top-$1$ optimality guarantee immediately implies top-$k$
optimality---you should verify this claim as an exercise.
It therefore suffices to state our requirements in terms of top-$1$ optimality alone.
So, ideally, $\mathcal{E}$ should guarantee that traversing $G$
in a best-first-search manner yields the optimal top-$1$ solution.

\subsection{The Research Question}

It is trivial to construct an edge set that provides the desired optimality guarantee:
Simply add an edge between every pair of nodes, completing the graph!
The greedy search algorithm described above will take us to the optimal solution.

However, such a graph not only has high space complexity, but it also has a linear query time complexity.
That is because, the very first step (which also happens to be the last step)
involves comparing the distance of $q$ to the entry node,
with the distance of $q$ to every other node in the graph! We are better off exhaustively
scanning the entire collection in a flat index.

\begin{svgraybox}
The research question that prompted the algorithms we are about to study in this chapter is
whether there exists a relatively \emph{sparse} graph that has the optimality guarantee we seek
or that can instead provide guarantees for the more relaxed,
$\epsilon$-approximate top-$k$ retrieval problem.
\end{svgraybox}

As we will learn shortly, with a few notable exceptions,
all constructions of $\mathcal{E}$ proposed thus far in the literature 
for high-dimensional vectors amount to
heuristics that attempt to \emph{approximate} a theoretical graph but
come with no guarantees. In fact, in almost all cases,
their worst-case complexity is no better than exhaustive search.
Despite that, many of these heuristics work remarkably well
in practice on real datasets, making graph-based methods one of the most
widely adopted solutions to the approximate top-$k$ retrieval problem.

\smallskip

In the remainder of this chapter, we will see classes of theoretical graphs
that were developed in adjacent scientific disciplines, but that are seemingly suitable
for the (approximate) top-$k$ retrieval problem. As we introduce these graphs,
we also examine representative algorithms that aim to build an approximation
of such graphs in high dimensions, and review their properties.

We note, however, that the literature on graph-based methods is vast and growing still.
There is a plethora of studies that experiment with (minor or major) adjustments to the basic idea
described earlier, or that empirically compare and contrast different algorithmic flavors
on real-world datasets.
This chapter does not claim to, nor does it intend to cover the explosion of
material on graph-based algorithms. Instead, it limits its scope to the foundational
principles and ground-breaking works that are theoretically somewhat interesting.
We refer the reader to existing reports and surveys for the full spectrum of
works on this topic~\citep{wang2021survey-graph-ann,li2020survey-ann}.

\section{The Delaunay Graph}
One classical graph that satisfies the conditions we seek and guarantees the optimality
of the solution obtained by best-first-search traversal
is the Delaunay graph~\citep{Delaunay_1934aa,fortune1997voronoi}.
It is easier to understand the construction of the Delaunay graph if we consider instead
its dual: the Voronoi diagram. So we begin with a description of the Voronoi diagram and
Voronoi regions.

\subsection{Voronoi Diagram}

For the moment, suppose $\delta$ is the Euclidean distance and that we are in $\mathbb{R}^2$.
Suppose further that we have a collection $\mathcal{X}$ of just two points $u$ and $v$
on the plane.
Consider now the subset of $\mathbb{R}^2$ comprising of all the points
to which $u$ is the closest point from $\mathcal{X}$.
Similarity, we can identify the subset to which $v$ is the closest point.
These two subsets are, in fact, partitions of the plane and are separated by a
line---the points on this line are equidistant to $u$ and $v$.
In other words, two points in $\mathbb{R}^2$ induce a partitioning of 
the plane where each partition is ``owned'' by a point and describes the
set of points that are closer to it than they are to the other point.

\begin{figure}[t]
    \centering
    \subfloat[Voronoi diagram]{
        \label{figure:graphs:delaunay:voronoi}
        \includegraphics[width=0.33\linewidth]{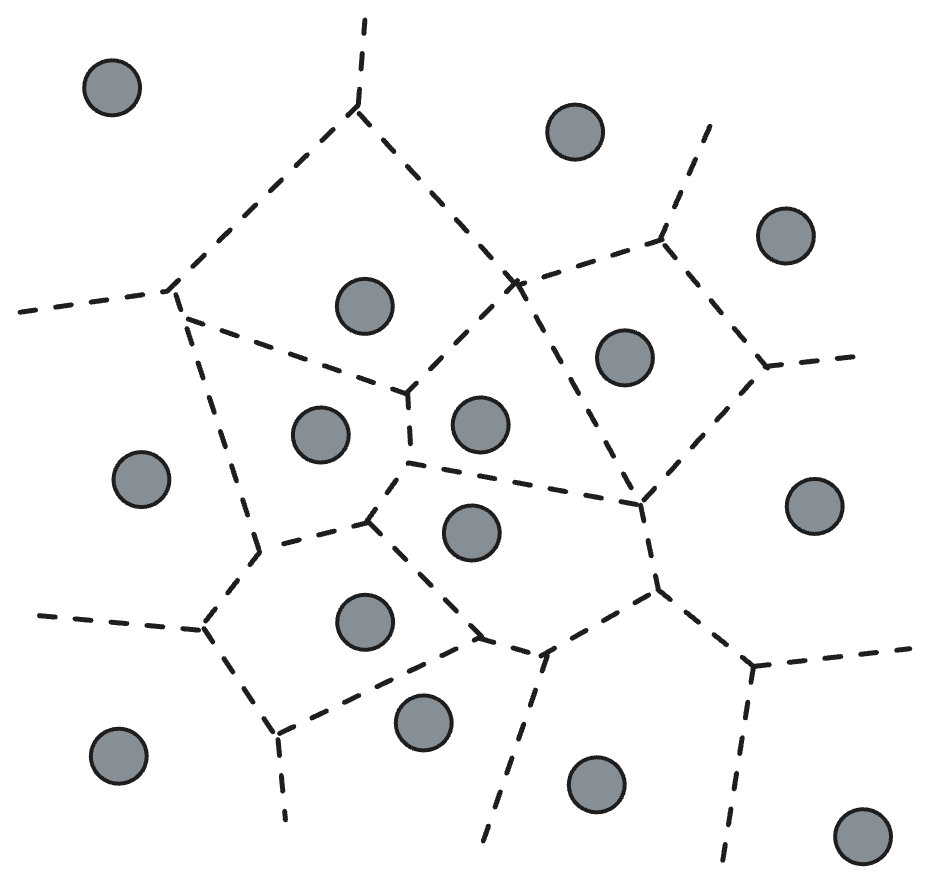}
    }
    \subfloat[Delaunay graph]{
        \label{figure:graphs:delaunay:graph}
        \includegraphics[width=0.33\linewidth]{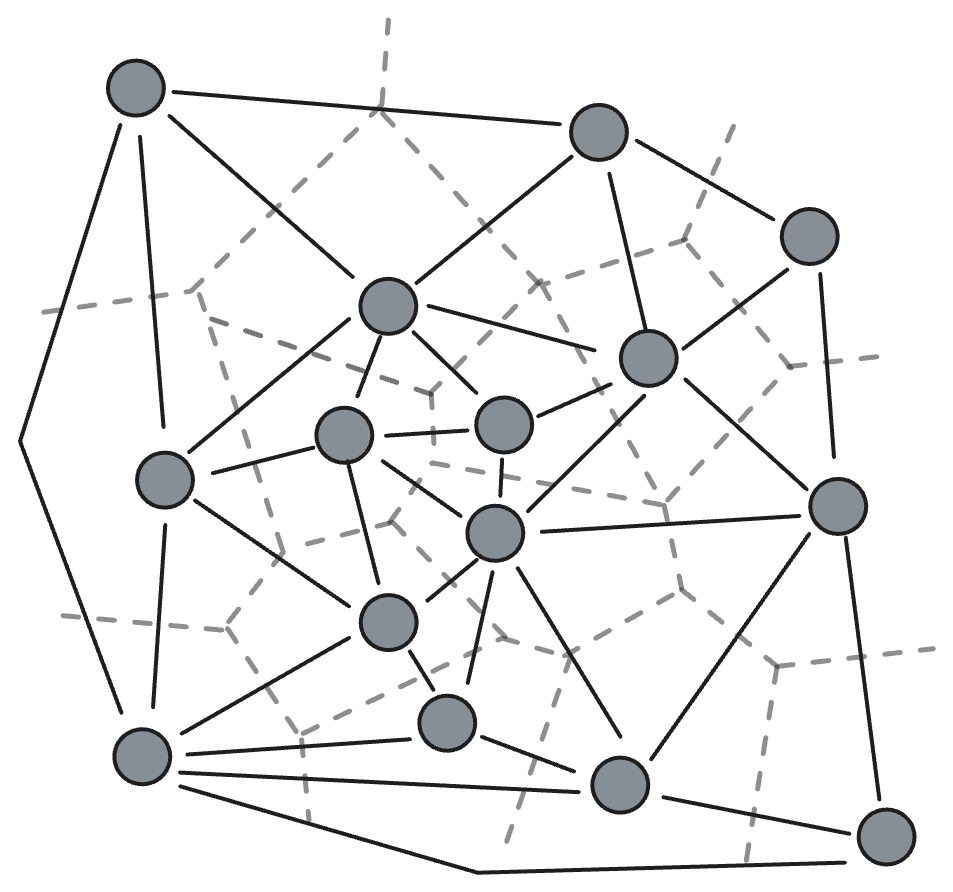}
    }
    \caption{Visualization of the Voronoi diagram (a) and its dual, the Delaunay graph (b)
    for an example collection $\mathcal{X}$ of points in $\mathbb{R}^2$.
    A Voronoi region associated with a point $u$
    (shown here as the area contained within the dashed lines)
    is a set of points whose nearest neighbor in $\mathcal{X}$
    is $u$. The Delaunay graph is an undirected graph whose nodes are points in $\mathcal{X}$ and two nodes
    are connected (shown as solid lines) if their Voronoi regions have a non-empty intersection.}
    \label{figure:graphs:delaunay}
\end{figure}

We can trivially generalize that notion to more than two points and, indeed, to higher dimensions.
A collection $\mathcal{X}$ of points in $\mathbb{R}^d$ partitions the space
into unique regions $\mathcal{R} = \bigcup_{u \in \mathcal{X}} \mathcal{R}_u$,
where the region $\mathcal{R}_u$ is owned by point
$u \in \mathcal{X}$ and represents the set of points to which $u$ is the closest
point in $\mathcal{X}$. Formally, $\mathcal{R}_u = \{ x \;|\; u = \argmin_{v \in \mathcal{X}} \delta(x, v) \}$. Note that, each region is a convex polytope that is the intersection of
half-spaces.
The set of regions is known as the Voronoi diagram for the collection $\mathcal{X}$ and is
illustrated in Figure~\subref*{figure:graphs:delaunay:voronoi} for an example collection in $\mathbb{R}^2$.

\subsection{Delaunay Graph}

The Delaunay graph for $\mathcal{X}$ is, in effect, a graph representation of its
Voronoi diagram.
The nodes of the graph are trivially the points in $\mathcal{X}$, as before.
We place an edge between two nodes $u$ and $v$ if their Voronoi regions have a non-empty
intersection: $\mathcal{R}_u \cap \mathcal{R}_v \neq \emptyset$. Clearly, by construction,
this graph is undirected. An example
of this graph is rendered in Figure~\subref*{figure:graphs:delaunay:graph}.

\medskip 

There is an important technical detail that is worth noting.
The Delaunay graph for a collection $\mathcal{X}$ is unique if the points
in $\mathcal{X}$ are in \emph{general position}~\citep{fortune1997voronoi}.
A collection of points are said to be in general position if the following two
conditions are satisfied. First, no $n$ points from $\mathcal{X} \subset \mathbb{R}^d$,
for $2 \leq n \leq d + 1$, must lie on a $(n-2)$-flat. Second, no $n+1$ points
must lie on any $(n-2)$-dimensional hypersphere. In $\mathbb{R}^2$, as an example,
for a collection of points to be in general position,
no three points may be co-linear, and no four points co-circular.

We must add that, the detail above is generally satisfied in practice.
Importantly, if the vectors in our collection are
independent and identically distributed, then the collection is almost surely in general
position. That is why we often take that technicality for granted.
So from now on, we assume that the Delaunay graph of a collection of points is unique.

\subsection{Top-\texorpdfstring{$1$}{1} Retrieval}
We can immediately recognize the importance of Voronoi regions: They geometrically
capture the set of queries for which a point from the collection is the solution
to the top-$1$ retrieval problem. But what is the significance of the dual representation
of this geometrical concept? How does the Delaunay graph help us solve the top-$1$
retrieval problem?

For one, the Delaunay graph is a compact representation of the Voronoi diagram.
Instead of describing polytopes, we need only to record edges between neighboring nodes.
But, more crucially, as the following claim shows, we can traverse the Delaunay graph greedily,
and reach the optimal top-$1$ solution from any node. In other words, the Delaunay graph
has the desired property we described in Section~\ref{section:graph:intuition}.

\begin{theorem}
    \label{theorem:graphs:delaunay}
    Let $G = (\mathcal{V}, \mathcal{E})$ be a graph that contains the Delaunay graph of $m$ vectors
    $\mathcal{X} \subset \mathbb{R}^d$.
    The best-first-search algorithm over $G$ gives the optimal solution to the top-$1$ retrieval problem
    for any arbitrary query $q$ if $\delta(\cdot, \cdot)$ is proper.
\end{theorem}

The proof of the result above relies on an important property of the Delaunay graph,
which we state first.

\begin{lemma}
\label{lemma:graphs:delaunay}
    Let $G = (\mathcal{V}, \mathcal{E})$ be the Delaunay graph of
    a collection of points $\mathcal{X} \subset \mathbb{R}^d$,
    and let $B$ be a ball centered at $\mu$ that contains two
    points $u, v \in \mathcal{X}$, with radius $r = \min (\delta(\mu, u), \delta(\mu, v))$,
    for a continuous and proper distance function $\delta(\cdot, \cdot)$.
    Then either $(u, v) \in \mathcal{E}$ or there exists a third point in $\mathcal{X}$
    that is contained in $B$.
\end{lemma}
\begin{proof}
Suppose there is no other point in $\mathcal{X}$ that is contained in $B$.
We must show that, in that case, $(u, v) \in \mathcal{E}$.

There are two cases.
The first and easy case is when $u$ and $v$ are on the surface of $B$.
Clearly, $u$ and $v$ are equidistant from $\mu$.
Because there are no other points in $B$,
we can conclude that $\mu$ lies in the intersection of $\mathcal{R}_u$
and $\mathcal{R}_v$, the Voronoi regions associated with $u$ and $v$.
That implies $(u, v) \in \mathcal{E}$.

\begin{figure}[t]
    \centering
    \includegraphics[width=0.25\linewidth]{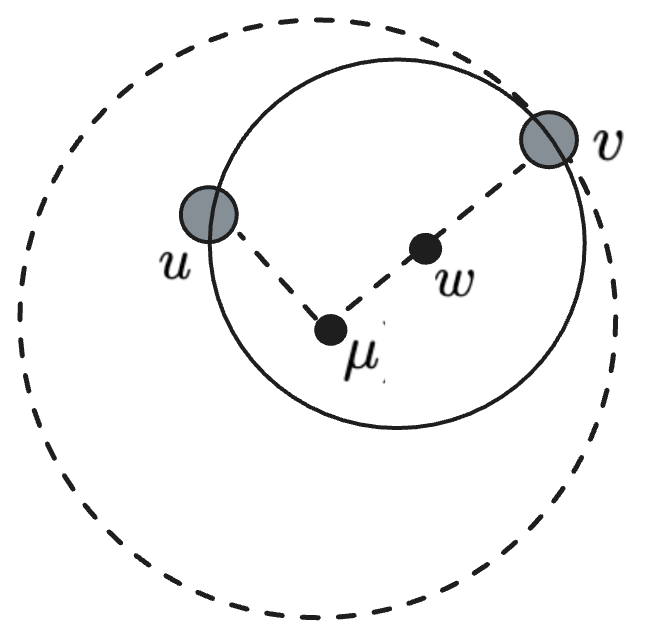}
    \caption{Illustration of the second case in the proof of Lemma~\ref{lemma:graphs:delaunay}.}
    \label{figure:graphs:delaunay:proof-lemma}
\end{figure}

In the second case, suppose $\delta(\mu, u) < \delta(\mu, v)$, so that $v$ is on the surface of $B$
and $u$ is in its interior. Consider the function
$f(\omega) = \delta(v, \omega) - \delta(u, \omega)$.
Clearly, $f(v) < 0$ and $f(\mu) > 0$. Therefore, there must be a point $w \in B$
on the line segment $\mu + \lambda (v - \mu)$ for $\lambda \in [0, 1]$ for which
$f(w) = 0$. That implies that $\delta(w, u) = \delta(w, v)$.
Furthermore, $v$ is the closest point on the surface of $B$ to $w$,
so that the ball centered at $w$ with radius $\delta(w, v)$ is entirely contained
in $B$. This is illustrated in Figure~\ref{figure:graphs:delaunay:proof-lemma}.

Importantly, no other point in $\mathcal{X}$ is closer to $w$ than $u$ and $v$.
So $w$ rests in the intersection of $\mathcal{R}_u$ and $\mathcal{R}_v$, and
$(u, v) \in \mathcal{E}$.
\end{proof}

\begin{proof}[Proof of Theorem~\ref{theorem:graphs:delaunay}]
    We prove the result for the case where $\delta$ is the Euclidean
    distance and leave the proof of the more general case as an exercise.
    (Hint: To prove the general case you should make the line segment
    argument as in the proof of Lemma~\ref{lemma:graphs:delaunay}.)

    Suppose the greedy search for $q$ stops at some local optimum $u$
    that is different from the global optimum, $u^\ast$,
    and that $(u, u^\ast) \notin \mathcal{E}$---otherwise,
    the algorithm must terminate at $u^\ast$ instead.
    Let $r = \delta(q, u)$.
    
    By assumption we have that the ball centered at $q$ with radius $r$,
    $B(q, r)$, is non-empty because it must contain $u^\ast$ whose distance to $q$ is less than $r$.
    Let $v$ be the point in this ball that is closest to $u$.
    Consider now the ball $B((u + v)/2, \delta(u, v)/2)$.
    This ball is empty: otherwise $v$ would not be the closest point to $u$.
    By Lemma~\ref{lemma:graphs:delaunay}, we must have that $(u, v) \in \mathcal{E}$.
    This is a contradiction because the greedy search cannot stop at $u$.
\end{proof}

\begin{svgraybox}
Notice that Theorem~\ref{theorem:graphs:delaunay} holds for any graph
that \emph{contains} the Delaunay graph.
The next theorem strengthens this result to show that the Delaunay graph represents
the minimal edge set that guarantees an optimal solution through greedy traversal.
\end{svgraybox}

\begin{theorem}
    \label{theorem:graphs:delaunay-is-minimal}
    The Delaunay graph is the minimal graph over which the best-first-search algorithm
    gives the optimal solution to the top-$1$ retrieval problem.
\end{theorem}

In other words, if a graph does not contain the Delaunay graph, then we can find queries
for which the greedy traversal from an entry point does not produce the optimal top-$1$
solution.

\begin{proof}[Proof of Theorem~\ref{theorem:graphs:delaunay-is-minimal}]
    Suppose that the data points $\mathcal{X}$ are in general position.
    Suppose further that $G = (\mathcal{V}, \mathcal{E})$ is a graph built from $\mathcal{X}$,
    and that $u$ and $v$ are two nodes in the graph.
    Suppose further that $(v, u) \notin \mathcal{E}$ but that that edge
    exists in the Delaunay graph of $\mathcal{X}$.

    If we could sample a query point $q$ such that
    $\delta(q, u) < \delta(q, v)$ but 
    $\delta(q, w) > \max(\delta(q, u), \delta(q, v))$ for all $w \neq u, v$,
    then we are done.
    That is because, if we entered the graph through $v$, then $v$ is a local optimum
    in its neighborhood: all other points that are connected to $v$ have a distance larger than $\delta(q, v)$.
    But $v$ is not the globally optimal solution, so that the greedy traversal does not converge to the
    optimal solution.

    It remains to show that such a point $q$ always exists.
    Suppose it did not. That is, for any point that is in the Voronoi region of $u$,
    there is a data point $w \neq v$ that is closer to it than $v$.
    If that were the case, then no ball whose boundary passes through $u$ and $v$
    can be empty, which contradicts Lemma~\ref{lemma:graphs:delaunay} (the ``empty-circle''
    property of the Delaunay graph).
\end{proof}

As a final remark on the Delaunay graph and its use in top-$1$ retrieval, we note that
the Delaunay graph only makes sense if we have precise knowledge of the structure of the space
(i.e., the metric). It is not enough to have just pairwise distances between points in a collection
$\mathcal{X}$. In fact,~\cite{navarro2002spatial-approximateion} showed that if pairwise
distances are all we know about a collection of points, then the only sensible
graph that contains the Delaunay graph and is amenable to greedy search is the complete graph.
This is stated as the following theorem.

\begin{theorem}
    Suppose the structure of the metric space is unknown, but we have pairwise distances between
    the points in a collection $\mathcal{X}$, due to an arbitrary, but proper distance function $\delta$.
    For every choice of $u, v \in \mathcal{X}$,
    there is a choice of the metric space such that $(u, v) \in \mathcal{E}$,
    where $G = (\mathcal{V}, \mathcal{E})$ is a Delaunay graph for $\mathcal{X}$.
\end{theorem}
\begin{proof}
    The idea behind the proof is to assume $(u, v) \notin \mathcal{E}$, then
    construct a query point that necessitates the existence of an edge between $u$ and $v$.
    To that end, consider a query point $q$ such that its distance to $u$ is $C + \epsilon$ for some
    constant $C$ and $\epsilon > 0$, its distance to $v$ is $C$, and its distance to 
    every other point in $\mathcal{X}$ is $C + 2 \epsilon$.

    This is a valid arrangement if we choose $\epsilon$ such that
    $\epsilon \leq 1/2 \min_{x, y \in \mathcal{X}} \delta(x, y)$
    and $C$ such that $C \geq 1/2 \max_{x, y \in \mathcal{X}} \delta(x, y)$. It is easy to verify that,
    if those conditions hold, a point $q$ with the prescribed distances can exist as the distances
    do not violate any of the triangle inequalities.

    Consider then a search starting from node $u$. If $(u, v) \notin \mathcal{E}$, then for the search
    algorithm to walk from $u$ to the optimal solution, $v$, it must first get \emph{farther} from $q$.
    But we know by the properties of the Delaunay graph that such an event implies that
    $u$ (which would be the local optimum) must be the global optimum. That is clearly not true.
    So we must have that $(u, v) \in \mathcal{E}$, giving the claim.
\end{proof}

\subsection{Top-\texorpdfstring{$k$}{k} Retrieval}

Let us now consider the general case of top-$k$ retrieval over the Delaunay graph.
The following result states that Algorithm~\ref{algorithm:graphs:greedy-search} is correct
if executed on any graph that contains the Delaunay graph, in the sense that it returns the
optimal solution to top-$k$ retrieval.

\begin{theorem}
    \label{theorem:graphs:delaunay:top-k}
    Let $G = (\mathcal{V}, \mathcal{E})$ be a graph that contains the Delaunay graph of $m$ vectors
    $\mathcal{X} \subset \mathbb{R}^d$.
    Algorithm~\ref{algorithm:graphs:greedy-search} over $G$ gives the optimal solution to the top-$k$ retrieval problem
    for any arbitrary query $q$ if $\delta(\cdot, \cdot)$ is proper.
\end{theorem}
\begin{proof}
    As with the proof of Theorem~\ref{theorem:graphs:delaunay},
    we show the result for the case where $\delta$ is the Euclidean
    distance and leave the proof of the more general case as an exercise.
    
    The proof is similar to the proof of Theorem~\ref{theorem:graphs:delaunay} but the argument
    needs a little more care when $k > 1$.
    Suppose Algorithm~\ref{algorithm:graphs:greedy-search} for $q$ stops at some local optimum set
    $Q$ that is different from the global optimum, $Q^\ast$. In other words,
    $Q \;\triangle\; Q^\ast \neq \emptyset$ where $\triangle$ denotes the symmetric difference between
    sets.

    Let $r = \max_{u \in Q} \delta(q, u)$ and consider the ball $B(q, r)$.
    Because $Q \;\triangle\; Q^\ast \neq \emptyset$, there must be at least $k$ points in the interior of this ball.
    Let $v \notin Q$ be a point in the interior and suppose $u \in Q$ is its closest point in the ball.
    Clearly, the ball $B((u + v)/2, \delta(u, v)/2)$ is empty:
    otherwise $v$ would not be the closest point to $u$.
    By Lemma~\ref{lemma:graphs:delaunay}, we must have that $(u, v) \in \mathcal{E}$.
    This is a contradiction because Algorithm~\ref{algorithm:graphs:greedy-search} would, before termination,
    place $v$ in $Q$ to replace the node that is on the surface of the ball.
\end{proof}

\subsection{The \texorpdfstring{$k$}{k}-NN Graph}

From our discussion of Voronoi diagrams and Delaunay graphs, it appears as though we have
found the graph we have been looking for. Indeed, the Delaunay graph of a collection of vectors
gives us the exact solution to top-$k$ queries, using such a strikingly simple search algorithm.
Sadly, the story does not end there and, as usual,
the relentless curse of dimensionality poses a serious challenge.

The first major obstacle in high dimensions actually concerns
the construction of the Delaunay graph itself. While there are many
algorithms~\citep{edelsbrunner1992delaunay,guibas1992delaunay,guibas1985voronoi}
that can be used to construct the Delaunay graph---or, to be more precise,
to perform Delaunay \emph{triangulation}---all suffer from an exponential dependence
on the number of dimensions $d$. So building the graph itself seems infeasible when $d$
is too large.

\begin{figure}[t]
    \centering
    \subfloat[Delaunay graph]{
        \includegraphics[width=0.33\linewidth]{figures/graphs-delaunay-graph-overlay.png}
    }
    \subfloat[$2$-NN graph]{
        \includegraphics[width=0.33\linewidth]{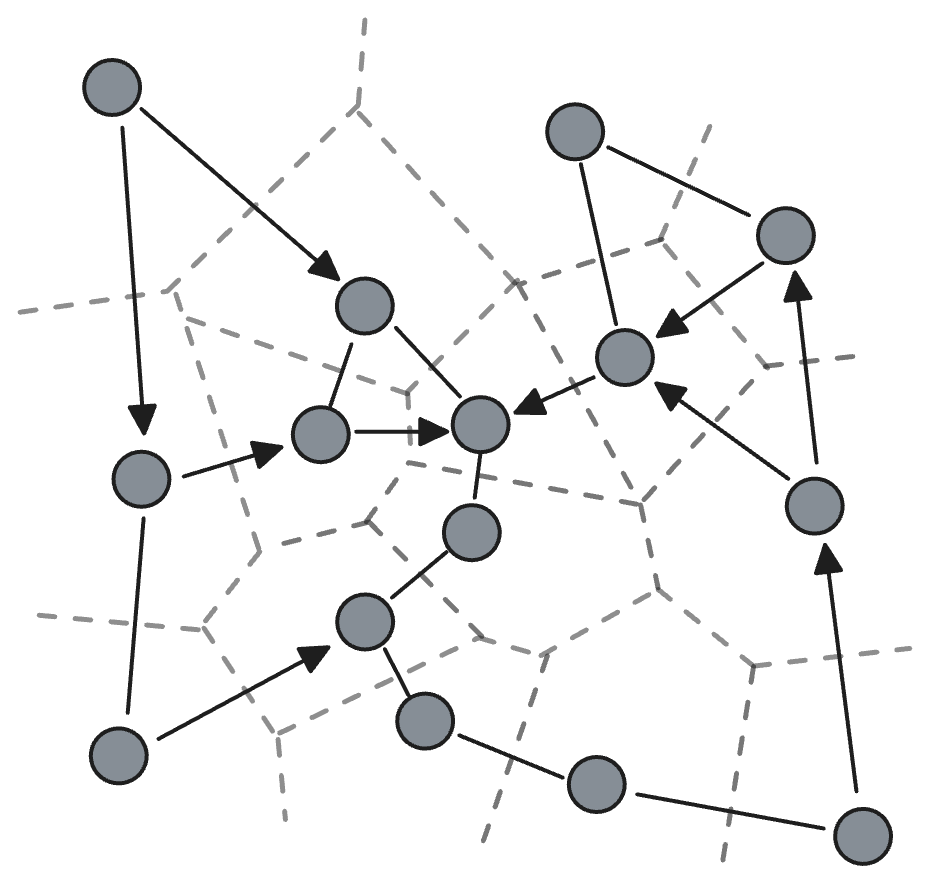}
    }
    \caption{Comparison of the Delaunay graph (a) with the $k$-NN graph for $k=2$ (b)
    for an example collection in $\mathbb{R}^2$.
    In the illustration of the directed $k$-NN graph, edges that go in both directions
    are rendered as lines without arrow heads.
    Notice that, the top left node cannot be reached from the rest
    of the graph.}
    \label{figure:graphs:knn-graph}
\end{figure}

Even if we were able to quickly construct the Delaunay graph for a large collection of
points, we would face a second debilitating issue: The graph is close to complete!
While exact bounds on the expected number of edges in the graph surely depend on the data
distribution, in high dimensions the graph becomes necessarily more dense.
Consider, for example, vectors that are independent and identically-distributed
in each dimension. Recall from our discussion from Chapter~\ref{chapter:instability},
that in such an arrangement of points, the distance between any pair of points tends to
concentrate sharply. As a result, the Delaunay graph has an edge between almost every
pair of nodes.

\bigskip

These two problems are rather serious, rendering the guarantees of the Delaunay graph for
top-$k$ retrieval mainly of theoretical interest. These same difficulties motivated
research to \emph{approximate} the Delaunay graph. One prominent method is known
as the $k$-NN graph~\citep{chavez2010knn,hajebi2011fast,fu2016efanna}.

The $k$-NN graph is simply a $k$-regular graph where every node (i.e., vector) is
connected to its $k$ closest nodes. So $(u, v) \in \mathcal{E}$ if $v \in \argmin^{(k)}_{w \in \mathcal{X}} \delta(u, w)$. Note that, the resulting graph may be directed, depending on the choice of $\delta$.
We should mention, however, that researchers have explored ways of turning the $k$-NN graph
into an undirected graph~\citep{chavez2010knn}. An example is depicted in
Figure~\ref{figure:graphs:knn-graph}.

We must remark on two important properties of the $k$-NN graph.
First, the graph itself is far more efficient to construct than the Delaunay
graph~\citep{chen2009fast-knngraph-construction,vaidya1989knn-construction,connor2010fast-knngraph,dong2011knn-graph}.
The second point concerns the connectivity of the graph.
As \cite{Brito1997knn-graph} show, under mild conditions governing the distribution of the vectors
and with $k$ large enough, the resulting graph has a high probability of being connected.
When $k$ is too small, on the other hand, the resulting graph may become too sparse,
leading the greedy search algorithm to get stuck in local minima.

Finally, at the risk of stating the obvious, the $k$-NN graph does not enjoy any of
the guarantees of the Delaunay graph in the context of top-$k$ retrieval.
That is simply because the $k$-NN graph is likely only a sub-graph of the Delaunay graph,
while Theorems~\ref{theorem:graphs:delaunay} and~\ref{theorem:graphs:delaunay:top-k}
are provable only for super-graphs of the Delaunay graph.
Despite these deficiencies, the $k$-NN graph remains an important component
of advanced graph-based, approximate top-$k$ retrieval algorithms.

\subsection{The Case of Inner Product}
Everything we have stated so far about the Voronoi diagrams and its duality with
the Delaunay graph was contingent on $\delta(\cdot, \cdot)$ being proper.
In particular, the proof of the optimality guarantees implicitly require non-negativity
and the triangle inequality. As a result, none of the results apply to MIPS \emph{prima facie}.
As it turns out, however, we could extend the definition of Voronoi regions and the Delaunay graph
to inner product, and present guarantees for MIPS (with $k=1$, but not with $k>1$).
That is the proposal by \cite{morozov2018ip-nsw}.

\subsubsection{The IP-Delaunay Graph}
Let us begin by characterizing the Voronoi regions for inner product.
The Voronoi region $\mathcal{R}_u$ of a vector $u \in \mathcal{X}$ comprises of
the set of points for which $u$ is the maximizer of inner product:
\begin{equation*}
    \mathcal{R}_u = \{ x \in \mathbb{R}^d \;|\; u = \argmax_{v \in \mathcal{X}} \langle x, v \rangle \}.
\end{equation*}
This definition is essentially the same as how we defined the Voronoi region for a proper $\delta$,
and, indeed, the resulting Voronoi diagram is a partitioning of the whole space.
The properties of the resulting Voronoi regions, however, could not be more different.

First, recall from Section~\ref{section:flavors:flavors:mips} that inner product
does not even enjoy what we called coincidence. That is, in general, $u = \argmax_{v \in \mathcal{X}} \langle u, v \rangle$
is not guaranteed. So it is very much possible that $\mathcal{R}_u$ is empty for some
$u \in \mathcal{X}$. Second, when $\mathcal{R}_u \neq \emptyset$, it is a convex cone
that is the intersection of half-spaces that \emph{pass through the origin}. So the Voronoi
regions have a substantially different geometry. Figure~\subref*{figure:graphs:ip-delaunay:ip-voronoi}
visualizes this phenomenon.

\begin{figure}[t]
    \centering
    \subfloat[Euclidean Delaunay]{
        \includegraphics[width=0.32\linewidth]{figures/graphs-delaunay-graph-overlay.png}
    }
    \subfloat[Inner Product Voronoi]{
        \label{figure:graphs:ip-delaunay:ip-voronoi}
        \includegraphics[width=0.32\linewidth]{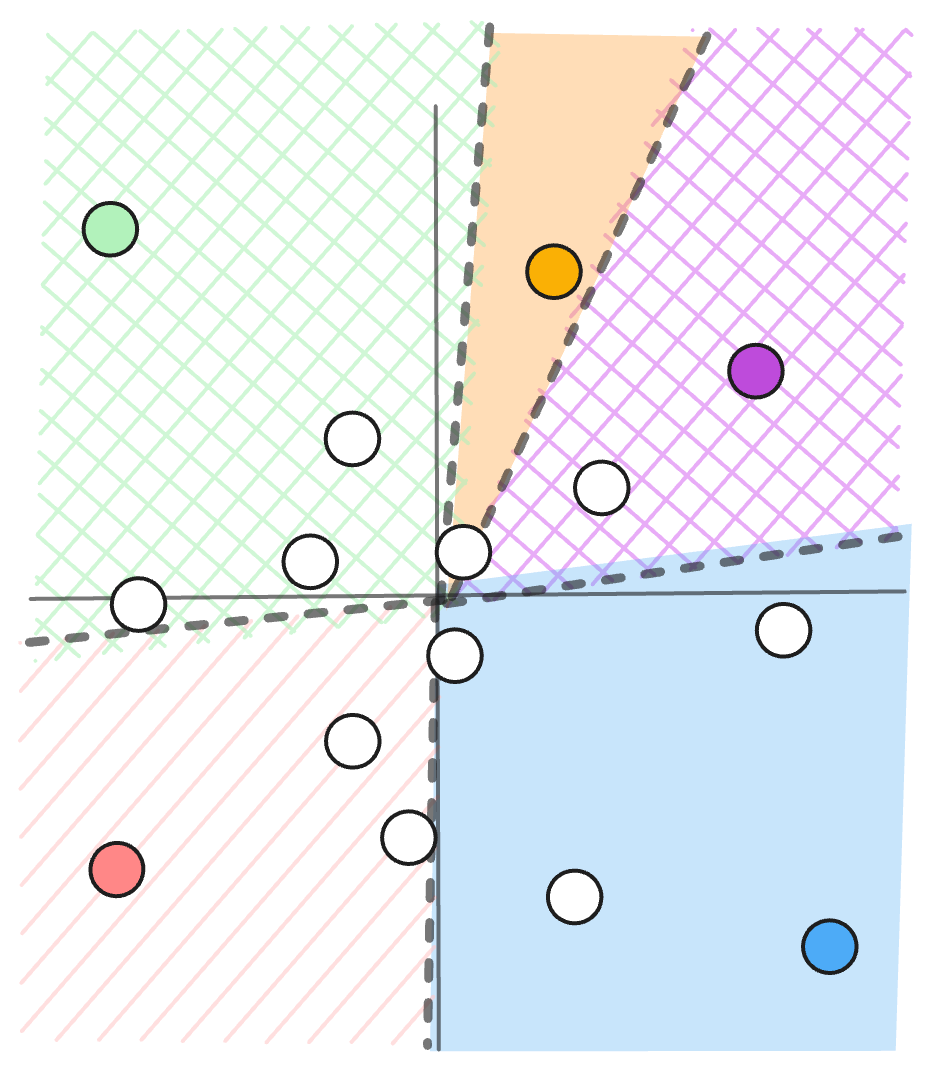}
    }
    \subfloat[IP-Delaunay]{
        \label{figure:graphs:ip-delaunay:ip-graph}
        \includegraphics[width=0.32\linewidth]{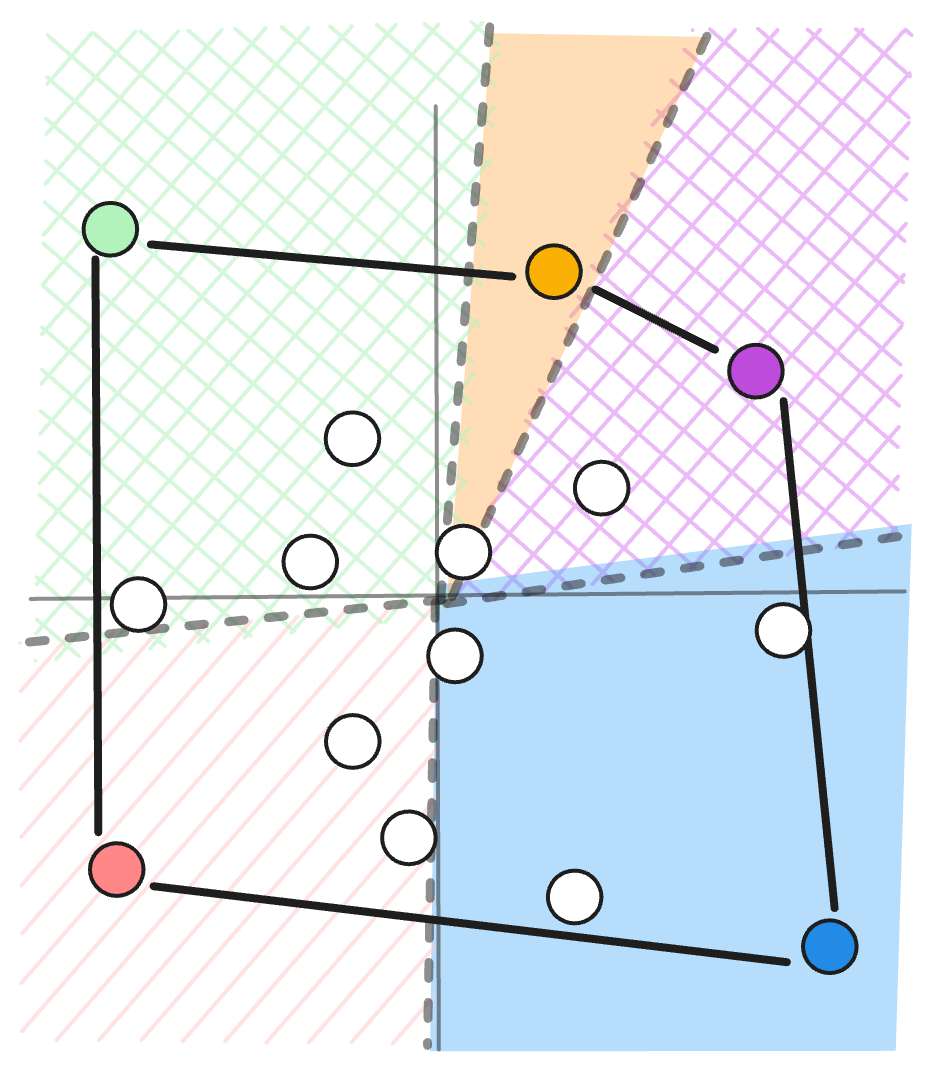}
    }
    \caption{Comparison of the Voronoi diagrams and Delaunay graphs for the same set of
    points but according to Euclidean distance versus inner product. Note that, for the non-metric
    distance function based on inner product, the Voronoi regions are convex cones determined by
    the intersection of half-spaces passing through the origin. Observe additionally that the
    inner product-induced Voronoi region of a point (those in white) may be an empty set.
    Such points can never be the solution to the $1$-MIPS problem.}
    \label{figure:graphs:ip-delaunay}
\end{figure}

Moving on to the Delaunay graph,~\cite{morozov2018ip-nsw} construct the graph in much the same
way as before and call the resulting graph the IP-Delaunay graph.
Two nodes $u, v \in \mathcal{V}$ in the IP-Delaunay graph are connected if their Voronoi
regions intersect: $\mathcal{R}_u \cap \mathcal{R}_v \neq \emptyset$. Note that, by the reasoning
above, the nodes whose Voronoi regions are empty will be isolated in the graph.
These nodes represent vectors that can never be the solution to MIPS for any
query---remember that we are only considering $k=1$. So it would be inconsequential
if we removed these nodes from the graph.
This is also illustrated in Figure~\subref*{figure:graphs:ip-delaunay:ip-graph}.

Considering the data structure above for inner product,~\cite{morozov2018ip-nsw} prove the following
result to give optimality guarantee for the greedy search algorithm for $1$-MIPS
(granted we enter the graph from a non-isolated node). Nothing, however, may be said about $k$-MIPS.

\begin{theorem}
    Suppose $G=(\mathcal{V}, \mathcal{E})$ is a graph that contains the
    IP-Delaunay graph for a collection $\mathcal{X}$ minus the isolated nodes.
    Invoking Algorithm~\ref{algorithm:graphs:greedy-search} with $k=1$
    and $\delta(\cdot, \cdot) = -\langle \cdot, \cdot \rangle$
    gives the optimal solution to the top-$1$ MIPS problem.
\end{theorem}
\begin{proof}
    If we showed that a local optimum is necessarily the global optimum, then we are done.
    To that end, consider a query $q$ for which Algorithm~\ref{algorithm:graphs:greedy-search} terminates
    when it reaches node $u \in \mathcal{X}$ which is distinct from the globally optimal
    solution $u^\ast \notin N(u)$. In other words, we have that
    $\langle q, u \rangle > \langle q, v \rangle$ for all $v \in N(u)$, but
    $\langle q, u^\ast \rangle > \langle q, u \rangle$ and
    $(u, u^\ast) \notin \mathcal{E}$.
    If that is true, then $q \notin \mathcal{R}_u$, the Voronoi region of $u$,
    but instead we must have that $q \in \mathcal{R}_{u^\ast}$.

    Now define the collection $\tilde{\mathcal{X}} \triangleq N(u) \cup \{ u \}$,
    and consider the Voronoi diagram of the resulting collection. It is easy to show
    that the Voronoi region of $u$ in the presence of points in $\tilde{\mathcal{X}}$
    is the same as its region given $\mathcal{X}$. From before, we also know that $q \notin \mathcal{R}_u$.
    Considering the fact that $\mathbb{R}^d = \bigcup_{v \in \tilde{\mathcal{X}}} \mathcal{R}_v$,
    $q$ must belong to $\mathcal{R}_v$ for some $v \in \tilde{\mathcal{X}}$ with $v \neq u$.
    That implies that $\langle q, v \rangle > \langle q, u \rangle$ for some $v \in \tilde{\mathcal{X}} \setminus u$.
    But because $v \in N(u)$ (by construction), the last inequality poses a contradiction to our
    premise that $u$ was locally optimal.
\end{proof}

\medskip

In addition to the fact that the IP-Delaunay graph does not answer top-$k$ queries,
it also suffers from the same deficiencies we noted for the Euclidean Delaunay graph earlier
in this section. Naturally then, to make the data structure more practical in high-dimensional
regimes, we must resort to heuristics and approximations, which in their simplest form may be
the $k$-MIPS graph (i.e., a $k$-NN graph where the distance function for finding the top-$k$
nodes is inner product).
This is the general direction~\cite{morozov2018ip-nsw} and a few other works
have explored~\citep{Liu2019UnderstandingAI,zhou2019mobius-mips}.

As in the case of metric distance functions, none of the guarantees stated above
port over to these approximate graphs. But, once again, empirical evidence gathered from
a variety of datasets show that these graphs perform reasonably well in practice, even for top-$k$
with $k > 1$.

\subsubsection{Is the IP-Delaunay Graph Necessary?}
\cite{morozov2018ip-nsw} justify the need for developing the IP-Delaunay graph
by comparing its structure with the following alternative:
First, apply a MIPS-to-NN \emph{asymmetric} transformation~\citep{xbox-tree}
from $\mathbb{R}^d$ to $\mathbb{R}^{d + 1}$.
This involves transforming a data point $u$ with
$\phi_d(u) = [u; \; \sqrt{1 - \lVert u \rVert_2^2}]$ and a query point $q$
with $\phi_q(v) = [v; 0]$. Next,
construct the standard (Euclidean) Delaunay graph over the transformed vectors.

What happens if we form the Delaunay graph on the transformed collection
$\phi_d \big(\mathcal{X} \big)$? Observe the Euclidean distance between
$\phi_d(u)$ and $\phi_d(v)$ for two vectors $u, v \in \mathcal{X}$:
\begin{align*}
    \lVert \phi_d(u) - \phi_d(v) \rVert_2^2 &=
        \lVert \phi_d(u) \rVert_2^2 + \lVert \phi_d(v) \rVert_2^2 - 2 \langle \phi_d(u), \phi_d(v) \rangle \\
    &= \big( \lVert u \rVert_2^2 + 1 - \lVert u \rVert_2^2 \big) +
        \big( \lVert v \rVert_2^2 + 1 - \lVert v \rVert_2^2 \big) \\
        &\quad - 2\langle u, v \rangle -
        2 \sqrt{\big( 1 - \lVert u \rVert_2^2 \big) \big( 1 - \lVert v \rVert_2^2 \big)}.
\end{align*}
Should we use these distances to construct the Delaunay graph, the resulting
structure will have nothing to do with the original MIPS problem.
That is because the $L_2$ distance between a pair of transformed data points
is not rank-equivalent to the inner product between the original data points.
For this reason,~\cite{morozov2018ip-nsw} argue that the IP-Delaunay graph is a more
sensible choice.

However, we note that their argument rests heavily on their particular choice of
MIPS-to-NN transformation. The transformation they chose makes sense in
contexts where we \emph{only} care about preserving the inner product between query-data
point pairs. But when forming the Delaunay graph, preserving inner product between pairs of
data points, too, is imperative. That is the reason why we lose rank-equivalence
between $L_2$ in $\mathbb{R}^{d+1}$ and inner product in $\mathbb{R}^d$.

There are, in fact, MIPS-to-NN transformations that are more appropriate for this problem
and would invalidate the argument for the need for the IP-Delaunay graph.
Consider for example $\phi_d: \mathbb{R}^d \rightarrow \mathbb{R}^{d + m}$ for a collection
$\mathcal{X}$ of $m$ vectors as follows:
$\phi_d(u^{(i)}) = u^{(i)} \oplus \big(\sqrt{1 - \lVert u^{(i)} \rVert_2^2}\big) e_{(d + i)}$,
where $u^{(i)}$ is the $i$-th data point in the collection, and
$e_j$ is the $j$-th standard basis vector. In other words, the $i$-th $d$-dimensional
data point is augmented with an $m$-dimensional
sparse vector whose $i$-th coordinate is non-zero.
The query transformation is simply $\phi_q(q) = q \oplus \mathbf{0}$, where $\mathbf{0} \in \mathbb{R}^m$
is a vector of $m$ zeros.

Despite the dependence on $m$, the transformation is remarkably easy to manage:
the sparse subspace of every vector has at most one non-zero coordinate,
making the doubling dimension of the sparse subspace $\mathcal{O}(\log m)$
by Lemma~\ref{lemma:intrinsic-dimensionality:sparse}.
Distance computation between the transformed vectors, too, has negligible overhead.
Crucially, we regain rank-equivalence between $L_2$ distance in $\mathbb{R}^{d + m}$
and inner product in $\mathbb{R}^d$ not only for query-data point pairs, but also
for pairs of data points:
\begin{align*}
    \lVert \phi_d(u) - \phi_d(v) \rVert_2^2 &=
        \lVert \phi_d(u) \rVert_2^2 + \lVert \phi_d(v) \rVert_2^2 - 2 \langle \phi_d(u), \phi_d(v) \rangle \\
    &= 2 - 2\langle u, v \rangle.
\end{align*}

Finally, unlike the IP-Delaunay graph, the standard Delaunay graph in $\mathbb{R}^{d + m}$
over the transformed vector collection has optimality guarantee for the 
top-$k$ retrieval problem per Theorem~\ref{theorem:graphs:delaunay:top-k}.
It is, as such, unclear if the IP-Delaunay graph is even necessary as a theoretical tool.

\begin{svgraybox}
    In other words, suppose we are given a collection of points $\mathcal{X}$
    and inner product as the similarity function. Consider a graph index where
    the presence of an edge is decided based on the inner product between data points.
    Take another graph index built for the transformed $\mathcal{X}$
    using the transformation described above from $\mathbb{R}^d$ to $\mathbb{R}^{d + m}$,
    and where the edge set is formed on the basis of the Euclidean distance between two
    (transformed) data points. The two graphs are equivalent.

    The larger point is that, MIPS over $m$ points in $\mathbb{R}^d$ is equivalent
    to NN over a transformation of the points in $\mathbb{R}^{d + m}$.
    While the transformation increases the apparent dimensionality,
    the intrinsic dimensionality of the data only increases by $\mathcal{O}(\log m)$.
\end{svgraybox}

\section{The Small World Phenomenon}
Consider, once again, the Delaunay graph but, for the moment, set aside
the fact that it is a prohibitively-expensive data structure to maintain for
high dimensional vectors. By construction, every node in the graph is only
connected to its Voronoi neighbors (i.e., nodes whose Voronoi region intersects
with the current node's). We showed that such a topology affords \emph{navigability},
in the sense that the greedy procedure in Algorithm~\ref{algorithm:graphs:greedy-search}
can traverse the graph only based on information about immediate neighbors of a node
and yet arrive at the globally optimal solution to the top-$k$ retrieval problem.

\begin{svgraybox}
Let us take a closer look at the traversal algorithm for the case of $k=1$.
It is clear that, navigating from the entry node to the solution takes us through
every Voronoi region along the path. That is, we cannot ``skip'' a Voronoi region
that lies between the entry node and the answer. This implies that the running time
of Algorithm~\ref{algorithm:graphs:greedy-search} is directly affected by the diameter
of the graph (in addition to the average degree of nodes).
\end{svgraybox}

Can we enhance this topology by adding \emph{long-range} edges between non-Voronoi
neighbors, so that we may skip over a fraction of Voronoi regions? After all,
Theorem~\ref{theorem:graphs:delaunay:top-k} guarantees navigability so long as the graph
\emph{contains} the Delaunay graph. Starting with the Delaunay graph and inserting
long-range edges, then, will not take away any of the guarantees.
But, what is the right number of long-range edges and how do we determine
which remote nodes should be connected? This section reviews the theoretical results
that help answer these questions.

\begin{figure}[t]
    \centering
    \centering
    \subfloat[]{
        \includegraphics[width=0.35\linewidth]{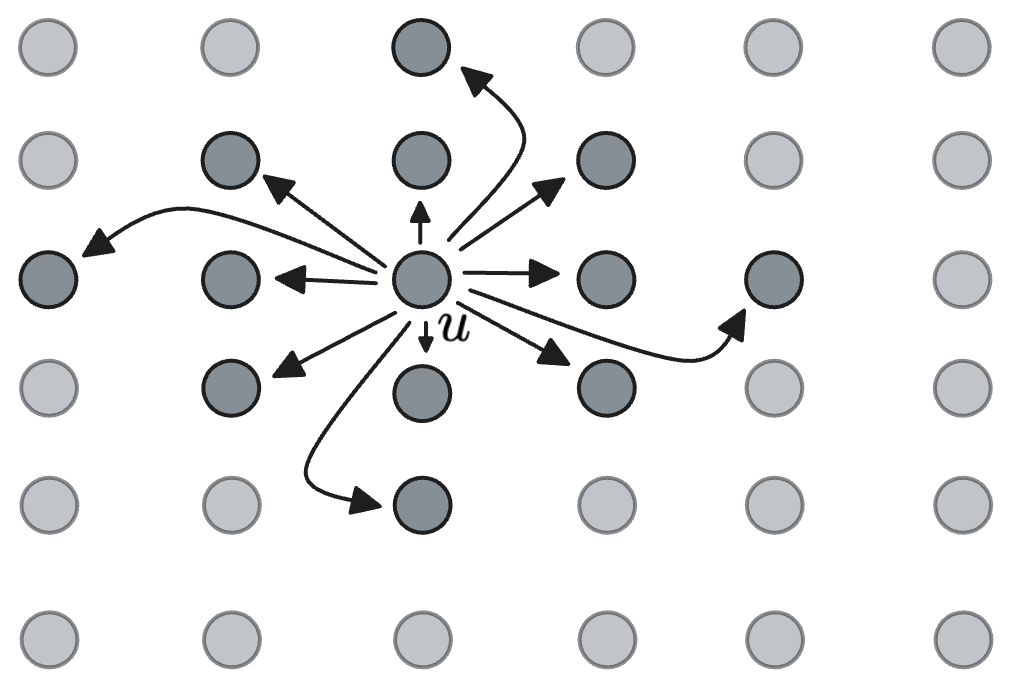}
    }
    \hspace{2em}
    \subfloat[]{
        \includegraphics[width=0.35\linewidth]{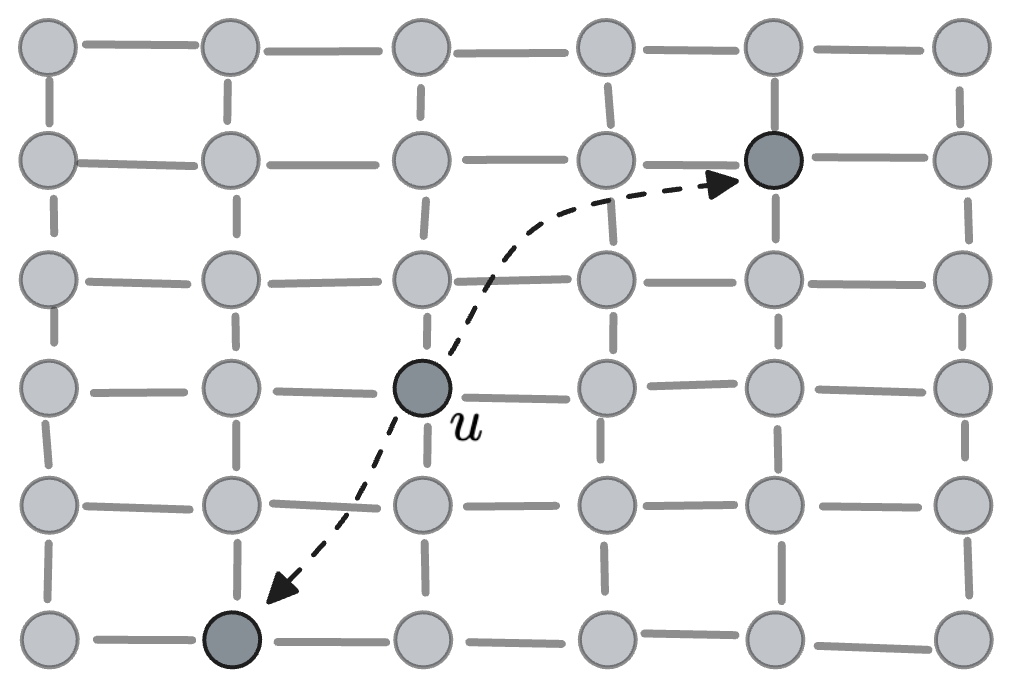}
    }
    \caption{Example graphs generated by the probabilistic model introduced
    by~\cite{kleinberg2000sw}. (a) illustrates the directed edges from $u$ for
    the following configuration: $r = 2$, $l = 0$. (b) renders the regular
    structure for $r=1$, where edges without arrows are bi-directional,
    and the long-range edges for node $u$ with configuration $l=2$.}
    \label{figure:graphs:lattice}
\end{figure}

\subsection{Lattice Networks}
Let us begin with a simple topology that is relatively easy to reason
about---we will see later how the results from this section can be generalized to the Delaunay graph.
The graph we have in mind is a lattice network where $m \times m$
nodes are laid on a two-dimensional grid.
Define the distance between two nodes as their \emph{lattice (Manhattan) distance}
(i.e., the minimal number of horizontal and vertical hops that connect two nodes).
That is the network examined by~\cite{kleinberg2000sw} in a seminal paper
that studied the effect of long-range edges on the time complexity of
Algorithm~\ref{algorithm:graphs:greedy-search}.

We should take a brief detour and note that,~\cite{kleinberg2000sw},
in fact, studied the problem of \emph{transmitting}
a message from a source to a known target using the best-first-search algorithm,
and quantified the average number
of hops required to do that in the presence of a variety of classes of
long-range edges. That, in turn, was inspired by a social phenomenon
colloquially known as the ``small-world phenomenon'': The empirical observation
that two strangers are linked by a short chain of
acquaintances~\citep{milgram_small-world_1967,travers1969sw}.

In particular,~\cite{kleinberg2000sw} was interested in explaining why
and under what types of long-range edges should our greedy algorithm
be able to navigate to the optimal solution, by only utilizing information
about immediate neighbors. To investigate this question,~\cite{kleinberg2000sw}
introduced the following probabilistic model of the lattice topology as an abstraction
of individuals and their social connections.

\subsubsection{The Probabilistic Model}
Every node in the graph has a (directed) edge with every other node within
lattice distance $r$, for some fixed hyperparameter $r \geq 1$. These connections
make up the regular structure of the graph. Overlaid with this structure is a set of
random, long-range edges that are generated according to the following
probabilistic model. For fixed constants $l \geq 0$ and $\alpha \geq 0$,
we insert a directed edge between every node $u$ and $l$ other nodes, where a node
$v \neq u$ is selected with probability proportional to $\delta(u, v)^{-\alpha}$
where $\delta(u, v) = \lVert u - v \rVert_1$ is the lattice distance.
Example graphs generated by this process are depicted in Figure~\ref{figure:graphs:lattice}.

The model above is reasonably powerful as it can express a variety of topologies.
For example, when $l = 0$, the resulting graph has no long-range edges.
When $l > 0$ and $\alpha = 0$, then every node $v \neq u$ in the graph
has an equal chance of being the destination of a long-range edge from $u$.
When $\alpha$ is large, the protocol becomes more biased to forming a long-range connection
from $u$ to nodes closer to it.

\subsubsection{The Claim}
\begin{svgraybox}
\cite{kleinberg2000sw} shows that, when $0 \leq \alpha < 2$,
the best-first-search algorithm must visit at least $\mathcal{O}_{r, l, \alpha}\big(m^{(2 - \alpha)/3}\big)$
nodes. When $\alpha > 2$, the number of nodes visited is at least
$\mathcal{O}_{r, l, \alpha}\big(m^{(\alpha - 2)/(\alpha - 1)} \big)$ instead. But, rather
uniquely, when $\alpha = 2$ and $r = l = 1$, we visit a number of nodes that
is at most poly-logarithmic in $m$.
\end{svgraybox}

Theorem~\ref{theorem:graphs:lattice} states this result formally.
But before we present the theorem, we state a useful lemma.

\begin{lemma}
    \label{lemma:graphs:lattice}
    Generate a lattice $G = (\mathcal{V}, \mathcal{E})$ of $m \times m$
    nodes using the probabilistic model above
    with $\alpha = 2$ and $l = 1$.
    The probability that there exists a long-range edge between two nodes
    $u, v \in \mathcal{V}$ is at least $\delta(u, v)^{-2} / 4 \ln(6m)$.
\end{lemma}
\begin{proof}
    $u$ chooses $v \neq u$ as its long-range destination
    with the following probability: $\delta(u, v)^{-2} / \sum_{w \neq u} \delta(u, w)^{-2}$.
    Let us first bound the denominator as follows:
    \begin{align*}
        \sum_{w \neq u} \delta(u, w)^{-2} &\leq \sum_{i = 1}^{2m - 2} (4i)(i^{-2})
        = 4 \sum_{i = 1}^{2m - 2} \frac{1}{j} \\
        &\leq 4 + 4 \ln \big( 2m - 2 \big) \leq 4 \ln \big( 6m \big).
    \end{align*}
    In the expression above, we derived the first inequality by iterating over
    all possible (lattice) distances between $m^2$ nodes on a two-dimensional grid
    (ranging from $1$ to $2m - 2$ if $u$ and $w$ are at diagonally opposite corners),
    and noticing that there are at most $4i$ nodes at distance $i$ from node $u$.
    From this we infer that the probability that node $(u, v) \in \mathcal{E}$ is at least
    $\delta(u, v)^{-2} / 4 \ln(6m)$.
\end{proof}

\begin{theorem}
    \label{theorem:graphs:lattice}
    Generate a lattice $G = (\mathcal{V}, \mathcal{E})$ of $m \times m$
    nodes using the probabilistic model above
    with $\alpha = 2$ and $r = l = 1$.
    The best-first-search algorithm beginning from any arbitrary node
    and ending in a target node visits at most $\mathcal{O}\big( \log^2 m \big)$ nodes on average.
\end{theorem}
\begin{proof}
    Define a sequence of sets $A_i$, where each $A_i$ consists of nodes whose distance
    to the target $u^\ast$ is greater than $2^i$ and at most $2^{i + 1}$.
    Formally, $A_i = \{ v \in \mathcal{V} \;|\; 2^i < \delta(u^\ast, v) \leq 2^{i + 1} \}$.
    Suppose that the algorithm is currently in node $u$
    and that $\log m \leq \delta(u, u^\ast) < m$, so that $u \in A_i$
    for some $\log \log m \leq i < \log m$. What is the probability that the algorithm
    exits the set $A_i$ in the next step?
    
    That happens when one of $u$'s neighbors has a distance to $u^\ast$ that is
    at most $2^i$. In other words, $u$ must have a neighbor that is in the set
    $A_{<i} = \cup_{j=0}^{j=i-1} A_j$.
    The number of nodes in $A_{<i}$ is at least:
    \begin{equation*}
        1 + \sum_{s=1}^{2^i} s = 1 + \frac{2^i \big( 2^i + 1\big)}{2} > 2^{2i - 1}.
    \end{equation*}

    How likely is it that $(u, v) \in \mathcal{E}$ if $v \in A_{<i}$?
    We apply Lemma~\ref{lemma:graphs:lattice}, noting that the distance of each of the
    nodes in $A_{<i}$ with $u$ is at most $2^{i + 1} + 2^i < 2^{i + 2}$.
    We obtain that, the probability that $u$ is connected to a node in $A_{<i}$ is at least
    $2^{2i - 1} (2^{i + 2})^{-2} / 4 \ln(6m) = 1 / 128 \ln (6m)$.

    Next, consider the total number of nodes in $A_i$ that are visited by the algorithm,
    and denote it by $X_i$. In expectation, we have the following:
    \begin{equation*}
        \ev \big[ X_i \big] = \sum_{j = 1}^\infty \probability \big[ X_i \geq j \big] \leq
        \sum_{j = 1}^\infty \Big( 1 - \frac{1}{128 \ln (6m)} \Big)^{j - 1} =
        128 \ln (6m).
    \end{equation*}
    We obtain the same bound if we repeat the arguments for $i = \log m$.
    When $0 \leq i < \log \log m$, the algorithm visits at most $\log m$
    nodes, so that the bound above is trivially true.

    Denoting by $X$ the total number of nodes visited, $X = \sum_{j = 0}^{\log m} X_j$,
    we conclude that:
    \begin{equation*}
        \ev \big[ X \big] \leq (1 + \log m)(128 \ln(6m)) = \mathcal{O}(\log^2 m),
    \end{equation*}
    thereby completing the proof.
\end{proof}

The argument made by~\cite{kleinberg2000sw} is that, in a lattice network
where each node is connected to its (at most four) nearest neighbors within
unit distance, and where every node has a long-range edge to one other node
with probability that is proportional to $1/\delta(\cdot, \cdot)^2$, then
the greedy algorithm visits at most a poly-logarithmic number of nodes.
Translating this result to the case of top-$1$ retrieval using
Algorithm~\ref{algorithm:graphs:greedy-search} over the same network,
we can state that
the time complexity of the algorithm is $\mathcal{O}(\log^2 m)$,
because the total number of neighbors per node is $\mathcal{O}(1)$.

While this result is significant, it only holds for the lattice network
with the lattice distance. It has thus no bearing on the time complexity of
top-$1$ retrieval over the Delaunay graph with the Euclidean distance.
In the next section, we will see how~\cite{voronet2007} close this gap.

\subsection{Extension to the Delaunay Graph}
We saw in the preceding section that, the secret to creating a 
provably navigable graph where the best-first-search algorithm
visits a poly-logarithmic number of nodes in the lattice network, was the
highly specific distribution from which long-range edges were sampled.
That element turns out to be the key ingredient when extending
the results to the Delaunay graph too, as~\cite{voronet2007} argue.

We will describe the algorithm for data in the two-dimensional unit square.
That is, we assume that the collection of data points $\mathcal{X}$
and query points are in $[0, 1]^2$.
That the vectors are bounded is not a limitation \emph{per se}---as we discussed previously,
we can always normalize vectors into the hypercube without loss of generality.
That the algorithm does not naturally extend to high dimensions is a serious
limitation, but then again, that is not surprising considering the Delaunay graph
is expensive to construct. However, in the next section, we will review
heuristics that take the idea to higher dimensions.

\subsubsection{The Probabilistic Model}
Much like the lattice network, we assume there is a base graph
and a number of randomly generated long-range edges between nodes.
For the base graph,~\cite{voronet2007}
take the Delaunay graph.\footnote{\cite{voronet2007} additionally
connect all nodes that are within
$\delta_\textsc{Min} \propto 1/m$ distance from each
other, where $\delta_\textsc{Min}$ is chosen
such that the expected number of uniformly-distributed points in a
ball of radius $\delta_\textsc{Min}$ is $1$. We reformulate their method
without $\delta_\textsc{Min}$ in the present monograph to simplify
their result.}
As for the long-range edges, each node has a directed edge to one other node
that is selected at random, following a process we will describe shortly.
Observe that in this model, the number of neighbors of each node is $\mathcal{O}(1)$.

We already know from Theorem~\ref{theorem:graphs:delaunay:top-k} that,
because the network above contains the Delaunay graph, it is navigable by Algorithm~\ref{algorithm:graphs:greedy-search}. What remains to be investigated
is what type of long-range edges could reduce the number of hops (i.e.,
the number of nodes the algorithm must visit as it navigates
from an entry node to a target node).
Because at each hop the algorithm needs to evaluate distances with $\mathcal{O}(1)$
neighbors, improving the number of steps directly improves the time complexity
of Algorithm~\ref{algorithm:graphs:greedy-search} (for the case of $k=1$).

\begin{figure}[t]
    \centering
    \includegraphics[width=0.35\linewidth]{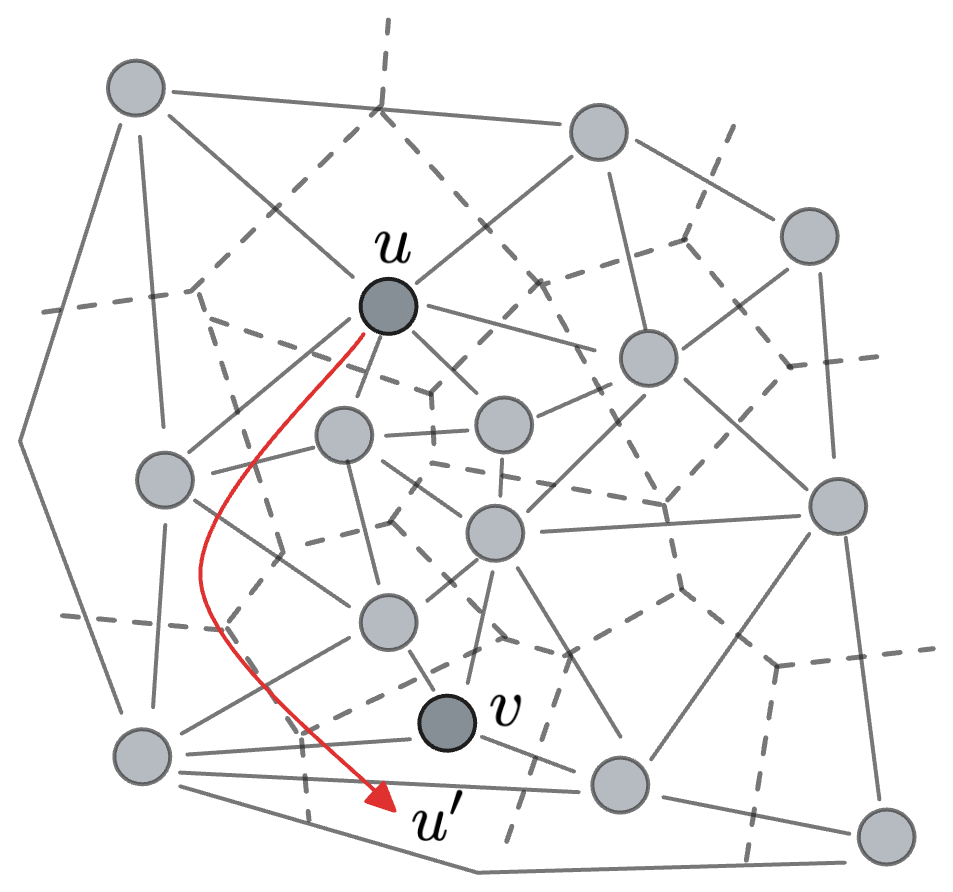}
    \caption{Formation of a long-range edge from $u$ to $v$ by the probabilistic
    model of~\cite{voronet2007}. First, we jump from $u$ to a random long-range end-point $u^\prime$,
    then designate its nearest neighbor ($v$) as the target of the edge.}
    \label{figure:graphs:voronet}
\end{figure}

\cite{voronet2007} show that, if long-range edges are chosen according to
the following protocol, then the number of hops is poly-logarithmic in $m$.
The protocol is simple: For a node $u$ in the graph,
first sample $\alpha$ uniformly from $[\ln \delta^\ast, \ln \delta_\ast]$,
where $\delta^\ast = \min_{v, w \in \mathcal{X}} \delta(v, w)$
and $\delta_\ast = \max_{v, w \in \mathcal{X}} \delta(v, w)$.
Then choose $\theta \sim [0, 2\pi]$, to finally obtain $u^\prime = u + z$
where $z$ is the vector $[e^\alpha \cos \theta, e^\alpha \sin \theta]$.
Let us refer to $u^\prime$ as the
``long-range end-point,'' and note that this point may escape
the $[0, 1]^2$ square. Next, we find the nearest node $v$ to $u^\prime$
and add a directed edge from $u$ to $v$.
This is demonstrated in Figure~\ref{figure:graphs:voronet}.

\subsubsection{The Claim}

Given the resulting graph,~\cite{voronet2007} state and prove that the
average number of hops taken by the best-first-search algorithm is poly-logarithmic.
Before we discuss the claim, however, let us state a useful lemma.

\begin{lemma}
    \label{lemma:graphs:voronet}
    The probability that the long-range end-point from a node $u$ lands in a ball
    centered at another node $v$ with radius $\beta \delta(u, v)$ for some small $\beta \in [0, 1]$
    is at least $K \beta^2 / (1 + \beta)^2$ where $K = (2 \ln \Delta)^{-1}$ and
    $\Delta = \delta_\ast / \delta^\ast$ is the aspect ratio.
\end{lemma}
\begin{proof}
    The probability that the long-range end-point lands in an area $dS$ that covers
    the distance $[r, r + dr]$ and angle $[\theta, \theta + d\theta]$, for small $dr$ and $d\theta$, is:
    \begin{equation*}
        \frac{d\theta}{2\pi} \frac{\ln (r + dr) - \ln r}{\ln \delta_\ast - \ln \delta^\ast} \approx
        \frac{d\theta}{2\pi} \frac{dr/r}{\ln \Delta} =
        \frac{1}{2\pi\ln \Delta} \frac{r d\theta dr}{r^2} \approx
        \frac{dS}{2 \pi \ln \Delta r^2}
    \end{equation*}.

    Observe now that the distance between a point $u$ and any point in the ball
    described in the lemma is at most $(1 + \beta) \delta(u, v)$. We can therefore
    see that the probability that the long-range end-point lands in $B(v, \beta \delta(u, v))$
    is at least:
    \begin{equation*}
        \frac{\pi \beta^2 \delta(u, v)^2}{2 \pi \ln \Delta (1 + \beta)^2 \delta(u, v)^2} =
        \frac{\beta^2}{2\ln(\Delta) (1 + \beta)^2},
    \end{equation*}
    as required.
\end{proof}

\begin{theorem}
    \label{theorem:graphs:voronet}
    Generate a graph $G = (\mathcal{V}, \mathcal{E})$ according to the probabilistic
    model described above, for vectors in $[0, 1]^2$ equipped with the
    Euclidean distance $\delta(\cdot, \cdot)$.
    The number of nodes visited by the best-first-search algorithm
    starting from any arbitrary node and ending at a target node is $\mathcal{O}(\log^2 \Delta)$.
\end{theorem}
\begin{proof}
    The proof follows the same reasoning as in the proof of Theorem~\ref{theorem:graphs:lattice}.

    Suppose we are currently at node $u$ and that $u^\ast$ is our target node.
    By Lemma~\ref{lemma:graphs:voronet}, the probability that the long-range end-point
    of $u$ lands in $B(u^\ast, \delta(u, u^\ast)/6)$ is at least $1/98 \ln \Delta$.
    As such, the total number of hops, $X$, from $u$ to a point in $B(u^\ast, \delta(u, u^\ast)/6)$
    has the following expectation:
    \begin{equation*}
        \ev [ X ] = \sum_{i = 1}^{\infty} \probability [X \geq i] \leq
        \sum_{i = 1}^\infty \Big( 1 - \frac{1}{98 \ln \Delta} \Big)^{i - 1} = 98 \ln \Delta.
    \end{equation*}
    Every time the algorithm moves from the current node $u$ to some other node in
    $B(u^\ast, \delta(u, u^\ast)/6)$, the distance is shrunk by a factor of $6/5$.
    As such, the total number of hops in expectation is at most:
    \begin{equation*}
        \Big( \ln_{6/5} \Delta \Big) \times \Big( 98 \ln \Delta \Big) = \mathcal{O}(\log^2 \Delta).
    \end{equation*}
\end{proof}

We highlight that,~\cite{voronet2007} choose the interval from which $\alpha$
is sampled differently. Indeed, $\alpha$ in their work is chosen uniformly from
the range $\delta_\textsc{Min} \propto 1/m$ and $\sqrt{2}$.
Substituting that configuration into Theorem~\ref{theorem:graphs:voronet} gives
an expected number of hops that is $\mathcal{O}(\log^2 m)$.

\subsection{Approximation}
The results of~\cite{voronet2007} are encouraging. In theory, so long as
we can construct the Delaunay graph, we not only have the optimality guarantee,
but we are also guaranteed to have a poly-logarithmic number of hops to
reach the optimal answer. Alas, as we have discussed previously, the Delaunay
graph is expensive to build in high dimensions.

\begin{svgraybox}
Moreover, the number of
neighbors per node is no longer $\mathcal{O}(1)$. So even if we inserted
long-range edges into the Delaunay graph, it is not immediate if the time
saved by skipping Voronoi regions due to long-range edges
offsets the additional time the algorithm
spends computing distances between each node along the path and its neighbors.
\end{svgraybox}

We are back, then, to approximation with the help of heuristics.
\cite{raynet2007} describe one such method in a follow-up study.
Their method approximates the Voronoi regions of every node by
resorting to a \emph{gossip} protocol. In this procedure, every node
has a list of $3d + 1$ of its current neighbors, where $d$ denotes the dimension
of the space. In every iteration of the algorithm,
every node passes its current list to its neighbors. When a node receives
this information, it takes the union of all lists, and finds the subset of
$3d + 1$ points with the minimal volume. This subset becomes the node's
current list of neighbors. While a na\"ive implementation of the protocol
is prohibitively expensive,~\cite{raynet2007} discuss an alternative to estimating
the volume induced by a set of $3d + 1$ points, and the search for the minimal volume.

\cite{nsw2014} take a different approach. They simply permute the vectors
in the collection $\mathcal{X}$, and sequentially add each vector to the graph. Every time
a vector is inserted into the graph, it is linked to its $k$ nearest neighbors
from the current snapshot of the graph.
The intuition is that, as the graph grows, the edges added earlier in the evolution
of the graph serve as long-range edges in the final graph,
and the more recent edges form an approximation
of the $k$-NN graph, which itself is an approximation of the Delaunay graph.
Later~\cite{hnsw2020} modify the algorithm by introducing a hierarchy of graphs.
The resulting graph has proved successful in practice and, despite its lack
of theoretical guarantees, is both effective and highly efficient.

\section{Neighborhood Graphs}

In the preceding section, our starting point was the Delaunay graph. We augmented it with
random long-range connections to improve the transmission rate through the network.
Because the resulting structure contains the Delaunay
graph, we get the optimality guarantee of Theorem~\ref{theorem:graphs:delaunay:top-k}
for free. But, as a result of the complexity of the Delaunay construction in high dimensions,
we had to approximate the structure instead, losing all guarantees in the process.
Frustratingly, the approximate structure obtained by the heuristics we discussed, is
certainly not a super-graph of the Delaunay graph, nor is it necessarily its sub-graph.
In fact, even the fundamental property of connectedness is not immediately guaranteed.
There is therefore nothing meaningful to say about the theoretical behavior of such graphs.

In this section, we do the opposite. Instead of adding edges to the Delaunay graph and
then resorting to heuristics to create a completely different graph,
we prune the edges of the Delaunay graph to find a structure that is its \emph{sub-graph}.
Indeed, we cannot say anything meaningful about the optimality of
\emph{exact} top-$k$ retrieval, but as we will later see, we can state formal results
for the approximate top-$k$ retrieval variant---albeit in a very specific case.
The structure we have in mind is known as the Relative Neighborhood
Graph (RNG)~\citep{Toussaint1980rng,relativeNeighborhoodGraphs}.

\begin{svgraybox}
In an RNG, $G = (\mathcal{V}, \mathcal{E})$, for a distance function $\delta(\cdot, \cdot)$,
there is an undirected edge between two nodes $u, v \in \mathcal{V}$ if and only if
$\delta(u, v) < \max \big( \delta(u, w), \delta(w, v) \big)$ for all
$w \in \mathcal{V} \setminus \{ u, v\}$.
That is, the graph guarantees that, if $(u, v) \in \mathcal{E}$, then there is no other
point in the collection that is simultaneously closer to $u$ and $v$, than $u$ and $v$ are to
each other. Conceptually, then, we can view constructing an RNG as pruning away
edges in the Delaunay graph that violate the RNG property.
\end{svgraybox}

The RNG was shown to contain the Minimum Spanning Tree~\citep{Toussaint1980rng},
so that it is guaranteed to be connected.
It is also provably contained in the Delaunay graph~\citep{OROURKE1982} in any metric
space and in any number of dimensions.
As a final property, it is not hard to see that such a graph $G$ comes with a weak
optimality guarantee for the best-first-search algorithm:
If $q = u^\ast \in \mathcal{X}$, then the greedy traversal algorithm
returns the node associated with $q$, no matter where it enters the graph.
That is due simply to the following fact: If the current node $u$ is a local optimum
but not the global optimum, then there must be an edge connecting $u$ to a node
that is closer to $u^\ast$. Otherwise, $u$ itself must be connected to $u^\ast$.

Later,~\cite{arya1993sng} proposed a \emph{directed} variant of the RNG,
which they call the \emph{Sparse Neighborhood Graph} (SNG) that is arguably
more suitable for top-$k$ retrieval.
For every node $u \in \mathcal{V}$, we apply the following procedure:
Let $\mathcal{U} = \mathcal{V} \setminus \{ u \}$.
Sort the nodes in $\mathcal{U}$ in increasing distance to $u$.
Then, remove the closest node (say, $v$) from $\mathcal{U}$ and add an edge between $u$ to $v$.
Finally, remove from $\mathcal{U}$ all nodes $w$ that satisfy $\delta(u, w) > \delta(w, v)$.
The process is repeated until $\mathcal{U}$ is empty.
It can be immediately seen that the weak optimality guarantee from before still holds in the SNG.

\medskip

Neighborhood graphs are the backbone of many graph algorithms for top-$k$
retrieval~\citep{nsw2014,hnsw2020,fanng2016,fu2019nsg,fu2022nssg,diskann}.
While many of these algorithms make for efficient methods in practice, the
Vamana construction~\citep{diskann} stands out as it introduces a novel
super-graph of the SNG that turns out to have provable theoretical properties.
That super-graph is what~\cite{indyk2023worstcase}
call an \emph{$\alpha$-shortcut reachable} SNG, which we will review next.
For brevity, though, we call this graph simply $\alpha$-SNG.

\begin{figure}[t]
    \centering
    \subfloat[$\alpha = 1$]{
        \includegraphics[width=0.45\linewidth]{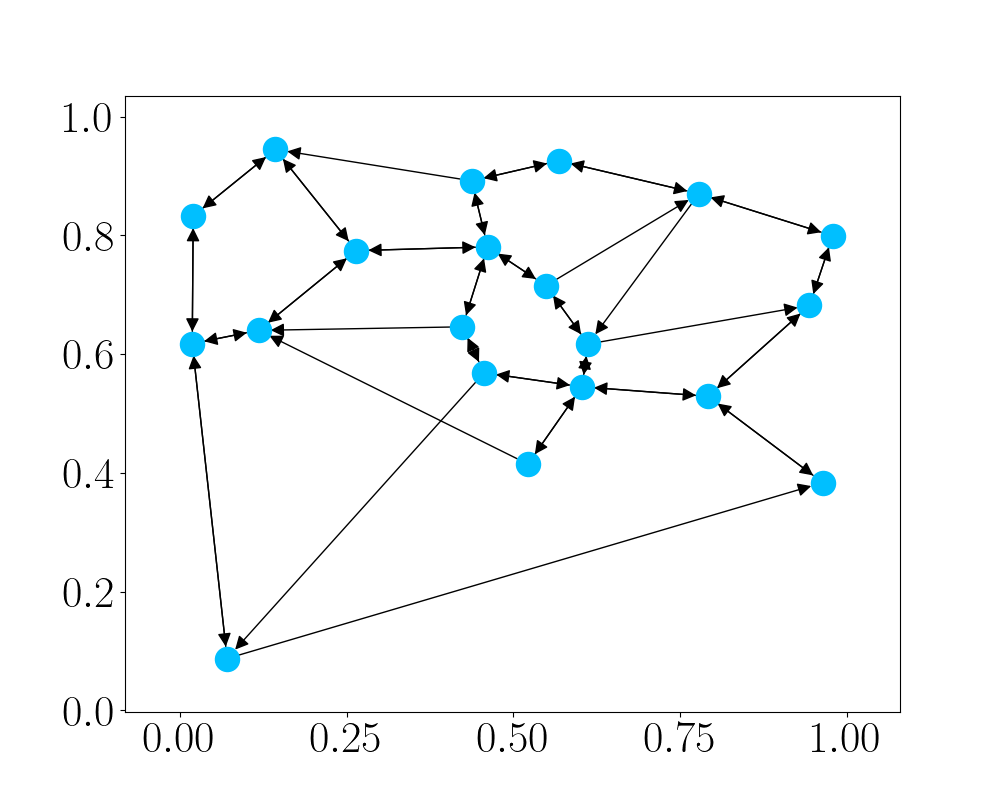}
    }
    \subfloat[$\alpha = 1.1$]{
        \includegraphics[width=0.45\linewidth]{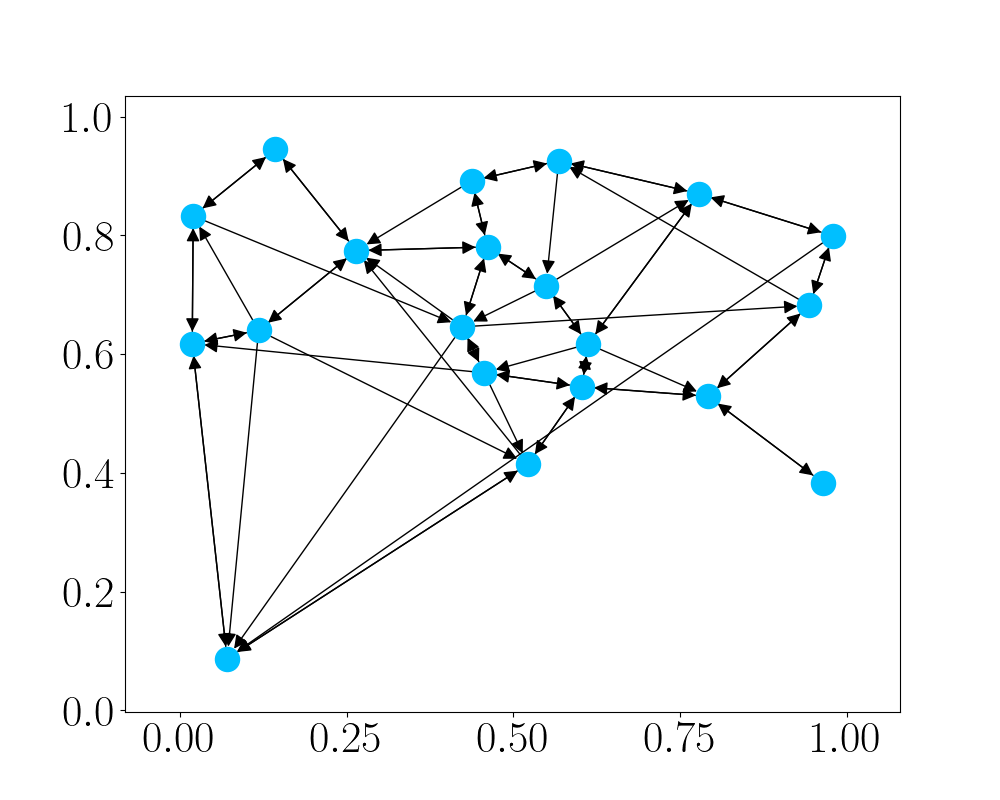}
    }
    
    \subfloat[$\alpha = 1.2$]{
        \includegraphics[width=0.45\linewidth]{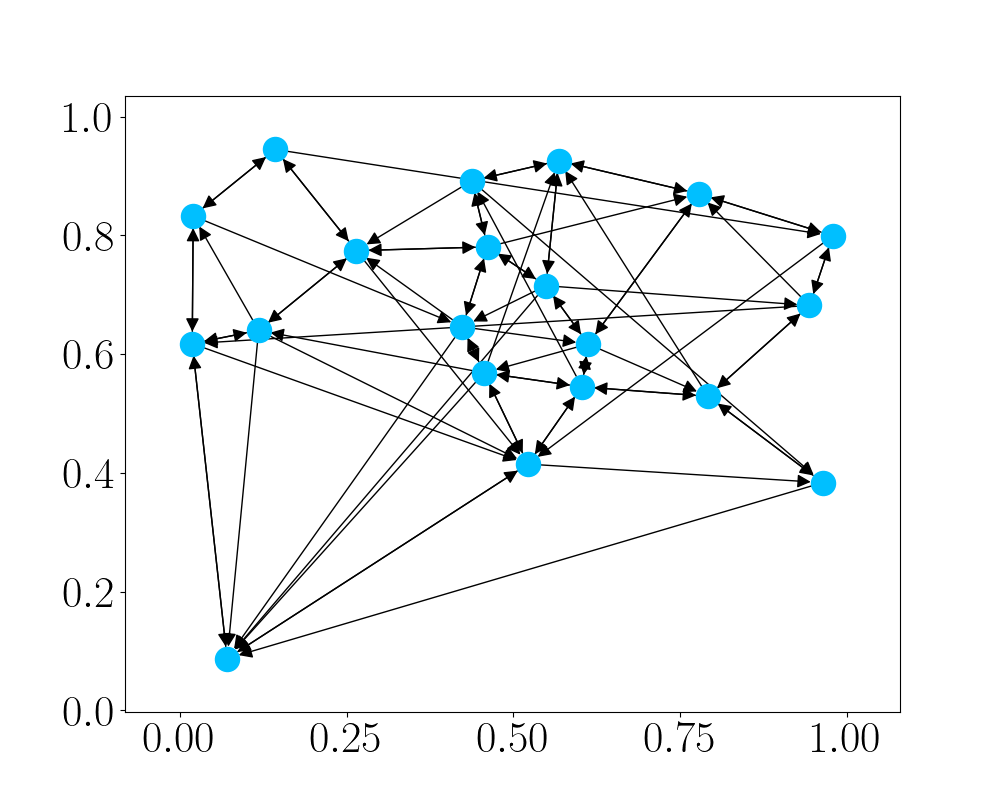}
    }
    \subfloat[$\alpha = 1.3$]{
        \includegraphics[width=0.45\linewidth]{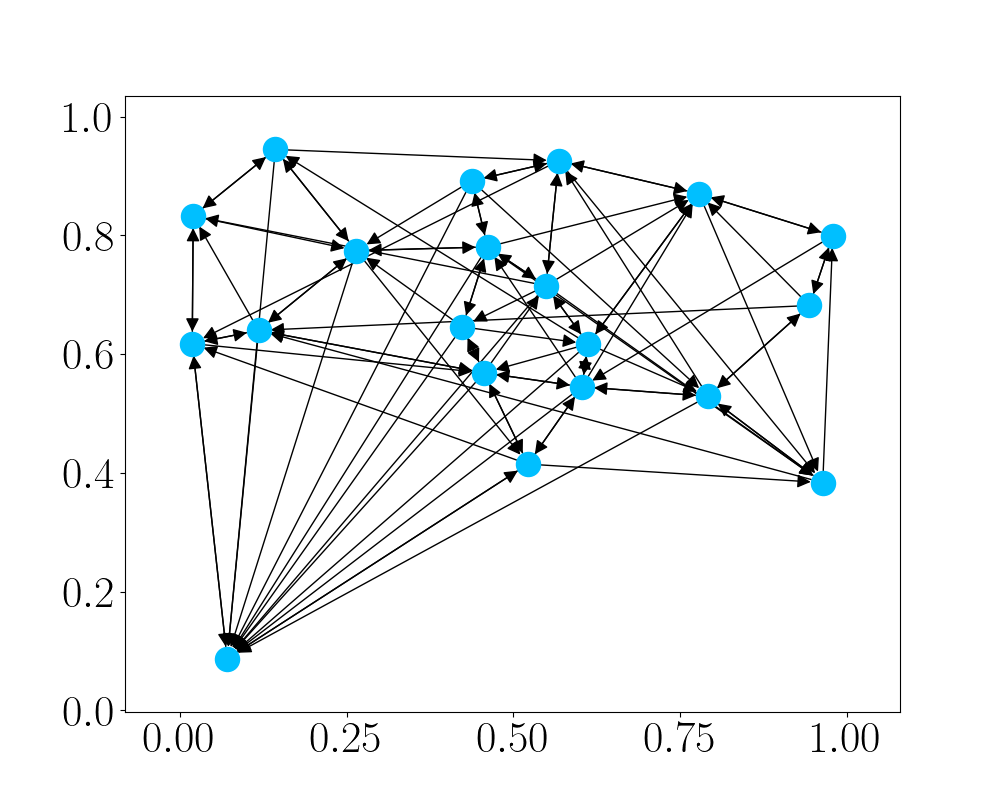}
    }
    \caption{Examples of $\alpha$-SNGs on a dataset of $20$ points drawn uniformly
    from $[0, 1]^2$ (blue circles).
    When $\alpha=1$, we recover the standard SNG. As $\alpha$ becomes larger, the
    resulting graph becomes more dense.}
    \label{figure:graphs:sng}
\end{figure}

\subsection{From SNG to \texorpdfstring{$\alpha$}{alpha}-SNG}
\cite{diskann} introduce a subtle adjustment to the SNG construction.
In particular, suppose we are processing a node $u$, have already extracted
the node $v$ whose distance to $u$ is minimal among the nodes in $\mathcal{U}$
(i.e., $v = \argmin_{w \in \mathcal{U}} \delta(u, w)$), and are now deciding
which nodes to discard from $\mathcal{U}$. In the standard SNG construction,
we remove a node $w$ for which $\delta(u, w) > \delta(w, v)$. But in the modified
construction, we instead discard a node $w$ if $\delta(u, w) > \alpha \delta(w, v)$
for some $\alpha > 1$. Note that, the case of $\alpha = 1$ simply gives the standard SNG.
Figure~\ref{figure:graphs:sng} shows a few examples of $\alpha$-SNGs on a toy dataset.

That is what~\cite{indyk2023worstcase} later refer to as an $\alpha$-shortcut reachable
graph. They define $\alpha$-shortcut reachability as the property where,
for any node $u$, we have that every other node $w$ is either the target of an edge
from $u$ (so that $(u, w) \in \mathcal{E}$), or that there is a node $v$
such that $(u, v) \in \mathcal{E}$ and $\delta(u, w) \geq \alpha \delta(w, v)$.
Clearly, the graph constructed by the procedure above is $\alpha$-shortcut
reachable by definition.

\subsubsection{Analysis}
\cite{indyk2023worstcase} present an analysis of the $\alpha$-SNG
for a collection of vectors $\mathcal{X}$ with \emph{doubling dimension} $d_\circ$
as defined in Definition~\ref{definition:doubling-dimension}.

For collections with a fixed doubling constant,~\cite{indyk2023worstcase} state two bounds. 
One gives a bound on the degree of every node in an $\alpha$-SNG. The other tells us
the expected number of hops from any arbitrary entry node to an $\epsilon$-approximate
solution to top-$1$ queries. The two bounds together give us an idea of the time complexity
of Algorithm~\ref{algorithm:graphs:greedy-search} over an $\alpha$-SNG as well as its
accuracy.

\begin{figure}[t]
    \centering
    \includegraphics[width=0.3\linewidth]{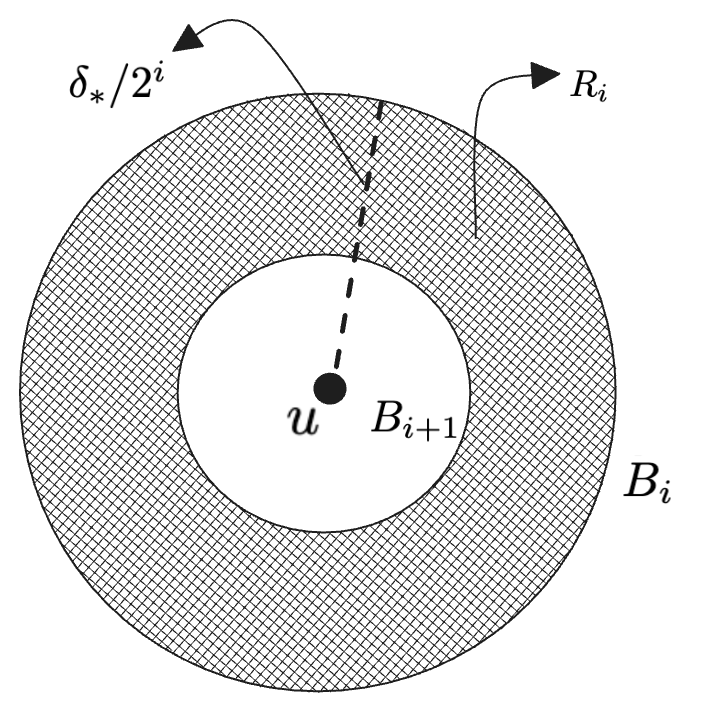}
    \caption{The sets $B_i$ and rings $R_i$ in the proof of Theorem~\ref{theorem:graphs:sng:degree}.}
    \label{figure:graphs:sng:proof}
\end{figure}

\begin{theorem}
    \label{theorem:graphs:sng:degree}
    The degree of any node in an $\alpha$-SNG is
    $\mathcal{O}\big( (4\alpha)^{d_\circ} \log \Delta \big)$ if the collection
    $\mathcal{X}$ has doubling dimension $d_\circ$ and aspect ratio $\Delta = \delta_\ast / \delta^\ast$.
\end{theorem}
\begin{proof}
    Consider a node $u \in \mathcal{V}$. For each $i \in [\log_2 \Delta]$, define
    a ball centered at $u$ with radius $\delta_\ast / 2^i$: $B_i = B(u, \delta_\ast / 2^i)$.
    From this, construct rings $R_i = B_i \setminus B_{i + 1}$. See Figure~\ref{figure:graphs:sng:proof}
    for an illustration.

    Because $\mathcal{X}$ has a constant doubling dimension, we can cover each $R_i$
    with $\mathcal{O}\big( (4\alpha)^{d_\circ} \big)$ balls of radius $\delta_\ast / \alpha 2^{i + 2}$.
    By construction, two points in each of these cover balls are at
    most $\delta_\ast / \alpha 2^{i + 1}$ apart. At the same time, the distance from $u$ to
    any point in a cover ball is at least $\delta_\ast / 2^{i + 1}$. By construction, all points
    in a cover ball except one are discarded as we form $u$'s edges in the $\alpha$-SNG.
    As such, the total number of edges from $u$ is bounded by the total number of cover balls,
    which is $\mathcal{O}\big( (4\alpha)^{d_\circ} \log \Delta \big)$.
\end{proof}

\begin{theorem}
    If $G = (\mathcal{V}, \mathcal{E})$ is an $\alpha$-SNG for collection $\mathcal{X}$,
    then Algorithm~\ref{algorithm:graphs:greedy-search} with $k = 1$ returns
    an $(\frac{\alpha + 1}{\alpha - 1} + \epsilon)$-approximate top-$1$ solution
    by visiting $\mathcal{O}\big( \log_\alpha \frac{\Delta}{(\alpha - 1)\epsilon} \big)$ nodes.
\end{theorem}
\begin{proof}
    Suppose $q$ is a query point and $u^\ast = \argmin_{u \in \mathcal{X}} \delta(q, u)$.
    Further assume that the best-first-search algorithm is currently in node $v_i$
    with distance $\delta(q, v_i)$ to the query. We make the following observations:
    \begin{itemize}
        \item By triangle inequality, we know that
        $\delta(v_i, u^\ast) \leq \delta(v_i, q) + \delta(q, u^\ast)$; and,
        \item By construction of the $\alpha$-SNG, $v_i$ is either connected to $u^\ast$
        or to another node whose distance to $u^\ast$ is shorter than $\delta(v_i, u^\ast) / \alpha$.
    \end{itemize}
    We can conclude that, the distance from $q$ to the next node the algorithm
    visits, $v_{i + 1}$, is at most:
    \begin{align*}
        \delta(v_{i+1}, q) &\leq \delta(v_{i + 1}, u^\ast) + \delta(u^\ast, q) \\
        &\leq \frac{\delta(v_i, u^\ast)}{\alpha} + \delta(u^\ast, q) \\
        &\leq \frac{\delta(v_i, q)}{\alpha} + (\alpha + 1) \delta(q, u^\ast).
    \end{align*}
    By induction, we see that, if the entry node is $s \in \mathcal{V}$:
    \begin{align*}
        \delta(v_i, q) &\leq \frac{\delta(s, q)}{\alpha^i} + (\alpha + 1) \delta(q, u^\ast) \sum_{j=1}^i \alpha^{-j} \\
        &\leq \frac{\delta(s, q)}{\alpha^i} + \frac{\alpha + 1}{\alpha - 1} \delta(q, u^\ast)
        \numberthis \label{equation:graphs:alpha-sng:hops-bound}.
    \end{align*}

    There are three cases to consider.

    \textbf{Case 1}: When $\delta(s, q) > 2 \delta_\ast$, then by triangle inequality,
    $\delta(q, u^\ast) > \delta(s, q) - \delta(s, u^\ast) > \delta(s, q) - \delta_\ast > \delta(s, q) / 2$.
    Plugging this into Equation~(\ref{equation:graphs:alpha-sng:hops-bound}) yields:
    \begin{align*}
        \delta(v_i, q) &\leq \frac{2\delta(q, u^\ast)}{\alpha^i} + \frac{\alpha + 1}{\alpha - 1} \delta(q, u^\ast) \\
        &\implies \frac{\delta(v_i, q)}{\delta(q, u^\ast)} \leq \frac{2}{\alpha^i} + \frac{\alpha + 1}{\alpha - 1}.
    \end{align*}
    As such, for any $\epsilon > 0$, the algorithm returns a $\big( \frac{\alpha + 1}{\alpha - 1} + \epsilon\big)$-approximate solution in $\log_\alpha 2/\epsilon$ steps.

    \textbf{Case 2}: $\delta(s, q) \leq 2 \delta_\ast$ and $\delta(q, u^\ast) \geq \frac{\alpha - 1}{4(\alpha + 1)}\delta^\ast$. By Equation~(\ref{equation:graphs:alpha-sng:hops-bound}), the algorithm returns a
    $\big( \frac{\alpha + 1}{\alpha - 1} + \epsilon \big)$-approximate solution as soon as $\delta(s, q)/\alpha^i < \epsilon \delta(q, u^\ast)$. So in this case:
    \begin{align*}
        \frac{\delta(v_i, q)}{\delta(q, u^\ast)} &\leq \frac{2 \delta_\ast}{\alpha^i \delta(q, u^\ast)} + \frac{\alpha + 1}{\alpha - 1} \\
        &\leq \frac{8 (\alpha + 1) \delta_\ast}{\alpha^i (\alpha - 1) \delta^\ast} + \frac{\alpha + 1}{\alpha - 1}.
    \end{align*}
    As such, the number of steps to reach the approximation level is $\log_\alpha \frac{8 (\alpha + 1) \Delta}{(\alpha - 1) \epsilon}$ which is $\mathcal{O}\big( \log_\alpha \Delta / (\alpha - 1) \epsilon \big)$.

    \textbf{Case 3}: $\delta(s, q) \leq 2 \delta_\ast$ and $\delta(q, u^\ast) < \frac{\alpha - 1}{4(\alpha + 1)}\delta^\ast$. Suppose $v_i \neq u^\ast$. Observe that: (a) $\delta(v_i, u^\ast) \geq \delta^\ast$;
    (b) $\delta(v_i, q) > \delta(q, u^\ast)$; and (c) $\delta(q, u^\ast) < \delta^\ast/2$ by assumption.
    As such, triangle inequality gives us: $\delta(v_i, q) > \delta(v_i, u^\ast) - \delta(u^\ast, q) >
    \delta^\ast - \delta^\ast/2 = \delta^\ast/2$. Together with Equation~(\ref{equation:graphs:alpha-sng:hops-bound}), we obtain:
    \begin{align*}
        \frac{\delta^\ast}{2} &\leq \delta(v_i, q) \leq \frac{2 \delta_\ast}{\alpha^i} + \frac{\delta^\ast}{4} \\
        &\implies \alpha^i \leq 8 \Delta \implies i \leq \log_\alpha 8\Delta.
    \end{align*}

    The three cases together give the desired result.
\end{proof}

In addition to the bounds above,~\cite{indyk2023worstcase} present negative results
for other major SNG-based graph algorithms by proving (via contrived examples)
linear-time lower-bounds on their performance. These results together show the
significance of the pruning parameter $\alpha$ in the $\alpha$-SNG construction.

\subsubsection{Practical Construction of $\alpha$-SNGs}
The algorithm described earlier to construct an $\alpha$-SNG for $m$ points
has $\mathcal{O}(m^3)$ time complexity. That is too expensive for even moderately
large values of $m$. That prompted~\cite{diskann} to approximate the $\alpha$-SNG
by way of heuristics.

The starting point in the approximate construction is a random $R$-regular graph:
Every node is connected to $R$ other nodes selected at random.
The algorithm then processes each node in random order as follows.
Given node $u$, it begins by searching the current snapshot of the graph for the top $L$
nodes for the query point $u$, using Algorithm~\ref{algorithm:graphs:greedy-search}.
Denote the returned set of nodes by $\mathcal{S}$. It then performs the
pruning algorithm by setting $\mathcal{U} = \mathcal{S} \setminus \{ u \}$, rather
than $\mathcal{U} = \mathcal{V} \setminus \{ u \}$.
That is the gist of the modified construction
procedure.\footnote{We have omitted minor but important details of the procedure
in our prose. We refer the interested reader to~\citep{diskann} for a description of
the full algorithm.}

Naturally, we lose all guarantees for approximate top-$k$ retrieval
as a result~\citep{indyk2023worstcase}. We do, however, obtain a more
practical algorithm instead that, as the authors show,
is both efficient and effective.

\section{Closing Remarks}

This chapter deviated from the pattern we got accustomed to so far in the monograph.
The gap between theory and practice in Chapters~\ref{chapter:branch-and-bound}
and~\ref{chapter:lsh} was narrow or none. That gap is rather wide, on the other hand,
in graph-based retrieval algorithms. Making theory work in practice required a great deal
of heuristics and approximations.

Another major departure is the activity in the respective bodies of literature.
Whereas trees and hash families have reached a certain level of maturity,
the literature on graph algorithms is still evolving, actively so.
A quick search through scholarly articles shows growing interest in this class
of algorithms. This monograph itself presented results that were obtained very recently.

There is good reason for the uptick in research activity.
Graph algorithms are among the most successful algorithms there are
for top-$k$ vector retrieval. They are often remarkably fast during retrieval
and produce accurate solution sets.

That success makes it all the more enticing to improve their other characteristics.
For example, graph indices are often large, requiring far too much memory.
Incorporating compression into graphs, therefore, is a low-hanging fruit
that has been explored~\citep{singh2021freshdiskann} but needs further investigation.
More importantly, finding an even sparser graph without losing accuracy
is key in reducing the size of the graph to begin with, and that boils down
to designing better heuristics.

Heuristics play a key role in the construction time of graph indices too.
Building a graph index for a collection of billions of points, for example,
is not feasible for the variant of the Vamana algorithm that offers theoretical
guarantees. Heuristics introduced in that work lost all such guarantees,
but made the graph more practical.

Enhancing the capabilities of graph indices too is an important practical
consideration. For example, when the graph is too large and, so, must rest
on disk, optimizing disk access is essential in maintaining the speed of
query processing~\citep{diskann}. When the collection of vectors is live
and dynamic, the graph index must naturally handle deletions and insertions
in real-time~\citep{singh2021freshdiskann}. When vectors come with metadata
and top-$k$ retrieval must be constrained to the vectors that pass
a certain set of metadata filters, then a greedy traversal of the graph
may prove sub-optimal~\citep{filtered-diskann2023}. All such questions
warrant extensive (often applied) research and go some way to make
graph algorithms more attractive to production systems.

There is thus no shortage of practical research questions.
However, the aforementioned gap between theory and practice should not
dissuade us from developing better theoretical algorithms.
The models that explained the small world phenomenon may not be directly
applicable to top-$k$ retrieval in high dimensions, but they inspired
heuristics that led to the state of the art. Finding theoretically-sound
edge sets that improve over the guarantees offered by Vamana
could form the basis for other, more successful heuristics too.

\bibliographystyle{abbrvnat}
\bibliography{biblio}

\chapter{Clustering}
\label{chapter:ivf}

\abstract{
We have seen index structures that manifest as trees, hash tables, and graphs.
In this chapter, we will introduce a fourth way of organizing data points: clusters.
It is perhaps the most natural and the simplest of the four methods,
but also the least theoretically-justified. We will see why that is as we describe
the details of clustering-based algorithms to top-$k$ retrieval.
}

\section{Algorithm}
\label{section:ivf:retrieval}

As usual, we begin by indexing a collection of $m$ data points $\mathcal{X} \subset \mathbb{R}^d$.
Except in this paradigm, that involves invoking a \textbf{clustering} function,
$\zeta: \mathbb{R}^d \rightarrow [C]$,
that is appropriate for the distance function $\delta(\cdot, \cdot)$,
to map every data point to one of $C$ clusters, where $C$ is an arbitrary parameter.
A typical choice for $\zeta$ is the KMeans algorithm
with $C = \mathcal{O}(\sqrt{m})$.
We then organize $\mathcal{X}$ into a \emph{table} whose row $i$ records the subset
of points that are mapped to the $i$-th cluster: $\zeta^{-1}(i) \triangleq \{ u \;|\; u \in \mathcal{X}, \; \zeta(u) = i \}$.

Accompanying the index is a \textbf{routing} function $\tau: \mathbb{R}^d \rightarrow [C]^\ell$.
It takes an arbitrary point $q$ as input and returns $\ell$ clusters that are more likely to contain
the nearest neighbor of $q$ with respect to $\delta$. In a typical instance of this framework
$\tau(\cdot)$ is defined as follows:
\begin{equation}
    \label{equation:ivf:centroid-based-routing}
    \tau(q) = \argmin^{(\ell)}_{i \in [C]} \delta \Bigg(q,\; \underbrace{\frac{1}{\lvert \zeta^{-1}(i) \rvert} \sum_{u \in \zeta^{-1}(i)} u}_{\mu_i} \Bigg),
\end{equation}
where $\mu_i$ is the \emph{centroid} of the $i$-th cluster.
In other words, $\tau(\cdot)$ simply solves the top-$\ell$ retrieval problem
over the collection of centroids!
We will assume that $\tau$ is defined as above in the remainder of this section.

When processing a query $q$, we take a two-step approach.
We first obtain the list of clusters returned by $\tau(q)$, then solve the top-$k$
retrieval problem over the union of the identified clusters.
Figure~\ref{figure:ivf:framework} visualizes this procedure.

Notice that, the search for top-$\ell$ clusters by using
Equation~(\ref{equation:ivf:centroid-based-routing})
and the secondary search over the clusters identified by $\tau$ are themselves instances of the approximate
top-$k$ retrieval problem. The parameter $C$ determines the amount of effort that must be spent
in each of the two phases of search: When $C=1$, the cluster retrieval problem is solved trivially,
whereas as $C \rightarrow \infty$, cluster retrieval becomes equivalent to top-$k$ retrieval over the
entire collection.
Interestingly, these operations can be delegated to a
subroutine that itself uses a tree-, hash-, graph-, or even a clustering-based solution.
That is, a clustering-based approach can be easily paired with any of the previously discussed methods!

This simple protocol---with some variant of KMeans as $\zeta$ and $\tau$ as in Equation~(\ref{equation:ivf:centroid-based-routing})---works well in
practice~\citep{auvolat2015clustering,pq,bruch2023bridging,invertedMultiIndex,chierichetti2007clusterPruning}. We present the results of our own experiments on various real-world datasets in
Figure~\ref{figure:ivf:clustering-performance}.
This method owes its success to the empirical phenomenon that real-world data points
tend to follow a multi-modal distribution,
naturally forming clusters around each mode. By identifying these clusters
and grouping data points together, we reduce the search space at the expense of retrieval
quality.

\begin{figure}[t]
    \centering
    \includegraphics[width=0.8\linewidth]{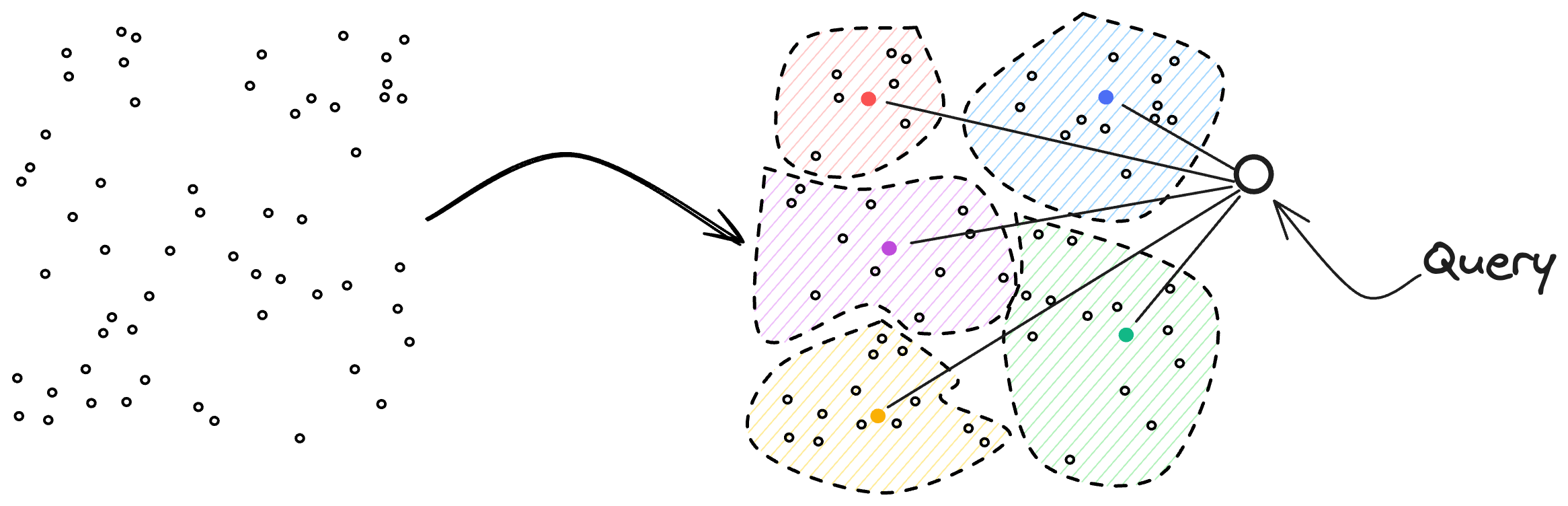}
    \caption{Illustration of the clustering-based retrieval method. The collection of points
    (left) is first partitioned into clusters (regions enclosed by dashed boundary on the right).
    When processing a query $q$
    using Equation~(\ref{equation:ivf:centroid-based-routing}), we compute $\delta(q, \cdot)$ for
    the centroid (solid circles) of every cluster and conduct our search over the $\ell$ ``closest'' clusters.}
    \label{figure:ivf:framework}
\end{figure}

\begin{figure}[t]
    \centering
    \subfloat[\textsc{MIPS}]{
        \includegraphics[width=0.5\linewidth]{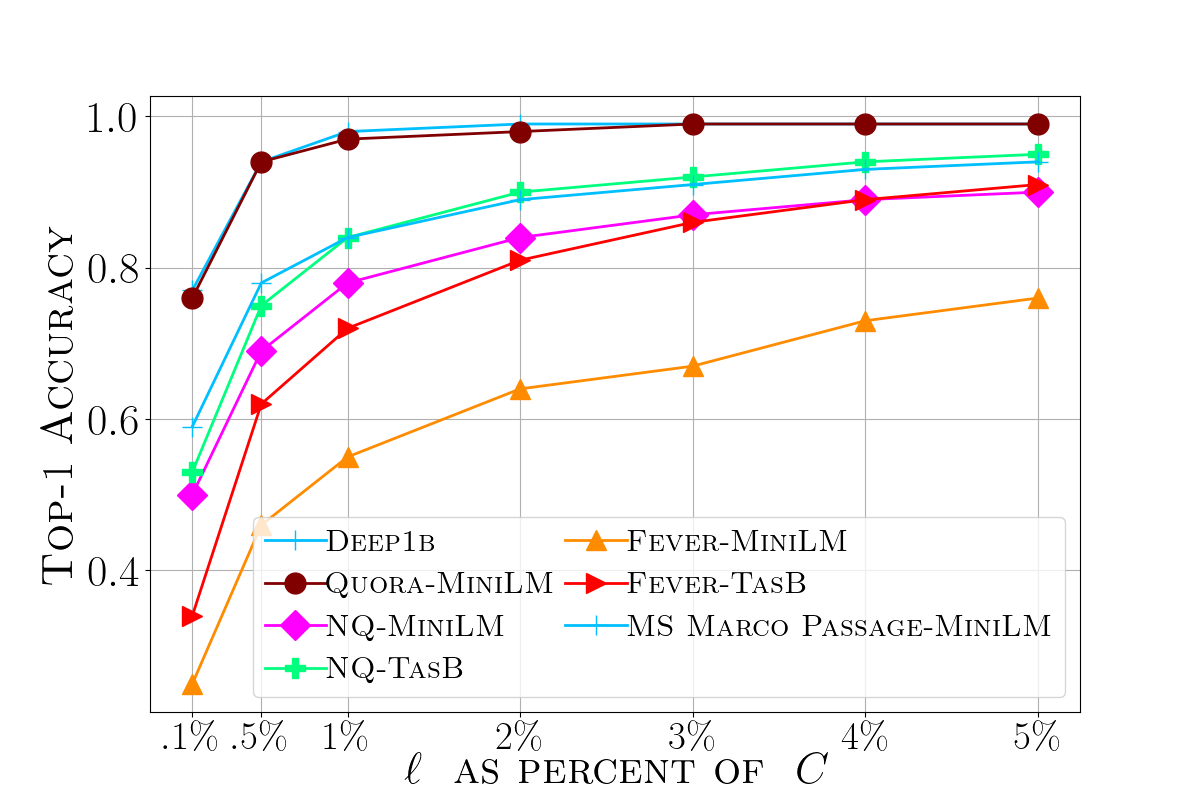}
    }
    \subfloat[\textsc{NN}]{
        \includegraphics[width=0.5\linewidth]{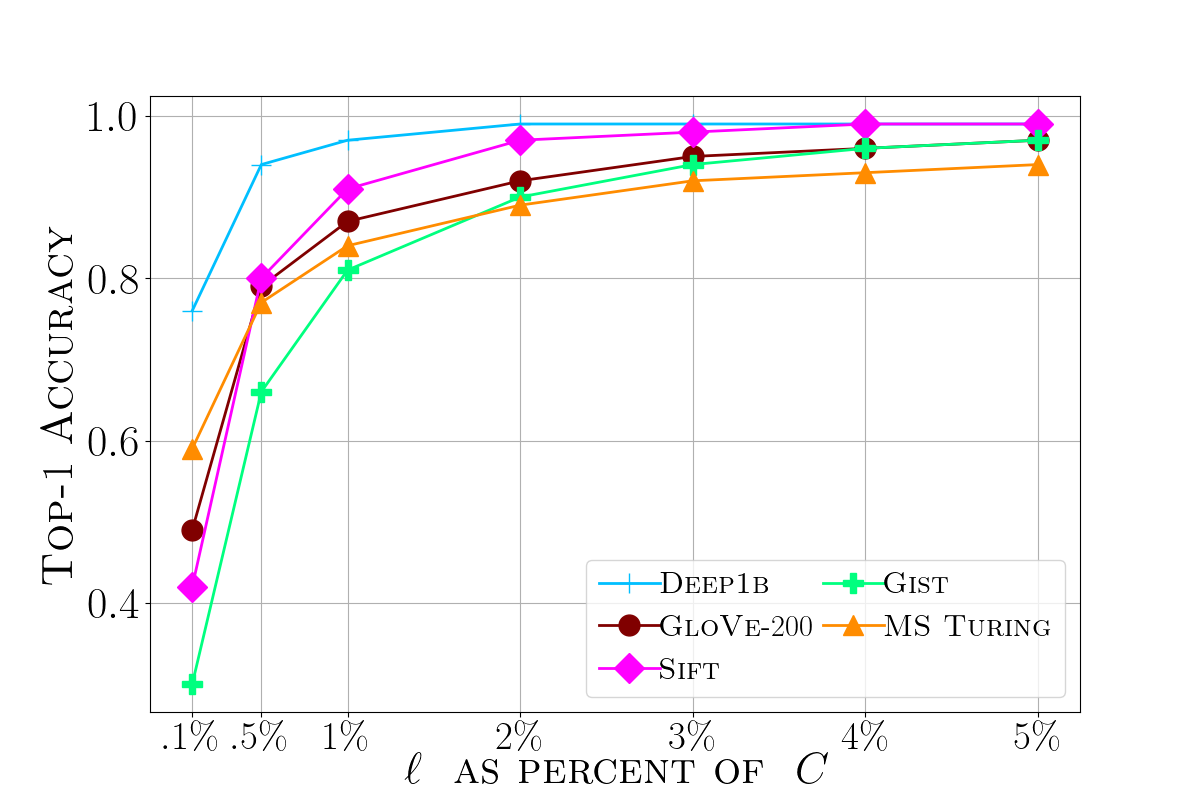}
    }
    \caption{Performance of the clustering-based retrieval method on various real-world collections,
    described in Appendix~\ref{appendix:collections}.
    The figure shows top-$1$ accuracy versus the
    number of clusters, $\ell$, considered by the routing function $\tau(\cdot)$ as a percentage of
    the number of clusters $C$. In these experiments, we set $C = \sqrt{m}$, where $m=\lvert \mathcal{X} \rvert$ is the size of the collection, and use spherical KMeans (MIPS) and standard KMeans (NN)
    to form clusters.}
    \label{figure:ivf:clustering-performance}
\end{figure}

However, to date, no formal analysis has been presented to quantify the retrieval error.
The choice of $\zeta$ and $\tau$, too, have been left largely unexplored, with KMeans and
Equation~(\ref{equation:ivf:centroid-based-routing}) as default answers. It is, for example, not known
if KMeans is the right choice for a given $\delta$. Or, whether clustering with
spillage, where each data point may belong to multiple clusters, might reduce the overall
error, as it did in Spill Trees. It is also an open question if,
for a particular choice of $\zeta$ and $\delta$, there exists
a more effective routing function---including learnt functions tailored to
a query distribution---that uses higher-order statistics from the cluster distributions.

In spite of these shortcomings, the algorithmic framework above contains a
fascinating insight that is actually useful for a rather different end-goal:
vector compression, or more precisely, \emph{quantization}.
We will unpack this connection in Chapter~\ref{chapter:quantization}.

\section{Closing Remarks}
This chapter departed entirely from the theme of this monograph.
Whereas we are generally able to say something intelligent about
trees, hash functions, and graphs, top-$k$ retrieval by clustering
has emerged entirely based on our intuition that data points naturally
form clusters. We cannot formally determine, for example, the behavior
of the retrieval system as a function of the clustering algorithm itself,
the number of clusters, or the routing function. All that must be determined empirically.

What we do observe often in practice, however, is that clustering-based
top-$k$ retrieval is \emph{efficient}~\citep{kmeanslsh,auvolat2015clustering,bruch2023bridging,pq},
at least in the case of Nearest Neighbor search with Euclidean distance,
where KMeans is a theoretically appropriate choice. It is efficient in the sense that
retrieval accuracy often reaches an acceptable level after probing
a few top-ranking clusters as identified by
Equation~(\ref{equation:ivf:centroid-based-routing}).

That we have a method that is efficient in practice, but
its efficiency and the conditions under which it is efficient are unexplained,
constitutes a substantial gap and thus presents multiple consequential open questions.
These questions involve optimal clustering, routing, and bounds on retrieval accuracy.

\bigskip

When the distance function is the Euclidean distance and our
objective is to learn the Voronoi regions of data points, the KMeans
clustering objective makes sense. We can even state formal results
regarding the optimality of the resulting clustering~\citep{kmeansplusplus_2007}.
That argument is no longer valid when the distance function is based on inner product,
where we must learn the inner product Voronoi cones, and where some points
may have an empty Voronoi region. What objective we must optimize for MIPS,
therefore, is an open question that, as we saw in this chapter,
has been partially explored in the past~\citep{scann}.

Even when we know what the right clustering algorithm is, there is still
the issue of ``balance'' that we must understand how to handle. It would,
for example, be far from ideal if the clusters end up having very different
sizes. Unfortunately, that happens quite naturally if the data points have
highly variable norms and the clustering algorithm is based on KMeans:
Data points with large norms become isolated, while vectors with small
norms form massive clusters.

\bigskip

What has been left entirely untouched is the routing machinery.
Equation~(\ref{equation:ivf:centroid-based-routing}) is the \emph{de facto}
routing function, but one that is possibly sub-optimal.
That is because, Equation~(\ref{equation:ivf:centroid-based-routing})
uses the mean of the data points within a cluster as the representative
or \emph{sketch} of that cluster. When clusters are highly concentrated
around their mean, such a sketch accurately reflects the potential of each
cluster. But when clusters have different shapes, higher-order statistics
from the cluster may be required to accurately route queries to clusters.

So the question we are faced with is the following: What is a good sketch
of each cluster? Is there a \emph{coreset} of data points within each cluster
that lead to better routing of queries during retrieval? Can we quantify the
probability of error---in the sense that the cluster containing the optimal
solution is not returned by the routing function---given a sketch?

We may answer these questions differently if we had some idea of what the query
distribution looks like. Assuming access to a set of training queries, it may
be possible to learn a more optimal sketch using supervised learning methods.
Concepts from learning-to-rank~\citep{bruch2023fntir}
seem particularly relevant to this setup. To see how, note that the
outcome of the routing function is to
identify the cluster that contains the optimal data point for a query.
We could view this as ranking clusters with respect to a query,
where we wish for the ``correct'' cluster to appear at the top
of the ranked list. Given this mental model, we
can evaluate the quality of a routing function using any of the many
ranking quality metrics such as Reciprocal Rank (defined as the reciprocal
of the rank of the correct cluster).
Learning a ranking function that maximizes Reciprocal Rank
can then be done indirectly by optimizing a custom
cross entropy-based surrogate, as proved by~\cite{bruch2019xendcg} and~\cite{bruch2021xendcg}.

\bigskip

Perhaps the more important open question is understanding when clustering
is efficient and why. Answering that question may require exploring the connection
between clustering-based top-$k$ retrieval, branch-and-bound algorithms, and LSH.

Take any clustering algorithm, $\zeta$. If one could show formally
that $\zeta$ behaves like an LSH family, then clustering-based top-$k$ retrieval
simply collapses to LSH. In that case, not only do the results from that literature
apply, but the techniques developed for LSH (such as multi-probe LSH) too
port over to clustering.

Similarly, one may adopt the view that finding the top cluster is a series
of decisions, each determining which side of a hyperplane a point falls.
Whereas in Random Partition Trees or Spill Trees, such decision hyperplanes were
random directions, here the hyperplanes are correlated.
Nonetheless, that insight could help us produce clusters with spillage, where
data points belong to multiple clusters, and in a manner that helps reduce
the overall error.

\bibliographystyle{abbrvnat}
\bibliography{biblio}

\chapter{Sampling Algorithms}
\label{chapter:sampling}

\abstract{
Nearly all of the data structures and algorithms we reviewed in the previous chapters are designed
specifically for either nearest neighbor search or maximum cosine similarity search.
MIPS is typically an afterthought. It is often cast as NN or MCS through a rank-preserving
transformation and subsequently solved using one of these algorithms. That is so because
inner product is not a proper metric, making MIPS different from the other vector retrieval variants.
In this chapter, we review algorithms that are specifically designed for MIPS
and that connect MIPS to the machinery underlying multi-arm bandits.
}

\section{Intuition}
\label{section:sampling:intuition}

That inner product is different can be a curse and a blessing.
We have already discussed that curse at length, but in this chapter,
we will finally learn something positive. And that is the fact that inner product
is a linear function of data points and can be easily decomposed into its parts,
thereby opening a unique path to solving MIPS.

The overarching idea in what we refer to as \emph{sampling algorithms} is to avoid
computing inner products. Instead, we either directly approximate the \emph{likelihood}
of a data point being the solution to MIPS (or, equivalently, its rank),
or estimate its inner product with a query (i.e., its score).
As we will see shortly, in both instances, we rely heavily on the linearity of inner product
to estimate probabilities and derive bounds.

Approximating the ranks or scores of data points uses some form of sampling:
we either sample data points according to a distribution defined by inner products,
or sample a dimension to compute partial inner products with and eliminate sub-optimal
data points iteratively. In the former, the more frequently a data point is sampled,
the more likely it is to be the solution to MIPS. In the latter, the more dimensions
we sample, the closer we get to computing full inner products.
Generally, then, the more samples we draw, the more accurate our solution to MIPS becomes.

\begin{svgraybox}
An interesting property of using sampling to solve MIPS is that,
regardless of \emph{what} we are approximating, we can decide when to stop!
That is, if we are given a time budget, we draw as many samples as our
time budget allows and return our best guess of the solutions based on
the information we have collected up to that point. The number of samples,
in other words, serves as a knob that trades off accuracy for speed.
\end{svgraybox}

The remainder of this chapter describes these algorithms in much greater detail.
Importantly, we will see how linearity makes the approximation-through-sampling
feasible and efficient.

\section{Approximating the Ranks}
We are interested in finding the top-$k$ data points
with the largest inner product with a query $q \in \mathbb{R}^d$,
from a collection $\mathcal{X} \subset \mathbb{R}^d$ of $m$ points.
Suppose that we had an efficient way of sampling a data point from $\mathcal{X}$
where the point $u \in \mathcal{X}$ has probability proportional
to $\langle q, u \rangle$ of being selected.

If we drew a sufficiently large number of samples, the data point with the largest
inner product with $q$ would be selected most frequently. The data point with the second
largest inner product would similarly be selected with the second highest frequency, and so on.
So, if we counted the number of times each data point has been sampled,
the resulting histogram would be a good approximation to the rank of each data point
with respect to inner product with $q$.

That is the gist of the sampling algorithm we examine in this section.
But while the idea is rather straightforward, making it work requires addressing
a few critical gaps. The biggest challenge is drawing samples according to
the distribution of inner products without actually computing any of the
inner products! That is because, if we needed to compute $\langle q, u \rangle$
for all $u \in \mathcal{X}$, then we could simply sort data points accordingly
and return the top-$k$; no need for sampling and the rest.

The key to tackling that challenge is the linearity of inner product.
Following a few simple derivations using Bayes' theorem,
we can break up the sampling procedure into two steps, each using marginal distributions
only~\citep{Lorenzen2021wedgeSampling,ballard2015diamondSampling,cohen1997wedgeSampling,pmlr-v89-ding19a}. Importantly, one of these marginal distributions can be computed offline
as part of indexing. That is the result we will review next.

\subsection{Non-negative Data and Queries}
We wish to draw a data point with probability that is proportional to its inner
product with a query: $\probability[u \;|\; q] \propto \langle q, u \rangle$.
For now, we assume that $u, q \succeq 0$ for all $u \in \mathcal{X}$ and queries $q$.

Let us decompose this probability along each dimension as follows:
\begin{equation}
    \probability[u \;|\; q] = \sum_{t = 1}^d \probability[t \;|\; q] \probability[u \;|\; t, q],
\end{equation}
where the first term in the sum is the probability of sampling a dimension $t \in [d]$
and the second term is the likelihood of sampling $u$ given a particular dimension.
We can model each of these terms as follows:
\begin{equation}
    \label{equation:sampling:wedge:prob_dimension}
    \probability[t \;|\; q] \propto \sum_{u \in \mathcal{X}} q_t u_t = q_t \sum_{u \in \mathcal{X}} u_t,
\end{equation}
and,
\begin{equation}
    \label{equation:sampling:wedge:prob_data_point}
    \probability[u \;|\; t, q] = \frac{\probability[u \land t \;|\; q]}{\probability[t \;|\; q]}
    \propto \frac{q_t u_t}{q_t \sum_{v \in \mathcal{X}} v_t} = \frac{u_t}{\sum_{v \in \mathcal{X}} v_t}.
\end{equation}
In the above, we have assumed that $\sum_{v \in \mathcal{X}} v_t \neq 0$; if that sum is $0$
we can simply discard the $t$-th dimension.

What we have done above allows us to draw a sample according to $\probability[u \;|\; q]$
by, instead, drawing a dimension $t$ according to $\probability[t \;|\; q]$ first,
then drawing a data point $u$ according to $\probability[u \;|\; t, q]$.

Sampling from these multinomial distribution requires constructing the distributions themselves.
Luckily, $\probability[u \;|\; t, q]$ is independent of $q$.
Its distribution can therefore be computed offline: we create $d$ tables,
where the $t$-th table has $m$ rows recording the probability of each data point
being selected given dimension $t$ using Equation~(\ref{equation:sampling:wedge:prob_data_point}).
We can then use the alias method~\citep{Walker1977theAliasMethod} to draw samples
from these distributions using $\mathcal{O}(1)$ operations.

The distribution over dimensions given a query, $\probability[t \;|\; q]$,
must be computed online using Equation~(\ref{equation:sampling:wedge:prob_dimension}),
which requires $\mathcal{O}(d)$ operations, assuming we compute $\sum_{u \in \mathcal{X}} u_t$
offline for each $t$ and store them in our index. Again, using the alias method,
we can subsequently draw samples with $\mathcal{O}(1)$ operations.

The procedure described above provides us with an efficient mechanism to
perform the desired sampling. If we were to draw $S$ samples, that could be done
in $\mathcal{O}(d + S)$, where $\mathcal{O}(d)$ term is needed to construct the
multinomial distribution that defines $\probability[t \;|\; q]$.

As we draw samples, we maintain a histogram over the $m$ data points, counting
the number of times each point has been sampled. In the end, we can identify the
top-$k^\prime$ (for $k^\prime \geq k$) points based on these counts, compute their
inner products with the query, and return the top-$k$ points as the final solution set.
All these operations together have time complexity
$\mathcal{O}(d + S + m \log k^\prime + k^\prime d)$, with $S$ typically being
the dominant term.

\subsection{The General Case}
When the data points or queries may be negative, the algorithm described in the
previous section will not work as is. To extend the sampling framework to general,
real vectors, we must make a few minor adjustments.

First, we must ensure that the marginal distributions are valid.
That is easy to do: In Equations~(\ref{equation:sampling:wedge:prob_dimension})
and~(\ref{equation:sampling:wedge:prob_data_point}), we replace each term
with its absolute value. So, $\probability[t \;|\; q]$ becomes proportional to
$\sum_{u \in \mathcal{X}} \lvert q_t u_t \rvert$, and 
$\probability[u \;|\; t, q] \propto \lvert u_t \rvert / \sum_{v \in \mathcal{X}} \lvert v_t \rvert$.

We then use the resulting distributions to sample data points as before,
but every time a data point $u$ is sampled, instead of incrementing its count in the
histogram by one, we add $\textsc{Sign}(q_t u_t)$ to its entry. As the following lemma
shows, in expectation, the final count is proportional to $\langle q, u \rangle$.

\begin{lemma}
    \label{lemma:sampling:wedge:expected-value}
    Define the random variable $Z$ as $0$ if data point $u \in \mathcal{X}$
    is not sampled and $\textsc{Sign}(q_t u_t)$ if it is for a query $q \in \mathbb{R}^d$
    and a sampled dimension $t$.
    Then $\ev[Z] = \langle q, u \rangle / \sum_{t = 1}^d \sum_{v \in \mathcal{X}} \lvert q_t v_t \rvert$.
\end{lemma}
\begin{proof}
    \begin{align*}
        \ev[Z \;|\; t] &= \textsc{Sign}(q_t u_t) \probability[u \;|\; t]
        = \textsc{Sign}(q_t u_t) \frac{\lvert u_t \rvert}{\sum_{v \in \mathcal{X}} \lvert v_t \rvert} \\
        &= \textsc{Sign}(q_t u_t) \frac{\lvert q_t u_t \rvert}{\sum_{v \in \mathcal{X}} \lvert q_t v_t \rvert}
        = \frac{q_t u_t}{\sum_{v \in \mathcal{X}} \lvert q_t v_t \rvert}.
    \end{align*}
    Taking expectation over the dimension $t$ yields:
    \begin{align*}
        \ev[Z] &= \ev\Big[\ev[Z \;|\; t]\Big] = \sum_{t = 1}^d \frac{q_t u_t}{\sum_{v \in \mathcal{X}} \lvert q_t v_t \rvert} \probability[t \;|\; q] \\
        &= \sum_{t = 1}^d \frac{q_t u_t}{\sum_{v \in \mathcal{X}} \lvert q_t v_t \rvert}
            \frac{\sum_{v \in \mathcal{X}} \lvert q_t v_t \rvert}{\sum_{l=1}^d \sum_{v \in \mathcal{X}} \lvert q_l v_l \rvert} \\
        &= \frac{\langle q, u \rangle}{\sum_{t = 1}^d \sum_{v \in \mathcal{X}} \lvert q_t v_t \rvert}.
    \end{align*}
\end{proof}

\subsection{Sample Complexity}
We have formalized an efficient way to sample data points according to the distribution
of inner products, and subsequently collect the most frequently-sampled points.
But how many samples must we draw in order to accurately identify the top-$k$ solution set?
\cite{pmlr-v89-ding19a} give an answer in the form of the following theorem for top-$1$ MIPS.

Before stating the result, it would be helpful to introduce a few shorthands.
Let $N = \sum_{t = 1}^d \sum_{v \in \mathcal{X}} \lvert q_t v_t \rvert$ be a normalizing factor.
For a vector $u \in \mathcal{X}$, denote by $\Delta_u$ the scaled gap between the
maximum inner product and the inner product of $u$ and $q$:
$\Delta_u = \langle q, u^\ast - u \rangle / N$.

If $S$ is the number of samples to be drawn, for a vector $u$, denote by $Z_{u,i}$ 
a random variable that is $0$ if $u$ was not sampled in round $i$,
and otherwise $\textsc{Sign}(q_t u_t)$ if $t$ is the sampled dimension.
Once the sampling has concluded, the final value for point $u$ is simply $Z_u = \sum_i Z_{u,i}$.
Note that, from Lemma~\ref{lemma:sampling:wedge:expected-value}, we have that
$\ev[Z_{u,i}] = \langle q, u \rangle / N$.

Given the notation above, let us also introduce the following helpful lemma.

\begin{lemma}
    Let $C_u = \sum_{t = 1}^d \lvert q_t u_t \rvert$ for a data point $u$. Then
    for a pair of distinct vectors $u, v \in \mathcal{X}$:
    \begin{equation*}
        \ev\Big[\big(Z_{u,i} - Z_{v,i} \big)^2\Big] = \frac{C_u + C_v}{N},
    \end{equation*}
    and,
    \begin{equation*}
        \var\Big[Z_u - Z_v\Big] = S \Big[ \frac{C_u + C_v}{N} - \frac{\langle q, u - v\rangle^2}{N^2} \Big].
    \end{equation*}
\end{lemma}
\begin{proof}
    The proof is similar to the proof of Lemma~\ref{lemma:sampling:wedge:expected-value}.
\end{proof}

\begin{theorem}
    \label{theorem:sampling:wedge:num-samples}
    Suppose $u^\ast$ is the exact solution to MIPS over $m$ points in $\mathcal{X}$ for query $q$.
    Define
    $\sigma_u^2 = \var\Big[Z_{u^\ast} - Z_u\Big]$ and
    let $\Delta = \min_{u \in \mathcal{X}} \Delta_u$.
    For $\delta \in (0, 1)$, if we drew $S$ samples such that:
    \begin{equation*}
        S \geq \max_{u \neq u^\ast}
        \frac{(1 + \Delta_u)^2}{\sigma^2_u h(\frac{\Delta_u (1 + \Delta_u)}{\sigma^2_u})}
        \log \frac{m}{\delta},
    \end{equation*}
    where $h(x) = (1 + x)\log(1 + x) - x$, then
    \begin{equation*}
        \probability[Z_{u^\ast} > Z_u \quad \forall \; u \neq u^\ast] \geq 1 - \delta.
    \end{equation*}
\end{theorem}

Before proving the theorem above, let us make a quick observation.
Clearly $\sigma^2_u \leq \mathcal{O}(d \Delta_u)$ and $(1 + \Delta_u) \approx 1$.
Because $h(\cdot)$ is monotone increasing in its argument ($\frac{\partial h}{\partial x} > 0$),
we can write:
\begin{align*}
    \sigma^2_u h(\frac{\Delta_u (1 + \Delta_u)}{\sigma^2_u}) &=
        \big( \sigma^2_u + \Delta_u(1 + \Delta_u) \big) \log \big( 1 + \frac{\Delta_u (1 + \Delta_u)}{\sigma^2_u} \big) - \Delta_u(1 + \Delta_u) \\
    &\approx \frac{\Delta_u^2 (1 + \Delta_u)^2}{\sigma^2_u} \geq \frac{\Delta}{d} (1 + \Delta)^2
    = \mathcal{O}(\frac{\Delta}{d}).
\end{align*}
Plugging this into Theorem~\ref{theorem:sampling:wedge:num-samples} gives us
$S \leq \mathcal{O}(\frac{d}{\Delta} \log \frac{m}{\delta})$.

\begin{svgraybox}
    Theorem~\ref{theorem:sampling:wedge:num-samples} tells us that, if we draw
    $\mathcal{O}(\frac{d}{\Delta} \log \frac{m}{\delta})$ samples, we can identify the top-$1$
    solution to MIPS with high probability. Observe that, $\Delta$ is a measure of the
    difficulty of the query: When inner products are close to each other, $\Delta$
    becomes smaller, implying that a larger number of samples would be needed to correctly
    identify the exact solution.
\end{svgraybox}

\begin{proof}[Proof of Theorem~\ref{theorem:sampling:wedge:num-samples}]
    Consider the probability that the registered value of a data point $u$
    is greater than or equal to the registered value of the solution $u^\ast$
    once sampling has concluded. That is, $\probability[Z_u \geq Z_{u^\ast}]$.
    Let us rewrite that quantity as follows:
    \begin{align*}
        \probability\Big[Z_u \geq Z_{u^\ast} \Big] &= 
            \probability \Big[\sum_i Z_{u,i} - Z_{{u^\ast},i} \geq 0 \Big] \\
        &= \probability\Big[ \sum_{i=1}^S \underbrace{Z_{u,i} - Z_{{u^\ast},i} + \Delta_u}_{Y_{u,i}} \geq \underbrace{S\Delta_u}_{y_u} \Big].
    \end{align*}
    Notice that $\ev[Y_{u,i}] = 0$ and that $Y_{u,i}$'s are independent.
    Furthermore, $Y_{u,i} \leq 1 + \Delta_u$. Letting $Y_u = \sum_i Y_{u,i}$,
    we can apply Bennett's inequality to bound the probability above:
    \begin{align*}
        \probability[ Y_u \geq y_u] &\leq \exp\Bigg(
            -\frac{S \sigma^2_u}{(1 + \Delta_u)^2}
            h\Big( \frac{(1 + \Delta_u) (S \Delta_u)}{S \sigma^2_u} \Big)
        \Bigg).
    \end{align*}
    Setting the right-hand-side to $\frac{\delta}{m}$, we arrive at:
    \begin{align*}
        \exp&\Bigg(
            -\frac{S \sigma^2_u}{(1 + \Delta_u)^2}
            h\Big( \frac{(1 + \Delta_u) \Delta_u}{\sigma^2_u} \Big)
        \Bigg) \leq \frac{\delta}{m} \\
        &\implies
        S (1 + \Delta_u)^{-2} \sigma^2_u h\Big( \frac{\Delta_u(1 + \Delta_u)}{\sigma^2_u} \Big)
        \geq \log \frac{m}{\delta}.
    \end{align*}
    It is easy to see that for $x > 0$, $h(x) > 0$. Observing that $\Delta_u(1 + \Delta_u)/\sigma^2_u$
    is positive, that implies that $h(\Delta_u(1 + \Delta_u) / \sigma^2_u) > 0$,
    and therefore we can re-arrange the expression above as follows:
    \begin{equation}
        \label{equation:sampling:proof:num-samples}
        S \geq \frac{(1+ \Delta_u)^2}{\sigma^2_u h\Big( \frac{\Delta_u (1 + \Delta_u)}{\sigma^2_u} \Big)} \log \frac{m}{\delta}.
    \end{equation}

    We have thus far shown that when $S$ satisfies the
    inequality in~(\ref{equation:sampling:proof:num-samples}),
    then $\probability[Y_u \geq y_u] \leq \frac{\delta}{m}$.
    Going back to the claim, we derive the following bound using the result above:
    \begin{align*}
        \probability[ &Z_{u^\ast} > Z_u \quad \forall u \in \mathcal{X}] \\
        &= 1 - \probability[ \exists \; u \in \mathcal{X} \; \mathit{s.t.} \quad Z_{u^\ast} \leq Z_u] \\
        &\geq 1 - m \frac{\delta}{m} = 1 - \delta,
    \end{align*}
    where we have used the union bound to obtain the inequality.
\end{proof}

\section{Approximating the Scores}
The method we have just presented avoids the computation of inner products altogether
but estimates the \emph{rank} of each data point with respect to a query using a sampling procedure.
In this section, we introduce another sampling method that approximates the \emph{inner product}
of every data point instead.

Let us motivate our next algorithm with a rather contrived example.
Suppose that our data points and queries are in $\mathbb{R}^2$, with the first
coordinate of vectors drawing values from $\mathcal{N}(0, \sigma_1^2)$
and the second coordinate from $\mathcal{N}(0, \sigma_2^2)$.
If we were to compute the inner product of $q$ with every vector $u \in \mathcal{X}$,
we would need to perform two multiplications and a sum: $u_1q_1 + u_2q_2$.
That gives us the exact ``score'' of every point with respect to $q$.
But if $\sigma_1^2 \gg \sigma_2^2$, then by computing $q_1u_1$ for all $u \in \mathcal{X}$,
it is very likely that we have a good approximation to the 
final inner product. So we may use the partial inner product as a high-confidence
estimate of the full inner product.

That is the core idea in this section. For each data point,
we sample a few dimensions without replacement,
and compute its partial inner product with the query along the chosen dimensions.
Based on the scores so far, we can eliminate data points whose full inner product
is projected, with high confidence, to be too small to make it to the top-$k$ set.
We then repeat the procedure by sampling
more dimensions for the remaining data points, until we reach a stopping criterion.

The process above saves us time by shrinking the set of data points
and computing only partial inner products in each round. But we must decide
how we should sample dimensions and how we should determine which data points to discard.
The objective is to minimize the number of samples needed to identify the solution set.
These are the questions that~\cite{Liu2019banditMIPS} answered in their work,
which we will review next. We note that, even though~\cite{Liu2019banditMIPS}
use the Bandit language~\citep{lattimore2020BanditAlgorithms} to
describe their algorithm, we find it makes for a 
clearer presentation if we avoided the Bandit terminology.

\subsection{The BoundedME Algorithm}

\begin{algorithm}[!t]
\SetAlgoLined
{\bf Input: }{Query point $q \in \mathbb{R}^d$; $k \geq 1$ for top-$k$ retrieval;
confidence parameters $\epsilon, \delta \in (0, 1)$; and data points $\mathcal{X} \subset \mathbb{R}^d$}\\
\KwResult{$(1-\delta)$-confident $\epsilon$-approximate top-$k$ set
to MIPS with respect to $q$.}

\begin{algorithmic}[1]

\STATE $i \leftarrow 1$
\STATE $\mathcal{X}_i \leftarrow \mathcal{X}$ \Comment*[l]{\footnotesize Initialize the solution set to $\mathcal{X}$.}
\STATE $\epsilon_i \leftarrow \frac{\epsilon}{4}$ and $\delta_i \leftarrow \frac{\delta}{2}$

\STATE $A_u \leftarrow 0 \quad \forall \; u \in \mathcal{X}_i$ \Comment*[l]{\footnotesize $A$ is a score accumulator.}

\STATE $t_0 \leftarrow 0$

\WHILE{$\lvert \mathcal{X}_i \rvert > k$}
    \STATE $t_i \leftarrow h\Bigg( \frac{2}{\epsilon_i^2} \log \Big(
        \frac{2(\lvert \mathcal{X}_i \rvert - k)}{\delta_i (\lfloor \frac{\lvert \mathcal{X}_i \rvert - k}{2} \rfloor + 1)}
    \Big) \Bigg)$ \label{algorithm:sampling:boundedme:num-samples}

    \FOR{$u \in \mathcal{X}_i$}
        \STATE Let $\mathcal{J}$ be $(t_i - t_{i - 1})$ dimensions sampled without replacement
        \label{algorithm:sampling:boundedme:sampled-dimensions}
        \STATE $A_u \leftarrow A_u + \sum_{j \in \mathcal{J}} u_j q_j$ \Comment*[l]{\footnotesize Compute partial inner product.}\label{algorithm:sampling:boundedme:partial-ip}
    \ENDFOR

    \STATE Let $\alpha$ be the $\lceil \frac{\lvert \mathcal{X}_i \rvert - k}{2} \rceil$-th score in $A$
    \STATE $\mathcal{X}_{i + 1} \leftarrow \{ u \in \mathcal{X}_i \; \mathit{ s.t. } \; A_u > \alpha \}$
    \STATE $\epsilon_{i + 1} \leftarrow \frac{3}{4} \epsilon_i$,
    $\delta_{i + 1} \leftarrow \frac{\delta_i}{2}$, and $i \leftarrow i + 1$
\ENDWHILE

\RETURN $X_i$
\end{algorithmic}
\caption{The BoundedME algorithm for MIPS.}
\label{algorithm:sampling:boundedme}
\end{algorithm}

The top-$k$ retrieval algorithm developed by~\cite{Liu2019banditMIPS} is presented in
Algorithm~\ref{algorithm:sampling:boundedme}. It is important to note that,
for the algorithm to be correct---as we will explain later---each partial inner product
must be bounded. In other words, for query $q$, any data point $u \in \mathcal{X}$,
and any dimension $t$, we must have that $q_t u_t \in [a, b]$ for some fixed interval.
This is not a restrictive assumption, however: $q$ can always be normalized without
affecting the solution to MIPS, and data points $u$ can be scaled into the hypercube.
In their work,~\cite{Liu2019banditMIPS} assume that partial inner products are 
in the unit interval.

This iterative algorithm begins with
the full collection of data points and removes almost half of the data points
in each iteration. It terminates as soon as the total number of data points left
is at most $k$.

In each iteration of the algorithm, it accumulates partial inner products
for all remaining data point along a set of sampled dimensions. Once a dimension
has been sampled, it is removed from consideration in all future
iterations---hence, sampling \emph{without} replacement.

The number of dimensions to sample is adaptive and changes from iteration to iteration.
It is determined using the quantity on Line~\ref{algorithm:sampling:boundedme:num-samples}
of the algorithm, where the function $h(\cdot)$ is defined as follows:
\begin{equation}
    \label{equation:sampling:boundedme:h}
    h(x) = \min \Big\{ \frac{1 + x}{1 + x/d}, \frac{x + x/d}{1 + x/d} \Big\}.
\end{equation}

At the end of iteration $i$ with the remaining data points in $\mathcal{X}_i$, the algorithm finds the
$\lceil \frac{\lvert \mathcal{X}_i \rvert - k}{2} \rceil$-th (i.e., close to the median)
partial inner product accumulated so far, and discards data points whose
score is less than that threshold. It then updates the confidence parameters
$\epsilon$ and $\delta$, and proceeds to the next iteration.

It is rather obvious that, the total number of dimensions along which
the algorithm computes partial inner products for any given data point
can never exceed $d$. That is simply because once
Line~\ref{algorithm:sampling:boundedme:partial-ip} is executed,
the dimensions in the set $\mathcal{J}$ defined on
Line~\ref{algorithm:sampling:boundedme:sampled-dimensions}
are never considered for sampling
in future iterations. As a result, in the worst case, the algorithm computes
full inner products in $\mathcal{O}(md)$ operations.

As for the time complexity of Algorithm~\ref{algorithm:sampling:boundedme},
it can be shown that it requires $\mathcal{O}(\frac{m \sqrt{d}}{\epsilon} \sqrt{\log(1/\delta)})$
operations. That is simply due to the fact that in each iteration,
the number of data points is cut in half, combined with the inequality
$h(x) \leq \mathcal{O}(\sqrt{dx})$ for $x > 0$.

\begin{theorem}
    \label{theorem:sampling:boundedme:complexity}
    The time complexity of Algorithm~\ref{algorithm:sampling:boundedme}
    is $\mathcal{O}(\frac{m \sqrt{d}}{\epsilon} \sqrt{\log(1/\delta)})$.
\end{theorem}

\begin{svgraybox}
    Theorem~\ref{theorem:sampling:boundedme:complexity} says that the time complexity
    of Algorithm~\ref{algorithm:sampling:boundedme} is linear in the number of data points
    $m$, but sub-linear in the number of dimensions $d$. That is a fundamentally different
    behavior than all the other algorithms we have presented thus far throughout the preceding
    chapters.
\end{svgraybox}

\begin{proof}[Proof of Theorem~\ref{theorem:sampling:boundedme:complexity}]
    Let us first show the following claim: $h(x) \leq \mathcal{O}(\sqrt{dx})$ for $x > 0$.
    To prove that, observe that $h(x)$ is the minimum of two positive values $a$ and $b$.
    As such, $h(x) \leq \sqrt{ab}$. Substituting $a$ and $b$ with the right expressions
    from Equation~(\ref{equation:sampling:boundedme:h}):
    \begin{align*}
        h(x) &\leq \sqrt{\frac{1 + x}{1 + x/d} \frac{x + x/d}{1 + x/d}}
        = \frac{1}{1 + x/d} \sqrt{x(1 + x)(1 + 1/d)} \\
        &= \frac{\mathcal{O}(x)}{1 + x/d}
        = \mathcal{O}(\frac{dx}{d + x}) \leq \mathcal{O}(\sqrt{dx}).
    \end{align*}

    Note that, in the $i$-th iteration there are at most $m/2^i$ data points to examine.
    Moreover, for each data point that is eliminated in round $i$, we will have computed
    at most $t_i$ partial inner products (see Line~\ref{algorithm:sampling:boundedme:num-samples}
    of Algorithm~\ref{algorithm:sampling:boundedme}).
    Using these facts, we can calculate the time complexity as follows:
    \begin{align*}
        \sum_{i = 1}^{\log m} \frac{m}{2^i} h(t_i) \leq 
        \sum_{i = 1}^{\log m} \frac{m}{2^i} \sqrt{d t_i} \leq
        \mathcal{O}\Big(\frac{m \sqrt{d}}{\epsilon} \sqrt{\log\frac{1}{\delta}} \Big).
    \end{align*}
\end{proof}

\subsection{Proof of Correctness}
Our goal in this section is to prove that Algorithm~\ref{algorithm:sampling:boundedme}
is correct, in the sense that it returns the $\epsilon$-approximate solution to $k$-MIPS
with probability at least $1 - \delta$:

\begin{theorem}
    \label{theorem:sampling:boundedme:correctness}
    Algorithm~\ref{algorithm:sampling:boundedme} is guaranteed to return the $\epsilon$-approximate
    solution to $k$-MIPS with probability at least $1 - \delta$.
\end{theorem}

The proof of Theorem~\ref{theorem:sampling:boundedme:correctness} requires the
concentration inequality due to~\cite{remi2015concentrationInequality}, repeated below
for completeness.

\begin{lemma}
    \label{lemma:sampling:concentration-inequality-sampling-wo-replacement}
    Let $\mathcal{J} \subset [0, 1]$ be a finite set of size $d$
    with mean $\mu$. Let $\{J_1, J_2, \ldots, J_n\}$ be $n < d$
    samples from $\mathcal{J}$ without replacement. Then for any $n \leq d$
    and any $\delta \in [0, 1]$ it holds:
    \begin{equation*}
        \probability\Big[
            \frac{1}{n} \sum_{t = 1}^n J_t - \mu \leq 
            \sqrt{\frac{\rho_n}{2n} \log \frac{1}{\delta}}
        \Big] \geq 1 - \delta,
    \end{equation*}
    where $\rho_n$ is defined as follows:
    \begin{equation*}
        \rho_n = \min \Big\{ 1 - \frac{n - 1}{d}, (1 - \frac{n}{d})(1 + \frac{1}{n}) \Big\}.
    \end{equation*}
\end{lemma}

The lemma above guarantees that, with probability at least $1 - \delta$,
the empirical mean of the samples does not exceed the mean of the universe
by a specific amount that depends on $\delta$. 
We now wish to adapt that result to derive a similar guarantee where the difference
between means is bounded by an arbitrary parameter $\epsilon$.
That is stated in the following lemma.

\begin{lemma}
    \label{lemma:sampling:boundedme:concentration}
    Let $\mathcal{J} \subset [0, 1]$ be a finite set of size $d$
    with mean $\mu$. Let $\{J_1, J_2, \ldots, J_n\}$ be $n < d$
    samples from $\mathcal{J}$ without replacement.
    Then for any $\epsilon, \delta \in (0, 1)$, if we have that:
    \begin{equation*}
        n \geq \min \Big\{ \frac{1 + x}{1 + x/d}, \frac{x + x/d}{1 + x/d} \Big\},
    \end{equation*}
    where $x = \log(1/\delta)/2\epsilon^2$, then the following holds:
    \begin{equation*}
        \probability\Big[
            \frac{1}{n} \sum_{t = 1}^n J_t - \mu \leq \epsilon
        \Big] \geq 1 - \delta.
    \end{equation*}
\end{lemma}
\begin{proof}
    By Lemma~\ref{lemma:sampling:concentration-inequality-sampling-wo-replacement}
    we can see that:
    \begin{equation*}
        \probability\Big[ \frac{1}{n} \sum_{t = 1}^n J_t - \mu \leq \epsilon \Big] \geq 1 - \delta,
    \end{equation*}
    so long as:
    \begin{equation*}
        \sqrt{\frac{\rho_n}{2n} \log \frac{1}{\delta}} \leq \epsilon \implies
        \frac{n}{\rho_n} \geq \frac{1}{2\epsilon^2} \log\frac{1}{\delta}.
    \end{equation*}

    There are two cases to consider. First, if $\rho_n = 1 - (n-1)/d$, then:
    \begin{align*}
        \frac{n}{\rho_n} \geq \underbrace{\frac{1}{2\epsilon^2} \log\frac{1}{\delta}}_x
        &\implies
        \frac{n}{1 - \frac{n - 1}{d}} \geq x \\
        &\implies n \geq \frac{x + x/d}{1 + x/d}.
    \end{align*}
    In the second case, $\rho_n = (1 - n/d)(1 + 1/n)$, which gives:
    \begin{align*}
        \frac{n}{\rho_n} \geq \underbrace{\frac{1}{2\epsilon^2} \log\frac{1}{\delta}}_x
        &\implies
        \frac{n}{(1 - \frac{n}{d})(1 + \frac{1}{n})} \geq x \\
        &\implies n \geq \Big[ 1 + \frac{1}{n} - \frac{n + 1}{d} \Big] x \\
        &\implies n^2 \geq nx + x - \frac{n^2}{d}x - \frac{n}{d}x \\
        &\implies (1 + \frac{x}{d}) n^2 - (x - \frac{x}{d})n - x \geq 0.
    \end{align*}
    To make the closed-form solution more
    manageable,~\cite{Liu2019banditMIPS} relax the problem above and solve $n$ in the following
    problem instead. Note that, any solution to the problem below is a valid solution
    to the problem above.
    \begin{align*}
        (1 + \frac{x}{d}) n^2 - (x - \frac{x}{d})n - x - 1\geq 0
        &\implies \Big[ (1 + \frac{x}{d}) n - x - 1 \Big][n + 1] \geq 0 \\
        &\implies n \geq \frac{1 + x}{1 + x/d}.
    \end{align*}

    By combining the two cases, we obtain:
    \begin{equation*}
        n \geq \min \{ \frac{1 + x}{1 + x/d}, \frac{x + x/d}{1 + x/d} \}.
    \end{equation*}
\end{proof}

\begin{svgraybox}
    Lemma~\ref{lemma:sampling:boundedme:concentration} gives us the minimum number of
    dimensions we must sample so that the partial inner product of a vector with a query
    is at most $\epsilon$ away from the full inner product, with probability at least
    $1 - \delta$.
\end{svgraybox}

Armed with this result, we can now proceed to proving the main theorem.

\begin{proof}[Proof of Theorem~\ref{theorem:sampling:boundedme:correctness}]
    Denote by $\zeta_i$ the $k$-th largest \emph{full} inner product among
    the set of data points $\mathcal{X}_i$ in iteration $i$. If we showed that,
    for two consecutive iterations, the difference between
    $\zeta_i$ and $\zeta_{i + 1}$ does not exceed $\epsilon_i$ with probability
    at least $1 - \delta_i$, that is:
    \begin{equation}
        \label{equation:sampling:boundedme:correctness:proof}
        \probability\Big[ \zeta_i - \zeta_{i + 1} \leq \epsilon_i \Big] \geq 1 - \delta_i,
    \end{equation}
    then the theorem immediately follows:
    \begin{equation*}
        \probability\Big[ \zeta_1 - \zeta_{\log m} \leq \epsilon \Big] \geq 1 - \delta,
    \end{equation*}
    because:
    \begin{equation*}
        \sum_{i = 1}^{\log m} \delta_i = \sum_{i = 1}^{\log m} \frac{\delta}{2^i} \leq \sum_{i = 1}^{\infty} \frac{\delta}{2^i} = \delta,
    \end{equation*}
    and,
    \begin{equation*}
        \sum_{i = 1}^{\log m} \epsilon_i = \sum_{i = 1}^{\log m} \frac{\epsilon}{4} \big(\frac{3}{4}\big)^{i - 1}
        \leq \sum_{i = 1}^{\infty} \frac{\epsilon}{4} \big(\frac{3}{4}\big)^{i - 1} = \epsilon.
    \end{equation*}
    So we focus on proving Equation~(\ref{equation:sampling:boundedme:correctness:proof}).

    Suppose we are in the $i$-th iteration. Collect in $\mathcal{Z}_{\epsilon_i}$
    every data point in $u \in \mathcal{X}_i$ such that $\zeta_i - \langle q, u \rangle \leq \epsilon_i$.
    That is: $\mathcal{Z}_{\epsilon_i} = \{ u \in \mathcal{X}_i \;|\; \zeta_i - \langle q, u \rangle \leq \epsilon_i \}$. If at least $k$ elements of $\mathcal{Z}_{\epsilon_i}$ end up in $\mathcal{X}_{i + 1}$,
    the event $\zeta_i - \zeta_{i + 1} \leq \epsilon_i$ succeeds. So, that event fails
    if there are more than $\lfloor \frac{\lvert \mathcal{X}_i \rvert - k}{2} \rfloor$
    data points in $\mathcal{X}_i \setminus \mathcal{Z}_{\epsilon_i}$ with partial inner products
    that are greater than partial inner products of the data points in $\mathcal{Z}_{\epsilon_i}$.
    Denote the number of such data points by $\beta$.

    What is the probability that a data point $u$ in $\mathcal{X}_i \setminus \mathcal{Z}_{\epsilon_i}$
    has a higher partial inner product than any data point in $\mathcal{Z}_{\epsilon_i}$?
    Assuming that $u^\ast$ is the data point that achieves $\zeta_i$, we can write:
    \begin{align*}
        \probability\big[ A_u \geq A_v \quad \forall \; v \in \mathcal{Z}_{\epsilon_i} \big]
        &\leq \probability\big[ A_u \geq A_{u^\ast} \big] \\
        &\leq \probability\big[
            A_u \geq \langle q, u \rangle + \frac{\epsilon_i}{2} \;\lor\;
            A_{u^\ast} \leq \zeta_i - \frac{\epsilon_i}{2}
        \big] \\
        &\leq \probability\big[
            A_u \geq \langle q, u \rangle + \frac{\epsilon_i}{2} \big]
        + \probability\big[
            A_{u^\ast} \leq \zeta_i - \frac{\epsilon_i}{2}
        \big].
    \end{align*}
    We can apply Lemma~\ref{lemma:sampling:boundedme:concentration} to obtain that,
    if the number of sampled dimensions is equal to the quantity on
    Line~\ref{algorithm:sampling:boundedme:num-samples} of
    Algorithm~\ref{algorithm:sampling:boundedme}, then the probability above
    would be bounded by:
    \begin{equation*}
        \frac{\lfloor \frac{\lvert \mathcal{X}_i \rvert - k}{2} \rfloor + 1}{\lvert \mathcal{X}_i \rvert - k} \delta_i.
    \end{equation*}

    Using this result along with Markov's inequality, we can bound the probability that
    $\beta$ is strictly greater than $\lfloor \frac{\lvert \mathcal{X}_i \rvert - k}{2} \rfloor$
    as follows:
    \begin{align*}
        \probability\Big[ \beta \geq \frac{\lvert \mathcal{X}_i \rvert - k}{2} + 1 \Big]
        &\leq \frac{\ev[\beta]}{\frac{\lvert \mathcal{X}_i \rvert - k}{2} + 1} \\
        &\leq \frac{ (\lvert \mathcal{X}_i \rvert - k) \frac{\lfloor \frac{\lvert \mathcal{X}_i \rvert - k}{2} \rfloor + 1}{\lvert \mathcal{X}_i \rvert - k} \delta_i }{\frac{\lvert \mathcal{X}_i \rvert - k}{2} + 1} \\
        &= \delta_i.
    \end{align*}
    That completes the proof of Equation~(\ref{equation:sampling:boundedme:correctness:proof})
    and, therefore, the theorem.
\end{proof}

\section{Closing Remarks}
The algorithms in this chapter were unique in two ways.
First, they directly took on the challenging problem of MIPS.
This is in contrast to earlier chapters where MIPS was only an afterthought.
Second, there is little to no pre-processing involved in the preparation of
the index, which itself is small in size.
That is unlike trees, hash buckets, graphs, and clustering
that require a generally heavy index that itself is computationally-intensive
to build.

The approach itself is rather unique as well. It is particularly interesting
because the trade-off between efficiency and accuracy can be adjusted during retrieval.
That is not the case with trees, LSH, or graphs, where the construction of the index
itself heavily influences that balance. With sampling methods, it is at least
theoretically possible to adapt the retrieval strategy to the hardness of the query
distribution. That question remains unexplored.

Another area that would benefit from further research is the sampling strategy itself.
In particular, in the BoundedME algorithm, the dimensions that are sampled next
are drawn randomly. While that simplifies analysis---which follows the analysis of
popular Bandit algorithms---it is not hard to argue that the strategy is sub-optimal.
After all, unlike the Bandit setup, where reward distributions are unknown and
samples from the reward distributions are revealed only gradually,
here we have direct access to all data points \emph{a priori}.
Whether and how adapting the sampling strategy to the underlying data or query distribution
may improve the error bounds or the accuracy or efficiency of the algorithm in practice
remains to be studied.

\bibliographystyle{abbrvnat}
\bibliography{biblio}

\begin{partbacktext}
\part{Compression}
\end{partbacktext}

\chapter{Quantization}
\label{chapter:quantization}

\abstract{
In a vector retrieval system, it is usually not enough to process queries as fast
as possible. It is equally as important to reduce the size of the index
by compressing vectors. Compression, however, must be done in such a way
that either decompressing the vectors during retrieval incurs a negligible cost, or
distances can be computed (approximately) in the compressed domain,
rendering it unnecessary to decompress compressed vectors during retrieval.
This chapter introduces a class of vector compression algorithms,
known as quantization, that is inspired by clustering.
}

\section{Vector Quantization}

Let us take a step back and present a different mental model of the clustering-based retrieval framework
discussed in Chapter~\ref{chapter:ivf}.
At a high level, we band together points that are placed by $\zeta(\cdot)$ into cluster $i$
and represent that group by $\mu_i$, for $i \in [C]$. In the first stage of the search for query $q$,
we take the following conceptual step:
First, we compute $\delta(q, \mu_i)$ for every $i$ and construct a ``table''
that maps $i$ to $\delta(q, \mu_i)$. We next approximate $\delta(q, u)$ for every $u \in \mathcal{X}$
using the resulting table: If $u \in \zeta^{-1}(i)$, then we look up an estimate of
its distance to $q$ from the $i$-th row of the table.
We then identify the $\ell$ closest distances, and perform a secondary search over
the corresponding vectors.

This presentation of clustering for top-$k$ retrieval highlights an important fact that does
not come across as clearly in our original description of the algorithm:
We have made an implicit assumption that $\delta(q, u) \approx \delta(q, \mu_i)$
for all $u \in \zeta^{-1}(i)$.
That is why we presume that if a cluster minimizes $\delta(q, \cdot)$, then the points within it
are also likely to minimize $\delta(q, \cdot)$.
That is, in turn, why we deem it sufficient to search over the points within the top-$\ell$ clusters.

Put differently, within the first stage of search, we appear to be approximating every point
$u \in \zeta^{-1}(i)$ with $\tilde{u} = \mu_i$. Because there are $C$ discrete choices to consider
for every data point, we can say that we \emph{quantize} the vectors into $[C]$.
Consequently, we can encode each vector using only
$\log_2 C$ bits, and an entire collection of vectors using $m \log_2 C$ bits!
All together, we can represent a collection $\mathcal{X}$ using $\mathcal{O}(Cd + m \log_2 C)$ space,
and compute distances to a query by performing $m$ look-ups into a table
that itself takes $\mathcal{O}(Cd)$ time to construct.
That quantity can be far smaller than $\mathcal{O}(md)$ given by the na\"ive distance computation
algorithm.

Clearly, the approximation error, $\lVert u - \tilde{u} \rVert$, is a function of $C$.
As we increase $C$, this approximation improves, so that $\lVert u - \tilde{u} \rVert \rightarrow 0$
and $\lvert \delta(q, u) - \delta(q, \tilde{u}) \rvert \rightarrow 0$.
Indeed, $C = m$ implies that $\tilde{u}=u$ for every $u$.
But increasing $C$ results in an increased space complexity and a less efficient distance computation.
At $C = m$, for example, our table-building exercise does not help speed up
distance computation for individual data points---because we
must construct the table in $\mathcal{O}(md)$ time anyway.
Finding the right $C$ is therefore critical
to space- and time-complexity, as well as the approximation or quantization error.

\subsection{Codebooks and Codewords}
What we described above is known as vector quantization~\citep{Gray1998Quantization}
for vectors in the $L_2$ space. We will therefore assume that $\delta(u, v) = \lVert u - v \rVert_2^2$
in the remainder of this section.
The function $\zeta: \mathbb{R}^d \rightarrow [C]$ is called a \emph{quantizer},
the individual centroids are referred to as \emph{codewords}, and the set of $C$
codewords make up a \emph{codebook}. It is easy to see that the set $\zeta^{-1}(i)$
is the intersection of $\mathcal{X}$ with the Voronoi region associated with codeword $\mu_i$.

The approximation quality of a given codebook is measured by the familiar mean squared error:
$\ev [\lVert \mu_{\zeta(U)} - U \rVert^2_2)]$, with $U$ denoting a random vector.
Interestingly, that is exactly the objective
that is minimized by Lloyd's algorithm for KMeans clustering. As such, an optimal codebook
is one that satisfies Lloyd's optimality conditions: each data point must be quantized to its
nearest codeword, and each Voronoi region must be represented by its mean. That is why
KMeans is our default choice for $\zeta$.

\section{Product Quantization}

As we noted earlier, the quantization error is a function of the number of clusters, $C$:
A larger value of $C$ drives down the approximation error, making the quantization
and the subsequent top-$k$ retrieval solution more accurate and effective.
However, realistically, $C$ cannot become too large, because then the framework
would collapse to exhaustive search, degrading its efficiency.
How may we reconcile the two seemingly opposing forces?

\cite{pq} gave an answer to that question in the form of Product Quantization (PQ).
The idea is easy to describe at a high level:
Whereas in vector quantization we quantize the entire vector
into one of $C$ clusters, in PQ we break up a vector into orthogonal subspaces
and perform vector quantization on individual chunks separately.
The quantized vector is then a concatenation of the quantized subspaces.

Formally, suppose that the number of dimensions $d$ is divisible by $d_\circ$,
and let $L = d/d_\circ$. Define a selector matrix $S_i \in \{0, 1\}^{d_\circ \times d}$, $1 \leq i \leq L$
as a matrix with $L$ blocks in $\{0, 1\}^{d_\circ \times d_\circ}$, where all blocks are $0$
but the $i$-th block is the identity.
The following is an example for $d = 6$, $d_\circ = 2$, and $i = 2$:
\begin{equation*}
    S_2 = 
    \begin{bmatrix}
        0 & 0 & 1 & 0 & 0 & 0 \\
        0 & 0 & 0 & 1 & 0 & 0
    \end{bmatrix}
\end{equation*}

For a given vector $u \in \mathbb{R}^d$, $S_i u$ gives the $i$-th $d_\circ$-dimensional subspace,
so that we can write: $u = \bigoplus_i\; S_i u$.
Suppose further that we have $n$ quantizers $\zeta_1$ through $\zeta_L$,
where $\zeta_i: \mathbb{R}^{d_\circ} \rightarrow [C]$ maps the subspace selected
by $S_i$ to one of $C$ clusters. Each $\zeta_i$ gives us $C$ centroids $\mu_{i,j}$
for $j \in [C]$.

Using the notation above, we can express the PQ code for a vector $u$ as $L$
cluster identifiers, $\zeta_i(S_i u)$, for $i \in [L]$. We can therefore quantize
a $d$-dimensional vector using $L \log_2 C$ bits. Observe that, when $L = 1$ (or equivalently,
$d_\circ = d$), PQ reduces to vector quantization. When $L = d$, on the other hand,
PQ performs scalar quantization per dimension.

Given this scheme, our approximation of $u$ is $\tilde{u} = \bigoplus_i \mu_{i, \zeta_i(u)}$.
It is easy to see that the quantization error $\ev[\lVert U - \tilde{U} \rVert_2^2]$, with $U$
denoting a random vector drawn from $\mathcal{X}$ and $\tilde{U}$ its reconstruction, is
the sum of the quantization error of individual subspaces:
\begin{align*}
    \ev[\lVert U - \tilde{U} \rVert_2^2] &=
        \frac{1}{m} \sum_{u \in \mathcal{X}} \Big[ \lVert u - \bigoplus_{i=1}^{L} \mu_{i, \zeta_i(u)} \rVert_2^2 \Big] \\
        &= \frac{1}{m} \sum_{u \in \mathcal{X}} \Big[ \sum_{i=1}^{L} \lVert S_i u - \mu_{i, \zeta_i(u)} \rVert_2^2 \Big].
\end{align*}
As a result, learning the $L$ codebooks can be formulated as $L$ independent
sub-problems. The $i$-th codebook can therefore be learnt by
the application of KMeans on $S_i \mathcal{X} = \{ S_i u \;|\; u \in \mathcal{X} \}$.

\subsection{Distance Computation with PQ}
In vector quantization, computing the distance of a vector $u$ to a query $q$
was fairly trivial. All we had to do was to precompute a table that maps
$i \in [C]$ to $\lVert q - \mu_i \rVert_2$, then look up the entry that corresponds to $\zeta(u)$.
The fact that we were able to precompute $C$ distances once per query, then simply look up
the right entry from the table for a vector $u$ helped us save a great deal of computation.
Can we devise a similar algorithm given a PQ code?

The answer is yes. Indeed, that is why PQ has proven to be an efficient
algorithm for distance computation. As in vector quantization, it first computes
$L$ distance tables, but the $i$-th table maps $j \in [C]$ to $\lVert S_i q - \mu_{i,j} \rVert_2^2$
(note the \emph{squared} $L_2$ distance). Using these tables, we can estimate the distance
between $q$ and any vector $u$ as follows:
\begin{align*}
    \lVert q - u \rVert_2^2 &\approx \lVert q - \tilde{u} \rVert_2^2 \\
    &= \lVert q - \bigoplus_{i=1}^L \mu_{i, \zeta_i(u)} \rVert_2^2 \\
    &= \lVert \bigoplus_{i = 1}^L \Big( S_i q - \mu_{i, \zeta_i(u)} \Big) \rVert_2^2 \\
    &= \sum_{i = 1}^L \lVert S_i q - \mu_{i, \zeta_i(u)} \rVert_2^2.
\end{align*}
Observe that, we have already computed the summands and recorded them in the distance tables.
As a result, approximating the distance between $u$ and $q$ amounts to $L$ table look-ups.
The overall amount of computation needed to approximate distances between $q$ and
$m$ vectors in $\mathcal{X}$ is then $\mathcal{O}(LCd_\circ + mL)$.

\bigskip

We must remark on the newly-introduced parameter $d_\circ$.
Even though in the context of vector quantization, the impact of $C$
on the quantization error is not theoretically known, there is nonetheless
a clear interpretation: A larger $C$ leads to better quantization.
In PQ, the impact of $d_\circ$ or, equivalently, $L$ on the quantization error
is not as clear. As noted earlier, we can say something about $d_\circ$ at the extremes,
but what we should expect from a value somewhere between $1$ and $d$ is largely
an empirical question~\citep{sun2023automating}.

\subsection{Optimized Product Quantization}
In PQ, we allocate an equal number of bits ($\log_2 C$) to each of the $n$ orthogonal subspaces.
This makes sense if our vectors have similar energy in every subspace.
But when the dimensions in one subspace are highly correlated, and in another
uncorrolated, our equal-bits-per-subspace allocation policy proves wasteful in the former
and perhaps inadequate in the latter. How can we ensure a more balanced energy across subspaces?

\cite{pq} argue that applying a random rotation $R \in \mathbb{R}^{d \times d}$ ($RR^T = I$)
to the data points prior to quantization is one way to reduce the correlation between dimensions.
The matrix $R$ together with $S_i$'s, as defined above, determines how we decompose the vector
space into its subspaces. By applying a rotation first, we no longer chunk up an input vector
into sub-vectors that comprise of consecutive dimensions.

Later,~\cite{opq} and~\cite{norouzi2013ckmeans} extended this idea and suggested that the matrix
$R$ can be learnt jointly with the codebooks. This can be done through an iterative algorithm
that switches between two steps in each iteration. In the first step, we freeze $R$ and learn
a PQ codebook as before. In the second step, we freeze the codebook and update the matrix $R$
by solving the following optimization problem:
\begin{align*}
    \min_R &\sum_{u \in \mathcal{X}} \lVert Ru - \tilde{u} \rVert_2^2, \\
    \mathit{s.t.} & \quad RR^T = I,
\end{align*}
where $\tilde{u}$ is the approximation of $u$ according to the frozen PQ codebook.
Because $u$ and $\tilde{u}$ are fixed in the above optimization problem,
we can rewrite the objective as follows:
\begin{align*}
    \min_R & \lVert RU - \tilde{U} \rVert_F, \\
    \mathit{s.t.} & \quad RR^T = I,
\end{align*}
where $U$ is a $d$-by-$m$ matrix where each column is a vector in $\mathcal{X}$,
$\tilde{U}$ is a matrix where each column is an approximation of the corresponding
column in $U$, and $\lVert \cdot \rVert_F$ is the Frobenius norm.
This problem has a closed-form solution as shown by~\cite{opq}.

\subsection{Extensions}

Since the study by~\cite{pq}, many variations of the idea have emerged in
the literature. In the original publication, for example,~\cite{pq}
used PQ codes in conjunction with the clustering-based retrieval framework presented earlier
in this chapter. In other words, a collection $\mathcal{X}$ is first clustered
into $C$ clusters (``coarse-quantization''), and each cluster
is subsequently represented using its own PQ codebook.
In this way, when the routing function identifies a cluster to search,
we can compute distances for data points within that cluster using their PQ codes.
Later,~\cite{invertedMultiIndex} extended this two-level quantization further
by introducing the ``inverted multi-index'' structure.

When combining PQ with clustering or coarse-quantization, instead of producing
PQ codebooks for raw vectors within each cluster, one could learn codebooks
for the \emph{residual} vectors instead. That means, if the centroid of the
$i$-th cluster is $\mu_i$, then we may quantize $(u - \mu_i)$ for each
vector $u \in \zeta^{-1}(i)$. This was the idea first introduced by~\cite{pq},
then developed further in subsequent works~\citep{locallyOptimizedPQ,multiscaleQuantization}.

The PQ literature does not end there. In fact, so popular, effective, and efficient is PQ
that it pops up in many different contexts and a variety of applications. Research into improving
its accuracy and speed is still ongoing. For example, there have been many works that speed
up the distance computation with PQ codebooks by leveraging hardware capabilities~\citep{pqWithGPU,Andre_2021,pqCacheLocality}.
Others that extend the algorithm to streaming (online) collections~\citep{onlinePQ},
and yet other studies that investigate other PQ
codebook-learning protocols~\citep{deepPQ,Yu_2018_ECCV,chen2020DifferentiablePQ,Jang_2021_ICCV,Klein_2019_CVPR,lu2023differeitableOPQ}.
This list is certainly not exhaustive and is still growing.

\section{Additive Quantization}

PQ remains the dominant quantization method for top-$k$ retrieval due to its overall simplicity
and the efficiency of its codebook learning protocol. There are, however, numerous generalizations of the
framework~\citep{additiveQuantization,chen2010approximate,Niu2023RVPQ,liu2015improvedRVQ,Ozan2016CompetitiveQuantization,krishnan2021projective}.
Typically, these generalized forms improve the approximation error but require more involved
codebook learning algorithms and vector encoding protocols. In this section, we review
one key algorithm, known as Additive Quantization (AQ)~\citep{additiveQuantization},
that is the backbone of all other methods.

Like PQ, AQ learns $L$ codebooks where each codebook consists of $C$ codewords.
Unlike PQ, however, each codeword is a vector in $\mathbb{R}^d$---rather than $\mathbb{R}^{d_\circ}$.
Furthermore, a vector $u$ is approximated as the \emph{sum}, instead of the concatenation,
of $L$ codewords, one from each codebook: $\tilde{u} = \sum_{i=1}^L \mu_{i, \zeta_i(u)}$,
where $\zeta_{i}: \mathbb{R}^d \rightarrow [C]$ is the quantizer associated with the
$i$-th codebook.

Let us compare AQ with PQ at a high level and understand how AQ is different.
We can still encode a data point using $L \log_2 C$ bits, as in PQ.
However, the codebooks for AQ are $L$-times larger than their PQ counterparts,
simply because each codeword has $d$ dimensions instead of $d_\circ$.
On the other hand, AQ does not decompose the space into orthogonal subspaces
and, as such, makes no assumptions about the independence between subspaces.

AQ is therefore a strictly more general quantization method than PQ.
In fact, the class of additive quantizers contains the class of product quantizers:
By restricting the $i$-th codebook in AQ to the set of codewords that are $0$
everywhere outside of the $i$-th ``chunk,'' we recover PQ.
Empirical comparisons~\citep{additiveQuantization,Matsui2018PQSurvey} confirm
that such a generalization is more effective in practice.

For this formulation to be complete, we have to specify how the codebooks are learnt,
how we encode an arbitrary vector, and how we perform distance computation.
We will cover these topics in reverse order in the following sections.

\subsection{Distance Computation with AQ}
Suppose for the moment that we have learnt AQ codebooks for a collection $\mathcal{X}$
and that we are able to encode an arbitrary vector into an AQ code (i.e.,
a vector of $L$ codeword identifiers). In this section, we examine how
we may compute the distance between a query point $q$ and a data point $u$
using its approximation $\tilde{u}$.

Observe the following fact:
\begin{equation*}
    \lVert q - u \rVert_2^2 = \lVert q \rVert_2^2 - 2 \langle q, u \rangle + \lVert u \rVert_2^2.
\end{equation*}
The first term is a constant that can be computed once per query and, at any rate,
is inconsequential to the top-$k$ retrieval problem. The last term, $\lVert u \rVert_2^2$
can be stored for every vector and looked up during distance computation,
as suggested by~\cite{additiveQuantization}. That means, the encoding of a vector
$u \in \mathcal{X}$ comprises of two components: $\tilde{u}$ and its (possibly scalar-quantized)
squared norm. This brings the total space required to encode $m$ vectors
to $\mathcal{O}(LCd + m (1 + L\log_2 C))$.

The middle term can be approximated by $\langle q, \tilde{u} \rangle$
and can be expressed as follows:
\begin{equation*}
    \langle q, u \rangle \approx \langle q, \tilde{u} \rangle =
    \sum \langle q, \mu_{i, \zeta_i(u)} \rangle.
\end{equation*}
As in PQ, the summands can be computed once for all codewords, and stored in a table.
When approximating the inner product, we can do as before and look up
the appropriate entries from these precomputed tables.
The time complexity of this operation is therefore $\mathcal{O}(LCd + mL)$
for $m$ data points, which is similar to PQ.

\subsection{AQ Encoding and Codebook Learning}
While distance computation with AQ codes is fairly similar to the process
involving PQ codes, the encoding of a data point is substantially different
and relatively complex in AQ. That is because we can no longer simply assign
a vector to its nearest codeword. Instead, we must find an arrangement of
$L$ codewords that together minimize the approximation error
$\lVert u - \tilde{u} \rVert_2$.

Let us expand the expression for the approximation error as follows:
\begin{align*}
    \lVert u - \tilde{u} \rVert_2^2 &= 
        \lVert u - \sum_{i = 1}^L \mu_{i, \zeta_i(u)} \rVert_2^2 \\
        &= \lVert u \rVert_2^2 - 2 \langle u, \sum_{i = 1}^L \mu_{i, \zeta_i(u)} \rangle + \lVert \sum_{i = 1}^L \mu_{i, \zeta_i(u)} \rVert_2^2 \\
        &= \lVert u \rVert_2^2 + \Big( \sum_{i = 1}^L - 2 \langle u, \mu_{i, \zeta_i(u)} \rangle + \lVert \mu_{i, \zeta_i(u)} \rVert_2^2 \Big) + \\
        &\quad \quad \sum_{1 \leq i < j \leq L} 2 \langle \mu_{i, \zeta_i(u)}, \mu_{j, \zeta_j(u)} \rangle.
\end{align*}
Notice that the first term is irrelevant to the objective function, so we may ignore it.
We must therefore find $\zeta_i$'s that minimize the remaining terms.

\cite{additiveQuantization} use a generalized Beam search to solve this optimization problem. 
The algorithm begins by selecting $L$ closest codewords from
$\bigcup_{i=1}^L \{ \mu_{i,1} \ldots \mu_{i,C} \}$ to $u$.
For a chosen codeword $\mu_{k,j}$, we compute the residual $u - \mu_{k,j}$
and find the $L$ closest codewords to it from $\bigcup_{i \neq k} \{ \mu_{i,1} \ldots \mu_{i,C} \}$.
After performing this search for all chosen codewords from the first round,
we end up with a maximum of $L^2$ unique pairs of codewords.
Note that, each pair has codewords from two different codebooks.

Of the $L^2$ pairs, the algorithm picks the top $L$ that minimize the approximation error.
It then repeats this process for a total of $L$ rounds, where in each
round we compute the residuals given $L$ tuples of codewords,
and for each tuple, find $L$ codewords from the remaining codebooks,
and ultimately identify the top $L$ tuples from the $L^2$ tuples.
At the end of the $L$-th round, the tuple with the minimal approximation error
is the encoding for $u$.

\bigskip

Now that we have addressed the vector encoding part, it remains to
describe the codebook learning procedure. Unsurprisingly, learning a codebook
is not so dissimilar to the PQ codebook learning algorithm.
It is an iterative procedure alternating between two steps to optimize
the following objective:
\begin{equation*}
    \min_{\mu_{i,j}} \sum_{u \in \mathcal{X}} \lVert u - \sum_{i=1}^L \mu_{i, \zeta_i(u)} \rVert_2^2.
\end{equation*}

One step of every iteration freezes the codewords and performs assignments $\zeta_i$'s,
which is the encoding problem we have already discussed above. The second step freezes
the assignments and updates the codewords, which itself is a least-squares problem
that can be solved relatively efficiently, considering that it decomposes over each dimension.

\section{Quantization for Inner Product}

The vector quantization literature has largely been focused on the Euclidean distance
and the approximate nearest neighbor search problem. Those ideas typically port over to
the maximum cosine similarity search with little effort, but not to MIPS under general conditions.
To understand why, suppose we wish to find a quantizer such that the 
inner product approximation error is minimized for a query distribution:
\begin{align*}
    \ev_{q} \Big[ \sum_{u \in \mathcal{X}} \big( \langle q, u \rangle - \langle q, \tilde{u} \rangle \big)^2 \Big] &= \sum_{u \in \mathcal{X}} \ev_q \Big[ \langle q, u - \tilde{u} \rangle^2 \Big] \\
    &= \sum_{u \in \mathcal{X}} \ev_q \Big[ (u - \tilde{u})^T qq^T (u - \tilde{u}) \Big] \\
    &= \sum_{u \in \mathcal{X}} (u - \tilde{u})^T \ev_q \big[ qq^T \big] (u - \tilde{u})
    \numberthis \label{equation:quantization:mips-objective},
\end{align*}
where $\tilde{u}$ is an approximation of $u$.
If we assumed that $q$ is isotropic, so that its covariance matrix is the identity matrix
scaled by some constant,
then the objective above reduces to the reconstruction error. In that particular case,
it makes sense for the quantization objective to be based on the reconstruction error,
making the quantization methods we have studied thus far appropriate for MIPS too.
But in the more general case, where the distribution of $q$ is anisotropic, there is
a gap between the true objective and the reconstruction error.

\cite{guo2016Quip} showed that, if we are able to obtain a small sample of queries
to estimate $\ev[qq^T]$, then we can modify the assignment step in
Lloyd's iterative algorithm for KMeans in order to minimize the objective in
Equation~(\ref{equation:quantization:mips-objective}).
That is, instead of assigning points to clusters by their Euclidean distance to 
the (frozen) centroids, we must instead use Mahalanobis distance characterized by
$\ev[qq^T]$. The resulting quantizer is arguably more suitable for inner product
than the plain reconstruction error.

\subsection{Score-aware Quantization}

Later,~\cite{scann} argued that the objective in Equation~(\ref{equation:quantization:mips-objective})
does not adequately capture the nuances of MIPS. Their argument rests on an observation and an intuition.
The observation is that, in Equation~(\ref{equation:quantization:mips-objective}), every single
data point contributes equally to the optimization objective. Intuitively, however,
data points are not equally likely to be the solution to MIPS.
The error from data points that are more likely to be the maximizers of inner product
with queries should therefore be weighted more heavily than others.

On the basis of that argument,~\cite{scann} introduce the following objective
for inner product quantization:
\begin{equation}
    \label{equation:quantization:scann-objective}
    \sum_{u \in \mathcal{X}} \underbrace{\ev_q \Big[  \omega(\langle q, u \rangle) \;
    \langle q, u - \tilde{u} \rangle^2 \Big]}_{\ell(u, \tilde{u}, \omega)}.
\end{equation}
In the above, $\omega: \mathbb{R} \rightarrow \mathbb{R}^+$ is an arbitrary weight function
that determines the importance of each data point to the optimization objective.
Ideally, then, $\omega$ should be monotonically non-decreasing in its argument.
One such weight function is $\omega(s) = \mathbbm{1}_{s \geq \theta}$ for some threshold $\theta$,
implying that only data points whose expected inner product is at least $\theta$ contribute to the
objective, while the rest are simply ignored. That is the weight function that~\cite{scann}
choose in their work.

\medskip

Something interesting emerges from Equation~(\ref{equation:quantization:scann-objective}) with the choice
of $\omega(s) = \mathbbm{1}_{s \geq \theta}$: It is more important for $\tilde{u}$ to preserve the
\emph{norm} of $u$ than it is to preserve its \emph{angle}. We will show why that is shortly, but
consider for the moment the reason this behavior is important for MIPS.
Suppose there is a data point whose norm is much larger than the rest of the data points.
Intuitively, such a data point has a good chance of maximizing inner product
with a query even if its angle with the query is relatively large.
In other words, being a candidate solution to MIPS is less sensitive to angles and more sensitive
to norms. Of course, as norms become more and more concentrated, angles take on a bigger role in determining
the solution to MIPS. So, intuitively, an objective that penalizes the distortion of norms more than
angles is more suitable for MIPS.

\subsubsection{Parallel and Orthogonal Residuals}

\begin{figure}[t]
    \centering
    \includegraphics[width=0.4\linewidth]{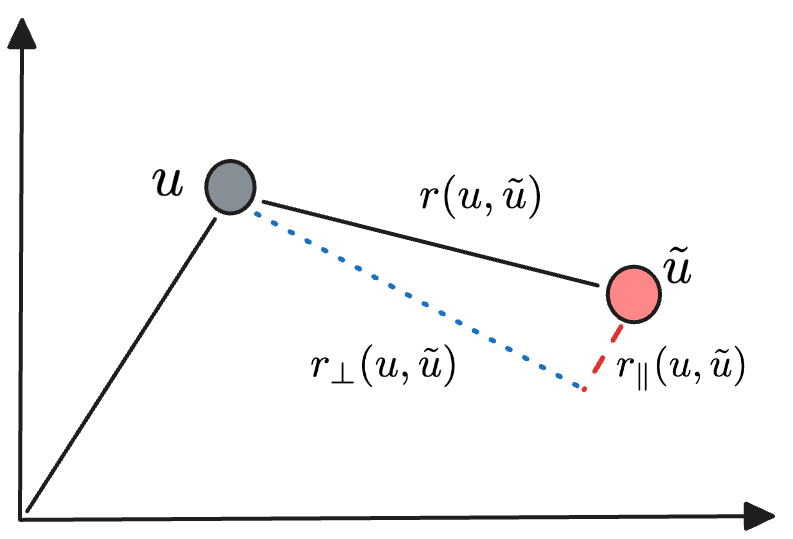}
    \caption{Decomposition of the residual error $r(u, \tilde{u}) = u - \tilde{u}$ for $u \in \mathbb{R}^2$
    to one component that is \emph{parallel} to the data point, $r_\parallel(u, \tilde{u})$,
    and another that is \emph{orthogonal} to it, $r_\perp(u, \tilde{u})$.}
    \label{figure:quantization:residual-error-decomposition}
\end{figure}

Let us present this phenomenon more formally and show why the statement above is true.
Define the residual error as $r(u, \tilde{u}) = u - \tilde{u}$.
The residual error can be decomposed into two components: one that is parallel to the
data point, $r_{\parallel}(u, \tilde{u})$, and another that is orthogonal to it, $r_\perp(u, \tilde{u})$,
as depicted in Figure~\ref{figure:quantization:residual-error-decomposition}.
More concretely:
\begin{equation*}
    r_\parallel(u, \tilde{u}) = \frac{ \langle u - \tilde{u}, u \rangle}{\lVert u \rVert^2}\; u,
\end{equation*}
and,
\begin{equation*}
    r_\perp(u, \tilde{u}) = r(u, \tilde{u}) - r_\parallel(u, \tilde{u}).
\end{equation*}

\cite{scann} show first that, regardless of the choice of $\omega$, the loss
defined by $\ell(u, \tilde{u}, \omega)$ in Equation~(\ref{equation:quantization:scann-objective})
can be decomposed as stated in the following theorem.

\begin{figure}[t]
    \centering
    \includegraphics[width=0.4\linewidth]{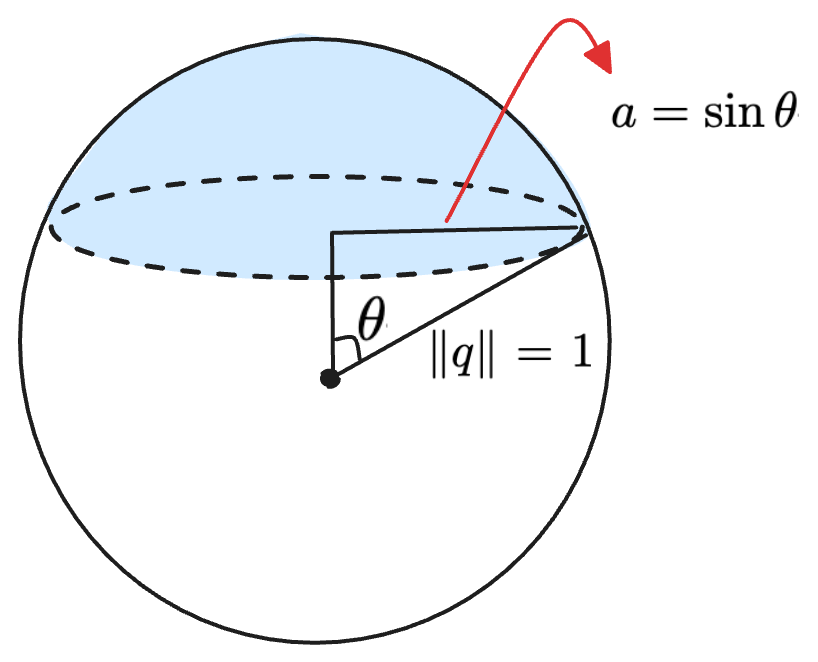}
    \caption{The probability that the angle between a fixed data point $u$ with a
    unit-normed query $q$ that is drawn from a spherically-symmetric distribution is at most $\theta$,
    is equal to the surface area of the spherical cap with base radius $a = \sin \theta$. This fact
    is used in the proof of Theorem~\ref{theorem:quantization:scann:decomposition}.}
    \label{figure:quantization:spherical-cap-area}
\end{figure}

\begin{theorem}
    \label{theorem:quantization:scann:decomposition}
    Given a data point $u$, its approximation $\tilde{u}$, and any weight function $\omega$,
    the objective of Equation~(\ref{equation:quantization:scann-objective}) can be decomposed as follows
    for a spherically-symmetric query distribution:
    \begin{equation*}
        \ell(u, \tilde{u}, \omega) \propto h_\parallel(\omega, \lVert u \rVert) \; \lVert r_\parallel(u, \tilde{u}) \rVert^2 +
        h_\perp(\omega, \lVert u \rVert) \; \lVert r_\perp(u, \tilde{u}) \rVert^2,
    \end{equation*}
    where,
    \begin{equation*}
        h_\parallel(\omega, t) = \int_0^\pi \omega(t \cos \theta) \Big( \sin^{d-2}\theta - \sin^d \theta \Big) \; d\theta,
    \end{equation*}
    and,
    \begin{equation*}
        h_\perp(\omega, t) = \frac{1}{d-1} \int_0^\pi \omega(t \cos \theta) \sin^d \theta \; d\theta.
    \end{equation*}
\end{theorem}
\begin{proof}
Without loss of generality, we can assume that queries are unit vectors (i.e., $\lVert q \rVert = 1$).
Let us write $\ell(u, \tilde{u}, \omega)$ as follows:
\begin{align*}
    \ell(u, \tilde{u}, \omega) &= \ev_q \Big[ \omega(\langle q, u \rangle) \; \langle q, u - \tilde{u} \rangle^2 \Big] \\
    &= \int_0^\pi \omega(\lVert u \rVert \cos \theta) \ev_q \Big[ \langle q, u - \tilde{u} \rangle^2 \big| \langle q, u \rangle = \lVert u \rVert \cos \theta \Big] \; d \probability\Big[ \theta_{q,u} \leq \theta \Big],
\end{align*}
where $\theta_{q, u}$ denotes the angle between $q$ and $u$.

Observe that
$\probability\Big[ \theta_{q,u} \leq \theta \Big]$ is the surface area of a spherical cap with base radius
$a = \lVert q \rVert \sin \theta = \sin \theta$---see Figure~\ref{figure:quantization:spherical-cap-area}.
That quantity is equal to:
\begin{equation*}
    \lVert q \rVert^{d - 1} \frac{\pi^{d/2}}{\Gamma(d/2)} I(a^2; \frac{d-1}{2}, \frac{1}{2}),
\end{equation*}
where $\Gamma$ is the Gamma function and $I(z; \cdot, \cdot)$ is the incomplete Beta function.
We may therefore write:
\begin{align*}
    \frac{d \probability\Big[ \theta_{q,u} \leq \theta \Big]}{d \theta} &\propto \Big[ (1 - a^2)^{\frac{1}{2} - 1} (a^2)^{\frac{d - 1}{2} - 1} \Big] \frac{d a}{d \theta} \\
    &= \frac{\sin^{d-3} \theta}{\cos \theta} (2 \sin \theta \cos \theta) \\
    &\propto \sin^{d-2} \theta,
\end{align*}
where in the first step we used the fact that $d I(z; s, t) = (1 - z)^{t - 1} z^{s-1} dz$.

Putting everything together, we can rewrite the loss as follows:
\begin{equation*}
    \ell(u, \tilde{u}, \omega) \propto \int_0^\pi \omega(\lVert u \rVert \cos \theta) \ev_q \Big[ \langle q, u - \tilde{u} \rangle^2 \big| \langle q, u \rangle = \lVert u \rVert \cos \theta \Big] \sin^{d-2}\theta \; d\theta.
\end{equation*}
We can complete the proof by applying the following lemma to the expectation over queries in the integral
above.

\begin{lemma}
\begin{equation*}
    \ev_q[\langle q, u - \tilde{u} \rangle^2 \big| \langle q, u \rangle = t] = 
    \frac{t^2}{\lVert u \rVert^2} \lVert r_\parallel(u, \tilde{u}) \rVert^2 +
    \frac{1 - t^2/\lVert u \rVert^2}{d - 1} \rVert r_\perp(u, \tilde{u}) \rVert^2.
\end{equation*}
\end{lemma}
\begin{proof}
    We use the shorthand $r_\parallel = r_\parallel(u, \tilde{u})$ and similarly
    $r_\perp = r_\perp(u, \tilde{u})$.
    Decompose $q = q_\parallel + q_\perp$ where $q_\parallel = \langle q, u \rangle \frac{u}{\lVert u \rVert^2}$ and $q_\perp = q - q_\parallel$. We can now write:
    \begin{align*}
        \ev_q[\langle q, u - \tilde{u} \rangle^2 \big| \langle q, u \rangle = t] =
        \ev_q[\langle q_\parallel, r_\parallel \rangle^2] \big| \langle q, u \rangle = t] +
        \ev_q[\langle q_\perp, r_\perp \rangle^2] \big| \langle q, u \rangle = t].
    \end{align*}
    All other terms are equal to $0$ either due to orthogonality or components or because of
    spherical symmetry. The first term is simply equal to
    $\lVert r_\parallel \rVert^2 \frac{t^2}{\lVert u \rVert^2}$.
    By spherical symmetry, it is easy to show that the second term reduces to
    $\frac{1 - t^2/\lVert u \rVert^2}{d - 1} \lVert r_\perp \rVert^2$.
    That completes the proof.
\end{proof}

Applying the lemma above to the integral, we obtain:
\begin{equation*}
    \ell(u, \tilde{u}, \omega) \propto \int_0^\pi \omega(\lVert u \rVert \cos \theta) 
    \Bigg(
    \cos^2 \theta \lVert r_\parallel(u, \tilde{u}) \rVert^2 + \frac{\sin^2 \theta}{d - 1} \lVert r_\perp(u, \tilde{u}) \rVert^2
    \Bigg)
    \sin^{d-2}\theta \; d\theta,
\end{equation*}
as desired.
\end{proof}

When $\omega(s) = \mathbbm{1}_{s \geq \theta}$ for some $\theta$,~\cite{scann} show that
$h_\parallel$ outweighs $h_\perp$, as the following theorem states. This implies that
such an $\omega$ puts more emphasis on preserving the parallel residual error as discussed earlier.

\begin{theorem}
    For $\omega(s) = \mathbbm{1}_{s \geq \theta}$ with $\theta \geq 0$, $h_\parallel(\omega, t) \geq h_\perp(\omega, t)$,
    with equality if and only if $\omega$ is constant over the interval $[-t, t]$.
\end{theorem}
\begin{proof}
    We can safely assume that $h_\parallel$ and $h_\perp$ are positive; they are $0$ if and only if
    $\omega(s) = 0$ over $[-t, t]$. We can thus express the ratio between them as follows:
    \begin{equation*}
        \frac{h_\parallel(\omega, t)}{h_\perp(\omega, t)} = (d - 1)
            \Bigg( \frac{\int_0^\pi \omega(t\cos \theta) \sin^{d-2}\theta \;d\theta }{\int_0^\pi \omega(t \cos \theta) \sin^d\theta \; d\theta} - 1\Bigg) = (d - 1) \Big( \frac{I_{d-2}}{I_d} -1 \Big),
    \end{equation*}
    where we denoted by $I_d = \int_0^\pi \omega(t \cos \theta) \sin^d \theta \; d\theta$.
    Using integration by parts:
    \begin{align*}
        I_d &= - \omega(t \cos \theta) \cos \theta \sin^{d-1}\theta \Big|_0^\pi + \\
        &\int_0^\pi \cos \theta \Big[ \omega(t \cos \theta)(d-1)\sin^{d-2}\theta \cos \theta -
        \omega^\prime(t \cos \theta) t \sin^d \theta \Big] \; d\theta \\
        &= (d - 1) \int_0^\pi \omega(t \cos \theta) \cos^2 \theta \sin^{d-2}\theta \; d\theta -
        t \int_0^\pi \omega^\prime(t \cos \theta) \cos \theta \sin^d \theta \; d\theta \\
        &= (d - 1) I_{d-2} - (d - 1) I_d -
        t \int_0^\pi \omega^\prime(t \cos \theta) \cos \theta \sin^d \theta \; d\theta.
    \end{align*}
    Because $\omega(s) = 0$ for $s < 0$, the last term reduces to an integral over $[0, \pi/2]$.
    The resulting integral is non-negative because sine and cosine are both non-negative over that interval.
    It is $0$ if and only if $\omega^\prime = 0$, or equivalently when $\omega$ is constant.
    We have therefore shown that:
    \begin{equation*}
        I_d \leq (d - 1) I_{d-2} - (d-1) I_d \implies (d - 1) \Big( \frac{I_{d-2}}{I_d} - 1 \Big) \geq 1
        \implies \frac{h_\parallel(\omega, t)}{h_\perp(\omega, t)} \geq 1,
    \end{equation*}
    with equality when $\omega$ is constant, as desired.
\end{proof}

\subsubsection{Learning a Codebook}

The results above formalize the intuition that the parallel residual plays a more important
role in quantization for MIPS. If we were to plug the formalism above into the objective
in Equation~(\ref{equation:quantization:scann-objective}) and optimize it to learn a codebook, we would
need to compute $h_\parallel$ and $h_\perp$ using Theorem~\ref{theorem:quantization:scann:decomposition}.
That would prove cumbersome indeed.

Instead,~\cite{scann} show that $\omega(s) = \mathbbm{1}_{s \geq \theta}$ results in a more
computationally-efficient optimization problem. 
Letting $\eta(t) = \frac{h_\parallel(\omega, t)}{h_\perp(\omega, t)}$,
they show that $\eta/(d-1)$ concentrates around $\frac{(\theta/t)^2}{1 - (\theta/t)^2}$
as $d$ becomes larger. So in high dimensions, one can rewrite the objective function
of Equation~(\ref{equation:quantization:scann-objective}) as follows:
\begin{equation*}
    \sum_{u \in \mathcal{X}} \frac{(\theta/\lVert u \rVert)^2}{1 - (\theta/\lVert u \rVert)^2}
    \lVert r_\parallel(u, \tilde{u}) \rVert^2 +
    \lVert r_\perp(u, \tilde{u}) \rVert^2.
\end{equation*}

\cite{scann} present an optimization procedure that is based on Lloyd's iterative algorithm for
KMeans, and use it to learn a codebook by minimizing the objective above.
Empirically, such a codebook outperforms the one that
is learnt by optimizing the reconstruction error.

\subsubsection{Extensions}
The score-aware quantization loss has, since its publication, been extended
in two different ways. \cite{Zhang_Liu_Lian_Liu_Wu_Chen_2022} adapted the objective
function to an Additive Quantization form. \cite{queryAwareQuantization}
updated the weight function $\omega(\cdot)$ so that the importance of a data point
can be estimated based on a given set of training queries. Both extensions lead
to substantial improvements on benchmark datasets.

\bibliographystyle{abbrvnat}
\bibliography{biblio}

\chapter{Sketching}
\label{chapter:sketching}

\abstract{
Sketching is a probabilistic tool to summarize high-dimensional
vectors into low-dimensional vectors, called \emph{sketches},
while approximately preserving properties of interest.
For example, we may sketch vectors in the Euclidean space
such that their $L_2$ norm is approximately preserved;
or sketch points in an inner product space such that the
inner product between any two points is maintained with high probability.
This chapter reviews a few \emph{data-oblivious} algorithms, cherry-picked
from the vast literature on sketching, that are tailored to sparse vectors
in an inner product space.
}

\section{Intuition}

We learnt about quantization as a form of vector compression in Chapter~\ref{chapter:quantization}.
There, vectors are decomposed into $L$ subspaces, with each subspace
mapped to $C$ geometrically-cohesive buckets.
By coding each subspace into only $C$ values, we can encode an entire vector in $L \log C$ bits,
often dramatically reducing the size of a vector collection, though at the cost of
losing information in the process.

The challenge, we also learnt, is that not enough can be said about the effects of
$L$, $C$, and other parameters involved in the process of quantization,
on the reconstruction error. We can certainly intuit the asymptotic behavior of quantization,
but that is neither interesting nor insightful. That leaves us no option
other than settling on a configuration empirically.

Additionally, learning codebooks can become involved and cumbersome.
It involves tuning parameters and running clustering algorithms,
whose expected behavior is itself ill-understood when handling improper distance functions.
The resulting codebooks too may become obsolete in the event of a distributional shift.

This chapter reviews a different class of compression techniques known as
\emph{data-oblivious sketching}. Let us break down this phrase and understand each part better.

The data-oblivious qualifier is rather self-explanatory:
We make no assumptions about the input data, and in fact, do not
even take advantage of the statistical properties of the data.
We are, in other words, completely agnostic and oblivious to our input.

\begin{svgraybox}
While oblivion may put us at a disadvantage and lead to a larger magnitude of error,
it creates two opportunities. First, we can often easily quantify the
average qualities of the resulting compressed vectors.
Second, by design, the compressed vectors are robust under any data drift.
Once a vector collection has been compressed, in other words, we can safely
assume that any guarantees we were promised will continue to hold.    
\end{svgraybox}

Sketching, to continue our unpacking of the concept, is a probabilistic tool
to reduce the dimensionality of a vector space while preserving
certain properties of interest \emph{with high probability}.
In its simplest form, sketching is a function $\phi: \mathbb{R}^d \rightarrow \mathbb{R}^{d_\circ}$,
where $d_\circ < d$. If the ``property of interest'' is the Euclidean distance
between any pair of points in a collection $\mathcal{X}$, for instance,
then $\phi(\cdot)$ must satisfy the following for random points $U$ and $V$:
\begin{equation*}
    \probability\Bigg[ \Big\lvert \lVert \phi(U) - \phi(V) \rVert_2 - 
    \lVert U - V \rVert_2 \Big\rvert < \epsilon \Bigg] > 1 - \delta,
\end{equation*}
for $\delta, \epsilon \in (0, 1)$.

The output of $\phi(u)$, which we call the \emph{sketch} of vector $u$,
is a good substitute for $u$ itself. If all we care about, as we do in top-$k$
retrieval, is the distance between pairs of points, then we retain the ability
to deduce that information with high probability just from the sketches of a
collection of vectors. Considering that $d_\circ$ is smaller than $d$,
we not only compress the collection through sketching, but, as with quantization,
we are able to perform distance computations directly on the compressed vectors.

\bigskip

The literature on sketching offers numerous algorithms that are designed
to approximate a wide array of norms, distances, and other properties of data.
We refer the reader to the excellent monograph by~\cite{woodruff2014sketching}
for a tour of this rich area of research. But to give the reader a better understanding
of the connection between sketching and top-$k$ retrieval, we use
the remainder of this chapter to delve into three algorithms.
To make things more interesting, we specifically review these algorithms in the
context of inner product for sparse vectors.

The first is the quintessential linear algorithm due to~\cite{JLLemma1984ExtensionsOL}.
It is linear in the sense that $\phi$ is simply a linear transformation,
so that $\phi(u) = \Phi u$ for some (random) matrix $\Phi \in \mathbb{R}^{d_\circ \times d}$.
We will learn how to construct the required matrix and discuss what guarantees it has to offer.

We then move to two sketching algorithms~\citep{bruch2023sinnamon} and~\citep{daliri2023sampling} whose
output space is \emph{not} Euclidean. Instead, the sketch of a vector is a data structure,
equipped with a distance function that approximates the inner product
between vectors in the original space.

\section{Linear Sketching with the JL Transform}

Let us begin by repeating the well-known result due to~\cite{JLLemma1984ExtensionsOL},
which we refer to as the JL Lemma:

\begin{lemma}
For $\epsilon \in (0, 1)$ and any set $\mathcal{X}$ of $m$ points
in $\mathbb{R}^d$, and an integer $d_\circ = \Omega(\epsilon^{-2} \ln m)$,
there exists a Lipschitz mapping $\phi: \mathbb{R}^d \rightarrow \mathbb{R}^{d_\circ}$ such that
\begin{equation*}
    (1 - \epsilon) \lVert u - v \rVert_2^2 \leq \lVert \phi(u) - \phi(v) \rVert_2^2
    \leq (1 + \epsilon) \lVert u - v \rVert_2^2,
\end{equation*}
for all $u, v \in \mathcal{X}$.
\end{lemma}

This result has been studied extensively and further developed since its introduction.
Using simple proofs, for example, it can be shown that the mapping $\phi$ may be
a linear transformation by a $d_\circ \times d$ random matrix $\Phi$ drawn
from a particular class of distributions. Such a matrix $\Phi$ is said to form a JL transform.

\begin{definition}
A random matrix $\Phi \in \mathbb{R}^{d_\circ \times d}$ forms a Johnson-Lindenstrauss
transform with parameters $(\epsilon, \delta, m)$, if with probability at least
$1 - \delta$, for any $m$-element subset $\mathcal{X} \subset \mathbb{R}^d$,
for all $u, v \in \mathcal{X}$ it holds that
$|\langle \Phi u, \Phi v \rangle - \langle u, v\rangle | \leq \epsilon \lVert u \rVert_2 \lVert v \rVert_2$.
\end{definition}

There are many constructions of $\Phi$ that form a JL transform.
It is trivial to show that when the entries of $\Phi$ are independently drawn from
$\mathcal{N}(0, \frac{1}{d_\circ})$, then $\Phi$ is a JL transform with parameters
$(\epsilon, \delta, m)$ if $d_\circ = \Omega(\epsilon^{-2} \ln (m / \delta))$.
In yet another construction, $\Phi=\frac{1}{\sqrt{d_\circ}} R$, where
$R \in \{ \pm 1 \}^{d_\circ \times d}$ is a matrix whose entries are independent Rademacher random variables.

We take the latter as an example due to its simplicity and analyze its properties.
As before, we refer the reader to~\citep{woodruff2014sketching} for a far more
detailed discussion of other (more efficient) constructions of the JL transform.

\subsection{Theoretical Analysis}
We are interested in analyzing the transformation above in the context of inner product.
Specifically, we wish to understand what we should expect if, instead of computing
the inner product between two vectors $u$ and $v$ in $\mathbb{R}^d$, we perform
the operation $\langle Ru, Rv \rangle$ in the transformed space in $\mathbb{R}^{d_\circ}$.
Is the outcome an unbiased estimate of the true inner product? How far off may this estimate be?
The following result is a first step to answering these questions for two \emph{fixed} vectors.

\begin{theorem}
    \label{theorem:sketching:jl-variance-fixed-vectors}
    Fix two vectors $u$ and $v \in \mathbb{R}^d$.
    Define $Z_\textsc{Sketch} = \langle \phi(u), \phi(v) \rangle$ as the random
    variable representing the inner product of sketches of size $d_\circ$,
    prepared using the projection $\phi(u) = R u$, with $R \in \{\pm 1/\sqrt{d_\circ}\}^{d_\circ \times d}$
    being a random Rademacher matrix.
    $Z_\textsc{Sketch}$ is an unbiased estimator of $\langle u, v \rangle$.
    Its distribution tends to a Gaussian with variance:
    \begin{equation*}
    \frac{1}{d_\circ} \Big( \lVert u \rVert_2^2 \lVert v \rVert_2^2 + \langle u, v \rangle^2 - 2 \sum_i u_i^2 v_i^2  \Big).
    \end{equation*}
\end{theorem}
\begin{proof}
    Consider the random variable $Z=\big( \sum_j R_j u_j \big)\big( \sum_k R_k v_k \big)$,
    where $R_i$'s are Rademacher random variables. It is clear that $d_\circ Z$ is the product
    of the sketch coordinate $i$ (for any $i$): $\phi(u)_i \phi(v)_i$.
    
    We can expand the expected value of $Z$ as follows:
    \begin{align*}
        \ev[Z] &= \ev\Big[ \big( \sum_j R_j u_j \big) \big( \sum_k R_k v_k \big) \Big] \\
        &= \ev[\sum_i R_i^2 u_i v_i] + \mathbb{E}[ \sum_{j \neq k} R_j R_k u_j v_k ] \\
        & = \sum_i u_i v_i \underbrace{\ev[R_i^2]}_1 + \sum_{j \neq k} u_j v_k \underbrace{\ev[R_j R_k ]}_0 \\
        &= \langle u, v \rangle.
    \end{align*}
    
    The variance of $Z$ can be expressed as follows:
    \begin{equation*}
        \var[Z] = \ev[Z^2] - \ev[Z]^2 =
        \ev\Big[\big( \sum_j R_j u_j \big)^2\big( \sum_k R_k v_k \big)^2\Big] - \langle u, v\rangle^2.
    \end{equation*}
    
    We have the following:
    \begin{align*}
        \ev&\Big[\big( \sum_j R_j u_j \big)^2\big( \sum_k R_k v_k \big)^2\Big] \\
        &= \ev\Big[
            \big( \sum_i u_i^2 + \sum_{i \neq j} R_i R_j u_i u_j \big)
            \big( \sum_k v_k^2 + \sum_{k \neq l} R_k R_l v_k v_l \big)
        \Big] \\
            &= \lVert u \rVert_2^2 \lVert v \rVert_2^2 +
                \underbrace{\ev\Big[\sum_i u_i^2 \sum_{k \neq l} R_k R_l v_k v_l\Big]}_0 \\
                &+ \underbrace{\ev\Big[\sum_k v_k^2 \sum_{i \neq j} R_i R_j u_i u_j\Big]}_0 +
                \ev\Big[\sum_{i \neq j} R_i R_j u_i u_j \sum_{k \neq l} R_k R_l v_k v_l\Big]. \numberthis
                \label{equation:sketching:jl-variance-fixed-vectors:variance}
    \end{align*}
    The last term can be decomposed as follows:
    \begin{align*}
        \ev&\Big[ \sum_{i \neq j \neq k \neq l} R_i R_j R_k R_l u_i u_j v_k v_l \Big] \\
        &+\ev\Big[ \sum_{i = k, j \neq l \lor i \neq k, j = l} R_i R_j R_k R_l u_i u_j v_k v_l \Big] \\
        &+ \ev\Big[ \sum_{i \neq j, i=k, j=l \lor i \neq j, i=l, j=k} R_i R_j R_k R_l u_i u_j v_k v_l \Big].
    \end{align*}
    The first two terms are $0$ and the last term can be rewritten as follows:
    \begin{equation}
    \label{equation:sketching:jl-variance-fixed-vectors:reduction}
        2 \ev\Big[ \sum_i u_i v_i \big( \sum_j u_j v_j - u_i v_i \big) \Big] = 2 \langle u, v \rangle^2 - 2 \sum_i u_i^2 v_i^2.
    \end{equation}
    
    We now substitute the last term in Equation~(\ref{equation:sketching:jl-variance-fixed-vectors:variance}) with
    Equation~(\ref{equation:sketching:jl-variance-fixed-vectors:reduction}) to obtain:
    \begin{equation*}
    \var[Z] = \lVert u \rVert_2^2 \lVert v \rVert_2^2 + \langle u, v \rangle^2 - 2 \sum_i u_i^2 v_i^2.
    \end{equation*}
    
    Observe that $Z_\textsc{Sketch} = 1/d_\circ \sum_i \phi(u)_i \phi(v)_i$ is
    the sum of independent, identically distributed random variables.
    Furthermore, for bounded vectors $u$ and $v$, the variance is finite.
    By the application of the Central Limit Theorem, we can deduce that
    the distribution of $Z_\textsc{Sketch}$ tends to a Gaussian distribution
    with the stated expected value. Noting that
    $\var[Z_\textsc{Sketch}] = 1/d_\circ^2 \sum_i \var[Z]$
    gives the desired result.
\end{proof}

\begin{svgraybox}
    Theorem~\ref{theorem:sketching:jl-variance-fixed-vectors} gives a clear model
    of the inner product error when two fixed vectors are transformed using our particular
    choice of the JL transform. We learnt that inner product of sketches is an ubiased
    estimator of the inner product between vectors, and have shown that the error
    follows a Gaussian distribution.
\end{svgraybox}

Let us now position this result in the context of top-$k$ retrieval
where the query point is fixed, but the data points are random.
To make the analysis more interesting, let us consider sparse vectors,
where each coordinate may be $0$ with a non-zero probability.

\begin{theorem}
    \label{theorem:sketching:jl-variance-fixed-query}
    Fix a query vector $q \in \mathbb{R}^d$ and let $X$ be
    a random vector drawn according to the following probabilistic model.
    Coordinate $i$, $X_i$, is non-zero with probability $p_i > 0$ and,
    if it is non-zero, draws its value from a distribution with mean $\mu$
    and variance $\sigma^2 < \infty$. Then, $Z_\textsc{Sketch} = \langle \phi(q), \phi(X) \rangle$,
    with $\phi(u) = R u$ and $R \in \{\pm 1/\sqrt{d_\circ} \}^{d_\circ \times d}$,
    has expected value $\mu \sum_i p_i q_i$ and variance:
    \begin{equation*}
        \frac{1}{d_\circ} \Bigg[
        (\mu^2 + \sigma^2)\Big( \lVert q \rVert_2^2 \sum_i p_i - \sum_i p_i q_i^2 \Big) +
        \mu^2 \Big( \big(\sum_i q_i p_i\big)^2 - \sum_i \big(q_i p_i\big)^2 \Big)
        \Bigg].
    \end{equation*}
\end{theorem}
\begin{proof}
    It is easy to see that:
    \begin{equation*}
        \ev[Z_\textsc{Sketch}] = \sum_i q_i \ev[X_i] = \mu \sum_i p_i q_i.
    \end{equation*}
    
    As for the variance, we start from Theorem~\ref{theorem:sketching:jl-variance-fixed-vectors}
    and arrive at the following expression:
    \begin{equation}
    \label{equation:sketching:jl-variance-fixed-query:variance}
    \frac{1}{d_\circ} \Big( \lVert q \rVert_2^2 \ev[\lVert X \rVert_2^2] + \ev[\langle q, X \rangle^2] - 2 \sum_i q_i^2 \ev[X_i^2]  \Big),
    \end{equation}
    where the expectation is with respect to $X$. Let us consider the terms inside the
    parentheses one by one. The first term becomes:
    \begin{align*}
        \lVert q \rVert_2^2 \ev[\lVert X \rVert_2^2] &= \lVert q \rVert_2^2 \sum_i\ev[X_i^2] \\
        &= \lVert q \rVert_2^2 (\mu^2 + \sigma^2) \sum_i p_i.
    \end{align*}
    The second term reduces to:
    \begin{align*}
    \ev\big[\langle q, X \rangle^2\big] &= \ev\big[ \langle q, X \rangle \big]^2 +
        \var\big[ \langle q, X \rangle \big] + \\
    &= \mu^2 (\sum_i q_i p_i)^2 + \sum q_i^2 \big[ (\mu^2 + \sigma^2) p_i - \mu^2 p_i^2 \big] \\
    &= \mu^2 \big( (\sum_i q_i p_i)^2 - \sum_i q_i^2 p_i^2 \big) + \sum_i q_i^2 p_i (\mu^2 + \sigma^2).
    \end{align*}
    Finally, the last term breaks down to:
    \begin{align*}
    - 2 \sum_i q_i^2 \ev[X_i^2] &= -2 \sum_i q_i^2 (\mu^2 + \sigma^2) p_i \\
    &= -2 (\mu^2 + \sigma^2) \sum_i q_i^2 p_i.
    \end{align*}
    
    Putting all these terms back into Equation~(\ref{equation:sketching:jl-variance-fixed-query:variance}) yields
    the desired expression for variance.
\end{proof}

Let us consider a special case to better grasp the implications of
Theorem~\ref{theorem:sketching:jl-variance-fixed-query}.
Suppose $p_i = \psi / d$ for some constant $\psi$ for all dimensions $i$.
Further assume, without loss of generality, that the (fixed) query vector has
unit norm: $\lVert q \rVert_2 = 1$. We can observe that the
variance of $Z_\textsc{Sketch}$ decomposes into a term that is
$(\mu^2 + \sigma^2) (1 - 1/d) \psi/d_\circ$,
and a second term that is a function of $1/d^2$.
The mean, on the other hand, is a linear function of the non-zero coordinates in the query:
$(\mu \sum_i q_i) \psi/d $.
As $d$ grows, the mean of $Z_\textsc{Sketch}$ tends to $0$ at
a rate proportional to the sparsity rate ($\psi/d$), while its variance tends to
$(\mu^2 + \sigma^2) \psi/d_\circ$.

The above suggests that the ability of $\phi(\cdot)$
to preserve the inner product of a query point with a randomly drawn
data point deteriorates as a function of the number of non-zero coordinates.
For example, when the number of non-zero coordinates becomes larger,
$\langle \phi(q), \phi(X) \rangle$ for a
fixed query $q$ and a random point $X$ becomes less reliable because
the variance of the approximation increases.

\section{Asymmetric Sketching}

Our second sketching algorithm is due to~\cite{bruch2023sinnamon}.
It is unusual in several ways. First, it is designed specifically for
retrieval. That is, the objective of the sketching technique is not
to preserve the inner product between points in a collection;
in fact, as we will learn shortly, the sketch is not even an unbiased estimator.
Instead, it is assumed that the setup is retrieval,
where we receive a query and wish to \emph{rank} data points in response.

That brings us to its second unusual property: asymmetry.
That means, only the data points are sketched while queries remain
in the original space. With the help of an asymmetric distance function,
however, we can easily compute an \emph{upper-bound} on the
query-data point inner product, using the raw query point and the sketch
of a data point.

Finally, in its original construction as presented in~\citep{bruch2023sinnamon},
the sketch was tailored specifically to sparse vectors.
As we will show, however, it is trivial to modify the algorithm and
adapt it to dense vectors.

In the rest of this section, we will first describe the sketching algorithm
for sparse vectors, as well as its extension to dense vectors.
We then describe how the distance between a query point in the original space
and the sketch of any data point can be computed asymmetrically.
Lastly, we review an analysis of the sketching algorithm.

\subsection{The Sketching Algorithm}

Algorithm~\ref{algorithm:sketching:sinnamon:sketching} shows the logic behind the sketching of sparse vectors.
It is assumed throughout that the sketch size, $d_\circ$, is even, so that $d_\circ/2$ is an integer.
The algorithm also makes use of $h$ independent random mappings
$\pi_o: [d] \rightarrow [d_\circ/2]$, where each $\pi_o(\cdot)$ projects
coordinates in the original space to an integer in the set $[d_\circ/2]$
uniformly randomly.

Intuitively, the sketch of $u \in \mathbb{R}^d$ is a data structure comprising of
the index of its set of non-zero coordinates (i.e., $\mathit{nz}(u)$), along with
an \emph{upper-bound sketch} ($\overline{u} \in \mathbb{R}^{d_\circ/2}$) and
a \emph{lower-bound sketch} ($\underline{u} \in \mathbb{R}^{d_\circ/2}$) on the non-zero values of $u$.
More precisely, the $k$-th coordinate of $\overline{u}$ ($\underline{u}$)
records the largest (smallest) value from the set of all non-zero coordinates in $u$
that map into $k$ according to at least one $\pi_o(\cdot)$.

\begin{algorithm}[!t]
\SetAlgoLined
{\bf Input: }{Sparse vector $u \in \mathbb{R}^d$.}\\
{\bf Requirements: }{$h$ independent random mappings $\pi_o: [d] \rightarrow [d_\circ/2]$.}\\
\KwResult{Sketch of $u$, $\{ \mathit{nz}(u);\; \underline{u};\; \overline{u} \}$ consisting of
the index of non-zero coordinates of $u$, the lower-bound sketch, and the upper-bound sketch.}

\begin{algorithmic}[1]
    \STATE Let $\overline{u}, \underline{u} \in \mathbb{R}^{d_\circ/2} $ be zero vectors
    \FORALL{$k \in [\frac{d_\circ}{2}]$}
        \STATE $\mathcal{I} \leftarrow \{ i \in \mathit{nz}(u) \;|\; \exists \;o\; \mathit{s.t.}\; \pi_o(i) = k \}$
        \STATE $\overline{u}_k \leftarrow \max_{i \in \mathcal{I}} u_i$ \label{algorithm:indexing:upper-bound-sketch}
        \STATE $\underline{u}_k \leftarrow \min_{i \in \mathcal{I}} u_i$
    \ENDFOR
    \RETURN $\{ \mathit{nz}(u),\; \underline{u},\; \overline{u} \}$
 \end{algorithmic}
 \caption{Sketching of sparse vectors}
\label{algorithm:sketching:sinnamon:sketching}
\end{algorithm}

\begin{svgraybox}
This sketching algorithm offers a great deal of flexibility.
When data vectors are non-negative, we may drop the lower-bounds
from the sketch, so that the sketch of $u$ consists only
of $\{ \mathit{nz}(u), \overline{u} \}$.
When vectors are dense, the sketch clearly does not need to store
the set of non-zero coordinates, so that the sketch of $u$ becomes
$\{ \overline{u}, \underline{u} \}$. Finally, when vectors are dense
and non-negative, the sketch of $u$ simplifies to $\overline{u}$.
\end{svgraybox}

\subsection{Inner Product Approximation}

Suppose that we are given a query point $q \in \mathbb{R}^d$ and wish to obtain an estimate
of the inner product $\langle q, u \rangle$ for some data vector $u$.
We must do so using only the sketch of $u$ as produced by
Algorithm~\ref{algorithm:sketching:sinnamon:sketching}.
Because the query point is not sketched and, instead, remains in the original
$d$-dimensional space, while $u$ is only known in its sketched form,
we say this computation is \emph{asymmetric}. This is not unlike the distance
computation between a query point and a quantized data point, as seen in
Chapter~\ref{chapter:quantization}.

\begin{algorithm}[t]
    \SetAlgoLined
    {\bf Input: }{Sparse query vector $q \in \mathbb{R}^d$;
    sketch of data point $u$: $\{ \mathit{nz}(u), \overline{u}, \underline{u} \}$}\\
    {\bf Requirements: }{$h$ independent random mappings $\pi_o: [d] \rightarrow [d_\circ/2]$.}\\
    \KwResult{Upper-bound on $\langle q, u \rangle$.}
    
    \begin{algorithmic}[1]
        \STATE $s \leftarrow 0$
        \FOR{$i \in \mathit{nz}(q) \cap \mathit{nz}(u)$}
            \STATE $\mathcal{J} \leftarrow \{ \pi_o(i) \;|\; o \in [h] \}$
            \IF{$q_i > 0$}
                \STATE $s \leftarrow s + \min_{j \in \mathcal{J}} \overline{u}_j$
            \ELSE
                \STATE $s \leftarrow s + \max_{j \in \mathcal{J}} \underline{u}_j$
            \ENDIF
        \ENDFOR
        \RETURN $s$
    \end{algorithmic}
    \caption{Asymmetric distance computation for sparse vectors}
    \label{algorithm:sketching:sinnamon:distance}
\end{algorithm}

This asymmetric procedure is described in Algorithm~\ref{algorithm:sketching:sinnamon:distance}.
The algorithm iterates over the intersection of the non-zero coordinates of the query vector
and the non-zero coordinates of the data point (which is included in the sketch).
It goes without saying that, if the vectors are dense, we may simply iterate over all coordinates.
When visiting the $i$-th coordinate, we first form the set of coordinates that $i$ maps to
according to the hash functions $\pi_o$'s; that is the set $\mathcal{J}$ in the algorithm.

\begin{svgraybox}
The next step then depends on the sign of the query at that coordinate.
When $q_i$ is positive, we find the \emph{least upper-bound} on the value of
$u_i$ from its upper-bound sketch. That can be determined by looking at $\overline{u}_j$
for all $j \in \mathcal{J}$, and taking the minimum value among those sketch coordinates.
When $q_i < 0$, on the other hand, we find the \emph{greatest lower-bound} instead.
In this way, it is always guaranteed that the partial inner product is an upper-bound on the actual
partial inner product, $q_i u_i$, as stated in the next theorem.
\end{svgraybox}

\begin{theorem}
    \label{theorem:sketching:sinnamon:upper-bound}
    The quantity returned by Algorithm~\ref{algorithm:sketching:sinnamon:distance} 
    is an upper-bound on the inner product of query and data vectors.
\end{theorem}

\subsection{Theoretical Analysis}

Theorem~\ref{theorem:sketching:sinnamon:upper-bound} implies that
Algorithm~\ref{algorithm:sketching:sinnamon:distance} always overestimates the
inner product between query and data points. In other words,
the inner product approximation error is non-negative.
But what can be said about the probability that such an error occurs?
How large is the overestimation error? We turn to these questions next.

Before we do so, however, we must agree on a probabilistic model of the data.
We follow~\citep{bruch2023sinnamon} and assume that a random sparse vector $X$
is drawn from the following distribution.
All coordinates of $X$ are mutually independent.
Its $i$-th coordinate is \emph{inactive} (i.e., zero) with probability $1 - p_i$.
Otherwise, it is \emph{active} and its value is a random variable, $X_i$,
drawn \emph{iid} from some distribution with probability density function (PDF)
$\phi$ and cumulative distribution function (CDF) $\Phi$.

\subsubsection{Probability of Error}
Let us focus on the approximation error of a single active coordinate.
Concretely, suppose we have a random vector $X$ whose $i$-th coordinate is active: $i \in \mathit{nz}(X)$.
We are interested in quantifying the likelihood that, if we estimated the value
of $X_i$ from the sketch, the estimated value, $\tilde{X}_i$, overshoots or undershoots the actual value.

Formally, we wish to model $\probability[\tilde{X}_i \neq X_i]$,
Note that, depending on the sign of the query's $i$-th coordinate
$\tilde{X}_i$ may be estimated from the upper-bound sketch ($\overline{X}$),
resulting in \emph{overestimation},
or the lower-bound sketch ($\underline{X}$), resulting in \emph{underestimation}.
Because the two cases are symmetric, we state the main result for the former case:
When $\tilde{X}_i$ is the least upper-bound on $X_i$, estimated from $\overline{X}$:
\begin{equation}
    \label{equation:sketching:sinnamon:upper-bound-estimator}
    \tilde{X}_i = \min_{j \in \{ \pi_o(i) \;|\; o \in [h] \}} \overline{X}_j.
\end{equation}

\begin{theorem}
    \label{theorem:sketching:sinnamon:sketching-error}
    For large values of $d_\circ$, an active $X_i$, and $\tilde{X}_i$ estimated 
    using Equation~(\ref{equation:sketching:sinnamon:upper-bound-estimator}),
    \begin{equation*}
        \probability\big[ \tilde{X}_i > X_i \big] \approx \int {\Bigg[ 1 - \exp\Big(-\frac{2h}{d_\circ} \big(1 - \Phi(\alpha)\big) \sum_{j \neq i} p_j \Big) \Bigg]^h \phi(\alpha) d\alpha },
    \end{equation*}
    where $\phi(\cdot)$ and $\Phi(\cdot)$ are the PDF and CDF of $X_i$.
\end{theorem}

Extending this result to the lower-bound sketch involves replacing $1 - \Phi(\alpha)$ with $\Phi(\alpha)$.
When the distribution defined by $\phi$ is symmetric, the probabilities of error
too are symmetric for the upper-bound and lower-bound sketches.

\begin{proof}[Proof of Theorem~\ref{theorem:sketching:sinnamon:sketching-error}]
    Recall that $\tilde{X}_i$ is estimated as follows:
    \begin{equation*}
        \tilde{X}_i = \min_{j \in \{ \pi_o(i) \;|\; o \in [h] \}} \overline{X}_j.
    \end{equation*}
    So we must look up $\overline{X}_j$ for values of $j$ produced by $\pi_o$'s.

    Suppose one such value is $k$ (i.e., $k = \pi_o(i)$ for some $o \in [h]$).
    The event that $\overline{X}_k > X_i$ happens only when there exists another
    active coordinate $X_j$ such that $X_j > X_i$ \emph{and} $\pi_o(j)=k$ for some $\pi_o$.
    
    To derive $\probability[\overline{X}_k > X_i]$, it is easier to think in terms
    of complementary events: $\overline{X}_k = X_i$ if every other active coordinate
    whose value is larger than $X_i$ maps to a sketch coordinate \emph{except} $k$.
    Clearly the probability that any arbitrary $X_j$ maps to a sketch coordinate other
    than $k$ is simply $1 - 2/d_\circ$. Therefore, given a vector $X$, the probability
    that no active coordinate $X_j$ larger than $X_i$ maps to the $k$-th coordinate of the sketch,
    which we denote by ``Event A,'' is:
    \begin{align*}
        \probability\Big[\text{Event A} \;|\; X \Big]
        = 1 - (1 - \frac{2}{d_\circ})^{h\sum_{j \neq i} \mathbbm{1}_{X_j \text{ is active}} \mathbbm{1}_{X_j > X_i}}.
    \end{align*}
    Because $d_\circ$ is large by assumption, we can approximate $e^{-1} \approx (1 - 2/d_\circ)^{d_\circ/2}$ and rewrite the expression above as follows:
    \begin{equation*}
        \probability\big[\text{Event A}  \;|\; X \big] \approx 1 - \exp\Big(-\frac{2h}{d_\circ}\sum_{j \neq i} \mathbbm{1}_{X_j \text{ is active}} \mathbbm{1}_{X_j > X_i} \Big).
    \end{equation*}
    Finally, we marginalize the expression above over $X_j$'s for $j \neq i$ to remove
    the dependence on all but the $i$-th coordinate of $X$.
    To simplify the expression, however, we take the expectation over the first-order
    Taylor expansion of the right hand side around $0$. This results in the following approximation:
    \begin{equation*}
        \probability\big[\text{Event A}  \;|\; X_i = \alpha \big] \approx 1 - \exp\Big(-\frac{2h}{d_\circ} (1 - \Phi(\alpha)) \sum_{j \neq i} p_j\Big).
    \end{equation*}
    
    For $\tilde{X}_i$ to be larger than $X_i$, event $A$ must take place for all $h$ sketch coordinates.
    That probability, by the independence of random mappings, is:
    \begin{equation*}
        \probability\big[ \tilde{X}_i > X_i \;|\; X_i = \alpha \big] \approx \big[ 1 - \exp\Big(-\frac{2h}{d_\circ} (1 - \Phi(\alpha)) \sum_{j \neq i} p_j \Big) \big]^h.
    \end{equation*}
    In deriving the expression above, we conditioned the event on the value of $X_i$.
    Taking the marginal probability leads us to the following expression for the event
    that $\tilde{X}_i > X_i$ for any $i$, concluding the proof:
    \begin{align*}
        \probability\big[ \tilde{X}_i > X_i \big] &\approx
            \int {\Big[ 1 - \exp\Big(-\frac{2h}{d_\circ} (1 - \Phi(\alpha)) \sum_{j \neq i} p_j \Big) \Big]^h d\probability(\alpha)} \\
            &\approx \int {\Big[ 1 - \exp\Big(-\frac{2h}{d_\circ} (1 - \Phi(\alpha)) \sum_{j \neq i} p_j\Big) \Big]^h \phi(\alpha) d\alpha}.
    \end{align*}
\end{proof}

\begin{svgraybox}
Theorem~\ref{theorem:sketching:sinnamon:sketching-error} offers insights into the
behavior of the upper-bound sketch. The first observation is that the sketching mechanism
presented here is more suitable for distributions where larger values occur with a smaller
probability such as sub-Gaussian variables. In such cases, the larger the value is,
the smaller its chance of being overestimated by the upper-bound sketch.
Regardless of the underlying distribution, empirically, the largest value in a
vector is always estimated \emph{exactly}.
\end{svgraybox}

The second insight is that there is a sweet spot for $h$ given a particular value
of $d_\circ$: using more random mappings helps lower the probability of error until
the sketch starts to saturate, at which point the error rate increases.
This particular property is similar to the behavior of a Bloom filter~\citep{bloom-filter}.

\subsubsection{Distribution of Error}
We have modeled the probability that the sketch of a vector overestimates a value.
In this section, we examine the shape of the distribution of error in the form of its
CDF. Formally, assuming $X_i$ is active and $\tilde{X}_i$ is estimated using
Equation~(\ref{equation:sketching:sinnamon:upper-bound-estimator}),
we wish to find an expression for
$\probability[\lvert \tilde{X}_i - X_i \rvert < \epsilon]$ for any $\epsilon > 0$.

\begin{theorem}
    \label{theorem:sketching:sinnamon:sketching-error-cdf}
    Suppose $X_i$ is active and draws its value from a distribution with PDF and CDF $\phi$ and $\Phi$.
    Suppose further that $\tilde{X}_i$ is the least upper-bound on $X_i$, obtained
    using Equation~(\ref{equation:sketching:sinnamon:upper-bound-estimator}).
    Then:
    \begin{equation*}
        \probability[\tilde{X}_i - X_i \leq \epsilon] \approx
        1 - \int \Big[ 1 - \exp\Big(-\frac{2h}{d_\circ} (1 - \Phi(\alpha + \epsilon)) \sum_{j \neq i} p_j \Big) \Big]^h \phi(\alpha) d \alpha.
    \end{equation*}
\end{theorem}
\begin{proof}
We begin by quantifying the conditional probability
$\probability[\tilde{X}_i - X_i \leq \epsilon \;|\; X_i = \alpha]$.
Conceptually, the event in question happens when all values that collide
with $X_i$ are less than or equal to $X_i + \epsilon$.
This event can be characterized as the complement of the event that
all $h$ sketch coordinates that contain $X_i$ collide with values greater
than $X_i + \epsilon$. Using this complementary event, we can write the
conditional probability as follows:
\begin{align*}
    \probability[\tilde{X}_i - X_i \leq \epsilon \;|\; X_i = \alpha] &=
    1 - \big[ 1 - (1 - \frac{2}{d_\circ})^{h (1 - \Phi(\alpha + \epsilon)) \sum_{j \neq i} p_j} \big]^h \\
    &\approx 1 - \Bigg[ 1 - \exp\Big(-\frac{2h}{d_\circ} (1 - \Phi(\alpha + \epsilon)) \sum_{j \neq i} p_j \Big) \Bigg]^h.
\end{align*}
We complete the proof by computing the marginal distribution over the support.
\end{proof}

Given the CDF of $\tilde{X}_i - X_i$ and the fact that $\tilde{X}_i - X_i \geq 0$,
it follows that its expected value conditioned on $X_i$ being active is:

\begin{lemma}
Under the conditions of Theorem~\ref{theorem:sketching:sinnamon:sketching-error-cdf}:
\begin{align*}
    \ev[\tilde{X}_i - X_i]
    \approx \int_{0}^{\infty} \int \Bigg[ 1 - \exp\Big(-\frac{2h}{d_\circ} (1 - \Phi(\alpha + \epsilon)) \sum_{j \neq i} p_j \Big) \Bigg]^h \phi(\alpha)\; d\alpha \; d\epsilon.
\end{align*}
\end{lemma}

\subsubsection{Case Study: Gaussian Vectors}
Let us make the analysis more concrete by applying the results to random Gaussian vectors.
In other words, suppose all active $X_i$'s are drawn from a zero-mean, unit-variance Gaussian
distribution. We can derive a closed-form expression for the overestimation probability
as the following corollary shows.

\begin{corollary}
    \label{corollary:sketching:sinnamon:sketching-error:gaussian}
    Suppose the probability that a coordinate is active, $p_i$, is equal to $p$
    for all coordinates of the random vector $X \in \mathbb{R}^d$.
    When an active $X_i$, drawn from $\mathcal{N}(0, 1)$, is estimated using the
    upper-bound sketch with Equation~(\ref{equation:sketching:sinnamon:upper-bound-estimator}),
    the overestimation probability is:
    \begin{equation*}
    \probability\big[ \tilde{X}_i > X_i \big] \approx 1 + \sum_{k=1}^{h} \binom{h}{k} (-1)^k \frac{d_\circ}{2kh(d-1)p} \big(1 - e^{-\frac{2kh(d-1)p}{d_\circ}} \big).
    \end{equation*}
\end{corollary}

We begin by proving the special case where $h=1$.

\begin{lemma}
    Under the conditions of Corollary~\ref{corollary:sketching:sinnamon:sketching-error:gaussian}
    with $h=1$, the probability that the upper-bound sketch overestimates the
    value of $X_i$ is:
    \begin{equation*}
        \probability\big[ \tilde{X}_i > X_i \big] \approx 1 - \frac{d_\circ}{2(d - 1)p} (1 - e^{-\frac{2(d - 1)p}{d_\circ}} ).
    \end{equation*}
\end{lemma}
\begin{proof}
    From Theorem~\ref{theorem:sketching:sinnamon:sketching-error} we have that:
    \begin{align*}
        \probability\big[ \tilde{X}_i > X_i \big] &\approx 
        \int {\big[ 1 - e^{-\frac{2h}{d_\circ} (1 - \Phi(\alpha)) (d - 1)p} \big]^h d\probability(\alpha)} \\
        &\overset{h=1}{=} \int {\big[ 1 - e^{-\frac{2 (1 - \Phi(\alpha)) (d - 1) p}{d_\circ}} \big] d\probability(\alpha)}.
    \end{align*}

    Given that $X_i$'s are drawn from a Gaussian distribution,
    and using the approximation above, we can rewrite the probability of error as:
    \begin{equation*}
        \probability\big[ \tilde{X}_i > X_i \big] \approx
        \frac{1}{\sqrt{2\pi}} \int_{-\infty}^{\infty} {\big[ 1 - e^{-\frac{2(d - 1)p}{d_\circ} (1 - \Phi(\alpha))} \big] e^{-\frac{\alpha^2}{2}} d\alpha}.
    \end{equation*}
    We now break up the right hand side into the following three sums,
    replacing $2(d - 1)p/d_\circ$ with $\beta$ for brevity:
    \begin{align}
        \probability\big[ \tilde{X}_i > X_i \big] &\approx
        \int_{-\infty}^{\infty} \frac{1}{\sqrt{2\pi}} e^{-\frac{\alpha^2}{2}} d\alpha \label{equation:sketching:sinnamon:proof:prob-error-gaussian:one} \\
        & - \int_{-\infty}^{0} \frac{1}{\sqrt{2\pi}} e^{-\beta (1 - \Phi(\alpha))} e^{-\frac{\alpha^2}{2}} d\alpha \label{equation:sketching:sinnamon:proof:prob-error-gaussian:negative} \\
        & - \int_{0}^{\infty} \frac{1}{\sqrt{2\pi}} e^{-\beta (1 - \Phi(\alpha))} e^{-\frac{\alpha^2}{2}} d\alpha \label{equation:sketching:sinnamon:proof:prob-error-gaussian:positive}.
    \end{align}
    The sum in~(\ref{equation:sketching:sinnamon:proof:prob-error-gaussian:one}) is equal to the quantity $1$. Let us turn to~(\ref{equation:sketching:sinnamon:proof:prob-error-gaussian:positive}) first. We have that:
    \begin{equation*}
        1 - \Phi(\alpha) \overset{\alpha > 0}{=} \frac{1}{2} - \underbrace{\int_{0}^{\alpha} \frac{1}{\sqrt{2\pi}} e^{-\frac{t^2}{2}}}_{\lambda(\alpha)} dt.
    \end{equation*}
    As a result, we can write:
    \begin{align*}
        \int_{0}^{\infty} \frac{1}{\sqrt{2\pi}} e^{-\beta (1 - \Phi(\alpha))} e^{-\frac{\alpha^2}{2}} d\alpha &= \int_{0}^{\infty} \frac{1}{\sqrt{2\pi}} e^{-\beta (\frac{1}{2} - \lambda(\alpha))} e^{-\frac{\alpha^2}{2}} d\alpha \nonumber \\
        &= e^{-\frac{\beta}{2}} \int_{0}^{\infty} \frac{1}{\sqrt{2\pi}} e^{\beta \lambda(\alpha)} e^{-\frac{\alpha^2}{2}} d\alpha \nonumber \\
        &= \frac{1}{\beta} e^{-\frac{\beta}{2}} e^{\beta\lambda(\alpha)} \big|_{0}^{\infty} = \frac{1}{\beta} e^{-\frac{\beta}{2}} \big( e^\frac{\beta}{2} - 1 \big) \nonumber \\
        &= \frac{1}{\beta} \big( 1 - e^{-\frac{\beta}{2}} \big)
    \end{align*}
    By similar reasoning, and noting that:
    \begin{equation*}
        1 - \Phi(\alpha) \overset{\alpha < 0}{=} \frac{1}{2} + \underbrace{\int_{\alpha}^{0} \frac{1}{\sqrt{2\pi}} e^{-\frac{t^2}{2}}}_{-\lambda(\alpha)} dt,
    \end{equation*}
    we arrive at:
    \begin{equation*}
        \int_{-\infty}^{0} \frac{1}{\sqrt{2\pi}} e^{-\beta (1 - \Phi(\alpha))} e^{-\frac{\alpha^2}{2}} d\alpha = \frac{1}{\beta} e^{-\frac{\beta}{2}} \big( 1 - e^{-\frac{\beta}{2}} \big)
    \end{equation*}

    Plugging the results above into Equations~(\ref{equation:sketching:sinnamon:proof:prob-error-gaussian:one}),~(\ref{equation:sketching:sinnamon:proof:prob-error-gaussian:negative}), and~(\ref{equation:sketching:sinnamon:proof:prob-error-gaussian:positive}) results in:
    \begin{align*}
        \probability\big[ \tilde{X}_i > X_i \big] &\approx 1 - \frac{1}{\beta} \big( 1 - e^{-\frac{\beta}{2}} \big) - \frac{1}{\beta} e^{-\frac{\beta}{2}} \big( 1 - e^{-\frac{\beta}{2}} \big) \\
        &= 1 - \frac{1}{\beta} \big( 1 - e^{-\frac{\beta}{2}} \big) \big( 1 + e^{-\frac{\beta}{2}} \big) \\
        &= 1 - \frac{d_\circ}{2(d-1)p} \big( 1 - e^{-\frac{2(d - 1)p}{d_\circ}} \big),
    \end{align*}
    which completes the proof.
\end{proof}

Given the result above, the solution for the general case of $h > 0$ is straightforward to obtain.

\begin{proof}[Proof of Corollary~\ref{corollary:sketching:sinnamon:sketching-error:gaussian}]
    Using the binomial theorem, we have that:
    \begin{align*}
        \probability\big[ \tilde{X}_i > X_i \big] &\approx
        \int {\big[ 1 - e^{-\frac{2h}{d_\circ} (1 - \Phi(\alpha)) (d - 1) p } \big]^h d\probability(\alpha)} \\
        &= \sum_{k=0}^{h} { \binom{h}{k} \int \big( -e^{-\frac{2h}{d_\circ} (1 - \Phi(\alpha)) (d - 1) p } \big)^{k} d\probability(\alpha)}.
    \end{align*}
    We rewrite the expression above for Gaussian variables to arrive at:
    \begin{equation*}
        \probability\big[ \tilde{X}_i > X_i \big] \approx \frac{1}{\sqrt{2\pi}} \sum_{k=0}^{h} \binom{h}{k} \int_{-\infty}^{\infty} {\big(- e^{-\frac{2h(d - 1)p}{d_\circ} (1 - \Phi(\alpha))} \big)^{k} e^{-\frac{\alpha^2}{2}} d\alpha}.
    \end{equation*}
    Following the proof of the previous lemma, we can expand the right hand side as follows:
    \begin{align*}
        \probability\big[ \tilde{X}_i > X_i \big] &\approx 
        1 + \frac{1}{\sqrt{2\pi}} \sum_{k=1}^{h} \binom{h}{k} (-1)^k \int_{-\infty}^{\infty} {e^{-\frac{2kh(d - 1)p}{d_\circ} (1 - \Phi(\alpha))} e^{-\frac{\alpha^2}{2}} d\alpha} \\
        &= 1 + \sum_{k=1}^{h} \binom{h}{k} (-1)^k \frac{d_\circ}{2kh(d - 1)p} (1 - e^{-\frac{2kh(d - 1)p}{d_\circ}} ),
    \end{align*}
    which completes the proof.
\end{proof}

Let us now consider the CDF of the overestimation error.

\begin{corollary}
    \label{corollary:sketching:sinnamon:sketching-error-cdf:gaussian}
    Under the conditions of Corollary~\ref{corollary:sketching:sinnamon:sketching-error:gaussian}
    the CDF of overestimation error for an active coordinate $X_i \sim \mathcal{N}(0, \sigma)$ is:
    \begin{equation*}
        \probability[\tilde{X}_i - X_i \leq \epsilon]
        \approx 1 - \Bigg[ 1 - \exp\Big(-\frac{2 h (d - 1) p}{d_\circ} (1 - \Phi^\prime(\epsilon)) \Big) \Bigg]^h,
    \end{equation*}
    where $\Phi^\prime(\cdot)$ is the CDF of a zero-mean Gaussian with standard deviation $\sigma\sqrt{2}$.
\end{corollary}
\begin{proof}
    When the active values of a vector are drawn from a Gaussian distribution,
    then the pairwise difference between any two coordinates has a Gaussian distribution with
    standard deviation $\sqrt{\sigma^2 + \sigma^2}=\sigma \sqrt{2}$.
    As such, we may estimate $1 - \Phi(\alpha + \epsilon)$ by considering the
    probability that a pair of coordinates (one of which having value $\alpha$) has 
    a difference greater than $\epsilon$: $\probability[X_i - X_j > \epsilon]$.
    With that idea, we may thus write:
    \begin{equation*}
        1 - \Phi(\alpha + \epsilon) = 1 - \Phi^\prime(\epsilon).
    \end{equation*}
    The claim follows by using the above identity in Theorem~\ref{theorem:sketching:sinnamon:sketching-error-cdf}.
\end{proof}

Corollary~\ref{corollary:sketching:sinnamon:sketching-error-cdf:gaussian}
enables us to find a particular sketch configuration given a desired bound on the probability of error,
as the following lemma shows.

\begin{lemma}
    Under the conditions of Corollary~\ref{corollary:sketching:sinnamon:sketching-error-cdf:gaussian},
    and given a choice of $\epsilon, \delta \in (0, 1)$ and the number of random mappings $h$,
    $\probability[\tilde{X}_i - X_i \leq \epsilon]$ with probability at least $1 - \delta$ if:
    \begin{equation*}
        d_\circ > - \frac{2h (d - 1)p (1 - \Phi^\prime(\epsilon))}{\log (1 - \delta^{1/h})}.
    \end{equation*}
\end{lemma}

\subsubsection{Error of Inner Product}
We have thus far quantified the probability that a value estimated
from the upper-bound sketch overestimates the original value of a randomly
chosen coordinate. We also characterized the distribution of the overestimation error
for a single coordinate and derived expressions for special distributions.
In this section, we quantify the overestimation error when approximating
the inner product between a fixed query point and a random data point using
Algorithm~\ref{algorithm:sketching:sinnamon:distance}.

To make the notation less cluttered, however, let us denote by $\tilde{X}_i$
our estimate of $X_i$. The estimated quantity is $0$ if $i \notin \mathit{nz}(X)$.
Otherwise, it is estimated either from the upper-bound sketch or the lower-bound sketch,
depending on the sign of $q_i$. Finally denote by $\tilde{X}$ a reconstruction of $X$
where each $\tilde{X}_i$ is estimated as described above.

Consider the expected value of $\tilde{X}_i - X_i$ conditioned on
$X_i$ being active---that is a quantity we analyzed previously.
Let $\mu_i = \ev[\tilde{X}_i - X_i; \; X_i \text{ is active}]$.
Similarly denote by $\sigma_i^2$ its variance when $X_i$ is active.
Given that $X_i$ is active with probability $p_i$ and inactive with probability $1 - p_i$,
it is easy to show that $\ev[\tilde{X}_i - X_i] = p_i \mu_i$ (note we have removed
the condition on $X_i$ being active)
and that its variance $\var[\tilde{X}_i - X_i] = p_i \sigma_i^2 + p_i (1 - p_i) \mu_i^2$.

With the above in mind, we state the following result.

\begin{theorem}
    Suppose that $q \in \mathbb{R}^d$ is a sparse vector.
    Suppose in a random sparse vector $X \in \mathbb{R}^d$,
    a coordinate $X_i$ is active with probability $p_i$ and, when active,
    draws its value from some well-behaved distribution (i.e., with finite expectation,
    variance, and third moment).
    If $\mu_i = \ev[\tilde{X}_i - X_i;\; X_i \text{ is active}]$ and
    $\sigma_i^2 = \var[\tilde{X}_i - X_i;\; X_i \text{ is active}]$,
    then the random variable $Z$ defined as follows:
    \begin{equation}
        Z \triangleq \frac{\langle q,\, \tilde{X} - X \rangle - \sum_{i \in \mathit{nz}(q)} q_i p_i \mu_i}{\sqrt{\sum_{i \in \mathit{nz}(q)} q_i^2 \big( p_i \sigma_i^2 + p_i (1 - p_i) \mu_i^2 \big)}},
    \end{equation}
    approximately tends to a standard Gaussian distribution as $\lvert \mathit{nz}(q) \rvert$ grows.
\end{theorem}
\begin{proof}
    Let us expand the inner product between $q$ and $\tilde{X} - X$ as follows:
    \begin{equation}
        \langle q, \tilde{X} - X \rangle = \sum_{i \in \mathit{nz}(q)} q_i \underbrace{(\tilde{X}_i - X_i)}_{Z_i}.
    \end{equation}

    The expected value of $\langle q, \tilde{X} - X \rangle$ is:
    \begin{equation*}
        \ev[\langle q, \tilde{X} - X \rangle] = \sum_{i \in \mathit{nz}(q)} q_i \ev[Z_i] =
        \sum_{i \in \mathit{nz}(q)} q_i p_i \mu_i.
    \end{equation*}
    Its variance is:
    \begin{equation*}
        \var[\langle q, \tilde{X} - X \rangle] = \sum_{i \in \mathit{nz}(q)} q_i^2 \var[Z_i] =
        \sum_{i \in \mathit{nz}(q)} q_i^2 \big( p_i \sigma_i^2 + p_i (1 - p_i) \mu_i^2 \big)
    \end{equation*}

    Because we assumed that the distribution of $X_i$ is well-behaved,
    we can conclude that $\var[Z_i] > 0$ and that $\ev[|Z_i|^3] < \infty$.
    If we operated on the assumption that $q_i Z_i$'s are independent---in reality,
    they are weakly dependent---albeit not identically distributed,
    we can appeal to the Berry-Esseen theorem to complete the proof.
\end{proof}

\subsection{Fixing the Sketch Size}

It is often desirable for a sketching algorithm to produce a sketch with a constant size.
That makes the size of a collection of sketches predictable, which is often required
for resource allocation. Algorithm~\ref{algorithm:sketching:sinnamon:sketching},
however, produces a sketch whose size is variable. That is because the sketch contains
the set of non-zero coordinates of the vector.

It is, however, straightforward to fix the sketch size. The key to that is
the fact that Algorithm~\ref{algorithm:sketching:sinnamon:distance} uses
$\mathit{nz}(u)$ of a vector $u$ only to ascertain if a query's non-zero coordinates
are present in the vector $u$. In effect, all the sketch must provide is a mechanism
to perform set membership tests. That is precisely what fixed-size signatures such
as Bloom filters~\citep{bloom-filter} do, albeit probabilistically.

\section{Sketching by Sampling}

Our final sketching algorithm is designed specifically for inner product
and is due to~\cite{daliri2023sampling}. The guiding principle is simple:
coordinates with larger values contribute more heavily to inner product than
coordinates with smaller values. That is an obvious fact that is a direct result
of the linearity of inner product: $\langle u, v \rangle = \sum_i u_i v_i$.

\cite{daliri2023sampling} use that insight as follows.
When forming the sketch of vector $u$, they sample coordinates (without replacement)
from $u$ according to a distribution defined by the magnitude of each coordinate.
Larger values are given a higher chance of being sampled, while smaller values are
less likely to be selected. The sketch, in the end, is a data structure that
is made up of the index of sampled coordinates, their values, and additional statistics.

The research question here concerns the sampling process: How must we sample
coordinates such that any distance computed from the sketch is an unbiased estimate
of the inner product itself? The answer to that question also depends, of course, on how
we compute the distance from a pair of sketches. Considering the non-linearity of the sketch,
distance computation can no longer be the inner product of sketches.

In the remainder of this section, we review the sketching algorithm,
describe distance computation given sketches, and analyze the expected error.
In our presentation, we focus on the simpler variant of the algorithm proposed
by~\cite{daliri2023sampling}, dubbed ``threshold sampling.''

\subsection{The Sketching Algorithm}

\begin{algorithm}[!t]
\SetAlgoLined
{\bf Input: }{Vector $u \in \mathbb{R}^d$.}\\
{\bf Requirements: }{a random mapping $\pi: [d] \rightarrow [0, 1]$.}\\
\KwResult{Sketch of $u$, $\{ \mathcal{I}, \mathcal{V}, \lVert u \rVert_2^2 \}$ consisting of
the index and value of sampled coordinates in $\mathcal{I}$ and $\mathcal{V}$, and the squared norm of the vector.}

\begin{algorithmic}[1]
    \STATE $\mathcal{I}, \mathcal{V} \leftarrow \emptyset$
    \FOR{$i \in \mathit{nz}(u)$}
        \STATE $\theta \leftarrow d_\circ \frac{u_i^2}{\lVert u \rVert_2^2}$
        \IF{$\pi(i) \leq \theta$}
            \STATE Append $i$ to $\mathcal{I}$, $u_i$ to $\mathcal{V}$
        \ENDIF
    \ENDFOR
    
    \RETURN $\{ \mathcal{I}, \mathcal{V}, \lVert u \rVert_2^2 \}$
 \end{algorithmic}
 \caption{Sketching with threshold sampling}
\label{algorithm:sketching:sampling:sketching}
\end{algorithm}

Algorithm~\ref{algorithm:sketching:sampling:sketching} presents the ``threshold sampling''
sketching technique by~\cite{daliri2023sampling}. It is assumed throughout that the desired
sketch size is $d_\circ$, and that the algorithm has access to a random hash function
$\pi$ that maps integers in $[d]$ to the unit interval.

The algorithm iterates over all non-zero coordinates of the input vector and makes a decision
as to whether that coordinate should be added to the sketch. The decision is made based
on the relative magnitude of the coordinate, as weighted by $u_i^2/\lVert u \rVert_2^2$.
If $u_i^2$ is large, coordinate $i$ has a higher chance of being sampled, as desired.

Notice, however, that the target sketch size $d_\circ$ is realized \emph{in expectation} only.
In other words, we may end up with more than $d_\circ$ coordinates in the sketch, or
we may have fewer entries. \cite{daliri2023sampling} propose a different variant of the algorithm
that is guaranteed to give a fixed sketch size; we refer the reader to their work for details.

\subsection{Inner Product Approximation}

When sketching a vector using a JL transform, we simply get a vector in the
$d_\circ$-dimensional Euclidean space, where inner product is well-defined.
So if $\phi(u)$ and $\phi(v)$ are sketches of two $d$-dimensional vectors $u$
and $v$, we approximate $\langle u, v \rangle$ with $\langle \phi(u), \phi(v) \rangle$.
It could not be more straightforward.

A sketch produced by Algorithm~\ref{algorithm:sketching:sampling:sketching}, however,
is not as nice. Approximating $\langle u, v \rangle$ from their sketches requires
a custom distance function defined for the sketch. That is precisely what
Algorithm~\ref{algorithm:sketching:sampling:distance} outlines.

\begin{algorithm}[!t]
\SetAlgoLined
{\bf Input: }{Sketches of vectors $u$ and $v$: $\{ \mathcal{I}_u, \mathcal{V}_u, \lVert u \rVert_2^2 \}$
and $\{ \mathcal{I}_v, \mathcal{V}_v, \lVert v \rVert_2^2 \}$.}\\
\KwResult{An unbiased estimate of $\langle u, v \rangle$.}

\begin{algorithmic}[1]
    \STATE $s \leftarrow 0$
    \FOR{$i \in \mathcal{I}_u \cap \mathcal{I}_v$}
        \STATE $s \leftarrow s + u_i v_i/\min ( 1, d_\circ u_i^2/\lVert u \rVert_2^2, d_\circ v_i^2/\lVert v \rVert_2^2 )$
    \ENDFOR
    
    \RETURN $s$
 \end{algorithmic}
 \caption{Distance computation for threshold sampling}
\label{algorithm:sketching:sampling:distance}
\end{algorithm}

In the algorithm, it is understood that $u_i$ and $v_i$ corresponding to
$i \in \mathcal{I}_u \cap \mathcal{I}_v$ are present in $\mathcal{V}_u$
and $\mathcal{V}_v$, respectively. These quantities, along with $d_\circ$
and the norms of the vectors are used to weight each partial inner product.
The final quantity, as we will learn shortly, is an unbiased estimate of
the inner product between $u$ and $v$.

\subsection{Theoretical Analysis}

\begin{theorem}
    Algorithm~\ref{algorithm:sketching:sampling:sketching} produces sketches
    that consist of at most $d_\circ$ coordinates in expectation.
\end{theorem}
\begin{proof}
    The number of sampled coordinates is $\lvert \mathcal{I} \rvert$.
    That quantity can be expressed as follows:
    \begin{equation*}
        \lvert \mathcal{I} \rvert = \sum_{i = 1}^d \mathbbm{1}_{i \in \mathcal{I}}.
    \end{equation*}
    Taking expectation of both sides and using the linearity of expectation,
    we obtain the following:
    \begin{equation*}
        \ev[\lvert \mathcal{I} \rvert] = \sum_i \ev[\mathbbm{1}_{i \in \mathcal{I}}]
        = \sum_i \min(1, d_\circ \frac{u_i^2}{\lVert u \rVert_2^2}) \leq d_\circ.
    \end{equation*}
\end{proof}

\begin{theorem}
    Algorithm~\ref{algorithm:sketching:sampling:distance} yields an unbiased estimate
    of inner product.
\end{theorem}
\begin{proof}
    From the proof of the previous theorem, we know that coordinate $i$ of an arbitrary
    vector $u$ is included in the sketch with probability equal to:
    \begin{equation*}
        \min(1, d_\circ \frac{u_i^2}{\lVert u \rVert_2^2}).
    \end{equation*}
    As such, the odds that $i \in \mathcal{I}_u \cap \mathcal{I}_v$ is:
    \begin{equation*}
        p_i = \min(1, d_\circ \frac{u_i^2}{\lVert u \rVert_2^2}, d_\circ \frac{v_i^2}{\lVert v \rVert_2^2}).
    \end{equation*}
    Algorithm~\ref{algorithm:sketching:sampling:distance} gives us a weighted sum of
    the coordinates that are present in $\mathcal{I}_u \cap \mathcal{I}_v$. We can rewrite that
    sum using indicator functions as follows:
    \begin{equation*}
        \sum_{i = 1}^d \mathbbm{1}_{i \in \mathcal{I}_u \cap \mathcal{I}_v} \frac{u_i v_i}{p_i}.
    \end{equation*}
    In expectation, then:
    \begin{equation*}
        \ev\Big[\sum_{i = 1}^d \mathbbm{1}_{i \in \mathcal{I}_u \cap \mathcal{I}_v} \frac{u_i v_i}{p_i}\Big]
        = \sum_{i = 1}^d p_i \frac{u_i v_i}{p_i} = \langle u, v \rangle,
    \end{equation*}
    as required.
\end{proof}

\begin{theorem}
    \label{theorem:sketching:sampling:variance}
    If $S$ is the output of Algorithm~\ref{algorithm:sketching:sampling:distance}
    for sketches of vectors $u$ and $v$, then:
    \begin{equation*}
        \var[S] \leq \frac{2}{d_\circ} \max \Big( \lVert u_\ast \rVert_2^2 \lVert v \rVert_2^2,
        \lVert u \rVert_2^2 \lVert v_\ast \rVert_2^2 \Big),
    \end{equation*}
    where $u_\ast$ and $v_\ast$ are the vectors $u$ and $v$ restricted to
    the set of non-zero coordinates common to both vectors (i.e., $\ast = \{ i \;|\; u_i \neq 0 \land v_i \neq 0 \}$).
\end{theorem}
\begin{proof}
    We use the same proof strategy as in the previous theorem.
    In particular, we write:
    \begin{align*}
        \var[S] =
        \var\Big[\sum_{i \in \ast} \mathbbm{1}_{i \in \mathcal{I}_u \cap \mathcal{I}_v} \frac{u_i v_i}{p_i}\Big]
        &= \sum_{i \in \ast} \var\big[ \mathbbm{1}_{i \in \mathcal{I}_u \cap \mathcal{I}_v} \frac{u_i v_i}{p_i} \big] \\
        &= \sum_{i \in \ast} \frac{u_i^2 v_i^2}{p_i^2} \var\big[ \mathbbm{1}_{i \in \mathcal{I}_u \cap \mathcal{I}_v} \big].
    \end{align*}
    Turning to the term inside the sum, we obtain:
    \begin{equation*}
        \var\big[ \mathbbm{1}_{i \in \mathcal{I}_u \cap \mathcal{I}_v} \big] = p_i - p_i^2,
    \end{equation*}
    which is $0$ if $p_i = 1$ and less than $p_i$ otherwise.
    Putting everything together, we complete the proof:
    \begin{align*}
        \var[S] &\leq \sum_{i \in \ast,\; p_i \neq 1} \frac{u_i^2 v_i^2}{p_i} =
            \lVert u \rVert_2^2 \lVert v \rVert_2^2
                \sum_{i \in \ast,\; p_i \neq 1} \frac{\big(u_i^2/\lVert u \rVert_2^2 \big) \big( v_i^2/\lVert v \rVert_2^2 \big)}{d_\circ \min\big( u_i^2/\lVert u \rVert_2^2, v_i^2/\lVert v \rVert_2^2 \big)} \\
        &= \frac{\lVert u \rVert_2^2 \lVert v \rVert_2^2}{d_\circ}
                \sum_{i \in \ast,\; p_i \neq 1} \max\big( u_i^2/\lVert u \rVert_2^2, v_i^2/\lVert v \rVert_2^2 \big) \\
        &= \frac{\lVert u \rVert_2^2 \lVert v \rVert_2^2}{d_\circ} \sum_{i \in \ast} \frac{u_i^2}{\lVert u \rVert_2^2} + \frac{v_i^2}{\lVert v \rVert_2^2} \\
        &= \frac{\lVert u \rVert_2^2 \lVert v \rVert_2^2}{d_\circ} \Big( \frac{\lVert u_\ast \rVert_2^2}{\lVert u \rVert_2^2} + \frac{\lVert v_\ast \rVert_2^2}{\lVert v \rVert_2^2} \Big) \\
        &= \frac{1}{d_\circ} \Big( \lVert u_\ast \rVert_2^2 \lVert v \rVert_2^2 + \lVert u \rVert_2^2 \lVert v_\ast \rVert_2^2 \Big) \\
        &\leq \frac{2}{d_\circ} \max \big( \lVert u_\ast \rVert_2^2 \lVert v \rVert_2^2, \lVert u \rVert_2^2 \lVert v_\ast \rVert_2^2 \big).
    \end{align*}
\end{proof}

\begin{svgraybox}
    Theorem~\ref{theorem:sketching:sampling:variance} tells us that, if we estimated $\langle u, v \rangle$
    for two vectors $u$ and $v$ using Algorithm~\ref{algorithm:sketching:sampling:distance},
    then the variance of our estimate will be bounded by factors that depend on the
    non-zero coordinates that $u$ and $v$ have in common. Because $\mathit{nz}(u) \cap \mathit{nz}(v)$
    has at most $d$ entries, estimates of inner product based on Threshold Sampling should generally be more accurate than
    those obtained from JL sketches. This is particularly the case when $u$ and $v$ are sparse.
\end{svgraybox}

\bibliographystyle{abbrvnat}
\bibliography{biblio}

\begin{partbacktext}
\part{Appendices}
\end{partbacktext}

\appendix

\chapter{Collections}
\label{appendix:collections}

\abstract{This appendix gives a description of the vector collections used in experiments
throughout this monograph. These collections demonstrate different operating points in
a typical use-case. For example, some consist of dense vectors, others of sparse vectors;
some have few dimensions and others are in much higher dimensions; some are relatively small
while others contain a large number of points.}

\bigskip

Table~\ref{table:appendix:collections:dense} gives a description of the dense vector collections
used throughout this monograph and summarizes their key statistics.

\begin{table*}[ht]
\caption{Dense collections used in this monograph along with select statistics.}
\scriptsize
\label{table:appendix:collections:dense}
\begin{center}
\begin{sc}
\begin{tabular}{p{5cm}|ccc}
\toprule
Collection & Vector Count & Query Count & Dimensions \\
\midrule
\textsc{GloVe}-$25$~\citep{pennington-etal-2014-glove} & $1.18$M & $10{,}000$ & $25$ \\
\textsc{GloVe}-$50$ & $1.18$M & $10{,}000$ & $50$ \\
\textsc{GloVe}-$100$ & $1.18$M & $10{,}000$ & $100$ \\
\textsc{GloVe}-$200$ & $1.18$M & $10{,}000$ & $200$ \\
\textsc{Deep1b}~\citep{deep1b} & $9.99$M & $10{,}000$ & $96$ \\
\textsc{MS Turing}~\citep{msturingDataset} & $10$M & $100{,}000$ & $100$ \\
\textsc{Sift}~\citep{Lowe2004DistinctiveIF} & $1$M & $10{,}000$ & $128$ \\
\textsc{Gist}~\citep{Oliva2001ModelingTS} & $1$M & $1{,}000$ & $960$ \\
\bottomrule
\end{tabular}
\end{sc}
\end{center}
\end{table*}

In addition to the vector collections above, we convert a few text collections
into vectors using various embedding models. These collections are described in
Table~\ref{table:appendix:collections:text}. Please see~\citep{nguyen2016msmarco} for
a complete description of the MS MARCO v1 collection and~\citep{thakur2021beir} for the others.

\begin{table*}[ht]
\caption{Text collections along with key statistics.
The rightmost two columns report the average number of non-zero
entries in data points and, in parentheses, queries for sparse vector
representations of the collections.}
\scriptsize
\label{table:appendix:collections:text}
\begin{center}
\begin{sc}
\begin{tabular}{c|cc|cc}
\toprule
Collection & Vector Count & Query Count & \splade{} & \esplade{}\\
\midrule
\textsc{MS Marco} Passage& $8.8$M & $6{,}980$ & 127 (49) & 185 (5.9) \\
NQ & $2.68$M & $3{,}452$ & 153 (51) & 212 (8) \\
\textsc{Quora} & $523$K & $10{,}000$ & 68 (65) & 68 (8.9) \\
\textsc{HotpotQA} & $5.23$M & $7{,}405$ & 131 (59) & 125 (13) \\
\textsc{Fever} & $5.42$M & $6{,}666$ & 145 (67) & 140 (8.6) \\
\textsc{DBPedia} & $4.63$M & $400$ & 134 (49) & 131 (5.9) \\
\bottomrule
\end{tabular}
\end{sc}
\end{center}
\end{table*}

When transforming the text collections of Table~\ref{table:appendix:collections:text}
into vectors, we use the following embedding models:
\begin{itemize}
    \item \textsc{AllMiniLM-l6-v2}:\footnote{Available at \url{https://huggingface.co/sentence-transformers/all-MiniLM-L6-v2}}
    Projects text documents into $384$-dimensional dense vectors for retrieval with angular distance.

    \item \textsc{Tas-B}~\citep{tas-b}: A bi-encoder model that was trained using supervision from a cross-encoder and a ColBERT~\citep{colbert2020khattab} model,
    and produces $768$-dimensional dense vectors that are meant for MIPS.
    The checkpoint used in this work is available on HuggingFace.\footnote{Available at \url{https://huggingface.co/sentence-transformers/msmarco-distilbert-base-tas-b}}

    \item \splade{}~\citep{formal2022splade}:\footnote{Pre-trained checkpoint from HuggingFace available at \url{https://huggingface.co/naver/splade-cocondenser-ensembledistil}}
    Produces sparse representations for text.
    The vectors have roughly $30{,}000$ dimensions, where each dimension corresponds
    to a term in the BERT~\citep{devlin2019bert} WordPiece~\citep{wordpiece} vocabulary.
    Non-zero entries in a vector reflect learnt term importance weights.

    \item \esplade{}~\citep{lassance2022sigir}:\footnote{Pre-trained checkpoints for document and
    query encoders were obtained from \url{https://huggingface.co/naver/efficient-splade-V-large-doc} and \url{https://huggingface.co/naver/efficient-splade-V-large-query},
    respectively.}
    This model produces queries that have far fewer non-zero entries than the original
    \splade{} model, but documents that may have a larger number of non-zero entries.
\end{itemize}

\bibliographystyle{abbrvnat}
\bibliography{biblio}

\chapter{Probability Review}
\label{appendix:probability}

\abstract{We briefly review key concepts in probability in this appendix.}

\section{Probability}
We identify a \emph{probability space} denoted by $(\Omega, \mathcal{F}, \probability)$
with an \emph{outcome space}, an \emph{events} set, and a \emph{probability measure}.
The outcome space, $\Omega$, is the set of all
possible outcomes. For example, when flipping a two-sided coin, the outcome
space is simply $\{0, 1\}$. When rolling a six-sided die, it is instead
the set $[6] = \{ 1, 2, \ldots, 6\}$.

The events set $\mathcal{F}$ is a set of subsets of $\Omega$ that
includes $\Omega$ as a member and is closed under complementation and
countable unions. That is, if $E \in \mathcal{F}$,
then we must have that $E^\complement \mathcal{F}$.
Furthermore, the union of countably many events $E_i$'s
in $\mathcal{F}$ is itself in $\mathcal{F}$: $\cup_i E_i \in \mathcal{F}$.
A set $\mathcal{F}$ that satisfies these properties is called a $\sigma$-algebra.

Finally, a function $\probability: \mathcal{F} \rightarrow \mathbb{R}$ is
a probability measure if it satisfies the following conditions: $\probability[\Omega] = 1$;
$\probability[E] \geq 0$ for any event $E \in \mathcal{F}$;
$\probability[E^\complement] = 1 - \probability[E]$; and, finally,
for countably many disjoint events $E_i$'s:
$\probability[\cup_i E_i] = \sum_i \probability[E_i]$.

We should note that, $\probability$ is also known as a ``probability distribution''
or simply a ``distribution.'' The pair $(\Omega, \mathcal{F})$ is called
a \emph{measurable space}, and the elements of $\mathcal{F}$ are
known as a \emph{measurable sets}. The reason they are called ``measurable''
is because they can be ``measured'' with $\probability$: The function
$\probability$ assigns values to them.

In many of the discussions throughout this monograph, we omit the outcome space
and events set because that information is generally clear from context.
However, a more formal treatment of our arguments requires a complete
definition of the probability space.

\section{Random Variables}
A random variable on a measurable space $(\Omega, \mathcal{F})$ is
a measurable function $X: \Omega \rightarrow \mathbb{R}$.
It is measurable in the sense that the \emph{preimage} of any Borel set $B \in \mathcal{B}$
is an event: $X^{-1}(B) = \{ \omega \in \Omega \;|\; X(\omega) \in B \} \in \mathcal{F}$.

A random variable $X$ generates a $\sigma$-algebra that comprises of the preimage
of all Borel sets. It is denoted by $\sigma(X)$
and formally defined as $\sigma(X) = \{ X^{-1}(B) \;|\; B \in \mathcal{B} \}$.

\bigskip

Random variables are typically categorized as discrete or continuous.
$X$ is \emph{discrete} when it maps $\Omega$ to a discrete set.
In that case, its \emph{probability mass function} is defined as $\probability[X = x]$
for some $x$ in its range.
A \emph{continuous} random variable is often associated with a
probability \emph{density} function, $f_X$, such that:
\begin{equation*}
    \probability[a \leq X \leq b] = \int_a^b f_X(x) dx.
\end{equation*}

Consider, for instance, the following probability density function over the real line for
parameters $\mu \in \mathbb{R}$ and $\sigma > 0$:
\begin{equation*}
    f(x) = \frac{1}{\sqrt{2 \pi \sigma^2}} e^{- \frac{(x - \mu)^2}{2\sigma^2}}.
\end{equation*}
A random variable with the density function above is said to follow a Gaussian
distribution with mean $\mu$ and variance $\sigma^2$, denoted by $X \sim \mathcal{N}(\mu, \sigma^2)$.
When $\mu = 0$ and $\sigma^2 = 1$, the resulting distribution is called the standard
Normal distribution.

Gaussian random variables have attractive properties.
For example, the sum of two independent Gaussian random variables is itself a Gaussian variable.
Concretely, $X_1 \sim \mathcal{N}(\mu_1, \sigma_1^2)$ and $X_2 \sim \mathcal{N}(\mu_2, \sigma_2^2)$,
then $X_1 + X_2 \sim \mathcal{N}(\mu_1 + \mu_2, \sigma_1^2 + \sigma_2^2)$.
The sum of the squares of $m$ independent Gaussian random variables, on the other hand,
follows a $\chi^2$-distribution with $m$ degrees of freedom.

\section{Conditional Probability}
Conditional probabilities give us a way to model how the probability of an event changes
in the presence of extra information, such as partial knowledge about a random outcome.
Concretely, if $(\Omega, \mathcal{F}, \probability)$ is a probability space and
$A, B \in \mathcal{F}$ such that $\probability[B] > 0$, then the \emph{conditional
probability} of $A$ given the event $B$ is denoted by $\probability[A \;\lvert\; B]$ and
defined as follows:
\begin{equation*}
    \probability[A \;\lvert\; B] = \frac{\probability[A \cap B]}{\probability[B]}.
\end{equation*}

We use a number of helpful results concerning conditional probabilities
in proofs throughout the monograph. One particularly useful inequality
is what is known as the \emph{union bound} and is stated as follows:
\begin{equation*}
    \probability[\cup_i A_i] \leq \sum_i \probability[A_i].
\end{equation*}

Another fundamental property is the law of total probability.
It states that, for mutually disjoint events $A_i$'s such that
$\Omega = \cup A_i$, the probability of any event $B$ can be expanded
as follows:
\begin{equation*}
    \probability[B] = \sum_i \probability[B \;\lvert\; A_i] \probability[A_i].
\end{equation*}
This is easy to verify: the summand is by definition equal to $\probability[B \cap A_i]$
and, considering the events $(B \cap A_i)$'s are mutually disjoint, their sum
is equal to $\probability[B \cap (\cup A_i)] = \probability[B]$.

\section{Independence}
Another tool that reflects the effect (or lack thereof) of additional knowledge on probabilities
is the concept of \emph{independence}. Two events $A$ and $B$ are said to be
\emph{independent} if $\probability[A \cap B] = \probability[A] \times \probability[B]$.
Equivalently, we say that $A$ is independent of $B$ if and only if
$\probability[A \;\lvert\; B] = \probability[A]$ when $\probability[B] > 0$.

\bigskip

Independence between two random variables is defined similarly but requires a bit more care.
If $X$ and $Y$ are two random variables and $\sigma(X)$ and $\sigma(Y)$ denote
the $\sigma$-algebras generated by them, then $X$ is independent of $Y$ if
all events $A \in \sigma(X)$ and $B \in \sigma(Y)$ are independent.

When a sequence of random variables are \emph{mutually} independent and are drawn
from the same distribution (i.e., have the same probability density function),
we say the random variables are drawn \emph{iid}: independent and identically-distributed.
We stress that \emph{mutual} independence is a stronger restriction than
\emph{pairwise} independence: $m$ events $\{ E_i \}_{i=1}^m$ are mutually independent if
$\probability[\cap_i E_i] = \prod_i \probability[E_i]$.

We typically assume that data and query points are drawn \emph{iid} from some
(unknown) distribution. This is a standard and often necessary assumption
that eases analysis.

\section{Expectation and Variance}

The \emph{expected value} of a discrete random variable $X$ is denoted by $\ev[X]$
and defined as follows:
\begin{equation*}
    \ev[X] = \sum_x x \probability[X = x].
\end{equation*}
When $X$ is continuous, its expected value is based on the following Lebesgue integral:
\begin{equation*}
    \ev[X] = \int_{\Omega} X d \probability.
\end{equation*}
So when a random variable has probability density function $f_X$, its expected value
becomes:
\begin{equation*}
    \ev[X] = \int x f_X(x) dx.
\end{equation*}

For a \emph{nonnegative} random variable $X$, it is sometimes more convenient to
unpack $\ev{X}$ as follows instead:
\begin{equation*}
    \ev[X] = \int_0^\infty \probability[X > x] dx.
\end{equation*}

A fundamental property of expectation is that it is a linear operator.
Formally, $\ev[X + Y] = \ev[X] + \ev[Y]$ for two random variables $X$ and $Y$.
We use this property often in proofs.

We state another important property for independent random variables
that is easy to prove.
If $X$ and $Y$ are independent, then $\ev[XY] = \ev[X]\ev[Y]$.

\bigskip

The \emph{variance} of a random variable is defined as follows:
\begin{equation*}
    \var[X] = \ev\Big[ (X - \ev[X])^2 \Big] = \ev[X]^2 - \ev[X^2].
\end{equation*}
Unlike expectation, variance is not linear unless the random variables involved
are independent. It is also easy to see that $\var[aX] = a^2 \var[X]$ for a
constant $a$.

\section{Central Limit Theorem}
The result known as the Central Limit Theorem is one of the most
useful tools in probability. Informally, it states that the average of \emph{iid}
random variables with finite mean and variance converges to a Gaussian distribution.
There are several variants of this result that extend the claim to, for example,
independent but not identically distributed variables. Below we repeat the formal
result for the \emph{iid} case.

\begin{theorem}
    Let $X_i$'s be a sequence of $n$ \emph{iid} random variables with finite mean $\mu$
    and variance $\sigma^2$. Then, for any $x \in \mathbb{R}$:
    \begin{equation*}
        \lim_{n \rightarrow \infty} \probability \Big[
            \underbrace{\frac{(1/n \sum_{i=1}^n X_i) - \mu}{\sigma^2/n}}_Z \leq x
        \Big] = \int_{-\infty}^x \frac{1}{\sqrt{2 \pi}} e^{-\frac{t^2}{2}} dt,
    \end{equation*}
    implying that $Z \sim \mathcal{N}(0, 1)$.
\end{theorem}

\chapter{Concentration of Measure}
\label{appendix:measure}

\abstract{
By the strong law of large numbers, we know that the average of a sequence
of $m$ \emph{iid} random variables with mean $\mu$ converges to $\mu$ with
probability $1$ as $m$ tends to infinity. But how far is that average from
$\mu$ when $m$ is finite? Concentration inequalities helps us answer that question
quantitatively. This appendix reviews important inequalities that are used
in the proofs and arguments throughout this monograph.
}

\section{Markov's Inequality}

\begin{lemma}
    \label{lemma:appendix:concentration:markov}
    For a nonnegative random variable $X$ and a nonnegative constant $a \geq 0$:
    \begin{equation*}
        \probability[X \geq a] \leq \frac{\ev[X]}{a}.
    \end{equation*}
\end{lemma}
\begin{proof}
    Recall that the expectation of a nonnegative random variable $X$ can be written
    as:
    \begin{equation*}
        \ev[X] = \int_0^\infty \probability[X \geq x] dx.
    \end{equation*}
    Because $\probability[X \geq x]$ is monotonically nonincreasing, we can expand
    the above as follows to complete the proof:
    \begin{equation*}
        \ev[X] \geq \int_0^a \probability[X \geq x] dx \geq \int_0^a \probability[X \geq a] dx = a \probability[X \geq a].
    \end{equation*}
\end{proof}

\section{Chebyshev's Inequality}

\begin{lemma}
    \label{lemma:appendix:concentration:chebyshev}
    For a random variable $X$ and a constant $a > 0$:
    \begin{equation*}
        \probability \Big[ \big\lvert X - \ev[X] \big\rvert \geq a \Big] \leq \frac{\var[X]}{a^2}.
    \end{equation*}
\end{lemma}
\begin{proof}
    \begin{equation*}
        \probability \Big[ \big\lvert X - \ev[X] \big\rvert \geq a \Big] =
        \probability \Big[ \big( X - \ev[X] \big)^2 \geq a^2 \Big] \leq \frac{\var[X]}{a^2},
    \end{equation*}
    where the last step follows by the application of Markov's inequality.
\end{proof}

\begin{lemma}
    Let $\{ X_i \}_{i=1}^n$ be a sequence of iid random variables
    with mean $\mu < \infty$ and variance $\sigma^2 < \infty$. For $\delta \in (0, 1)$,
    with probability $1 - \delta$:
    \begin{equation*}
        \Big\lvert \frac{1}{n} \sum_{i = 1}^n X_i - \mu \Big\rvert \leq \sqrt{\frac{\sigma^2}{\delta n}}.
    \end{equation*}
\end{lemma}
\begin{proof}
    By Lemma~\ref{lemma:appendix:concentration:chebyshev}, for any $a > 0$:
    \begin{equation*}
        \probability \Bigg[ \Big\lvert \frac{1}{n}\sum_{i=1}^n X_i - \mu \Big\rvert \geq a \Bigg]
        \leq \frac{\sigma^2/n}{a^2}.
    \end{equation*}
    Setting the right-hand-side to $\delta$, we obtain:
    \begin{equation*}
        \frac{\sigma^2}{n a^2} = \delta \implies a = \sqrt{\frac{\sigma^2}{\delta n}},
    \end{equation*}
    which completes the proof.
\end{proof}

\section{Chernoff Bounds}

\begin{lemma}
    Let $\{ X_i \}_{i=1}^n$ be independent Bernoulli variables with success probability $p_i$.
    Define $X = \sum_i X_i$ and $\mu = \ev[X] = \sum_i p_i$. Then:
    \begin{equation*}
        \probability \Big[ X > (1 + \delta) \mu \Big] \leq e^{-h(\delta) \mu},
    \end{equation*}
    where,
    \begin{equation*}
        h(t) = (1 + t) \log(1 + t) - t.
    \end{equation*}
\end{lemma}
\begin{proof}
    Using Markov's inequality of Lemma~\ref{lemma:appendix:concentration:markov}
    we can write the following for any $t > 0$:
    \begin{equation*}
        \probability\Big[ X > (1 + \delta)\mu \Big] =
            \probability\Big[ e^{tX} > e^{t(1 + \delta)\mu} \Big] \leq
            \frac{\ev\big[ e^{tX} \big]}{e^{t (1 + \delta) \mu}}.
    \end{equation*}
    Expanding the expectation, we obtain:
    \begin{align*}
        \ev\big[e^{tX}\big] &= \ev\Big[ e^{t \sum_i X_i} \Big] = \ev\Big[ \prod_i e^{tX_i} \Big]
        = \prod_i \ev[e^{tX_i}] \\
        &= \prod_i \Big( p_i e^t + (1 - p_i) \Big) \\
        &= \prod_i \big( 1 + p_i (e^t - 1) \big) \\
        &\leq \prod_i e^{p_i(e^t - 1)} = e^{(e^t - 1)\mu}. && \text{by $(1 + t \leq e^t)$} \\
    \end{align*}
    Putting all this together gives us:
    \begin{equation}
        \label{equation:appendix:concentration:chernoff:proof}
        \probability\Big[ X > (1 + \delta)\mu \Big] \leq 
        \frac{e^{(e^t - 1) \mu}}{e^{t (1 + \delta) \mu}}.
    \end{equation}
    This bound holds for any value $t > 0$, and in particular a value of $t$ that
    minimizes the right-hand-side. To find such a $t$, we may differentiate
    the right-hand-side, set it to $0$, and solve for $t$ to obtain:
    \begin{align*}
        \frac{\mu e^t e^{(e^t - 1) \mu}}{e^{t (1 + \delta) \mu}} &-
        \mu ( 1 + \delta ) \frac{e^{(e^t - 1) \mu}}{e^{t (1 + \delta) \mu}} = 0 \\
        &\implies \mu e^t = \mu (1 + \delta) \\
        &\implies t = \log(1 + \delta).
    \end{align*}
    Substituting $t$ into Equation~(\ref{equation:appendix:concentration:chernoff:proof})
    gives the desired result.
\end{proof}

\section{Hoeffding's Inequality}

We need the following result, known as Hoeffding's Lemma, to present
Hoeffding's inequality.

\begin{lemma}
    \label{lemma:appendix:concentration:hoeffding-lemma}
    Let $X$ be a zero-mean random variable that takes values in $[a, b]$.
    For any $t > 0$:
    \begin{equation*}
        \ev\big[ e^{tX} \big] \leq \exp\Big( \frac{t^2 (b - a)^2}{8} \Big).
    \end{equation*}
\end{lemma}
\begin{proof}
    By convexity of $e^{tx}$ and given $x \in [a, b]$ we have that:
    \begin{equation*}
        e^{tx} \leq \frac{b - x}{b - a} e^{ta} +
            \frac{x - a}{b - a} e^{tb}.
    \end{equation*}
    Taking the expectation of both sides, we arrive at:
    \begin{equation*}
        \ev\Big[e^{tx}\Big] \leq
            \frac{b}{b - a} e^{ta} - \frac{a}{b - a} e^{tb}.
    \end{equation*}
    To conclude the proof, we first write the right-hand-side as
    $\exp(h(t(b - a)))$ where:
    \begin{equation*}
        h(x) = \frac{a}{b - a} x + \log \Big( \frac{b}{b - a} - \frac{a}{b - a} e^x \Big).
    \end{equation*}
    By expanding $h(x)$ using Taylor's theorem, it can be shown that
    $h(x) \leq x^2/8$. That completes the proof.
\end{proof}

We are ready to present Hoeffding's inequality.

\begin{lemma}
    Let $\{ X_i \}_{i=1}^n$ be a sequence of iid random variables
    with finite mean $\mu$ and suppose $X_i \in [a, b]$ almost surely.
    For all $\epsilon > 0$:
    \begin{equation*}
        \probability\Bigg[ \Big\lvert \frac{1}{n} \sum_{i=1}^n X_i - \mu \Big\rvert > \epsilon \Bigg] \leq 2 \exp\Big({-\frac{2n \epsilon^2}{(b - a)^2}}\Big).
    \end{equation*}
\end{lemma}
\begin{proof}
    Let $X = 1/n \sum_i X_i - \mu$. Observe by Markov's inequality that:
    \begin{equation*}
        \probability[X \geq \epsilon] = \probability\Big[ e^{tX} \geq e^{t\epsilon} \Big]
        \leq e^{-t\epsilon} \ev[e^{tX}].
    \end{equation*}
    By independence of $X_i$'s and
    the application of Lemma~\ref{lemma:appendix:concentration:hoeffding-lemma}:
    \begin{align*}
        \ev[e^{tX}] &= \ev \Bigg[ \prod_i e^\frac{t(X_i - \mu)}{n} \Bigg] \\
        &= \prod_i \ev \Big[ e^{\frac{t(X_i-\mu)}{n}} \Big] \\
        &\leq \prod_i \exp\Big( \frac{t^2 (b - a)^2}{8 n^2} \Big) \\
        &= \exp\Big( \frac{t^2 (b - a)^2}{8 n} \Big).
    \end{align*}
    We have shown that:
    \begin{equation*}
        \probability[X \geq \epsilon] \leq \exp\Big( -t \epsilon + \frac{t^2 (b - a)^2}{8 n} \Big).
    \end{equation*}
    That statement holds for all values of $t$ and in particular one that minimizes
    the right-hand-side. Solving for that value of $t$ gives us
    $t = 4n\epsilon / (b - a^2)$, which implies:
    \begin{equation*}
        \probability[X \geq \epsilon] \leq e^{-\frac{2n \epsilon^2}{(b - a)^2}}.
    \end{equation*}
    By a symmetric argument we can bound $\probability[X \leq -\epsilon]$. The claim
    follows by the union bound over the two cases.
\end{proof}

\section{Bennet's Inequality}

\begin{lemma}
    Let $\{ X_i \}_{i=1}^n$ be a sequence of independent random variables with zero mean
    and finite variance $\sigma_i^2$. Assume that $\lvert X_i \rvert \leq a$ almost surely for all $i$. Then:
    \begin{equation*}
        \probability\Big[\sum_i X_i \geq t \Big] \leq 
        \exp \Bigg( -\frac{\sigma^2}{a^2} h\Big( \frac{a t}{\sigma^2} \Big) \Bigg),
    \end{equation*}
    where $h(x) = (1 + x) \log(1 + x) - x$ and $\sigma^2 = \sum_i \sigma_i^2$.
\end{lemma}
\begin{proof}
    As usual, we take advantage of Markov's inequality to write:
    \begin{align*}
        \probability\Big[\sum_i X_i \geq t \Big] &\leq
            e^{-\lambda t} \ev \Big[ e^{\lambda \sum_i X_i} \Big] \\
        &= e^{-\lambda t} \ev \Big[ \prod_i e^{\lambda X_i} \Big] \\
        &= e^{-\lambda t} \prod_i \ev \Big[ e^{\lambda X_i} \Big] \\
    \end{align*}
    Using the Taylor expansion of $e^x$, we obtain:
    \begin{align*}
        \ev \Big[ e^{\lambda X_i} \Big] &= \ev \Big[ \sum_{k=0}^\infty \frac{\lambda^k X_i^k}{k!} \Big] \\
        &= 1 + \sum_{k=2}^\infty \frac{\lambda^k \ev[X_i^2 X_i^{k - 2}]}{k!} \\
        &\leq 1 + \sum_{k=2}^\infty \frac{\lambda^k \sigma_i^2 a^{k-2}}{k!} \\
        &= 1 + \frac{\sigma_i^2}{a^2} \sum_{k=2}^\infty \frac{\lambda^k a^k}{k!} \\
        &= 1 + \frac{\sigma_i^2}{a^2} \big( e^{\lambda a} - 1 - \lambda a \big) \\
        &\leq \exp\Big( \frac{\sigma_i^2}{a^2} \big( e^{\lambda a} - 1 - \lambda a \big) \Big).
    \end{align*}
    Putting it all together:
    \begin{align*}
        \probability\Big[\sum_i X_i \geq t \Big] &\leq
            e^{-\lambda t} \prod_i \exp\Big( \frac{\sigma_i^2}{a^2} \big( e^{\lambda a} - 1 - \lambda a \big) \Big) \\
        &= e^{-\lambda t} \exp\Big( \frac{\sigma^2}{a^2} \big( e^{\lambda a} - 1 - \lambda a \big) \Big).
    \end{align*}
    This inequality holds for all values of $\lambda$, and in particular one that minimizes the
    right-hand-side. Setting the derivative of the right-hand-side to $0$ and solving for $\lambda$
    leads to the desired result.
\end{proof}

\chapter{Linear Algebra Review}
\label{appendix:linear-algebra}

\abstract{
This appendix reviews basic concepts from Linear Algebra that are useful
in digesting the material in this monograph.
}

\section{Inner Product}

Denote by $\mathbb{H}$ a vector space.
An inner product $\langle \cdot, \cdot \rangle: \mathbb{H} \times \mathbb{H} \rightarrow \mathbb{R}$
is a function with the following properties:
\begin{itemize}
    \item $\forall \; u \in \mathbb{H},\; \langle u, u \rangle \geq 0$;
    \item $\forall \; u \in \mathbb{H},\; \langle u, u \rangle = 0 \Leftrightarrow u = 0$;
    \item $\forall \; u, v \in \mathbb{H},\; \langle u, v \rangle = \langle v, u \rangle$; and,
    \item $\forall \; u, v, w \in \mathbb{H}, \textit{ and } \alpha, \beta \in \mathbb{R},\; 
    \langle \alpha u + \beta v, w \rangle = \alpha \langle u, w \rangle + \beta \langle v, w \rangle$.
\end{itemize}

We call $\mathbb{H}$ together with the inner product $\langle \cdot, \cdot \rangle$
an \emph{inner product space}.
As an example, when $\mathbb{H} = \mathbb{R}^d$, given two vectors
$u = \sum_{i=1}^d u_i e_i$ and $v = \sum_{i=1}^d v_i e_i$, where $e_i$'s
are the standard basis vectors, the following is an inner product:
\begin{equation*}
    \langle u, v \rangle = \sum_{i = 1}^d u_i v_i.
\end{equation*}

We say two vectors $u$ and $v$ in an inner product space are \emph{orthogonal}
if their inner product is $0$: $\langle u, v \rangle = 0$.

\section{Norms}

A function $\Phi: \mathbb{H} \rightarrow \mathbb{R}_+$ is a norm on
$\mathbb{H}$ if it has the following properties:
\begin{itemize}
    \item Definiteness: For all $u \in \mathbb{H}$, $\Phi(u) = 0 \Leftrightarrow u = 0$;
    \item Homogeneity: For all $u \in \mathbb{H}$ and $\alpha \in \mathbb{R}$,
        $\Phi(\alpha u) = \lvert \alpha \rvert \Phi(u)$; and,
    \item Triangle inequality: $\forall \; u, v \in \mathbb{H}, \; \Phi(u + v) \leq \Phi(u) + \Phi(v)$.
\end{itemize}

Examples include the absolute value on $\mathbb{R}$,
and the $L_p$ norm (for $p \geq 1$) on $\mathbb{R}^d$ denoted by $\lVert \cdot \rVert_p$
and defined as:
\begin{equation*}
    \lVert u \rVert_p = \Big( \sum_{i=1}^d \lvert u_i \rvert^p \Big)^{\frac{1}{p}}.
\end{equation*}
Instances of $L_p$ include the commonly used $L_1$, $L_2$ (Euclidean),
and $L_\infty$ norms, where $\lVert u \rVert_\infty = \max_i \lvert u_i \rvert$.

Note that, when $\mathbb{H}$ is an inner product space, then
the function $\lVert u \rVert = \sqrt{\langle u, u \rangle}$ is a norm.

\section{Distance}
A norm on a vector space induces a notion of distance between two vectors.
Concretely, if $\mathbb{H}$ is a normed space equipped with $\lVert \cdot \rVert$,
then we define the distance between two vectors $u, v \in \mathbb{H}$ as follows:
\begin{equation*}
    \delta(u, v) = \lVert u - v \rVert.
\end{equation*}

\section{Orthogonal Projection}

\begin{lemma}
    Let $\mathbb{H}$ be an inner product space and suppose $u \in \mathbb{H}$ and $u \neq 0$.
    Any vector $v \in \mathbb{H}$ can be uniquely decomposed along $u$ as:
    \begin{equation*}
        v = v_{\perp} + v_{\parallel},
    \end{equation*}
    such that $\langle v_\perp, v_\parallel \rangle = 0$. Additionally:
    \begin{equation*}
        v_\parallel = \frac{\langle u, v \rangle}{\langle u, u \rangle} u,
    \end{equation*}
    and $v_\perp = v - v_\parallel$.
\end{lemma}
\begin{proof}
    Let $v_\parallel = \alpha u$ and $v_\perp = v - v_\parallel$.
    Because $v_\parallel$ and $v_\perp$ are orthogonal, we deduce that:
    \begin{align*}
        \langle v_\parallel, v_\perp \rangle = 0 \implies
            \langle \alpha u, v_\perp \rangle = 0 \implies
            \langle u, v_\perp \rangle = 0.
    \end{align*}
    That implies:
    \begin{align*}
        \langle v, u \rangle = \alpha \langle u, u \rangle \implies
        \alpha = \frac{\langle u, v \rangle}{\langle u, u \rangle},
    \end{align*}
    so that:
    \begin{equation*}
        v_\parallel = \frac{\langle u, v \rangle}{\langle u, u \rangle} u.
    \end{equation*}

    We prove the uniqueness of the decomposition by contradiction.
    Suppose there exists another decomposition of $v$ to $v_\parallel^\prime + v_\perp^\prime$.
    Then:
    \begin{align*}
        v_\parallel + v_\perp = v_\parallel^\prime + v_\perp^\prime &\implies
        \langle u, v_\parallel + v_\perp \rangle = \langle u,  v_\parallel^\prime + v_\perp^\prime\rangle \\
        &\implies \langle u, v_\parallel \rangle = \langle u,  v_\parallel^\prime \rangle \\
        &\implies \langle u, \alpha u \rangle = \langle u, \beta u \rangle \\
        &\implies \alpha = \beta.
    \end{align*}
    We must therefore also have that $v_\perp = v_\perp^\prime$.
\end{proof}

\printindex

\end{document}